\documentclass[a4paper,fleqn]{cas-sc}

\usepackage[authoryear,longnamesfirst]{natbib}

\def\tsc#1{\csdef{#1}{\textsc{\lowercase{#1}}\xspace}}
\tsc{WGM}
\tsc{QE}
\tsc{EP}
\tsc{PMS}
\tsc{BEC}
\tsc{DE}

\usepackage{graphicx}
\usepackage{caption}
\usepackage{subcaption}
\usepackage{xcolor}
\usepackage{bm}

\makeatletter
\def\namedlabel#1#2{\begingroup
   \def\@currentlabel{#2}%
   \phantomsection\label{#1}\endgroup
}
\makeatother

\usepackage{soul}
\usepackage{makecell}
\usepackage{tabularx}
\usepackage{mathtools,amsfonts,amsmath}

\usepackage{booktabs}

\definecolor{ao(english)}{rgb}{0.0, 0.5, 0.0}
\definecolor{americanrose}{rgb}{1.0, 0.01, 0.24}
\definecolor{cerisepink}{rgb}{0.93, 0.23, 0.51}
\definecolor{darkorchid}{rgb}{0.6, 0.2, 0.8}
\definecolor{applegreen}{rgb}{0.55, 0.71, 0.0}
\definecolor{brightpink}{rgb}{1.0, 0.0, 0.5}
\definecolor{azure(colorwheel)}{rgb}{0.0, 0.5, 1.0}
\definecolor{blue-violet}{rgb}{0.54, 0.17, 0.89}
\definecolor{deepmagenta}{rgb}{0.8, 0.0, 0.8}
\definecolor{fashionfuchsia}{rgb}{0.96, 0.0, 0.63}
\definecolor{armygreen}{rgb}{0.29, 0.33, 0.13}
\definecolor{dogwoodrose}{rgb}{0.84, 0.09, 0.41}
\definecolor{bole}{rgb}{0.47, 0.27, 0.23}
\definecolor{coolgrey}{rgb}{0.55, 0.57, 0.67}

\usepackage{makecell}
\usepackage{tabularx}

\usepackage{amsthm}
\usepackage{amsmath}
\DeclareMathOperator*{\argmax}{arg\,max}
\DeclareMathOperator*{\argmin}{arg\,min}

\newtheorem{condition}{Condition}
\newtheorem{assumption}{Assumption}
\newtheorem{definition}{Definition}
\newtheorem{remark}{Remark}
\newtheorem{proposition}{Proposition}
\newtheorem{corollary}{Corollary}
\newtheorem{lemma}{Lemma}
\newtheorem{theorem}{Theorem}

\ExplSyntaxOn
\cs_gset:Npn \__first_footerline:
  { \group_begin: \small \sffamily \__short_authors: \group_end: }
\ExplSyntaxOff

\begin{document}
\let\WriteBookmarks\relax
\def\floatpagepagefraction{1}
\def\textpagefraction{.001}

\shorttitle{Dynamic Adverse Selection with Off-Menu Actions}

\shortauthors{Tao Zhang and Quanyan Zhu}

\title [mode = title]{A Dynamic Adverse Selection Multiagent Model with Off-Menu Actions}                      
\tnotemark[1,2]

\author{Tao Zhang}
\ead{tz636@nyu.edu}
\cormark[1]

\author{Quanyan Zhu}
\ead{qz494@nyu.edu}

\cortext[cor1]{Corresponding author}

\affiliation{organization={Electrical and Computer Engineering, New York University},
            addressline={}, 
            city={Brooklym},
            postcode={11201}, 
            state={NY},
            country={US}}

\begin{abstract}
This paper addresses a notable gap in dynamic mechanism design literature by proposing a framework that accounts for agents' strategic, dynamic participation. It augments the classic principal-multiagent model by integrating agents' coupled decisions on participation and action. The principal, grappling with adverse selection, designs a mechanism with a task policy profile, a coupling policy profile, and an off-switch function profile. We introduce the payoff-flow conservation principle to ensure dynamic incentive compatibility and present a unique process, termed persistence transformation, to derive a closed-form off-switch function. This guarantees sufficient incentive compatibility for agents' coupled decisions, aligning them with principal's objectives. Additionally, our work extends the traditional envelope theorem by introducing a necessary condition for incentive compatibility that leverages principal-desired actions' coupled optimality. This method allows explicit formulation of both the coupling and off-switch functions. We also provide envelope-like conditions solely on task policies, facilitating the application of the first-order approach.
\end{abstract}





\maketitle


\section{Introduction}

Distributed dynamic decision-making is prevalent in today's interconnected socio-technical systems, such as communication networks (\citet{jiang2021opinion,du2022interference}), transportation and energy systems (\citet{levering2022framework,zhang2018dynamic,aid2022optimal}), cyber-security systems (\citet{etesami2019dynamic,chen2021dynamic}), and electronic commerce platforms and marketplaces (\citet{shirata2020evolution,chenavaz2021dynamic}).
Interactions among people and technological infrastructures within these distributed dynamic systems yield significant economic and social outcomes.
Typically modeled by dynamic games, these interactions often involve self-interested entities possessing time-evolving local perceptions and choosing individual actions over time.
In the inverse problems of dynamic games, a central authority, or principal, aims to achieve optimal system performance (e.g., maximized social welfare or revenue) by designing mechanisms (e.g., game rules or incentive policies) that effectively coordinate self-interested agents, incentivizing them to choose actions that result in the desired system performance.
Such incentive design typically intervenes in the utility functions by determining the payoff-relevant consequences of agents' actions, either through monetary incentives or by designing utility functions as a general paradigm.
Information asymmetries between the principal and the agents, which arise due to limited communication capabilities or individual's local informational advantages/disadvantages, present challenges in designing incentive mechanisms, leading to inefficient outcomes.
Generally, there are two forms of information asymmetries: \textit{adverse selection} and \textit{moral hazard} (\citet{grossman1992analysis,einav2013selection}).
Adverse selection occurs when agents possess private payoff-relevant information that the principal is unaware of, while their realized actions are commonly observable. 
On the other hand, moral hazard arises when agents' actions are unobserved by the principal, but their states are common information.

%
In dynamic mechanisms with adverse selection, agents are typically assumed to be far-sighted, seeking to optimize their expected cumulative utilities by dynamically choosing optimal single-period actions and adaptively planning future plays as their private information evolves.
The central problem addressed by classic dynamic mechanism design literature is incentivizing agents to truthfully report their sequentially-arriving private states, which is ensured by imposing a sequence of interim incentive compatibility constraints over time (\citet{bergemann2010dynamic,kakade2013optimal,pavan2014dynamic,zhang2022incentive}).
Surprisingly, dynamic voluntary participation constraints have received less attention, as most works focus on ex-ante or initial-period interim participation constraints, assuming agents cannot quit after participating at the beginning.
If agents accept the mechanism but can terminate participation later, then participation constraints must be placed in all periods.
However, the literature on dynamic participation decisions (\citet{bergemann2010dynamic,kakade2013optimal,bergemann2020scope,krahmer2015optimal,heumann2020information,chang2021optimal,ashlagi2023sequential,ashlagi2016sequential}) mainly considers a myopic setting, with agents not strategically planning their future participation decisions.
That is, when agents decide whether to participate in each period, they look into the entire future time horizon with planned action strategies (e.g., reporting strategies), assuming they will not quit in the future.

Voluntary participation is typically secured by imposing either interim or ex-post individual rationality (IR) constraints.
Most dynamic mechanism models typically place one-shot interim IR constraints in the initial round (\citet{pavan2017dynamic,bergemann2019dynamic}) after each agent observes his initial private information, which requires the expected payoff from current and all future rounds to be non-negative.
In mechanisms with periodic interim IR constraints (\citet{mirrokni2016dynamic}), it requires that at the beginning of each round, every agent's expected payoff from current and all future rounds is non-negative given the agent's local private information.
On the contrary, most research on ex-post voluntary participation focuses on periodic ex-post IR constraints (\citet{bergemann2010dynamic,cavallo2012optimal,kakade2013optimal,ashlagi2016sequential,zhang2022incentive}) that are imposed at the end of each round or one-shot ex-post IR constraints (\citet{krahmer2015optimal,bergemann2020scope,heumann2020information,ashlagi2023sequential}) that are placed at the end of the mechanism.
In the case of periodic ex-post IR, every agent's expected payoff from all future rounds given all agents' contemporaneous private information (and history) is non-negative (\citet{kakade2013optimal}), while in mechanisms with one-shot ex-post IR, every agent's realized payoff is non-negative (\citet{ashlagi2023sequential}).
A stronger form of periodic ex-post IR is linked to the concept of (per-round) limited liability, as discussed in works such as \citet{krahmer2015optimal,ashlagi2016sequential,braverman2021prior}.
Limited liability mandates that the single-stage payoff at each round must be non-negative, which guarantees that the payoff received by every agent from their participation in the mechanism up until the end of each round is non-negative. 

In this paper, we study a general dynamic mechanism design problem with adverse selection, in the context of task delegations (\citet{alonso2008optimal,chen2022dynamic,ambrus2017delegation}), for a finite group of risk-neutral agents over a finite horizon.
Our model integrates time-varying private information with Markov dynamics and strategic participation to offer novel insights into the structure of optimal mechanisms when the agents possess the dynamic autonomy to disengage during an ongoing long-term game.
We move beyond the classic myopic participation decision-making to consider a far-sighted setting where each agent's current-period decisions take into account planned future participation decisions and regular actions-taking.
We refer to such strategic voluntary participation decisions as \textit{off-menu actions} $\{\mathtt{OM}_{i,t}\in\{0,1\}\}$ (i.e., participate if $\mathtt{OM}_{i,t}=0$; or leave if $\mathtt{OM}_{i,t}=1$), explored as the counterpart to the agents' \textit{on-menu actions}, i.e., the regular actions that are listed in the menu of available actions. 
The primary objective is to examine the conditions under which a principal can successfully design a dynamic game with adverse selection, ensuring that agents' equilibrium behaviors—encompassing both off-menu and regular actions—align with the principal's predictions and intentions.

Our study of voluntary participation presents fundamental differences from the mainstream literature.
Specifically, we incorporate a strategic context by considering agents' periodic interim decisions to participate as off-menu actions, in addition to the regular actions (i.e., reporting the private information and then receiving an allocation in classic mechanisms), that are adaptively chosen and planned over time.
Additionally, our model not only studies the dynamic participation guarantee (i.e., \textit{dynamic individual rationality}) but also explores scenarios where the principal seeks to incentivize agents to quit at specific periods.
We refer to the latter as the \textit{switchability} property of the mechanism.
Therefore, the notions of incentive compatibility and individual rationality in this setting involve the agents' voluntary departures, which allows for a more realistic and dynamic representation of the system.
These settings also relate our work to mechanism design with population dynamics (\citet{parkes2007online,garrett2016intertemporal,gershkov2014dynamic}) and the optimal stopping problems (\citet{kovac2013optimal,kruse2015optimal,kruse2018inverse,gensbittel2019zero}).
To the best of the authors' knowledge, the design of mechanisms with adverse selection for classic dynamic multiagent systems incorporating far-sighted participation is new to the literature.

Strategic engagement through off-menu actions enables agents to address uncertainty more effectively and enhance risk management.
This is achieved by factoring in the option value of deferring decisions and the prospect of making improved choices in the future.
Nevertheless, the presence of time-evolving private information held by agents integrated with the autonomy of choosing off-menu actions presents added complexities for the principal in crafting a dynamic, incentive-compatible mechanism.

In this paper, we extends the Myersonian approach (\citet{pavan2014dynamic,pavan2017dynamic}) for classic dynamic mechanism design by integrating time-varying adverse selection and strategic participation, and explore structural results in the context of continuous states and actions.
Our dynamic delegation mechanism, denoted as $\Theta = <\sigma, \rho, \phi>$, encompasses a task policy profile $\sigma = (\sigma_{i,t})$, a coupling policy profile $\rho = (\rho_{i,t})$, and an off-switch function profile $\phi = (\phi_{i,t})$.
The mechanism $\Theta$ directly intervenes in the base game model $\mathcal{G}$, resulting in a modified game $\mathcal{G}^{\Theta}$. 
Within the broader framework of optimal allocation problems, the task policy and coupling policy serve purposes analogous to the allocation rule and transfer rule, respectively. Each task policy $\sigma_{i,t}$ defines an action menu, denoted by $A_{i,t}[\sigma]\subseteq A_{i,t}$, where $A_{i,t}$ is the all possible actions for agent $i$ in period $t$, contingent upon the agent's state space $S_{i,t}$.
The action menu $A_{i,t}[\sigma]\{a_{i,t}=\sigma_{i,t}(s_{i,t}, h_{t}): \forall s_{i,t}\in S_{i,t}\}$, where $h_{t}$ is the public history, represents the set of available actions for every agent $i$ in period $t$.
Hence, If an agent chooses an action depending on his true state $s_{i,t}$, then it has the same effect as if the agent truthfully reports his true state and then receives an allocation $\left(a_{i,t}=\sigma_{i,t}(s_{i,t},h_{t})\right)$ under a classic mechanism (see Remark \ref{remark:equivalence_truthful_report}).
To address the agents' off-menu action, the mechanism incorporates the off-switch function, which takes the public history of actions as inputs and remains independent of the contemporaneous states and actions. At the beginning of each period $t$, agent $i$ may choose to discontinue participation. In such instances, the off-switch function $\phi_{i,t}$ realizes a compensation or penalty for the agent based on the history of actions taken by all agents.
Therefore, the future off-switch functions induce a persistent effect of the agents' current decisions on future payoffs.

Each agent $i$ uses a pure off-menu strategy $\tau_{i,t}$ to choose $\mathtt{OM}_{i,t}\in\{0,1\}$ and a pure policy $\pi_{i,t}$ to choose an action $a_{i,t} = \pi_{i,t}(s_{i,t}, h_{t})$ from the menu $A_{i,t}[\sigma]$ in every period $t$. 
The principal endeavors to develop a dynamic mechanism $\Theta$ aimed at optimizing her expected objective subject to the \textit{dynamic obedient incentive compatibility} (DOIC) under \textit{perfect Bayesian equilibria} (PBE), which encompasses the obedient incentive compatibility (IC) for the off-menu actions (OAIC) and the regular actions (RAIC).
Our concept of DOIC builds upon the foundations of classic dynamic (interim) incentive compatibility by leveraging the agents' coupled obedience.
In particular, an agent is considered obedient if he possesses the incentive to engage in both desired off-menu and regular actions as determined by the principal.
That is, under a DOIC mechanism agents' $\tau=(\tau_{i,t})$ coincides with the principal's desired $\tau^{d}=(\tau^{d}_{i,t})$ and each agent $i$'s $\pi_{i,t}(s_{i,t}, h_{t}) = \sigma_{i,t}(s_{i,t}, h_{t})$.

The off-menu and the regular actions exhibit fundamental differences, preventing a unified characterization through traditional methodologies for classic mechanism design, such as the first-order approach (\citet{kapivcka2013efficient,pavan2014dynamic,pavan2017dynamic}).
Given that agents initiate off-menu actions prior to regular actions, it is not feasible to leverage the \textit{inverse Stackelberg interaction} (\citet{ho1981information}) to design each off-switch function by making it as a function of the regular actions.
Moreover, the nature of off-menu actions as discrete choices, juxtaposed with the continuous nature of regular actions, precludes the utilization of the envelope theorem (\citet{milgrom2002envelope,pavan2014dynamic}).
This prevents us from characterizing the DOIC by constructing $\rho$ and $\phi$ as functions of $\sigma$ based on the envelope theorem.
Therefore, successfully tackling the integration of dynamic adverse selection and strategic participation demands a more intricate interplay amongst $\sigma$, $\rho$, and $\phi$, imposing further hurdles in the characterization of the DOIC.

Our model has broad applications in the design of dynamic mechanisms augmented with ``cancel-anytime" policies or deadline design and is relevant in various contexts in  economics and operations research, such as sequential screening (\citet{pinar2017robust,li2017discriminatory}), selling options (\citet{board2007selling}), bandit auctions (\citet{pavan2014dynamic,kakade2013optimal}), repeated sales (\citet{grubb2009selling,bergemann2015dynamic}), optimal delegations (\citet{krahmer2016optimal,chen2022dynamic}), as well as utility and incentive design for multiagent systems (\citet{parkes2015economic,ratliff2019perspective}) in control, engineering, and AI communities.

\subsection{Summary of Main Results}



The core element throughout our analysis lies in the application of a set of auxiliary \textit{carrier functions}, $\{g_{i,t}\}$.
Each $g_{i,t}$ is initially introduced in Section \ref{sec:up_persistence_transforms} as an implicit function of $\sigma$ alongside agent $i$'s contemporaneous state and action, devoid of association with any specific formulations and independently functioning from $\rho$ and $\phi$.
Theorem \ref{thm:payoff_flow_conservation_sufficient} articulates the principle of \textit{payoff-flow conservation} when the principal wants every agent to participate until the game $\mathcal{G}^{\Theta}$ terminates\textemdash the mechanism is individually rational (IR).
This principle delineates a sufficient condition for the RAIC under the assumption that the corresponding OAIC is assured.
We refer to the DOIC mechanism with individual rationality as IR-DOIC.
The payoff-flow conservation reveals an interconnection among the base game model $\mathcal{G}$, the mechanism $\Theta=<\sigma, \rho, \phi>$, and the carrier functions $g$.
Within these established relationships, the OAIC mechanism becomes RAIC, and consequently, demonstrates DOIC.

The structure of the off-switch functions is critical in ensuring the OAIC of a RAIC mechanism.
In a game $\mathcal{G}^{\Theta}$, each agent $i$' freedom of choosing off-menu actions from $\{0,1\}$ leads to a partitioning of his period-$t$ state space $S_{i,t}$.
This partitioning divides each $S_{i,t}$ into two distinct areas: an \textit{on region} (ONR) and an \textit{off region} (OFR), which do not need to be continuous intervals.
Each agent $i$ takes $\mathtt{OM}_{i,t}=1$ (resp. $\mathtt{OM}_{i,t}=0$) if his state $s_{i,t}\in S^{\mathtt{off}}_{i,t}$ (resp. $s_{i,t}\not\in S^{\mathtt{off}}_{i,t}$).
The premise of OAIC dictates that the actual OFR (of each agent $i$ in each period $t$) aligns with the principal-desired OFR, denoted by $S^{\mathtt{off}}_{i,t}\subseteq S_{i,t}$.
We consider a specific format of the off-switch function known as the \textit{region-cutting off-switch} or cutoff-switch, expressed as $\{\phi_{i,t}(\cdot|c_{i,t})\}$ in Definition \ref{def:cutoff_switch_function}, in which the boundary profile, $c_{i,t}$, demarcates the boundaries of $S^{\mathtt{off}}_{i,t}$.
The cutoff-switch functions enable the principal to directly leverage her desired OM strategy profile $\tau^{d}$ into the construction of the mechanism (i.e., the off-switch function).
The principal strives to design each $\phi_{i,t}(\cdot|c_{i,t})$, alongside $\sigma$ and $\rho$, in a way that provides every agent $i$ with an incentive to choose corresponding off-menu actions, leading to an alignment between $S^{\mathtt{off}}_{i,t}$ and the actual OFR for every agent in each period.

\begin{equation}\label{eq:informal_on_rent}
    \begin{aligned}
        \textup{on-rent} = G_{i,t}(a_{i,t}|s_{i,t}, \phi_{i,k}) - \phi_{i,t}(h_{t}|c_{i,t}).
    \end{aligned}
\end{equation}

The notion of \textit{on-rent} (informally defined by (\ref{eq:informal_on_rent})), intrinsically capturing the economic rent (\citet{lado1997competition}) incurred by each agent $i$'s discretion to opt for off-menu actions in every period $t$, is defined as the discrepancy between the prospective payoff-to-go (informally, $G_{i,t}(\cdot, \phi_{i,k})$ in (\ref{eq:informal_on_rent})) if agent $i$ postpones his decision $\mathtt{OM}_{i,k}=1$ to a later period $k>t$, and the value delineated by $\phi_{i,t}(h_{t}|c_{i,t})$, realized if agent $i$ takes $\mathtt{OM}_{i,t}=1$ in period $t$.
Subsequently, strictly positive and strictly negative on-rents respectively encapsulate the strict incentives to take $\mathtt{OM}_{i,t}=0$ and $\mathtt{OM}_{i,t}=1$.
In the event of an agent's on-rent equating to zero, the agent exhibits indifference between options $\mathtt{OM}_{i,t}=0$ and $\mathtt{OM}_{i,t}=1$.
In scenarios of zero on-rent, we presuppose a tie-breaking rule favoring the principal.
Therefore, the principal necessitates each $s_{i,t}\in S^{\mathtt{off}}_{i,t}$ to engender a weakly negative on-rent for an obedient agent $i$ adhering to the regular actions.
Diverging from the myopic participation decision-making, our on-rent in each non-terminal period $t$ relies not solely on the contemporaneous-period $\phi_{i,t}$, but also potentially future $\phi_{i,k}$ for $k>t$ contained in $G_{i,t}(\cdot, \phi_{i,k})$.
This dependence precludes a direct characterization of the on-rent for the design of $\phi_{i,t}$ in terms of $\sigma$ and $\rho$ incorporated in $G_{i,t}(\cdot, \phi_{i,k})$ for $k>t$.
%

\begin{figure*}
  \centering
    \includegraphics[width=0.7\textwidth]{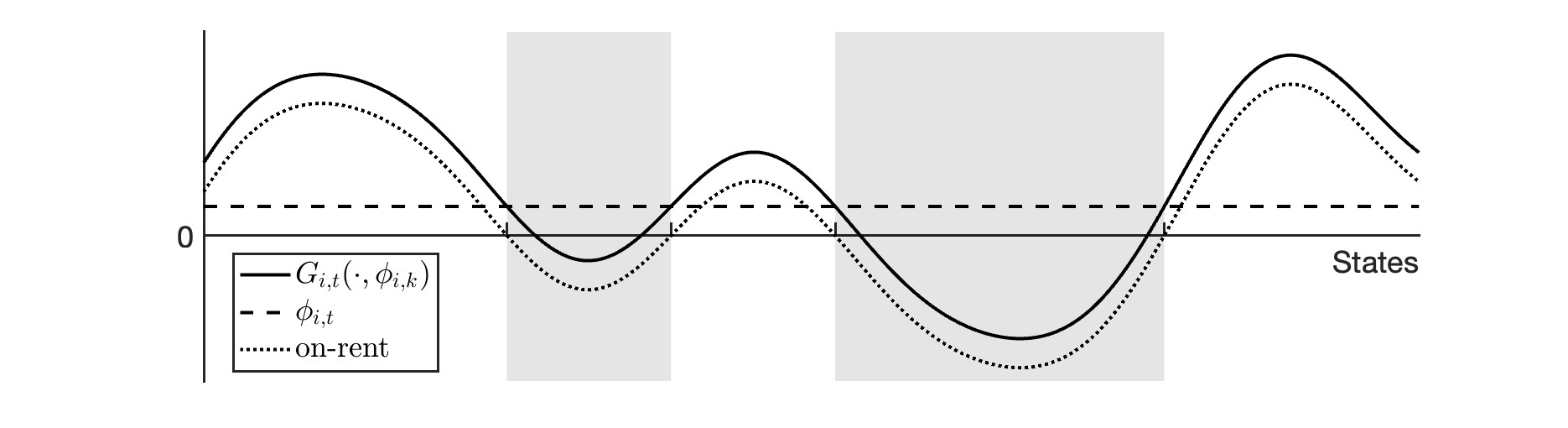}
    \caption{The principal's desired OFR $S^{\mathtt{off}}_{i,t}$ for agent $i$ in period $t$ is composed of two continuous intervals (grey regions). Any state in $S^{\mathtt{off}}_{i,t}$ induces a weakly negative on-rent.
     }
    \label{fig:intro_on_rent}
\end{figure*}

We introduce a novel procedure dubbed the \textit{up-persistence transformation}. 
This process essentially involves two key steps: state projection and subsequent state transition, specifically designed to manipulate the persistence of agents' past decisions.
The term persistence refers to the fact that agents' actions exert a stochastic influence on future states.
In the up-persistence transformation, a specific protocol is followed based on the realized state.
If a state is realized within any \textit{continuous (sub-)interval} of $S^{\mathtt{off}}_{i,t}$, it is projected to a predetermined state existing within the same interval.
The selection of this state is determined by applying a specific criterion based on the carrier functions.
If a state is realized outside of $S^{\mathtt{off}}_{i,t}$, it is left unchanged\textemdash no projection operation is performed.
Following the state projection phase, the state\textemdash whether it is the projected state or the original one\textemdash undergoes a transition based on the state dynamics provided by the base game model $\mathcal{G}$.

Leveraging the up-persistence transformation, we devise each cutoff-switch function as an implicit (given the carrier functions) closed-form expression of the task policy profile $\sigma$.
As delineated in Theorem \ref{thm:SPIR_conditions}, we articulate a sufficient condition that ensures the mechanism, integrated with the implicitly defined closed-form cutoff-switch functions, preserves IR-DOIC.
This condition, contingent solely on $\sigma$ (directly and indirectly through the carrier functions) in the context of the base game model, imposes a structural constraint encapsulated by the concept of \textit{essential regions} (Definition \ref{def:essential_region}).
The corresponding cutoff-switch function profile effectively induces the actual OFRs in which each state invariably results in an on-rent of zero for every agent across all periods, thereby rendering the agents indifferent between continued participation or withdrawal.

Due to the presupposed tie-breaking rule in scenarios of zero on-rents, the principal possesses the capacity to covert the IR-DOIC mechanism into a $S^{\mathtt{off}}$-switchable RAIC mechanism, denoted as the $S^{\mathtt{off}}$-DOIC mechanism, in which each $S^{\mathtt{off}}_{i,t}$ is an indifference region.
In such a $S^{\mathtt{off}}$-DOIC mechanism, the value of $G_{i,t}(a_{i,t}|s_{i,t}, \phi_{i,k})$ equals the constant value dictated by $\phi_{i,t}(h_{t}|c_{i,t})$ for all instances of $s_{i,t}$ in $S^{\mathtt{off}}_{i,t}$.
When each $S^{\mathtt{off}}_{i,t}$ is not an indifference region, the $S^{\mathtt{off}}$-DOIC mechanism necessitates a weakly negative on-rent for every instance of $s_{i,t}$ encapsulated within each $S^{\mathtt{off}}_{i,t}$.
With reference to Fig. \ref{fig:intro_on_rent}, the design of cutoff-switch functions subsequently requires the determination of a fixed line (i.e., the constant value given by $\phi_{i,t}(h_{t}|c_{i,t})$ for each history $h_{t}$) to cut the curve $G_{i,t}(\cdot, \phi_{i,k})$.
This requires that $G_{i,t}(\cdot, \phi_{i,k})$ resides below (or above) the fixed line, including the line itself, for states within the boundaries of $S^{\mathtt{off}}_{i,t}$ (or outside of $S^{\mathtt{off}}_{i,t}$).
The integration of $S^{\mathtt{off}}$-switchability and RAIC introduces heightened complexity to the mechanism design.
This complexity stems from the interplay between the principal's selection of $S^{\mathtt{off}}$, the mechanism, and the base game model.
Theorem \ref{thm:MAO_switchability_conjecture_gen} extends the results of Theorem \ref{thm:SPIR_conditions} to and establishes a sufficient conditions that guarantees the mechanism satisfies general $S^{\mathtt{off}}$-DOIC, incorporating an implicit closed-form formulation for each cutoff-switch functions in terms of the task policy profile. 

%
A necessary condition for $S^{\mathtt{off}}$-DOIC akin to the envelope-like condition, as specified in Proposition \ref{prop:necessary_DOIC_sigma_tau}, is explicitly imposed on the task policy profile $\sigma$ and the principal's preferred off-menu strategy profile $\tau^{d}$ given $S^{\mathtt{off}}$.
This condition is independent of the coupling policy profile $\rho$ and the off-switch function profile $\phi$.
It provides insight into the explicit expression of each carrier function, and we derive closed-form expressions for $\rho$ and $\phi$ as functions of $\sigma$.
Should $S^{\mathtt{off}}$ results in zero on-rents, Theorem \ref{thm:sufficient_condition_without_MSO} presents a sufficient condition for $S^{\mathtt{off}}$-DOIC, termed \textit{constrained monotone}.
If the task policy profile $\sigma$ fulfills this constrained monotone condition, then by utilizing the derived closed-form expressions of $\rho$ and $\phi$ in terms of $\sigma$, we can ascertain a mechanism for $S^{\mathtt{off}}$-DOIC.

%
The necessary condition obtained in Proposition \ref{prop:necessary_DOIC_sigma_tau} depends on the principal's desired off-menu strategy profile $\tau^{d}$.
However, $\tau^{d}$ is not a direct decision variable of the principal.
Rather, it needs to be induced by the mechanism $\Theta$ such that the game $\mathcal{G}^{\Theta}$ admits a PBE in which $\tau^{d}$ is the off-menu strategy profile.
This prevents a direct application of the first-order approach (\citet{kapivcka2013efficient,pavan2014dynamic,pavan2017dynamic}) to the design of $\Theta$ by selecting a $\sigma$ from a relaxed optimization problem.
In Section \ref{sec:discussion_1st_order_approach}, we show a stringent condition known as \textit{maximum-sensitive obedience}, which facilitates a first-order necessary condition in a dynamic envelop formulation for $S^{\mathtt{off}}$-DOIC that depends only on $\sigma$.
With maximum-sensitive obedience, the constrained monotone condition of $\sigma$ is both necessary and sufficient for $S^{\mathtt{off}}$-DOIC with zero on-rent.

The remainder of the paper is organized as follows.
Section \ref{sec:Preliminaries_and_Setup} lays the groundwork by introducing the preliminaries and the basic setups.
In Section \ref{sec:problem_formulation}, we delve into the details of our dynamic delegation mechanism, highlighting the specific class of off-switch functions known as the region-cutting off-switch or cutoff-switch functions.
Section \ref{sec:up_persistence_transforms} is dedicated to the detailed definition of the up-persistence transformation.
Our investigation into IR-DOIC takes center stage in Section \ref{sec:DOIC_D-IROD}, while Section \ref{sec:switchability_D-SOD} elucidates the characteristics of the general $S^{\mathtt{off}}$-DOIC.
Section \ref{sec:switchability_D-SOD} continues the exploration of $S^{\mathtt{off}}$-DOIC, further delineating its characteristics and providing closed-form formulations for each carrier function.
Lastly, in Section \ref{sec:discussion_1st_order_approach}, we delve into a discussion about the first-order approach, establishing conditions for the existence of a first-order necessary condition, and conclude the paper.
Proofs of the results are presented in the online supplementary document.


\section{Preliminaries and Setup}\label{sec:Preliminaries_and_Setup}

We first provide some conventions.
For any measurable set $Y$, $\Delta(Y)$ denotes the set of probability measures over $Y$.
Any function defined on a measurable set is assumed to be measurable.
Given a discrete time set $\mathbb{T}\equiv(1,2,\dots,T)$, the truncated time interval is denoted by $\mathbb{T}_{t,t+k} = (t,t+1,\dots, t+k)$, for $0\leq t \leq t+k \leq T$.
Let $y_{i,t}\in Y_{i,t}$ denote any term (e.g., random variable, parameter, or strategy) for agent $i$ in period $t$, where $Y_{i,t}$ is a compact set of $y_{i,t}$. Then, we write $y^{t}_{i}\equiv(y_{i,s})_{s\in \mathbb{T}_{0,t}}$, $y_{t} \equiv (y_{i,t})_{i\in I}$,  $y\equiv (y_{t})_{t\in\mathbb{T}}=(y_{i,t})_{i\in\mathcal{N},t\in\mathbb{T}}$, and $y^{t,L}_{i}\equiv(y_{i,k})^{L}_{k=t}$ for $L\in \mathbb{T}_{t,T}$.
We use $-i$ to denote the agents other than agent $i$ such that, for example, $y_{-i,t} \equiv \{y_{j,t}\}_{j\in \mathcal{N}\backslash\{i\}}$.
For the sets of profiles, we denote $Y_{t}\equiv \{Y_{i,t}\}_{i\in\mathcal{N}}$ and $Y\equiv  \{Y_{t}\}_{t\in\mathbb{T}}$.
For simplicity, we may omit the explicit listing of the index sets in the definitions of vectors and collections; e.g., $y\equiv(y_{i,t})=(y_{i,t})_{i\in\mathcal{N}, t\in\mathbb{T}}$, $Y\equiv\{Y_{i,t}\}=\{Y_{i,t}\}_{i\in\mathcal{N}, t\in\mathbb{T}}$ represent a vector and a collection, respectively, with entries indexed by $i$ and $t$, for all $i\in\mathcal{N}$, $t\in\mathbb{T}$.
A list main notations is provided in Table \ref{table:notations} in the online supplementary document.

\subsection{Base Game Model}

The \textit{base game} model is defined by a tuple $\mathcal{G}\equiv<\mathcal{E}, \mathcal{P}, \mathcal{Z}>$, where $ \mathcal{E}$ is the \textit{event model}, $\mathcal{P}$ is the \textit{state-dynamic model}, and $\mathcal{Z}$ is the \textit{reward model}.

\subsubsection{Event Model}
The event model $\mathcal{E}\equiv<\mathcal{N}, \mathbb{T}, S, A>$ is composed of four parts.
$\mathcal{N}\equiv[n]$ with $1\leq n < \infty$ is a finite set of agents.
$\mathbb{T}=(1,2,\dots, T)$ with $T< \infty$ is a finite set of discrete time indexes.
$S\equiv\{S_{i,t}\}$ is a joint set of \textit{states}, where each $S_{i,t}\equiv[\underline{s}_{i,t}, \bar{s}_{i,t}]$ is an ordered, continuous and compact set of states for agent $i$ in period $t$.
$A=\{A_{i,t}\}$ is a joint set of \textit{actions}, where each $A_{i,t}\equiv [\underline{a}_{i,t}, \bar{a}_{i,t}]$ is a continuous and compact set of actions for agent $i$ in period $t$. Each action $a_{i,t}\in A_{i,t}$ represents a \textit{task} for agent $i$ in $t$.
Each agent $i$ privately observes a state $s_{i,t}\in S_{i,t}$ at the beginning of each period $t$ and takes an action $a_{i,t}\in A_{i,t}$.
We consider that a state $s_{i,t}$, for all $i\in \mathcal{N}$ and $t\in\mathbb{T}$, is generated by Nature (e.g., \textit{Harsanyi's type}, \citet{harsanyi1967games}), which summarizes the payoff-relevant information of each agent $i$ including the agent's taste or preference about the decision-making in period $t$, his belief over other agents' information and their beliefs about him.

\subsubsection{State-Dynamic Model}
The state-dynamic model 
$\mathcal{P}\equiv<\{\kappa_{i,t},F_{i,t}\}_{i\in \mathcal{N}, t\in \mathbb{T}}, \Omega, W>$ comprises four elements.
We define the \textit{history} of period $t$ as a sequence of realized actions up to period $t-1$; i.e., $h_{t}\equiv(a_{k})_{k=1}^{t-1}\in H_{t}\equiv \prod_{k=1}^{t-1}A_{k}$ with $H_{1}=\emptyset$.
If agent $i$'s state in period $t$ is $s_{i,t}$, the history is $h_{t}$, and each agent $i$'s period-$t+1$ shock is $\omega_{i,t+1}$, then the agent's period-$t+1$ state is given by
\begin{equation}\label{eq:recommendation_dynamics}
    s_{i,t+1} = \kappa_{i,t+1}(s_{i,t}, h_{t+1}, \omega_{i,t+1}),
\end{equation}
with $s_{i,0} = \kappa_{i,0}(\omega_{i,0})$.
%
%
We refer to $\kappa_{i,t}: S_{i,t-1}\times H_{t}\times \Omega \mapsto S_{i,t}$ as the \textit{state-dynamic function} of agent $i$ in period $t$.
We assume that agents' contemporaneous states are independent of each other, which implies that the random shocks are idiosyncratic.
This assumption rules out aggregate random shocks that are common to all agents.
$F_{i,t}(s_{i,t-1}, h_{t})\in \Delta(S_{i,t})$ is the probability distribution of agent $i$'s period-$t$ state corresponding to the state-dynamic function $\kappa_{i,t}$. 
%
%
If the shock $\Tilde{\omega}_{i,t}$ is distributed according to $W$, then the state $\tilde{s}_{i,t}$ is distributed according to $F_{i,t}(s_{i,t-1}, h_{t})$.
With abuse of notation, we denote the cumulative distribution function (CDF) of $F_{i,t}(s_{i,t-1}, h_{t})$ by $F_{i,t}(\cdot|s_{i,t},h_{t})\in[0,1]$.
That is, $F_{i,t}(s_{i,t}|s_{i,t-1}, h_{t}) = \mathbb{P}\big( \kappa_{i,t}(s_{i,t-1}, h_{t}, \tilde{\omega}_{i,t}) \leq s_{i,t} \big)$.
The following assumption presupposes that the state dynamics have full support.
\begin{assumption}\label{assp:full_support}
For all $s_{i,t}\in S_{i,t}$, $h_{t}\in  H_{t}$, $i\in\mathcal{N}$, $t\in\mathbb{T}\backslash\{T\}$, $F_{i,t+1}(\cdot|s_{i,t}, h_{t})$ is strictly increasing on $S_{i,t+1}$.
\end{assumption}

%
Let $\mathcal{N}_{t}\subseteq \mathcal{N}$ denote the set of agents who participate (i.e., $\mathtt{OM}_{i,t}=0$ for all $i\in\mathcal{N}_{t}$ and $\mathtt{OM}_{j,t}=1$ for all $j\not\in \in\mathcal{N}_{t}$) in period $t$.
Let $a^{\dagger}_{t}=\{a_{j,t}\}_{j\in \mathcal{N}_{t}}$ denote the action profile taken by agents in $\mathcal{N}_{t}$.
To cope with each agent's possible earlier termination of his participation, we assume that $\{\kappa_{i,t+1}, F_{i,t+1}\}$ is well-defined for every subset $\mathcal{N}_{t}$ of $\mathcal{N}$.
That is, we have $\kappa_{i,t+1}(s_{i,t}, h_{t}, a^{\dagger}_{t}, \omega_{i,t})$ and $F_{i,t+1}(\cdot|s_{i,t}, h_{t}, a^{\dagger}_{t})$ for all $\mathcal{N}_{t}\subseteq \mathcal{N}$, $i\in\mathcal{N}_{t}$, $t\in\mathbb{T}\backslash\{T\}$.
We take the state-dynamic model $\mathcal{P}$ as primitive.

\subsubsection{Reward Model}
In the reward model $\mathcal{Z}\equiv\{u_{i,t}\}_{i\in \mathcal{N}, t\in \mathbb{T}}$, each $u_{i,t}:S_{i,t}\times A_{t}\mapsto \mathbb{R}$ is each agent $i$'s (intrinsic) reward function.
When agent $i$'s state is $s_{i,t}$, his action is $a_{i,t}$, and other agents' actions are $a_{-i,t}=(a_{j,t})_{j\neq i}$, the agent receives a single-period reward $u_{i,t}(s_{i,t}, a_{i,t}, a_{-i,t})$.

\subsubsection{Strategies}

In this work, the agents are \textit{ex-interim} decision makers. That is, each agent $i$ takes an action after observing $(s_{i,t}, h_{t})$ for every $t\in \mathbb{T}$, $s_{i,t}\in S_{i,t}$, $h_{t}\in H_{t}$, given the period-$t$ prior beliefs about other agents' state based on $\mathcal{P}$. 
Define a collection of pure-strategy \textit{non-anticipating strategy profiles} available to the agents,
\[
    \begin{aligned}
    \Pi\equiv\Big\{\pi=(\pi_{i,t}) \Big| \pi_{i,t}: S_{i,t}\times H_{t} \mapsto  A_{i,t}, \forall i\in \mathcal{N}\Big\}.
    \end{aligned}
\]
Agents play $\mathcal{G}$ according to a pure-strategy \textit{equilibrium} $\Pi[\mathcal{G}]\subseteq \Pi$, which is a set of policy profiles satisfying a certain equilibrium criterion.
An \textit{outcome} of $\mathcal{G}$, $\gamma\in \Delta(S\times A)$, is a probability measure of states and actions over time which describes the ex-ante stochasticity of the agents' collective responses and the states.
%
%
Given a game $\mathcal{G}$, let $O(\cdot; \mathcal{G}): \Pi[\mathcal{G}]\mapsto \Delta(S\times A)$ denote the correspondence from a policy profile $\pi\in \Pi[\mathcal{G}]$ to an outcome $\gamma$; i.e., $\gamma = O(\pi; \mathcal{G})$.
%


\subsubsection{Information Asymmetry}
In this work, we assume that each agent $i$'s period-$t$ state $s_{i,t}$ is privately observed by the agent; once all the agents have made decisions in a period, their decisions and private states become publicly observable. Hence, the history $h_{t}=(a_{k})_{k=1}^{t-1}$ is public information at the beginning of each period $t$, and each agent does not need to form beliefs about other agents' past actions.

%

\subsection{Mechanism Design without Participation Constraints}

The principal's objective, denoted by $Q(\gamma; \mathcal{G})\in \mathbb{R}$, depends on the outcome $\gamma$ of the game $\mathcal{G}$. We refer to $g^{d}=\max_{\gamma}Q(\gamma; \mathcal{G})$ as the \textit{goal} of the principal.
The principal may aim to achieve \textit{optimal social welfare} (e.g., optimal multiagent system operations) when $Q(\cdot;\mathcal{G})$ is the expected sum of all agents' payoffs; she may want to obtain \textit{optimal revenue} when $Q(\cdot;\mathcal{G})$ is the optimized expected benefit (e.g., payments) transferred from the agents' payoffs.

Let $\gamma^{d}$ denote the principal's desired outcome; i.e., $\gamma^{d}\in\arg\max_{\gamma} Q(\gamma;\mathcal{G})$ or equivalently $g^{d} = Q(\gamma^{d};\mathcal{G})$.
However, the principal has no ability to directly change any outcomes from the agents' arbitrary behaviors to $\gamma^{d}$ or force the agents to take certain actions that lead to $\gamma^{d}$.
Instead, the principal aims to design a \textit{mechanism} $\Theta$ that directly intervenes in the base game $\mathcal{G}$ such that her goal can be achieved.
The direct intervention leads to a new game, denoted by $\mathcal{G}^{\Theta}$.
%

The principal is \textit{ex-ante} decision maker and is assumed to have the commitment power.
That is, the principal needs to design a mechanism $\Theta$ and commit to it before the agents start to play the game $\mathcal{G}^{\Theta}$. 
After $\mathcal{G}^{\Theta}$ has started, the principal cannot change the mechanism $\Theta$.
On the other hand, each agent $i$ dynamically determines policy $\pi_{i,t}$ and adaptively plans future policies $\pi^{t,T}_{i}=(\pi_{i,k})^{T}_{k=t}$ to choose an optimal action for each period after observing his private state that arrives over multiple periods.

Informally, let the correspondence $D_{i,t}(s_{i,t},h_{t}|\pi_{-(i,t)},\mathcal{G}^{\Theta})\subseteq A^{\Theta}_{i,t}$ denote the set of each agent $i$'s period-$t$ best responses to other agents' strategies $\pi^{t,T}_{-i}=(\pi_{-i,k})_{k=t}^{T}$ and his planned future strategies $\pi^{t+1,T}_{i}=(\pi_{i,k})_{k=t+1}^{T}$, when the agent's private state is $s_{i,t}$ and the public history is $h_{t}$.
Note that given $h_{t}$, $D_{i,t}(\cdot)$ is independent of $\pi^{t-1}$. Here, we denote $\pi_{-(i,t)}=(\pi^{1,t-1}, \pi^{t+1,T}_{i}, \pi^{t, T}_{-i})$ to simplify the notation.
For the dynamic game $\mathcal{G}^{\Theta}$ induced by the mechanism $\Theta$, we consider \textit{perfect Bayesian equilibrium} (PBE) as the equilibrium solution concept.

\begin{definition}[PBE]\label{def:PBE_0}
    A \textup{PBE} of $\mathcal{G}^{\Theta}$ is a pair $<\xi, \pi>$ consisting of a \textup{belief profile} $\xi=(\xi_{i,t})$, where $\xi_{i,t}:H_{t} \mapsto \Delta(S_{-i,t})$, and a \textup{policy profile} $\pi=(\pi_{i,k}, \pi_{-i,k})_{k=1}^{T}$ such that \textit{(i)} each belief $\xi_{i,t}(\cdot|h_{t})$ updates according to the Bayes' law and \textit{(ii)} for all $i\in\mathcal{N}$, $t\in\mathbb{T}$, $s_{i,t}\in S_{i,t}$, $h_{t}\in H_{t}$,
    %
    \begin{equation}\label{eq:PBE_0}
    \begin{aligned}
    \pi_{i,t}(s_{i,t}, h_{t}) \in D_{i,t}(s_{i,t},h_{t}|\pi_{-(i,t)},\mathcal{G}^{\Theta}).
    \end{aligned}
\end{equation}
Let $\Pi^{P}[\mathcal{G}^{\Theta}]$ denote the set of all possible PBE of the game $\mathcal{G}^{\Theta}$.
\hfill$\triangle$
\end{definition}

Due to the subgame perfection of PBE, the condition (\ref{eq:PBE_0}) satisfies the \textit{one-shot deviation principle} (\citet{hendon1996one}).
In particular, $\pi$ is a PBE if and only if $\pi_{i,t}$ is a best response to $\pi_{-(i,t)}=(\pi^{t-1}_{i},\pi^{T}_{i,t+1},\pi_{-i})$ for every $i\in\mathcal{N}$, $t\in\mathbb{T}$, $s_{i,t}\in S_{i,t}$, $h_{t}\in H_{t}$ when agent $i$ plans $\pi^{t+1, T}_{i}$.

The principal faces adverse selection in her interaction with the agents because $s_{i,t}$ is unobserved by the principal at the beginning of each period $t$ for every $i\in\mathcal{N}$. 
The principal as an ex-ante decision-maker has an objective $Q(\gamma;\mathcal{G}^{\Theta})$, where $\gamma$ is the outcome due to the agents' interim decision-making processes $\{D_{i,t}(\cdot|\pi_{-(i,t)},\mathcal{G}^{\Theta})\}$ that are unknown by the principal due to the unobservability of the private states that arrive over time. 
Hence, the principal's mechanism $\Theta$ needs to provide enough incentives for the agents such that the principal's desired equilibrium $\pi^{d}$ coincides with one of the actual equilibria of the game $\mathcal{G}^{\Theta}$.
That is, the principal's mechanism design problem can be described as follows:
\begin{equation}
    \begin{aligned}
        \max\limits_{\Theta} Q\Big(O(\pi^{d}; \mathcal{G}^{\Theta} ), \mathcal{G}^{\Theta}\Big), \textup{ s.t., } \pi^{d}\in \Pi^{P}[\mathcal{G}^{\Theta}].
    \end{aligned}
\end{equation}
For every $\Theta$ such that $\pi^{d}\in \Pi^{P}[\mathcal{G}^{\Theta}]$, we say that the mechanism $\Theta$ is \textit{incentive compatible} in PBE $\pi^{d}$.
The notion of incentive compatibility captures the situation when the mechanism provides incentives for the agents to take actions that coincide with the principal's prediction and intention.

\section{Problem Formulation}\label{sec:problem_formulation}

\begin{figure*}
  \centering
    \includegraphics[width=0.7\textwidth]{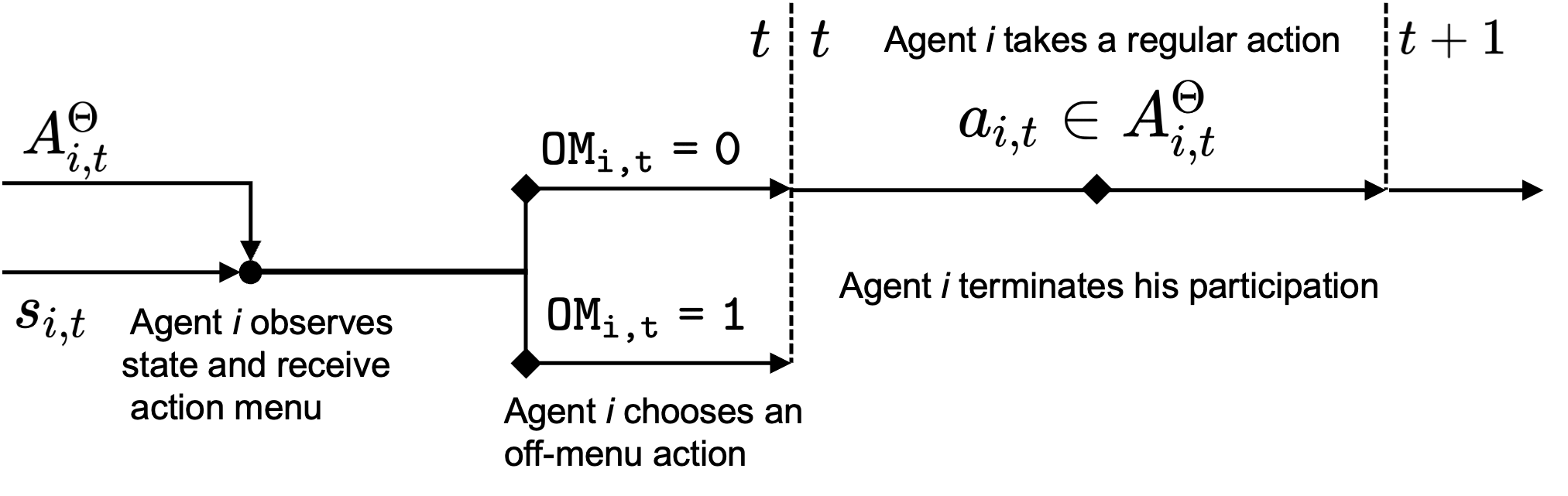}
    \caption{Dynamic delegation with off-menu actions. At the beginning of each period, the principal generates an action menu $A^{\Theta}_{i,t}$ for each agent $i$. Agent $i$ with private state $s_{i,t}$ first chooses $\mathtt{OM}_{i,t}\in\{0,1\}$. If $\mathtt{OM}_{i,t}=1$, then agent $i$ terminates his participation; otherwise, he takes a regular action from $A^{\Theta}_{i,t}$, and continues to the next period.
     }
    \label{fig:delegation_timing_off_menu}
\end{figure*}

In this work, we consider a scenario where each agent $i$ decides whether to participate at the beginning of each period $t$ after observing $s_{i,t}$ and then chooses an action from $A^{\Theta}_{i,t}\subseteq A_{i,t}$ if he has decided to participate.
We refer to each agent $i$'s decision of participation at each period $t$ as the \textit{off-menu action}, denoted by $\mathtt{OM}_{i,t}\in\{0,1\}$, for all $i\in\mathcal{N}$, $t\in\mathbb{T}$.
With reference to Fig. \ref{fig:delegation_timing_off_menu}, if agent $i$ takes $\mathtt{OM}_{i,t}=1$, then he withdraws from participation; otherwise, he participates in period $t$ by taking $\mathtt{OM}_{i,t}=0$ and then chooses a regular action $a_{i,t}$ from $A^{\Theta}_{i,t}$, following which the game progresses to the next period.

\subsection{Dynamic Delegation Mechanism}

The principal's \textit{dynamic delegation mechanism} $\Theta$ is defined by a tuple $\Theta\equiv<\sigma, \rho, \phi>$,
where $\sigma=(\sigma_{i,t})$ is a profile of \textit{task allocation policies} (task policies), $\rho=(\rho_{i,t})$ is a profile of \textit{coupling policies}, and $\phi=(\phi_{i,t})$ is a profile of \textit{off-switch functions}.

\subsubsection{Task Policy}

The principal uses each task policy $\sigma_{i,t}:S_{i,t}\times H_{t}\mapsto A_{i,t}$ determines an \textit{action} \textit{menu} $A^{\Theta}_{i,t}=A_{i,t}[\sigma_{i,t}, h_{t}]$ $\subseteq A_{i,t}$ defined by
\[
A_{i,t}[\sigma_{i,t},h_{t}] \equiv \Big\{\sigma_{i,t}(s_{i,t}, h_{t}): s_{i,t}\in S_{i,t}\Big\}\subseteq A_{i,t},
\]
for all $i\in\mathcal{N}$, $t\in\mathbb{T}$, $h_{t}\in H_{t}$.
We write $A_{i,t}[\sigma]$ for simplicity by omitting $h_{t}$ and the indexes $(i,t)$ of $\sigma_{i,t}$.

\subsubsection{Coupling Policy}
%
Each coupling policy $\rho_{i,t}:A_{t}[\sigma]\times H_{t}\mapsto \mathbb{R}$ specifies a \textit{coupling value} $m_{i,t}=\rho_{i,t}(a_{t}, h_{t})\in \mathbb{R}$ for all $i\in\mathcal{N}$, $t\in\mathbb{T}$, $a_{t}\in A_{t}[\sigma]=\prod_{i\in\mathcal{N}} A_{i,t}[\sigma_{i,t}; h_{t}]$, $h_{t}\in H_{t}$, such that each agent $i$ receives a single-period utility 
\[
z_{i,t}(a_{t}, s_{i,t}, m_{i,t})\equiv u_{i,t}(s_{i,t}, a_{t}) + m_{i,t}.
\]
The coupling value could be a payment (i.e., $m_{i,t}\in \mathbb{R}_{-}$), a compensation (i.e., $m_{i,t}\in \mathbb{R}_{+}$), or a reward signal.
With abuse of notation, we write $z_{i,t}(a_{t}, s_{i,t}) = z_{i,t}\left(a_{t}, s_{i,t}, \rho_{i,t}(a_{t}, h_{t})\right)$ when $m_{i,t}$ is replaced by $\rho_{i,t}(a_{t}, h_{t})$.

\begin{remark}\label{remark:equivalence_truthful_report}
    In $\Theta$, the principal proposes an action menu for each agent and lets the agent freely choose one action from the menu. Agent $i$ with $s_{i,t}$ would choose an action $a_{i,t}=\sigma_{i,t}(s_{i,t}, h_{t})$ (or $\hat{a}_{i,t}=\sigma_{i,t}(\hat{s}_{i,t}, h_{t})$) depending on the true state $s_{i,t}$ (or any arbitrary state $\hat{s}_{i,t}$), which has the same effect as if the agent first reports his state $s_{i,t}$ truthfully (or misreport a state $\hat{s}_{i,t}$) and then the principal allocates an action using $\sigma_{i,t}$.
    For the coupling policy, given $\sigma_{i,t}$, we have $\overline{\rho}_{i,t}(\hat{s}_{t}, h_{t}) = \rho_{i,t}\left(\sigma_{t}(\hat{s}_{t}, h_{t}), h_{t}\right)$ which depends on the ``reports" $\hat{s}_{t}=(\hat{s}_{i,t})_{i\in\mathcal{N}}$.
    \hfill $\triangle$
\end{remark}

\subsubsection{Off-Switch Function}
%
Each off-switch (function) $\phi_{i,t}(\cdot|c_{i,t}):H_{t}\mapsto \mathbb{R}$ is parameterized by $c_{i,t}$.
The off-switch realizes an off-switch value, $o_{i,t}=\phi_{i,t}(h_{t}|c_{i,t})$, contingent on any history $h_{t}\in H_{t}$, to agent $i$ if and only if agent $i$ takes $\mathtt{OM}_{i,t}=1$.
Let $\hat{z}_{i,t}(\mathtt{OM}_{i,t}, a_{t}, s_{i,t})$ denote single-period utility involving the OM actions.
That is,
\[
\hat{z}_{i,t}(\mathtt{OM}_{i,t}, a_{t}, s_{i,t})=
\begin{cases}
\phi_{i,t}(h_{t}|c_{i,t}),&\textup{ if } \mathtt{OM}_{i,t}=1,\\
z_{i,t}(a_{t}, s_{i,t}),& \textup{ if } \mathtt{OM}_{i,t}=0.
\end{cases}
\]
In every period $t$, the off-switch value $o_{i,t}$ is a posted value that is independent of agent $i$'s current states or any regular actions, and it is only determined by $\phi_{i,t}$ after the history $h_{t}$ is realized.
Here, the principal's mechanism $\Theta$ directly intervenes in the event model via $\sigma$ and the reward model by $\rho$ and $\phi$.


\subsection{Agents' Dynamic Decision-Making}

Before using the pure strategy $\pi$ to choose regular actions from the action menu, the agents use a pure OM strategy profile $\tau=(\tau_{i,t})$ to choose OM actions in each period.
As they plan their future regular actions using $\pi^{t,T}$, agents also strategize future OM actions and incorporate this plan into their decision-making for the current period.
Since each agent $i$ can take $\mathtt{OM}_{i,t}=1$, up to one time in some period $t$, we consider that each $\tau_{i,t}$ determines $\mathtt{OM}_{i,L^{i} }=1$ for any $L^{i}\in \mathbb{T}_{t,T+1}\equiv \mathbb{T}_{t,T}\cup\{T+1\}$.
That is, $\mathtt{OM}_{i,L^{i} } = 1 \textup{ if and only if } \tau_{i,t}(s_{i,t},h_{t}) = L^{i}\in \mathbb{T}_{t,T+1}$.
Hence,
\begin{equation}\label{eq:agent_OM_strategy_plan}
    \begin{aligned}
        \tau_{i,t}(s_{i,t},h_{t}) = L^{i} \Longrightarrow \left(\mathtt{OM}_{i,t}=\mathbf{1}_{\{L^{i}=t\}}, \mathtt{OM}_{i,t+1}=\mathbf{1}_{\{L^{i}=t+1\}},\dots, \mathtt{OM}_{i,T+1}=\mathbf{1}_{\{L^{i}=T+1\}}\right),
    \end{aligned}
\end{equation}
where $\mathbf{1}_{\{\cdot\}}$ is the indicator function.
That is, when $\tau_{i,t}(s_{i,t},h_{t}) = L^{i}$, agent $i$ \textit{(i)} takes $\mathtt{OM}_{i,t}=\mathbf{1}_{\{L^{i}=t\}}$ and \textit{(ii)} plans $\left(\mathtt{OM}_{i,k}=\mathbf{1}_{\{k=L^{i}\}}\right)_{k\in\mathbb{T}_{t+1, T+1}}$.
Since the executions of $\tau$ lead to possible decreases in population over time, we let $\mathcal{N}_{t}\subseteq \mathcal{N}$ to denote the set of agents who have taken $\{\mathtt{OM}_{i,t-1}=0\}_{i\in \mathcal{N}_{t}}$ in period $t-1$. 
For notational compactness, we still use $a_{t}$ and $h_{t}$ to denote the joint actions and history without indicating the population change.
Similar to the regular actions, we assume that the principal discloses who has taken $\mathtt{OM}_{i,t}=1$ at the end of each period $t$.
Therefore, each agent knows $\mathcal{N}_{t}$ (i.e., who remains in the game at the beginning of each period) and does not have to form beliefs about it.
%

\subsubsection{Expected Future History}

For ease of exposition, consider when there are only two agents, $\mathcal{N}_{t}=\{i,j\}$, in period $t$.
Agent $i$'s expected payoff-to-go depends on agent $j$'s $\tau_{j,t}$ through the \textit{expected future history}.
In particular, suppose that agent $j$ with state $s_{j,t}$ has chosen $L^{j}=\tau_{j,t}(s_{j,t}, h_{t})$ (i.e., $\mathtt{OM}_{j, L^{j}}=1$). Then, agent $i$'s expected future history up to $L^{i}=\tau_{i,t}(s_{i,t}, h_{t})$ can be expressed as
\begin{equation}\label{eq:expected_future_history}
\tilde{h}_{t}[L^{i};L^{j}]=\begin{cases}
\left(h_{t}, \tilde{h}^{t+1,L^{j}}, \tilde{h}^{L^{j}+1, L^{i}}\right), & \textup{ if } L^{j}<L^{i},\\
\left(h_{t}, \tilde{h}^{t+1,L^{i}}\right), & \textup{ if } L^{j} \geq L^{i}.
\end{cases}
\end{equation}
%
That is, when agent $i$ looks into the future (to make plans), he needs to anticipate agent $j$'s future participation to obtain $\tilde{h}_{t}[L^{i};L^{j}]$, which has a direct influence on agent $i$'s expected future payoffs.
%
%
Let $L^{-i}=\{L^{j}\}_{j\neq i} = \tau_{-i,t}(s_{-i,t}, h_{t})$ denote the profile generated by other agents' OM strategy profile $\tau_{-i,t}$.
By $\tilde{h}_{t}[L^{i};L^{-i}]$, we generalize agent $i$'s expected future history $\tilde{h}_{t}[L^{i};L^{j}]$ to the game in period $t$ with $|\mathcal{N}_{t}|\geq 2$, in which the components of  $\tilde{h}_{t}[L^{i};L^{-i}]$ from period $L^{l}+1$ to $L^{i}$ (suppose $L^{l}\leq L^{i}$) excludes agent $l$'s participation (according to agent $l$'s plan).
To simplify the notation, we use $\tilde{h}_{k}$ for $k>t$ to denote agent $i$'s expected future history without marking others' OM decisions (i.e., future population dynamics).

\subsubsection{Conjectures}

However, agent $i$ does not observe others' private states at the beginning of each period. 
Instead, given other agents' policies $\tau_{-i,t}$, agent $i$ forms a \textit{conjecture} $x_{i,t}(\cdot|h_{t})\in \Delta\Big(\big(\mathbb{T}_{t,T}\big)^{|\mathcal{N}|-1}\Big)$ over others' choices of $L^{-i}\in \Big(\mathbb{T}_{t,T}\Big)^{|\mathcal{N}|-1}$.
That is, 
\begin{equation}\label{eq:x_coincides_tau}
\begin{aligned}
     x_{i,t}(L^{-i}|h_{t})\equiv\mathbb{P}\big( \tau_{-i,t}(s_{-i, t}, h_{t}) = L^{-i}\big|h_{t}\big)=\mathbb{E}^{F_{-i,t}}\Big[\mathbf{1}_{\{\tau_{-i,t}(\tilde{s}_{-i,t}, h_{t}) = L^{-i}\}}\Big|h_{t}\Big].
\end{aligned}
\end{equation}

According to the Ionescu Tulcea theorem (\citet{hernandez2012discrete}), the state-dynamic model $\mathcal{P}$, the agents' strategy profiles $\pi$ and $\tau$, and the principal's task policy profile $\sigma$ uniquely define an outcome of the game $\mathcal{G}^{\Theta}$, denoted by $\gamma^{\sigma}_{\pi,\tau}=O(\pi,\tau; \mathcal{G}^{\Theta})\in \Delta(S\times A)$.
We denote the expectation operator corresponding to $\gamma^{\sigma}_{\pi,\tau}$ by $\mathbb{E}^{\sigma}_{\pi,\tau}[\cdot]$.
In addition, given any history $h_{t}\in H_{t}$, state $s_{i,t}\in S_{i,t}$, and agent $i$'s conjecture $x_{i,t}$ (which coincides with $\tau_{-i,t}$ by satisfying (\ref{eq:x_coincides_tau})), the Ionescu Tulcea theorem also pins down a unique \textit{continuing outcome} perceived by agent $i$, denoted by $\gamma^{\sigma}_{\pi}[s_{i,t}, h_{t}, x_{i,t}] \in \Delta\left( S_{-i,t}\times A_{-i,t}\times \big(\prod^{T}_{k=t+1} S_{k}\times A_{k}\big) \right)$, according to $\{\{\kappa^{t,T}_{i}, F^{t,T}_{i}\}_{i\in \mathcal{N}}, \Omega, W\}$, $\pi^{t, T}$, and $\sigma^{t,T}$. 
Note that the continuing outcome $\gamma^{\sigma}_{\pi}[s_{i,t}, h_{t}, x_{i,t}]$ is defined by assuming that agent $i$ will remain participating in the game until the end.
Hence, $\gamma^{\sigma}_{\pi}[s_{i,t}, h_{t}, x_{i,t}]$ is independent of $\tau_{i,t}$.
We denote the corresponding expectation operator by $\mathbb{E}^{\sigma}_{\pi}[\cdot|s_{i,t}, h_{t},x_{i,t}]$.

\subsubsection{Perfect Bayesian Equilibrium}

For all $i\in\mathcal{N}$, $t\in\mathbb{T}$, $s_{i,t}\in S_{i,t}$, $h_{t}\in H_{t}$, and $x_{i,t}$, define the \textit{prospect function} up to period $L\in \mathbb{T}_{t,T}$ by 
\begin{equation}\label{eq:offswitch_payoff_to_go}
    \begin{aligned}
&G_{i,t}(a_{i,t}|s_{i,t}, h_{t}, L, x_{i,t})\equiv \mathbb{E}^{\sigma}_{\pi}\left[\sum\limits_{k=t}^{L} z_{i,k}\big(\tilde{s}_{i,k}, \tilde{a}_{i,k}\big) +\phi_{i,L+1}(\tilde{h}_{L+1}|c_{i,L+1}) \Big| s_{i,t}, h_{t}, x_{i,t} \right],
\end{aligned}
\end{equation}
with $\phi_{i,T+1}(\cdot) = 0$.
Then, each agent $i$'s period-$t$ interim expected payoff-to-go is defined by
\begin{equation}\label{eq:to_go_off_menu_switch}
    \begin{aligned}
    &\Lambda_{i,t}\left(\mathtt{OM}_{i,t}, a_{i,t}|s_{i,t}, h_{t}, x_{i,t}\right)\equiv \phi_{i,t}(h_{t}|c_{i,t}) \mathbf{1}_{\{ \mathtt{OM}_{i,t} =1 \}}+ \max\limits_{L\in\mathbb{T}_{t,T}} G_{i,t}(a_{i,t}|s_{i,t}, h_{t}, L, x_{i,t}) \mathbf{1}_{\{ \mathtt{OM}_{i,t} =0 \}}.
    \end{aligned}
\end{equation}
We omit the action $a_{i,t}$ in $G_{i,t}$ and $\Lambda_{i,t}$ when $a_{i,t}$ is obedient. 
Hence, the set of the best responses of agent $i$ in period $t$ is given by the correspondence
\[
\begin{aligned}
    D_{i,t}\left(s_{i,t},h_{t}|x_{i,t}, \pi_{-(i,t)}, \mathcal{G}^{\Theta}\right)=\argmax_{\mathtt{OM}_{i,t}\in\{0,1\}, a_{i,t} \in A_{i,t}[\sigma] } \Lambda_{i,t}(\mathtt{OM}_{i,t}, a_{i,t}|s_{i,t}, h_{t}, x_{i,t}).
\end{aligned}
\]

\begin{definition}[PBE]\label{def:PBE_1}
    A \textup{PBE} of $\mathcal{G}^{\Theta}$ is a pair $<\xi, \tau,\pi>$ consisting of a \textup{posterior belief profile} $\xi=(\xi_{i,t})$, where $\xi_{i,t}:H_{t} \mapsto \Delta(S_{-i,t})$, and \textup{policy profiles} $\tau=(\tau_{i,t},\tau_{-i,t})_{t=1}^{T}$ and $\pi=(\pi_{i,t}, \pi_{-i,t})_{t=1}^{T}$ such that \textit{(i)} each belief $\xi_{i,t}(\cdot|h_{t})$ updates according to the Bayes' law and \textit{(ii)} for all $i\in\mathcal{N}$, $t\in\mathbb{T}$, $s_{i,t}\in S_{i,t}$, $h_{t}\in H_{t}$,
    %
    \begin{equation}\label{eq:PBE_1}
    \begin{aligned}
    \left<\tau_{i,t}(s_{i,t}, h_{t}),\pi_{i,t}(s_{i,t}, h_{t})\right> \in D_{i,t}(s_{i,t},h_{t}|x_{i,t}, \pi_{-(i,t)}, \mathcal{G}^{\Theta}),
    \end{aligned}
\end{equation}
where $x_{i,t}$ satisfies (\ref{eq:x_coincides_tau}) given $\tau$.
Let $\Pi^{P}[\mathcal{G}^{\Theta}]$ denote the set of all possible PBE of the game $\mathcal{G}^{\Theta}$.
\hfill$\triangle$
\end{definition}

\begin{definition}[DIC]
    A mechanism $\Theta=\left<\sigma, \rho, \phi\right>$ is \textit{dynamic incentive compatible (DIC)} in $\left<\tau, \pi\right>$ if $\left<\tau, \pi\right>$ is a PBE of the game $\mathcal{G}^{\Theta}$.
\hfill$\triangle$
\end{definition}

\begin{assumption}\label{assp:bounded_prospect}
    For any available mechanism $\Theta=\left<\sigma, \rho, \phi\right>$ and $\pi$, the prospect function is bounded; i.e., for all $i\in\mathcal{N}$, $a_{i,t}\in A_{i,t}[\sigma]$, $s_{i,t}\in S_{i,t}$, $h_{t}\in H_{t}$, $x_{i,t}$, $\max\limits_{L}|G_{i,t}(a_{i,t}|s_{i,t}, h_{t}, L, x_{i,t})|<\infty$.
\end{assumption}

Assumptions \ref{assp:full_support} and \ref{assp:bounded_prospect} are the standing assumptions that will be held constant throughout the rest of the paper.
The timing of the principal's decision-making and the agents' decision-making processes in the game $\mathcal{G}^{\Theta}$ are as follows.
\begin{itemize}
    \item At the ex-ante stage, the principal designs and commits to a mechanism $\Theta\left<\sigma, \rho, \phi\right>$ and then publicly discloses it. After the agents decide to participate, the game $\mathcal{G}^{\Theta}$ starts.
    \item Given the (public) history $h_{t}$, the principal's $\sigma_{i,t}$ generates an action menu $A_{i,t}[\sigma]$ to each agent $i$ at the beginning of each period $t$.
    \item After observing his private state $s_{i,t}\in S_{i,t}$, agent $i$ takes $\mathtt{OM}_{i,t}\in\{0,1\}$ according to $\tau_{i,t}$. If $\mathtt{OM}_{i,t}=1$, then agent $i$ terminates his participation and leaves the game. If $\mathtt{OM}_{i,t}=0$, then agent $i$ takes an action $a_{i,t}\in A_{i,t}[\sigma]$ using $\pi_{i,t}$.
    \item After all agents have taken actions $\mathtt{OM}_{t}$ and $a_{t}$, the principal uses the off-switch function $\phi$ to assign a off-switch value $o_{j,t}$ for agent $i$ who has taken $\mathtt{OM}_{i,t}=1$ and uses the coupling policy $\rho$ to specify a coupling value $m_{j,t}$ for agent $j$ who has taken $\mathtt{OM}_{j,t}=0$.
    \item Each agent $i$ receives an immediate utility $\hat{z}_{i,t}(\mathtt{OM}_{i,t}, a_{t}, $ $s_{i,t})$.
    \item At the end of period $t$, the states $s_{t}$, $\mathtt{OM}_{t}$, and $a_{t}$ become public information and the history $h_{t}$ is updated to $h_{t+1}$. Agents' states are then transitioned according to the state-dynamic model.
\end{itemize}

\subsection{Dynamic Obedient Incentive Compatibility}

In this section, we define the notion of dynamic obedient incentive compatibility.

\begin{definition}[Dynamic Obedience]\label{def:obedience}
Fix a base game $\mathcal{G}$ and a mechanism $\Theta=\left<\sigma, \rho, \phi\right>$.
We say that a PBE $\left<\xi,\tau^{d}, \pi^{d}\right>$ of the game $\mathcal{G}^{\Theta}$ is \textup{dynamically obedient (OPBE)} if $\pi^{d}_{i,t}(s_{i,t}, h_{t}) = \sigma_{i,t}(s_{i,t}, h_{t})$ for all $i\in \mathcal{N}$, $t\in \mathbb{T}$, $s_{i,t}\in S_{i,t}$, $h_{t}\in H_{t}$. 
%
\hfill$\triangle$
\end{definition}

In pure-strategy OPBE, the public history provides complete information on agents' past true states.
As such, there is no requirement for any agent to construct and update posterior beliefs concerning the past private states of others.
Consequently, an agent's beliefs about the contemporaneous states of others, at the beginning of any period, are aligned with the period's prior distributions as dictated by the dynamic models of other agents' states.
Given these conditions, we make a simplification by omitting the posterior belief $\xi$ in OPBE.
The incentive compatibility for dynamic obedience (DOIC) requires the optimality of $\left<\tau^{d}_{i,t}(s_{i,t}, h_{t}), \sigma_{i,t}(s_{i,t}, h_{t})\right>$ at every period for each agent $i$ with state $s_{i,t}$ and history $h_{t}$.
That is, for all $i\in\mathcal{N}$, $t\in\mathbb{T}$, $s_{i,t}\in S_{i,t}$, $h_{t}\in H_{t}$, 
\begin{equation}\tag{$\mathtt{DOIC}$}\label{eq:def_DOIC}
    \begin{aligned}
        \left<\tau^{d}_{i,t}, \sigma_{i,t}\right>\in D_{i,t}(s_{i,t},h_{t}|x_{i,t}, \sigma_{-(i,t)}, \mathcal{G}^{\Theta}).
    \end{aligned}
\end{equation}
%
The constraint (\ref{eq:def_DOIC}) ensures that obedience is one of each agent $i$'s best response to other agents' obedience for every possible $s_{i,t}\in S_{i,t}$ and $h_{t}\in H_{t}$ in every period $t\in\mathbb{T}$.

\begin{proposition}\label{prop:dynamic_revelation_principle}
    Fix any base game model $\mathcal{G}$. 
    For any mechanism $\Theta'=\left<\sigma',\rho', \phi'\right>$ that is DIC in $\left<\tau', \pi'\right>$, there exists a DOIC mechanism $\Theta=\left<\sigma, \rho, \phi\right>$ that induces an OPBE $\left<\tau,\sigma\right>$ such that 
    \[
    Q(\gamma^{\sigma'}_{\pi',\tau'};\mathcal{G}^{\Theta'}) = Q(\gamma^{\sigma}_{\sigma,\tau};\mathcal{G}^{\Theta}).
    \]
\end{proposition}

Proposition \ref{prop:dynamic_revelation_principle} shows a dynamic revelation principle for the design of $\Theta$, which allows us to restrict attention to DOIC mechanisms.
The principal's DOIC mechanism design problem is defined by the following mathematical programming with equilibrium constraints:
\begin{equation}\label{eq:principal_DOD_Obj}
    \max\limits_{\Theta=<\sigma, \rho, \phi>} Q\big(\gamma^{\Theta}_{\sigma,\tau^{d}}; \mathcal{G}^{\Theta } \big), \text{ s.t., (\ref{eq:def_DOIC})}. 
\end{equation}
Here, $\tau^{d}$ is one of the anticipated collective responses of the agents, which is payoff relevant to the principal.
Nonetheless, it should be noted that $\tau^{d}$ does not serve as a decision variable within the principal's mechanism design problem.
Rather, it is an outcome that arises from the implementation of the mechanism $\Theta$.
In Section \ref{sec:region_cutting_off_switch}, we will present a class of off-switch functions.
These functions enable the principal to effectively incorporate her desired $\tau^{d}$ into the formulation of the off-switch functions, thereby granting her greater control and influence over the system.

\subsection{Applications}

Our model presents a versatile framework with potential applicability across a broad range of fields that incorporate the ``cancel-anytime" policy. The following are two illustrative examples:

\subsubsection*{Dynamic Subscription Model}
A firm provides services on a monthly subscription basis (e.g., a streaming service provider such as Netflix) in which the actions refer to the services or products offered to the subscribers.
Every subscriber $i$ has a local state $s_{i,t}$ in each month $t$ that represents the subscriber's expected value derived from the service in period $t$.
For example, in the context of streaming services, the state could affect the activity degree of the subscriber (e.g., heavy, moderate, or light viewer).
The state of the subscribers can change from period to period following a Markovian dynamic in an endogenous and exogenous manner.
For example, a heavy viewer (with higher state $s_{i,t}$) may become a moderate viewer (with a decreased state $s_{i,t+1}$) according to $\kappa_{i,t+1}(s_{i,t}, h_{t+1}, \omega_{i,t+1})$, perhaps due to the changes in their free time (captured by the shock $\omega_{i,t+1}$) or the interest in current content offerings (the new action menu).
Also, if a subscriber has been satisfied with the service offered in the past, they are likely to continue their subscription and may even upgrade to a higher plan. On the other hand, if they have been disappointed with the service, they might downgrade their plan or cancel their subscription.
This is due to endogenous factors.
The firm can observe what and how much service a subscriber consumes in each period (i.e., $a_{i,t}$) but the firm does not know the state of the subscriber.
Subscribers with certain states (e.g., high demands) could lead to disproportionate costs for, e.g., content licensing, data bandwidth, and infrastructure, producing expenses, and customer service.
If the revenue gained from these subscribers does not offset the higher costs, the firm may prefer these customers to unsubscribe.

\subsubsection*{Cyber Insurance with Changing Risk Preferences}

A cyber insurance provider sells seasonal insurance policies to businesses.
Each business $i$ has a risk profile that is represented by its state $s_{i,t}$ that affects its perceived need for cyber insurance in a given period $t$.
The risk profile can be influenced by a variety of factors, such as the size of the business, its industry, the existing cybersecurity measures adopted, and its past experience with cyber threats.
The risk profile evolves over multiple periods.
For example, a small business $i$ (with a smaller state $s_{i,t}$) may grow and become a medium-sized business (with increased state $s_{i,t+1}$) according to $\kappa_{i,t+1}(s_{i,t}, h_{t+1}, \omega_{i,t+1})$.
The growth could be due to factors such as an increase in customers, expansion to new markets, or the launch of a successful new product. The evolution of risk profile due to such growths could be interpreted as the exogenous factors, captured by the shock $\omega_{i,t+1}$.
The risk profile could also change over time due to endogenous factors.
Businesses often learn from their past experiences. If they have experienced a cyber incident in the past and their insurance policy helped mitigate the financial impact, they may update their risk profile to acknowledge the value of maintaining an insurance policy, leading to higher perceived risk and an increased willingness to pay for insurance.
Furthermore, the terms and conditions of past policies might influence the business's risk management practices. For example, insurance companies often offer reduced premiums or other incentives to businesses that implement certain cybersecurity measures, such as regular security audits or employee training programs. Over time, these practices can reduce the business's vulnerability to cyber threats, effectively changing its risk profile.
The insurance provider aims to design dynamic cyber insurance policies with the goal of revenue maximization.
However, there may be businesses with a high risk of cyber-attacks or data breaches, due to their industry, poor cybersecurity measures, or frequent past incidents. If these high-risk businesses are not willing to pay premium rates that fully reflect their risk level, it might be more profitable for the insurance provider to lose these customers rather than risk having to pay out expensive claims.

\subsection{Region-Cutting Off-Switch}\label{sec:region_cutting_off_switch}

\begin{figure}
  \centering
    \includegraphics[width=0.7\linewidth]{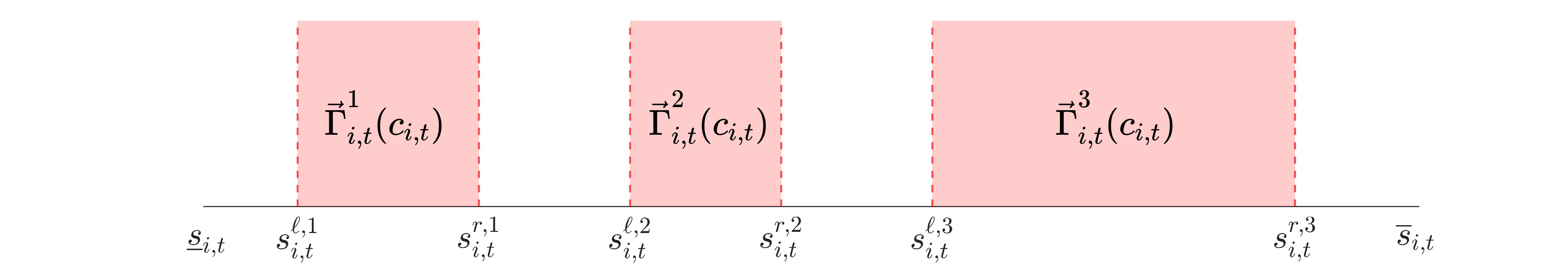}
    \caption{The cutoff-switch function $\phi_{i,t}(\cdot|c_{i,t})$ partitions agent $i$'s period-$t$ state space $S_{i,t}=[\underline{s}_{i,t}, \bar{s}_{i,t}]$ into multiple regions, in which the red-colored regions are sub-OFRs that constitute $S^{\mathtt{off}}_{i,t}$.} \label{fig:cutoff_def}
\end{figure}

In this section, we introduce a class of off-switch functions, known as the \textit{region-cutting off-switch} functions.
For the rest of this section, we assume all agents are obedient in taking the regular actions, all other agents (except agent $i$) are obedient in taking the OM actions, and each $x_{i,t}$ is agent $i$'s conjecture about others' obedient OM actions; unless otherwise stated.

\subsubsection{Off and On Regions}

First, we define the notion of \textit{on-rent} as one \textit{economic rent} (\citet{schoemaker1990strategy}) for each agent due to the agent's autonomy of strategic participation.

\begin{definition}[On-Rent]
    Given an action $a_{i,t}\in A_{i,t}[\sigma]$, the (period-$t$) \textit{on-rent} for each agent $i$ is defined as the difference in the agent's expected payoff-to-go between the selection of $\mathtt{OM}_{i,t}=0$ and $\mathtt{OM}_{i,t}=1$.
    %
    %
    That is, for all $i\in\mathcal{N}$, $t\in\mathbb{T}$, $a_{i,t}\in A_{i,t}[\sigma]$, $s_{i,t}\in S_{i,t}$, $h_{t}\in H_{t}$, $x_{i,t}$,
\begin{equation}\tag{\texttt{OnR}}\label{eq:function_z}
    \begin{aligned}
    Z_{i,t}(a_{i,t}|s_{i,t}, h_{t}, x_{i,t}; \Theta)&\equiv 
    \max_{L\in \mathbb{T}_{t,T}} G_{i,t}(a_{i,t}|s_{i,t}, h_{t}, L, x_{i,t})- \phi_{i,t}(h_{t}|c_{i,t}).
    \end{aligned}
\end{equation}
For simplicity, we write $Z_{i,t}(s_{i,t}, h_{t},x_{i,t}; \Theta)=Z_{i,t}(a_{i,t}|s_{i,t},$  $h_{t}, x_{i,t}; \Theta)$ when the regular action $a_{i,t}=\sigma_{i,t}(s_{i,t}, h_{t})$ is obedient.
    \hfill $\triangle$
\end{definition}

Then, each agent $i$'s expected payoff to-go satisfies
\begin{equation}\label{eq:payofftogo_onrent}
    \Lambda_{i,t}(\mathtt{OM}_{i,t}|s_{i,t}, h_{t}, x_{i,t})=\begin{cases}
    \phi_{i,t}(h_{t}|c_{i,t}),& \textup{ if } Z_{i,t}(s_{i,t}, h_{t},x_{i,t}; \Theta)\leq 0, \\ 
    \max\limits_{L\in \mathbb{T}_{t,T}} G_{i,t}(a_{i,t}|s_{i,t}, h_{t}, L, x_{i,t}), & \textup{ if } Z_{i,t}(s_{i,t}, h_{t},x_{i,t}; \Theta)\geq 0.
\end{cases}
\end{equation}

When the agents are obedient in taking the regular actions, the \textit{strict off region} (strict OFR) in period $t$ is the set of the true states for each agent $i$ in which the period-$t$ on-rent is strictly negative.
That is, given a conjecture $x_{i,t}$, for all $i\in\mathcal{N}$, $t\in\mathbb{T}$, $h_{t}\in H_{t}$,
\begin{equation}
    \begin{aligned}
        &\mathtt{S}^{\dagger}_{i,t}(h_{t}, x_{i,t})\equiv \left\{s_{i,t}\in S_{i,t}: Z_{i,t}(s_{i,t}, h_{t}, x_{i,t}; \Theta)< 0 \right\}.
    \end{aligned}
\end{equation}
When the agents are obedient in taking the regular actions, the \textit{indifference region}, denoted by $\mathtt{Id}^{\dagger}_{i,t}(h_{t}, x_{i,t})\subseteq S_{i,t}$, is defined as the set of true states for each agent $i$ with which the agent's period-$t$ on-rent is zero. 
That is, given a conjecture $x_{i,t}$, for all $i\in \mathcal{N}$, $t\in\mathbb{T}$, $h_{t}\in H_{t}$,
\begin{equation}\label{eq:indifference_region_def}
    \begin{aligned}
        &\mathtt{Id}^{\dagger}_{i,t}(h_{t},x_{i,t})\equiv \left\{s_{i,t}\in S_{i,t}:  Z_{i,t}(a_{i,t}|s_{i,t}, h_{t}, x_{i,t}; \Theta)= 0 \right\}.
    \end{aligned}
\end{equation}
When agent $i$ is obedient in taking the regular actions, if the agent's true state $s_{i,t}\in \mathtt{Id}^{\dagger}_{i,t}(h_{t},x_{i,t})$, then he is indifferent between taking $\mathtt{OM}_{i,t}=0$ and $\mathtt{OM}_{i,t}=1$.

Given the strict OFR $\mathtt{S}^{\dagger}_{i,t}(h_{t},x_{i,t})$ and any sub-region $\mathtt{Id}_{i,t}(h_{t},x_{i,t})\subseteq\mathtt{Id}^{\dagger}_{i,t}(h_{t}, x_{i,t})$ (with the same conjecture $x_{i,t}$), we define the \textit{(weak) OFR} by for all $i\in\mathcal{N}$, $t\in\mathbb{T}$, $h_{t}\in H_{t}$,
\begin{equation}\tag{\textup{OFR}}\label{eq:def_irrational_region}
    \begin{aligned}
    &\mathtt{S}_{i,t}(h_{t}, x_{i,t})\equiv \mathtt{S}^{\dagger}_{i,t}(h_{t},x_{i,t}) \cup \mathtt{Id}_{i,t}(h_{t},x_{i,t}).
    \end{aligned}
\end{equation}
We assume that when the agents are indifferent between $\mathtt{OM}_{i,t}=0$ and $\mathtt{OM}_{i,t}=1$, the tie-breaking is in the principal's favor.
In other words, it is the principal's decision whether agents with the true state within $\mathtt{Id}^{\dagger}_{i,t}(c_{i,t},x_{i,t})$ should take $\mathtt{OM}_{i,t}=0$ or $\mathtt{OM}_{i,t}=1$.
Here, we define the (weak) OFR $\mathtt{S}_{i,t}(h_{t}, x_{i,t})$ by (\ref{eq:def_irrational_region}), in which each agent $i$ with $s_{i,t}$ has the strict incentive to take $\mathtt{OM}_{i,t}=1$ when $s_{i,t}\in\mathtt{S}^{\dagger}_{i,t}(h_{t}, x_{i,t})$ or he is indifferent between $\mathtt{OM}_{i,t}=0$ and $\mathtt{OM}_{i,t}=1$ but the principal's tie-breaking rule makes him to take $\mathtt{OM}_{i,t}=1$ when $s_{i,t}\in\mathtt{Id}_{i,t}(h_{t},x_{i,t})$.
In addition, agent $i$ with a state in $\mathtt{Id}^{\dagger}_{i,t}(h_{t},x_{i,t})\backslash \mathtt{Id}_{i,t}(h_{t},x_{i,t})$ takes $\mathtt{OM}_{i,t}=0$ by the principal's tie-breaking rule.
Thus, for all states in the OFR, each agent $i$ has the incentive to take $\mathtt{OM}_{i,t}=1$ for all $i\in\mathcal{N}$, $t\in\mathbb{T}$.

\subsubsection{Cutoff-Switch Functions}

We formally define the region-cutting off-switch (cutoff-switch) function as follows.
\begin{definition}[Cutoff-Switch Function]\label{def:cutoff_switch_function}
    Let $\tau^{d}=(\tau^{d}_{i,t})$ be the principal's desired OM strategy profile, and let $x^{d}=(x^{d}_{i,t})$ be conjecture profile corresponding to $\tau^{d}$ (i.e., satisfying (\ref{eq:x_coincides_tau})).
    An off-switch function $\phi_{i,t}(\cdot|c_{i,t})$ is a \textit{region-cutting off-switch} \textit{(cutoff-switch)} function if $c_{i,t}$ is an ordered finite sequence of period-$t$ states, i.e., $c_{i,t} = \big(s^{\ell,1}_{i,t}, s^{r,1}_{i,t}, s^{\ell,2}_{i,t},\cdots,s^{\ell,B}_{i,t} , s^{r,B}_{i,t}\big)$ with $\underline{s}_{i,t}\leq s^{\ell,1}_{i,t}\leq s^{r,1}_{i,t}\leq \dots \leq s^{\ell,B}_{i,t} \leq  s^{r,B}_{i,t} \leq \bar{s}_{i,t}$ and $1\leq B<\infty$, such that
\begin{itemize}
    \item[(i)] $S^{\mathtt{off}}_{i,t} \equiv\bigcup_{b\in[B]}\vec{\Gamma}^{b}_{i,t}(c_{i,t})\subseteq S_{i,t}$ is the principal's desired OFR for agent $i$ in period $t$, where each $\vec{\Gamma}^{b}_{i,t}(c_{i,t})\equiv [s^{\ell,b}_{i, t}, s^{r,b}_{i,t}]$, for all $b\in[B]$, and
    \item[(ii)] $Z_{i,t}(s'_{i,t}, h_{t},x^{d}_{i,t};\Theta)=0$, for all $s'_{i,t}\in c_{i,t}$ and $h_{t}\in H_{t}$.
\end{itemize}
Let $\vec{\Gamma}_{i,t}(c_{i,t})\equiv\{\vec{\Gamma}^{b}_{i,t}(c_{i,t})\}_{b\in[B]}$.
We refer to such $c$ as a \textit{boundary profile}.
\hfill $\triangle$
\end{definition}

Fig. \ref{fig:cutoff_def} shows how a cutoff-switch function specifies the $S^{\mathtt{off}}_{i,t}$ of period $t$ for agent $i$, where $\vec{\Gamma}^{b}_{i,t}(c_{i,t})$ denotes the $b$-th sub-off region (sub-OFR), which is a continuous interval $[s^{\ell,b}_{i,t}, s^{r,b}_{i,t}]$.
Let $S^{\mathtt{off}}=\{S^{\mathtt{off}}_{i,t}\}$ be the collection of the principal's desired OFR (OFR${}^{d}$) for the agents.
The principal's desired OM strategy profile $\tau^{d}$ coincides with $S^{\mathtt{off}}$ as follows: for all $i\in\mathcal{N}$, $t\in\mathbb{T}$, $h_{i,t}\in H_{t}$,
\begin{equation}\label{eq:def_switchability_tau}
        \begin{cases}
            \tau^{d}_{i,t}(s_{i,t}, h_{t}) = t, & \textup{ if } s_{i,t}\in S^{\mathtt{off}}_{i,t},\\
            \tau^{d}_{i,t}(s_{i,t}, h_{t})>t,& \textup{ if } s_{i,t}\in S_{i,t}\backslash S^{\mathtt{off}}_{i,t}.
        \end{cases}
\end{equation}
%
%
Suppose that the boundary profile $c_{i,t}$ partitions the state $S_{i,t}$ into $0<B'<\infty$ regions.
Let $\Psi_{i,t}(c_{i,t})\equiv\left\{\Psi^{e}_{i,t}(c_{i,t})\right\}_{e=1}^{B'-B}$ denote the principal-desired \textit{on region} (ONR${}^{d}$) of agent $i$ in period $t$ with each $\Psi^{e}_{i,t}(c_{i,t})\in \Psi_{i,t}(c_{i,t})$ is the $e$-th ONR${}^{d}$ such that for all $i\in\mathcal{N}$, $t\in\mathbb{T}$, we have $\left(\bigcup_{e=1}^{B'-B}\Psi^{e}_{i,t}(c_{i,t})\right)\bigcup S^{\mathtt{off}}_{i,t}=S_{i,t}$.

Fix $\sigma$ and $\rho$, any off-switch function $\phi_{i,t}(\cdot|c_{i,t})$ (not necessarily cutoff-switch) induces an OFR for each agent $i$ in period $t$ (if it exists).
The structure of the cutoff-switch function directly depends on $S^{\mathtt{off}}$ through the parameter $c_{i,t}$. 
The principal aims to find a mechanism such that OFR${}^{d}$ coincides with the actual OFR of the game $\mathcal{G}^{\Theta}$.
That is, the game $\mathcal{G}^{\Theta}$ should induces a conjecture $x_{i,t}$ for each agent $i$, such that for all $i\in\mathcal{N}$, $t\in\mathbb{T}$, $h_{t}\in H_{t}$,
\begin{equation}\label{eq:cutoff_switch_def}
    \begin{aligned}
    \mathtt{S}_{i,t}\left(h_{t},x_{i,t}\right)&=S^{\mathtt{off}}_{i,t},
    \end{aligned}
\end{equation}
where the OM strategy profile $\tau$ corresponding to $x=(x_{i,t})$ satisfies (\ref{eq:def_switchability_tau}) given $S^{\mathtt{off}}$.
The principal aims to incentivize each agent $i$ with state $s_{i,t}\in S^{\mathtt{off}}_{i,t}$ (resp. $s_{i,t}\not\in S^{\mathtt{off}}_{i,t}$) to take $\mathtt{OM}_{i,t}=1$ (resp. $\mathtt{OM}_{i,t}=0$) in every period $t$.
Given the principal's desired $S^{\mathtt{off}}$ (determined by the boundary profile $c$), it is the design of the mechanism $\Theta=\left<\sigma, \rho, \phi\right>$ that needs to guarantee (\ref{eq:cutoff_switch_def}).

The following corollary is straightforward.

\begin{corollary}\label{corollary:SPIR_cutoff}
Suppose that the agents are obedient in taking the regular actions (i.e., the mechanism is RAIC).
Then, the delegation mechanism $\Theta=\left<\sigma, \rho, \phi\right>$ is DOIC if there exists a cutoff-switch profile $\phi(\cdot|c)$ such that \textit{(i)} the game $\mathcal{G}^{\Theta}$ induces an OFR $\mathtt{S}_{i,t}\left(h_{t},x_{i,t}\right)\subseteq \mathtt{Id}^{\dagger}_{i,t}(h_{t}, x_{i,t})$ is an indifference region and \textit{(ii)} (\ref{eq:cutoff_switch_def}) holds for all $i\in\mathcal{N}$, $t\in\mathbb{T}$, $h_{t}\in H_{t}$, in which each $S^{\mathtt{off}}_{i,t}$ is determined by $c_{i,t}$.
\end{corollary}

\proof{Proof.}
The proof of Corollary \ref{corollary:SPIR_cutoff} is straightforward. If $S^{\mathtt{off}}_{i,t}=\mathtt{Id}_{i,t}(c_{i,t})$, for all $i\in\mathcal{N}$, $t\in\mathbb{T}$, then $Z_{i,t}(s_{i,t}, h_{t},x_{i,t}; \Theta)= 0$ for all $s_{i,t}\in S^{\mathtt{off}}_{i,t}$. 
%
From (\ref{eq:def_irrational_region}), we conclude that $Z_{i,t}(s_{i,t}, h_{t},x_{i,t};\Theta)\geq 0$ for all $s_{i,t}\in S_{i,t}$. Hence, the RAIC mechanism is OAIC.
\hfill $\square$
\endproof
Corollary \ref{corollary:SPIR_cutoff} states that a mechanism with cutoff-switch functions in which agents are obedient in taking the regular actions is also obedient in taking the OM actions if the induced OFR is an indifference region.
Since an empty boundary profile $c_{i,t}$ is infeasible for each cutoff-switch $\phi_{i,t}$, any possible choice of $S^{\mathtt{off}}$ is non-empty (at least one state). 
Based on the tie-breaking rule, the principal can incentivize each agent $i$ to take $\mathtt{OM}_{i,t}=0$ for all $i\in\mathcal{N}$, $t\in\mathbb{T}$.

\section{Persistence Transformation}\label{sec:up_persistence_transforms}

\begin{figure}
  \centering
    \includegraphics[width=0.7\linewidth]{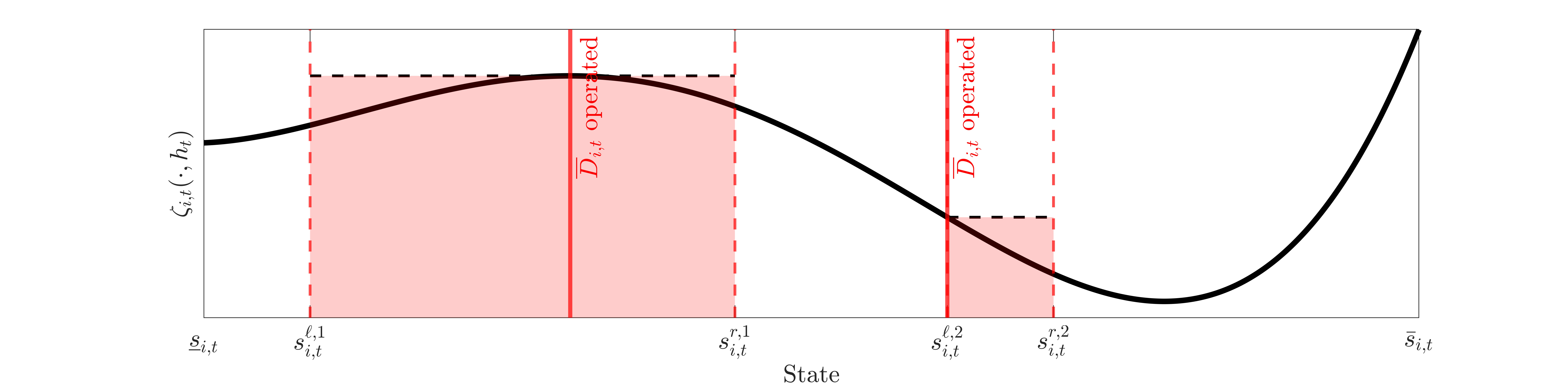}
    \caption{Example of $\overline{d}_{i,t}$, in which $S^{\mathtt{off}}_{i,t}$ is composed of two sub-OFR${}^{d}$s, $\vec{\Gamma}^{1}_{i,t}(c_{i,t})=[s^{\ell,1}_{i,t}, s^{r,1}_{i,t}]$ and $\vec{\Gamma}^{2}_{i,t}(c_{i,t})=[s^{\ell,2}_{i,t}, s^{r,2}_{i,t}]$.} \label{fig:example_up_op}
\end{figure}
%

%
This section introduces an auxiliary procedure known as the persistence transformation, a key step in the construction of cutoff-switch functions in Sections \ref{sec:DOIC_D-IROD} and \ref{sec:switchability_D-SOD}. 
This process integrates multiple \textit{reflecting barriers} (\citet{kruse2015optimal,kruse2018inverse}) into the state-dynamic functions, transforming the original state evolution into a constrained one. The concept of \textit{persistence} (\citet{jarque2010repeated,bohren2018using}) refers to the phenomenon where agents' past actions stochastically influence the distributions of future states. This persistence transformation fundamentally relies on a series of implicit functions associated with the task policy profile, which are known as the carrier functions.

\subsection{Carrier Functions}

Given an obedient conjecture $x_{i,t}$, we define a set of auxiliary functions termed \textit{carrier functions}, $g_{i,t}(\cdot,x_{i,t}|\theta_{i,t}):A_{i,t}[\sigma]\times S_{i,t} \times H_{t} \times \mathbb{T}_{t,T}\mapsto \mathbb{R}$, which is parameterized by $\theta_{i,t}$ and depends on the base game $\mathcal{G}$ and the task policy profile $\sigma$, for all $i\in\mathcal{N}$, $t\in\mathbb{T}$.
Each carrier function is an implicit term that \textit{carries} the task policy profile $\sigma$.
%
%
In the following parts of this section, as well as in Sections \ref{sec:DOIC_D-IROD} and \ref{sec:switchability_D-SOD}, we draw upon the existence of carrier functions as inherent aspects of the task policies. The detailed, closed-form formulation of each carrier function will be further explored in Section \ref{sec:formulation_of_carrier_functions}.

The \textit{maximum carrier} (MC) function is  then defined by, for all $i\in\mathcal{N}$, $t\in\mathbb{T}$,
\begin{equation}\label{eq:def_maximum_carrier}
    \begin{aligned}
        &\mathtt{Mg}_{i,t}(a_{i,t}, s_{i,t}, h_{t}, x_{i,t}|\theta_{i,t})\equiv \max\limits_{L\in\mathbb{T}_{t,T}}g_{i,t}(a_{i,t}, s_{i,t}, h_{t}, L, x_{i,t}|\theta_{i,t}),
    \end{aligned}
\end{equation}
with $\mathtt{Mg}_{i,t}(s_{i,t}, h_{t}, x_{i,t}|\theta_{i,t})=\mathtt{Mg}_{i,t}(a_{i,t}, s_{i,t}, h_{t}, x_{i,t}|\theta_{i,t})$ when $a_{i,t}$ is obedient.
Unless otherwise stated, we will omit the parameter $\theta_{i,t}$ in the carrier and the MC functions for simplicity.

Define the \textit{intertemporal marginal maximum carrier} (marginal carrier), for all $i\in\mathcal{N}$, $s_{i,t}\in S_{i,t}$, $h_{t}\in H_{t}$,
\begin{equation}\label{eq:def_marginal_carrier}
    \begin{aligned}
    &\zeta_{i,t}(s_{i,t}, h_{t}, x_{i,t})\equiv \mathtt{Mg}_{i,t}(s_{i,t}, h_{t},x_{i,t})- \mathbb{E}^{\sigma}\Big[\mathtt{Mg}_{i,t+1}(\tilde{s}_{i,t+1}, \tilde{h}_{t+1})\Big|s_{i,t}, h_{t},x_{i,t}\Big].
    \end{aligned}
\end{equation}
That is, the marginal carrier in period $t$ is the difference between the current period MC and the expected next-period MC.
From the definition of the carrier functions, the marginal carrier $\zeta_{i,t}(s_{i,t}, h_{t})$ depends only on the task policy.
%

\subsection{Up-Persistence Transformation}

Given a profile $c$, we define for all $i\in \mathcal{N}$, $b\in[B]$, 
\begin{equation}\label{eq:operator_d_bar}
    \begin{aligned}
    &\overline{d}_{i,t}(b,h_{t},x_{i,t}|c_{i,t}) \equiv  \sup\left\{ s'_{i,t}\in\argmax\limits_{s''_{i,t} \in \vec{\Gamma}^{b}_{i,t}(c_{i,t}) } \zeta_{i,t}(s''_{i,t}, h_{t}, x_{i,t})  \right\}.
    \end{aligned}
\end{equation}
With reference to Fig. \ref{fig:example_up_op}, $\overline{d}_{i,t}(b,h_{t}, x_{i,t}|$ $c_{i,t})$ (hereafter, $\overline{d}_{i,t}(b)$; unless otherwise stated) projects any state in the $b$th sub-OFR $\vec{\Gamma}^{b}_{i,t}(x_{i,t},c_{i,t})$ to the state in $\vec{\Gamma}^{b}_{i,t}(x_{i,t},c_{i,t})$ that maximizes the marginal carrier $\zeta_{i,t}(s_{i,t}, h_{t},x_{i,t})$ for any $h_{t}\in H_{t}$.
When there are multiple such maximizing states, we choose the largest one.

\begin{figure}
  \centering
    \includegraphics[width=0.7\linewidth]{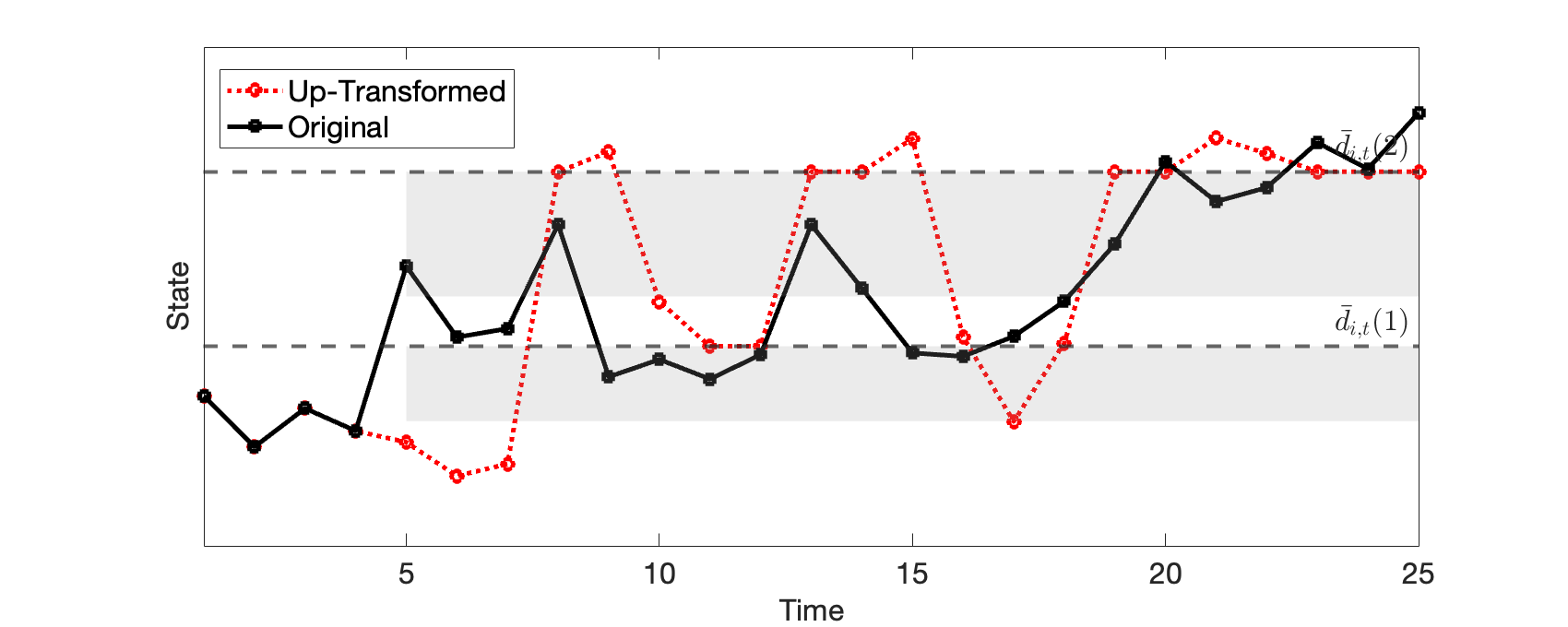}
    \caption{Example of \texttt{utPT} with OFR${}^{d}$ consisting of two continuous intervals (two grey regions).
    The \texttt{utPT} is performed starting from period $5$. The \texttt{utPT}ed state dynamics (red dotted line) never realizes a state within the OFR${}^{d}$ except at $\overline{d}_{i,t}(1)$ or $\overline{d}_{i,t}(2)$. The \texttt{utPT}ed state dynamic shares the same state transition as the original state dynamic (black solid line) when the states are outside of the OFR${}^{d}$.} 
    \label{fig:example_UT_transformed_process}
\end{figure}

Given $\overline{d}_{i,t}$ in (\ref{eq:operator_d_bar}), define the \textit{up transform} (\texttt{uT}) operator for agent $i$ in period $t$ as follows: for all $i\in\mathcal{N}$, $s_{i,t}\in S_{i,t}$, $b\in[B]$,
\begin{equation}\tag{$\mathtt{uT}$}\label{eq:def_up_transform}
    \begin{aligned}
    &\overline{D}_{i,t}\circ s_{i,t}\equiv 
    \begin{cases}
    \overline{d}_{i,t}(b),& \text{ if } s_{i,t}\in \vec{\Gamma}^{b}_{i,t}(x_{i,t},c_{i,t}),\\
    s_{i,t}, &\text{ if } s_{i,t}\not\in S^{\mathtt{off}}_{i,t}.
    \end{cases}
    \end{aligned}
\end{equation}

The \texttt{uT} operator turns each $\vec{\Gamma}^{b}_{i,t}(c_{i,t})$ as a \textit{reflecting interval} in which the corresponding boundaries $s^{b,\ell}_{i,t}$ and $s^{b,r}_{i,t}$ serve a similar purpose as the reflecting barriers (\citet{hamadene2007starting,kruse2015optimal,kruse2018inverse}) that constrains the state dynamics.
Specifically, whenever the state dynamics pass through any boundary to enter $\vec{\Gamma}^{b}_{i,t}(c_{i,t})$, the \texttt{uT} operator projects the state to $\overline{d}_{i,t}(b)$; otherwise, the state remains unchanged.
Given the \texttt{uT} operator, we define the \textit{up-persistence transform} (\texttt{upPT}) operator $\overline{D}^{t,t+1}_{i}$ for each agent $i$ in each period $t$ as follows, for all $s_{i,t}\in S_{i,t}$, $a^{t}\in A^{t}$, $h_{t+1}=(h_{t}, a_{t})\in H_{t}$,
\begin{equation}\label{eq:uppt_marginal}
    \begin{aligned}
    \overline{D}^{t,t+1}_{i}\left[s_{i,t}, h_{t+1}, \omega_{i,t+1}\right]\equiv \overline{D}_{i,t+1}\circ \kappa_{i,t+1}\left(s_{i,t}, h_{t+1}, \omega_{i,t+1}\right),
    \end{aligned}
\end{equation}
with $\overline{D}^{t,t}_{i}[s_{i,t}, \cdot]=s_{i,t}$, for some $\omega_{i,t+1}\in  \Omega$ with $W(\omega_{i,t+1})>0$. 
Again, we omit the actions when agents are obedient; i.e., $\overline{D}^{t,t+1}_{i}[s_{i,t},\omega_{i,t+1}] =  \overline{D}^{t,t+1}_{i}[s_{i,t}, h_{t+1},$  $ \omega_{i,t+1}]$; unless otherwise stated.
Given $\{\overline{D}^{t,t+1}_{i}\}_{t\in \mathbb{T}\backslash\{T\}}$, we can construct $\overline{D}^{t,k}_{i}$, for all $k\in \mathbb{T}_{t+1, T}$, recursively as follows:
\begin{equation}\label{eq:uppt_sequence}
    \begin{aligned}
    \overline{D}^{t,k}_{i}\big[s_{i,t}, \omega^{k}_{i,t+1} \big]\equiv \overline{D}^{k-1, k}_{i}\left[\overline{D}^{t,k-1}_{i}\big[s_{i,t}, \omega^{k-1}_{i,t+1}\big],\omega_{i,k-1}\right],
    \end{aligned}
\end{equation}
where $\omega^{k}_{i,t+1}=(\omega_{i,l})_{l=t+1}^{k}$.
Here, the obedient actions from period $t+1$ to $k$ are those corresponding to the \texttt{upPT}ed states.
Fig. \ref{fig:example_UT_transformed_process} shows an example of the operation performed by $\overline{D}^{t,k}_{i}[\cdot]$.

Define the set of all possible \texttt{upPT}ed states realizations in period $\tau\in\mathbb{T}_{t+1,T}$ when period-$t$ state is $s_{i,t}$ as follows:
\begin{equation}\label{eq:set_uppt}
    \overline{S}_{i,k}[s_{i,t}] \equiv\left\{
    \begin{aligned}
        & s_{i,k}\in S_{i,k}: \exists \omega^{k}_{i,t+1}\in \Omega^{k-t}, \textup{with } W(\omega_{i,l})>0 \textup{ for all } \omega_{i,l}\in \omega^{k}_{i,t+1} , \textup{ s.t. }  \overline{D}^{t,k}_{i}[s_{i,t}; \omega^{k}_{i,t+1} ] = s_{i,k}.
    \end{aligned}
    \right\}
\end{equation}
and $\overline{S}_{i,t}[s_{i,t}] =\{\overline{D}_{i,t}\circ s_{i,t}:s_{i,t}\in S_{i,t}\}$.

\begin{lemma}\label{lemma:upPTed_states_sets}
Under Assumption \ref{assp:full_support}, $\overline{S}_{i,k}[s_{i,t}]=\left\{ \overline{D}_{i,k}\circ s_{i,k}: s_{i,k}\in S_{i,k}\right\}$ for all $i\in\mathcal{N}$, $t\in\mathbb{T}$, $k\in\mathbb{T}_{t+1,T}$, $s_{i,t}\in S_{i,t}$.
\end{lemma}

Lemma \ref{lemma:upPTed_states_sets} implies that under Assumption \ref{assp:full_support} (full support), the set of all possible states in each period $k\in\mathbb{T}_{t+1,T}$ that is \texttt{upPT}ed by period-$t$ state $s_{i,t}$ is independent of $s_{i,t}$ and the history $h_{t}$.

Let $\widetilde{\mathtt{us}}_{i,t}$ represent the \texttt{upPT}ed random variable of period-$t$ state on $S_{i,t}$ with $\overline{F}_{i,t}$ as probability distribution such that $\widetilde{\mathtt{us}}_{i,t+1}\sim \overline{F}_{i,t+1}(s_{i,t}, h_{t+1})$ if $\tilde{s}_{i,t+1}\sim F_{i,t+1}(s_{i,t}, h_{t+1})$.
%
%
Additionally, let $K(\cdot|h_{t}):S_{i,t} \mapsto \mathbb{R}$ be any measurable function of $s_{i,t}\in S_{i,t}$ that is integrable, which may depend on history $h_{t}\in H_{t}$ for each $i\in\mathcal{N}$, $t\in\mathbb{T}$.
Given the \texttt{upPT} operators, we define the \texttt{upPT} \textit{process} for each $k\in\mathbb{T}_{t+1, T}$ as follows:
\begin{equation}\label{eq:upPT_process_def}
    \begin{aligned}
\overline{\mathtt{DE}}^{t,L}\left[\sum\limits_{l=t+1}^{L} K_{i,k}(\widetilde{\mathtt{us}}_{i,l} |\tilde{h}_{k})\Big| s_{i,t}, h_{t},x_{i,t} \right] \equiv \mathbb{E}^{\sigma}\left[\sum\limits_{k=t+1}^{L}K_{i,k}\Big( \overline{D}^{t,k}_{i}\big[s_{i,t}; \tilde{\omega}^{k}_{i,t+1} \big]\Big|\tilde{h}_{k}\Big)  \Big|s_{i,t}, h_{t},x_{i,t} \right],
\end{aligned}
\end{equation}
where the expected history $\Tilde{h}_{k}=\{h_{t}, \Tilde{h}^{k}_{t+1}\}$ is composed of the realized history $h_{t}$ and the random variable $\Tilde{h}^{k}_{t+1}$ consisting expected actions corresponding to the \texttt{upPT}ed states from period $t+1$ to $k$ for $k\in\mathbb{T}_{t+1, L}$.
%
%

\subsection{Essential Region}\label{sec:essential_partition}

Define for all $i\in\mathcal{N}$, $t\in \mathbb{T}\backslash\{T\}$, $s_{i,t+1}\in S_{i,t+1}$, $h_{t+1}=(h_{t}, a_{t})\in H_{t+1}$,
\[
\begin{aligned}
&U_{i,t+1}(s_{i,t+1},s_{i,t}|h_{t+1}, x_{i,t})\equiv\mathtt{Mg}_{i,t+1}(s_{i,t+1}, h_{t+1},x_{i,t})-\mathbb{E}^{\sigma}\left[\mathtt{Mg}_{i,t+1}(\tilde{s}_{i,t+1}, \tilde{h}_{t+1}))\Big|s_{i,t}, h_{t},x_{i,t}\right].
\end{aligned}
\]
Here, $U_{i,t+1}(s_{i,t+1},s_{i,t}|h_{t+1})$ is the deviation of a sampled MC given period-$t+1$ $(s_{i,t+1}, h_{t+1})$ from its expectation given period-$t$ $(s_{i,t}, h_{t})$, in which $h_{t+1}$ contains $h_{t}$.
Similarly, each $U_{i,k}(\mathtt{us}_{i,k}, \mathtt{us}_{i,k-1}|h_{k})$ captures the deviation from a sampled MC given the period-$t+1$ \texttt{upPT}ed state $\mathtt{us}_{i,t+1}$ (and $h_{t+1}$) and its \textit{non-}\texttt{upPT}ed expectation given the period-$t$ $\mathtt{us}_{i,t}$ (and $h_{t}$). Since the agents are assumed to be obedient, each $a_{i,t}=\sigma_{i,t}(\mathtt{us}_{i,t},h_{t})$ in $h_{t+1}=(h_{t},a_{t})$ for all $i\in\mathcal{N}_{t}$. 

Let, for all $i\in\mathcal{N}$, $t\in\mathbb{T}$, $s_{i,t}\in S_{i,t}$, $h_{t}\in H_{t}$,
\begin{equation}\label{eq:uppted_delta_function}
    \begin{aligned}
    \overline{\delta}_{i,t}( s_{i,t}, h_{t}, x_{i,t})\equiv \overline{\mathtt{DE}}^{t,T}\left[\sum\limits_{k=t+1}^{T} U_{i,k}(\widetilde{\mathtt{us}}_{i,k}, \widetilde{\mathtt{us}}_{i,k-1}|\tilde{h}_{k})\Big| s_{i,t}, h_{t}, x_{i,t} \right]\Big|_{\widetilde{\mathtt{us}}_{i,t} = s_{i,t}}.
    \end{aligned}
\end{equation}
If we replace the \texttt{upPT} process in (\ref{eq:uppted_delta_function}) by the standard expectation operator $\mathbb{E}^{\sigma}[\cdot|s_{i,t}, h_{t}, x_{i,t}]$ (i.e., each $\widetilde{\mathtt{us}}_{i,k}=\tilde{s}_{i,k}$ for all $k\in\mathbb{T}_{t+1,T}$), then $\overline{\delta}_{i,t}( s_{i,t}, h_{t}, x_{i,t})=0$ for all $s_{i,t}\in S_{i,t}$, $h_{t}\in H_{t}$.

Denote
\begin{align*}
    \overline{\Xi}_{i,t}(s_{i,t}, s'_{i,t}| h_{t}, x_{i,t}) \equiv \overline{\delta}_{i,t}( s_{i,t}, h_{t},\chi_{i,t})- \overline{\delta}_{i,t}(s'_{i,t}, h_{t}, x_{i,t}),&\\
    \Upsilon_{i,t}(s_{i,t},s'_{i,t} |h_{t}, x_{i,t}) \equiv \mathtt{Mg}_{i,t}(s_{i,t}, h_{t},x_{i,t})- \mathtt{Mg}_{i,t}(s'_{i,t}, h_{t},x_{i,t}),&
\end{align*}
%
for all $i\in\mathcal{N}$, $s_{i,t}, s'_{i,t}\in S_{i,t}$, $h_{t}\in H_{t}$, where $\overline{\Xi}_{i,t}$ depends on the up-persistence transformation while $\Upsilon_{i,t}(s_{i,t},s'_{i,t} |h_{t})$ does not.
For any continuous interval $S'_{i,t}=[\underline{s}'_{i,t}, \overline{s}'_{i,t}]\subseteq S_{i,t}$, let 
\begin{equation}\label{eq:ordered_sequence_boundary}
    r_{i,t}\equiv \big(s^{[1]}_{i,t} = \underline{s}'_{i,t} , s^{[2]}_{i,t},\cdots s^{[B']}_{i,t}, s^{[B'+1]}_{i,t} = \overline{s}'_{i,t}\big),
\end{equation}
denote a boundary profile that partitions $S'_{i,t}$ into $B'$ intervals, for all $i\in\mathcal{N}$, $t\in\mathbb{T}$, $B'\geq 2$.
In particular, $[s^{[k]}_{i,t}, s^{[k+1]}_{i,t})\subseteq S'_{i,t}$ is the $k$-th partitioned interval of $S'_{i,t}$ with $[s^{[B']}_{i,t}, s^{[B'+1]}_{i,t}]$, and $S'_{i,t}=\left(\bigcup_{k=1}^{B'-1}[s^{k}_{i,t}, s^{[k+1]}_{i,t})\right)\bigcup\big[s^{[B']}_{i,t},$ $ s^{[B'+1]}_{i,t}\big]$.

Given a boundary profile $r_{i,t}$, a \textit{partition} of any $S'_{i,t}\subseteq S_{i,t}$ is a finite collection of $1<B'<\infty$ continuous intervals, denoted by $\Pi_{i,t}[r_{i,t}]=\big\{ \Pi^{k}_{i,t}[r_{i,t}]\big\}_{k=1}^{B'}$ where each  $\Pi^{k}_{i,t}[r_{i,t}]=[s^{[k]}_{i,t}, s^{[k+1]}_{i,t})$ for $k<B'$ and $\Pi^{B'}_{i,t}[r_{i,t}]=[s^{[B']}_{i,t}, s^{[B'+1]}_{i,t}]$, such that $S'_{i,t}=\bigcup_{k=1}^{B'} \Pi^{k}_{i,t}[r_{i,t}]$.

\begin{definition}[Essential Region]\label{def:essential_region}
    Fix a base game $\mathcal{G}$ and a task policy profile $\sigma$.
    Suppose that the agents are obedient in taking regular actions.
    Let $r_{i,t}$ be a boundary profile that partitions $S'_{i,t}\subseteq S_{i,t}$ into $\Pi_{i,t}[r_{i,t}]=\{\vec{\Pi}_{i,t}[r_{i,t}], \hat{\Pi}_{i,t}[r_{i,t}]\}$, which are decomposed into two collections where $\vec{\Pi}_{i,t}[r_{i,t}]=\{\vec{\Pi}^{b}_{i,t}[r_{i,t}]\}_{b\in [B]}$ with $1\leq B <B'$ and $\hat{\Pi}_{i,t}[r_{i,t}]=\{\hat{\Pi}^{e}_{i,t}[r_{i,t}]\}_{e\in [B'-B]}$.
    $\vec{\Pi}_{i,t}[r_{i,t}]$ is a collection of \textit{essential intervals} of $S'_{i,t}$ if there exists a point (not necessarily unique) $\vec{s}_{i,t}\in\vec{\Pi}^{b}_{i,t}[r_{i,t}]$ for every $b\in[B]$ such that
    \begin{itemize}
    \item[\textup{(i)}] for all $s'_{i,t}\in \vec{\Pi}^{b}_{i,t}[r_{i,t}]$, $\Upsilon_{i,t}(\vec{s}_{i,t},s'_{i,t} |h_{t},x_{i,t})+\overline{\Xi}_{i,t}(\vec{s}_{i,t}, s'_{i,t}|h_{t},x_{i,t})\geq0$, 
    \item[\textup{(ii)}] for all $s_{i,t}\not\in \bigcup_{b\in[B]}\vec{\Pi}^{b}_{i,t}[r_{i,t}]$, $\Upsilon_{i,t}(\vec{s}_{i,t},s_{i,t} |h_{t}, x_{i,t}) + \overline{\Xi}_{i,t}(\vec{s}_{i,t}, s_{i,t}| h_{t},x_{i,t}) \leq 0$.
\end{itemize}
    We refer to the union of the essential intervals, $\cup_{b\in[B]}\vec{\Pi}^{b}_{i,t}(c_{i,t})$, as \textit{essential region} of $S'_{i,t}$ and the point $\vec{s}_{i,t}$ as the \textit{essential point} of the $b$-th \textit{essential interval} $\vec{\Pi}^{b}_{i,t}[r_{i,t}]$, for all $b\in[B]$.
    \hfill $\triangle$
\end{definition}

Given any partitions of the state spaces, the notion of essential partition places a set of conditions for the base game model $\mathcal{G}$ and the task policy profile $\sigma$ (thus the principal's desired action menus and its time evolution).
From Definition \ref{def:essential_region}, an essential partition depends on the action menu (i.e. $\sigma$) and the base game model $\mathcal{G}$.

\subsection{Characterizing the Essential Region}

Next, we characterize the essential region by showing a set of exemplary conditions of the state-dynamic models and the reward model, under which we can guarantee the existence of the essential partition.
We fix a conjecture profile $x=(x_{i,t})$ in this section.
Let $Y_{i,t}(h_{t},x_{i,t})\subset \mathbb{R}$ denote the codomain of $\zeta_{i,t}(\cdot, h_{t}, x_{i,t})$ given any history $h_{t}\in H_{t}$.
Let $y_{i,t}$ denote a typical value in $Y_{i,t}(h_{t},x_{i,t})$.
Define a subset of states
\begin{equation}
    \begin{aligned}
    \vec{S}_{i,t}[y_{i,t}] \equiv \left\{s_{i,t}\in S_{i,t}: \zeta_{i,t}(s_{i,t},h_{t},x_{i,t}) \leq y_{i,t}\right\}.
    \end{aligned}
\end{equation}
We refer to $\vec{S}_{i,t}[y_{i,t}]$ (which is not necessarily continuous) as the $y_{i,t}$-cut region, in which each state leads to a value of the marginal carrier that is no larger than a value $y_{i,t}\in Y_{i,t}(h_{t}, x_{i,t})$.
Then, for any $s_{i,t}\in \vec{S}_{i,t}[y_{i,t}]$, $s'_{i,t}\in S_{i,t}\backslash \vec{S}_{i,t}[y_{i,t}]$, $h_{t}\in H_{t}$, $i\in \mathcal{N}$, $t\in\mathcal{N}$, we have
\begin{equation}\label{eq:condition_RE-CR}
    \zeta_{i,t}(s_{i,t},h_{t},x_{i,t}) \leq \zeta_{i,t}(s'_{i,t},h_{t},x_{i,t}).
\end{equation}
%
%
A special case is when $\zeta_{i,t}(s_{i,t},h_{t},x_{i,t})$ is monotone in $s_{i,t}$.
Suppose $\zeta_{i,t}(s_{i,t},h_{t},x_{i,t})$ is strictly monotone.
Then, every $y_{i,t}\in Y_{i,t}(h_{t},x_{i,t})$ has a unique corresponding state $s'_{i,t}\in S_{i,t}$ such that $\zeta_{i,t}(s'_{i,t},h_{t},x_{i,t}) = y_{i,t}$.
%
%


\begin{definition}[Dynamic Crossing Region]
Suppose that agents are obedient in taking regular actions.
A region (not necessarily continuous interval) $\mathring{S}_{i,t}\subseteq S_{i,t}$ is a \textit{dynamic crossing region (DCR)} if for any obedient action history $h_{t}\in H_{t}$, $s'_{i,t}\in \mathring{S}_{i,t}$, $s_{i,t}\in S_{i,t}$, $i\in\mathcal{N}$, $t\in\mathbb{T}\backslash\{T\}$, $s_{i,t+1}\in S_{i,t+1}$, $a_{t}\in A_{t}[\sigma]$,
\begin{equation}\label{eq:dcr_conditions}
    F_{i,t}( s_{i,t+1} |  s_{i,t} , a_{t}, h_{t}) \leq F_{i,t}(s_{i,t+1}|  s'_{i,t},a_{t}, h_{t}).
\end{equation}
%
%
\hfill $\triangle$
\end{definition}


Given any DCR $\mathring{S}_{i,t}\subseteq S_{i,t}$, define the set $M_{i,t}(y_{i,t},h_{t}| \mathring{S}_{i,t})$ by for any $y_{i,t}\in Y_{i,t}(h_{t},x_{i,t})$, any $h_{t}\in H_{t}$, $i\in\mathcal{N}$, $t\in\mathbb{T}$, 
\[
M_{i,t}(y_{i,t},h_{t}|\mathring{S}_{i,t})\equiv \vec{S}_{i,t}[y_{i,t}]\cap \mathring{S}_{i,t}.
\]

\begin{definition}[Dominated Region]\label{def:dominated_region}
Suppose that agents are obedient in taking regular actions.
Given a DCR $\mathring{S}_{i,t}$, the \textup{dominated region} of $\sigma$ for agent $i$ in $t$ is $\overrightarrow{M}_{i,t}(h_{t}|\mathring{S}_{i,t}) = M_{i,t}(y^{*}_{i,t},h_{t}|\mathring{S}_{i,t})\subseteq S_{i,t}$, for all $i\in \mathcal{N}$, $t\in\mathbb{T}$, $h_{t}\in H_{t}$, where $y^{*}_{i,t} \in \arg\max_{y_{i,t}\in Y_{i,t}(h_{t},x_{i,t}) } \big| M_{i,t}(y_{i,t},h_{t}|\mathring{S}_{i,t}) \big|.$
%
\hfill $\triangle$
\end{definition}

Given a DCR $\mathring{S}_{i,t}$, a dominated region for agent $i$ in $t$ is the largest intersection of $\mathring{S}_{i,t}$ and the $y_{i,t}$-cut region for $y_{i,t}\in Y_{i,t}(h_{t},x_{i,t})$.

\begin{proposition}\label{prop:essential_region_DR_SD}
Suppose that agents are obedient in taking regular actions.
Given a DCR $\vec{S}_{i,t}$, the dominated region $\overrightarrow{M}_{i,t}(h_{t}|\vec{S}_{i,t})$ is an essential region if
\begin{itemize}
    \item[(i)] $\overrightarrow{M}_{i,t}(h_{t}|\mathring{S}_{i,t}) = \mathring{S}_{i,t}$, and
    \item[(ii)] there exists a state $\vec{s}_{i,t}\in \overrightarrow{M}_{i,t}(h_{t}|\mathring{S}_{i,t})$ such that $\overrightarrow{M}_{i,t}(h_{t}|\mathring{S}_{i,t})\backslash \{\vec{s}_{i,t}\}$ is a dominated region of $\overrightarrow{M}_{i,t}(h_{t}|\mathring{S}_{i,t})$.
\end{itemize}
\end{proposition}



%
\begin{definition}[Monotonic Environment]\label{def:monotone_environment}
In the \textup{monotonic environment}, the task policy profile $\sigma$ and the base game model $\mathcal{G}$ satisfy the following.
%
%
\begin{itemize}
    \item[(i)] $\zeta_{i,t}(s_{i,t},h_{t},\chi_{i,t})$ is non-decreasing (resp. non-increasing) in $s_{i,t}$ for all $i\in\mathcal{N}$, $t\in\mathbb{T}$, $h_{t}\in H_{t}$,
    \item[(ii)] $F_{i,t}(s_{i,t+1}|s_{i,t}, h_{t})$ is non-increasing (resp. non-decreasing) in $s_{i,t}$ for all $i\in\mathcal{N}$, $t\in\mathbb{T}\backslash\{T\}$, $s_{i,t+1}\in S_{i,t+1}$, $h_{t}\in H_{t}$,
\end{itemize}
\hfill $\triangle$
\end{definition}

Part \textit{(ii)} of Definition \ref{def:monotone_environment} is the \textit{first-order stochastic dominance} (FO-SD) condition, which is a special case of (\ref{eq:dcr_conditions}).
In particular, the (\ref{eq:dcr_conditions}) becomes the FO-SD if for every $\hat{s}_{i,t}\in(\underline{s}_{i,t},  \overline{s}_{i,t})$, $\hat{S}_{i,t}[\hat{s}_{i,t}]$ is a DCR.
That is, for every $s_{i,t}\leq s'_{i,t} \in S_{i,t}$ (or $s_{i,t}\geq s'_{i,t} \in S_{i,t}$), we have $F_{i,t}( s_{i,t+1} |  s'_{i,t} , h_{t})\geq  F_{i,t}(s_{i,t+1}|  s_{i,t}, h_{t})$ (or $F_{i,t}( s_{i,t+1} |  s'_{i,t} , h_{t})\leq  F_{i,t}(s_{i,t+1}|  s_{i,t},h_{t})$).
Without loss of generality, we refer to the monotonic environment as when $\zeta_{i,t}(s_{i,t},h_{t})$ is non-decreasing in $s_{i,t}$ and $F_{i,t}(s_{i,t+1}|s_{i,t}, h_{t})$ is non-increasing in $s_{i,t}$.

\begin{corollary}\label{lemma:monotone_dominated_region}
In the monotonic environment, $S'_{i,t}=[\underline{s}_{i,t}, s'_{i,t}]\subseteq S_{i,t}$ is an essential region of $S_{i,t}$ for any $s'_{i,t}\in S_{i,t}$.
\end{corollary}

\section{Dynamic Individual Rational Mechanism}\label{sec:DOIC_D-IROD}

In this section, we consider a class of mechanisms in which each DOIC mechanism is \textit{dynamically individually rational} (DIR).
The notion of DIR is the one, where at the equilibrium, agents choose to participate in every period.
We refer to such mechanism as IR-DOIC.
In classic mechanism design, the common setting is to demand that the agents' expected payoff-to-go with a fixed time horizon $T$ be nonnegative. 
This exhibits myopic decision-making of participation because the agents do not plan future participation decisions (i.e., the OM actions).
In our strategic setting, the principal's IR-DOIC mechanism requires her desired profile $\tau^{d}$ to take $\mathtt{OM}_{i,t}=0$ for all $i\in\mathcal{N}$, $t\in\mathbb{T}$.
That is, for all $i\in\mathcal{N}$, $t\in\mathbb{T}$, $s_{i,t}\in S_{i,t}$, $h_{t}\in H_{t}$, $\tau^{d}_{i,t}(s_{i,t}, h_{t}) = T+1$.
In this section, we consider that each agent $i$ forms conjecture $x_{i,t}$ by assuming all other agents are obedient in taking the OM actions.

The following corollary is straightforward and reformulates the IR-DOIC mechanisms into the \textit{dynamic incentive compatibility for OM actions} (OAIC) and the \textit{dynamic incentive compatibility for regular actions} (RAIC), in terms of the on-rent given by (\ref{eq:function_z}).
\begin{corollary}\label{corollary:OSOD}
A delegation mechanism $\Theta=<\sigma, \rho,\phi>$ with individual rationality is \textup{IR-DOIC} if it is \textup{OAIC} and \textup{RAIC}: for all $i\in \mathcal{N}$, $t\in\mathbb{T}$, $s_{i,t}\in S_{i,t}$, $h_{t}\in H_{t}$, any $a'_{i,t}\in A_{i,t}[\sigma]$,
\begin{align}
    &Z_{i,t}(s_{i,t}, h_{t}, x_{i,t}; \Theta)\geq 0,\tag{$\mathtt{OAIC}$}\label{eq:PIR_off_menu_1}\\
    &Z_{i,t}(s_{i,t}, h_{t}, x_{i,t}; \Theta) \geq Z_{i,t}(a'_{i,t}|s_{i,t}, h_{t}, x_{i,t}; \Theta),\tag{$\mathtt{RAIC}$}\label{eq:DPIC_off_menu_1}
\end{align}
where $Z_{i,t}$ is the on-rent given by (\ref{eq:function_z}).
%
\end{corollary}

Note that in general OAIC (resp. RAIC) cannot imply RAIC (resp. OAIC).
The principal's optimal IR-DOIC mechanism design problem is then given by
\begin{equation}\label{eq:principal_objective_osod}
    \max\limits_{\Theta=<\sigma, \rho, \phi>} Q\big(\gamma^{\Theta}_{\sigma,\tau^{d}}; \mathcal{G}^{\Theta } \big), \text{ s.t., (\ref{eq:PIR_off_menu_1}), (\ref{eq:DPIC_off_menu_1})}.
\end{equation}
Compared with (\ref{eq:principal_DOD_Obj}), the principal's desired OM strategy profile $\tau^{d}$ is captured by the construction of the cutoff-switch profile $\phi(\cdot|c)$ through the boundary profile $c$. The obedience concerning $\tau^{d}$ is incentivized by the OAIC constraint.

\subsection{Payoff-Flow Conservation}\label{sec:characterizations_of_SPIC}

In this section, we characterize the RAIC of IR-DOIC by assuming that (\ref{eq:PIR_off_menu_1}) is satisfied. 
We take the perspective of a typical agent $i$ and assume that \textit{(i)} all other agents are obedient, and \textit{(ii)} agent $i$ makes unilateral one-shot deviations.
For ease of exposition, we use $a_{i,t}=\sigma_{i,t}(s_{i,t}, h_{t})$ and $\hat{a}_{i,t}=\sigma_{i,t}(\hat{s}_{i,t}, h_{t})=\hat{\pi}_{i,t}(s_{i,t}, h_{t})$, respectively, to denote the typical obedient action and any action taken by any policy $\hat{\pi}_{i,t}(s_{i,t}, h_{t})$ from the action menu, when the history is $h_{t}$ and agent $i$'s true state is $s_{i,t}$.
As described by Remark \ref{remark:equivalence_truthful_report}, taking $\hat{a}_{i,t}\in A_{i,t}[\sigma]$ can be interpreted as agent $i$'s pretending to have $\hat{s}_{i,t}\in S_{i,t}$ as his true state while his actual true is $s_{i,t}\in S_{i,t}$.
As shown by (\ref{eq:agent_OM_strategy_plan}), each agent $i$'s OM strategy $\tau_{i,t}$ already includes the plan of future OM actions. 
Hence, $x_{i,t}$ includes agent $i$'s conjecture about other agents' future OM actions.
Thus, each agent $i$ does not predict or estimate future conjectures in each period $t$.
Therefore, we drop $x_{i,k}$ for any $k\in\mathbb{T}_{t+1,T}$ in the notations of functions whose expectation is evaluated in period $t$.

For simplicity, we let, for all $i\in\mathcal{N}$, $t\in\mathbb{T}$, $s_{i,t}\in S_{i,t}$, $h_{t}\in H_{t}$, $L\in\mathbb{T}_{t,T}$,
\[
\begin{aligned}
    \lambda_{i,t}(\hat{a}_{i,t}, s_{i,t}, h_{t}, L, x_{i,t})&\equiv G_{i,t}(\hat{a}_{i,t}|s_{i,t}, h_{t}, L, x_{i,t})-\mathbb{E}^{\sigma}_{\hat{a}_{i,t}}\left[ \phi_{i,L+1}(\tilde{h}_{L+1}|c_{i,L+1})-\rho_{i,t}(\hat{a}_{i,t}, \tilde{a}_{-i,t}, h_{t})\Big|s_{i,t}, h_{t}, x_{i,t}\right].
\end{aligned}
\]
That is, $\lambda_{i,t}(\cdot, L)$ is the prospect function up-to $L$ without \textit{(i)} current-period $\rho_{i,t}$ and \textit{(ii)} period-$L+1$ off-switch function $\phi_{i,L+1}$
When agent $i$ is obedient, we omit the regular action $\hat{a}_{i,t}$.
In addition, let

\begin{equation}\label{eq:feasible_rho_0}
            \begin{aligned}
                &\mathtt{E}z_{i,t}(s_{i,t}, h_{t}, x_{i,t}; \rho_{i,t}) = \mathtt{Mg}_{i,t}(s_{i,t}, h_{t},x_{i,t}) - \mathbb{E}^{\sigma}\Big[ \mathtt{Mg}_{i,t+1}(\tilde{s}_{i,t+1}, \tilde{h}_{t+1})\Big| s_{i,t}, h_{t}, x_{i,t} \Big].
            \end{aligned}
        \end{equation}

Given a base game $\mathcal{G}$, a task policy profile $\sigma$, and a set of carrier functions $g=(g_{i,t})$, define the following set of coupling policy profiles,
%
\begin{equation}\tag{$\mathtt{\mathcal{P}[\sigma,g,\mathcal{G}]}$}\label{eq:feasible_rho_C1_original}
 \mathcal{P}[\sigma, g ,\mathcal{G}]\equiv\left\{\begin{aligned}
    \rho=(\rho_{i,t}):& \mathbb{E}^{F_{-i,t}}\left[ \rho_{i,t}(a_{i,t}, \tilde{a}_{-i,t}, h_{t}) \Big|h_{t}, x_{i,t}\right] =-\mathbb{E}^{F_{-i,t}}\left[u_{i,t}(s_{i,t}, a_{i,t}, \tilde{a}_{-i,t})\Big|h_{t}, x_{i,t}\right]\\
    &+\mathtt{Mg}_{i,t}(s_{i,t}, h_{t},x_{i,t})- \mathbb{E}^{\sigma}\left[ \mathtt{Mg}_{i,t+1}(\tilde{s}_{i,t+1}, \tilde{h}_{t+1})\Big| s_{i,t}, h_{t}, x_{i,t} \right], \forall i\in\mathcal{N}, t\in\mathbb{T}. 
\end{aligned}  \right\}.
\end{equation}

The following theorem shows a sufficient condition for RAIC.
\begin{theorem}\label{thm:payoff_flow_conservation_sufficient} 
    Fix a base game $\mathcal{G}$ and a task policy profile $\sigma$. 
    A mechanism $\Theta=\left<\sigma, \rho, \phi\right>$ is RAIC if there exists a collection of carrier functions, denoted by $g=(g_{i,t})$, such that the followings jointly hold.
    \begin{description}
        \item[\textup{(C1)}\namedlabel{itm:C1}{(C1)}] $\rho\in \textup{\ref{eq:feasible_rho_C1_original}}$.
        %
    \item[\textup{(C2)}\namedlabel{itm:C2}{(C2)}] There exists $\eta_{i,t}(\cdot|c_{i,t}): H_{t}\mapsto \mathbb{R}$ such that for all $i\in\mathcal{N}$, $t\in\mathbb{T}$, $a_{i,t-1}=\sigma_{i,t-1}(s_{i,t-1}, h_{t-1})$, $h_{t}\in H_{t}$, the off-switch function satisfies 
    \begin{equation}\label{eq:thm_payoff_conservation_phi}
    \begin{aligned}
    \phi_{i,t}\left(h_{t}|c_{i,t}\right) + \mathtt{E}z_{i,t-1}\left(s_{i,t-1}, h_{t-1}, x_{i,t-1}; \rho_{i,t-1}\right)=\eta_{i,t}\left(h_{t}\middle|c_{i,t}\right) + g_{i,t-1}\left(s_{i,t-1},h_{t},L, x_{i,t}\right).
    \end{aligned}
    \end{equation}
\item[\textup{(C3)}\namedlabel{itm:C3}{(C3)}] For all $i\in\mathcal{N}$, $t\in\mathbb{T}$, $s_{i,t}, \hat{s}_{i,t}\in S_{i,t}$, $h_{t}\in H_{t}$, and for all $\eta=(\eta_{i,t})$ satisfying (\ref{eq:thm_payoff_conservation_phi}),
    \begin{equation}\label{eq:thm_payoff_conservation_lambda_2}
        \begin{aligned}
            &\max\limits_{L\in \mathbb{T}_{t,T}}\left( g_{i,t}\left(\hat{a}_{i,t},\hat{s}_{i,t}, L, x_{i,t}\right) -  g_{i,t}\left(a_{i,t}, s_{i,t}, L, x_{i,t}\right)\right)\\
            &\leq \min\limits_{ L\in\mathbb{T}_{t,T}} \left( \lambda_{i,t}\left( \hat{s}_{i,t}, h_{t}, L, x_{i,t}\right)-\lambda_{i,t}\left(\hat{a}_{i,t}, s_{i,t}, h_{t}, L, x_{i,t}\right)-\eta_{i,L+1}\left(h_{L+1}|c_{i,L+1}\right)\right).
        \end{aligned}
    \end{equation}
    \end{description}
\end{theorem}

Theorem \ref{thm:payoff_flow_conservation_sufficient} delineates a sufficient condition for the RAIC of IR-DOIC.
The condition is anchored on what we denote as the \textit{payoff-flow conservation}, a principle encapsulated within Theorem \ref{thm:payoff_flow_conservation_sufficient}.
In particular, \ref{itm:C1} requires that each coupling policy should satisfy a configuration as articulated in \ref{eq:feasible_rho_C1_original}.
The condition \ref{itm:C1} can be construed as a directive that insists on the conservation of the expected immediate payoff for each agent $i$ in a single period due to the unobservability of others' states.
This directive is formulated within the constraints of task policies and the carrier functions.
In condition \ref{itm:C2}, the left-hand side (LHS) of (\ref{eq:thm_payoff_conservation_phi}) is defined as the sum of the off-switch value and the preceding period's utility, provided that agent $i$ takes $\mathtt{OM}_{i,t}=1$ in period $t$.
To ensure the conservation of the LHS of (\ref{eq:thm_payoff_conservation_phi}), we mandate the following: \textit{(i)} an input from agent $i$'s previous-period action captured by the carrier function $g_{i,t-1}(\cdot, L)$ for any $L\in\mathbb{T}_{t-1, T}$, and \textit{(ii)} a value specific to period $t$, encapsulated by $\eta_{i,t}(h_{t}|c_{i,t})$, which is independent of period-$t$ actions and states.
The term $\eta_{i,t}(h_{t}|c_{i,t})$ is a posted factor.
It is contingent upon the current history and embodies the parameters $c_{i,t}$ of the off-switch function $\phi_{i,t}(h_{t}|c_{i,t})$.
Here, (\ref{eq:thm_payoff_conservation_phi}) outlines the implicit construction of the off-switch function, which is an expression that depends on the task policies and is conveyed through the medium of the utility functions, the coupling policy, the carrier functions, and the posted factor.
The sufficient conditions \ref{itm:C1}-\ref{itm:C3} establish implicit relationships between $\sigma$, $\rho$, and $\phi$.
Provided we can derive $\eta_{i,t}(\cdot|c_{i,t})$ and definitively determine the carrier functions in closed-form conforming to (\ref{eq:thm_payoff_conservation_lambda_2}), then condition \ref{itm:C1} can provide guidance on the formulation of each coupling policy $\rho_{i,t}$.
Indeed, \ref{itm:C1} proposes a condition on how the expected coupling policy $\mathtt{E}\rho_{i,t}(a_{i,t}, h_{t}, x_{i,t})=\mathbb{E}^{\sigma}[\rho_{i,t}(a_{i,t}, \tilde{a}_{-i,t}, h_{t})|s_{i,t}, h_{t}, x_{i,t}]$ is specified in terms of the task policies.
As it depends on other agents' obedient policies, the expected coupling policy $\mathtt{E}\rho_{i,t}(a_{i,t}, h_{t}, x_{i,t})$ emerges as a viable coupling policy that complies with condition \ref{itm:C1}.
With the establishment of each $\rho_{i,t}$ and the precise determination of each $\eta_{i,t}(\cdot|c_{i,t})$, the off-switch function $\phi_{i,t}(\cdot|c_{i,t})$ can be conceived in line with (\ref{eq:thm_payoff_conservation_phi}), which depends solely on the task policies.
As such, condition \ref{itm:C3} evolves into a constraint of the task policy (and the carrier functions) that ensures the RAIC.

%
%

\subsection{Sufficient Conditions for IR-DOIC}

%
Given a boundary profile $c$, construct each cutoff-switch $\phi_{i,t}(\cdot;b|c_{i,t})$ for each $b\in[B]$ by
\begin{equation}\label{eq:cutoff_switch_horizontal}
    \begin{aligned}
    \phi_{i,t}\left(b, h_{t}, x_{i,t}\middle|c_{i,t}\right) = \mathtt{Mg}_{i,t}\left(\overline{d}_{i,t}(b), h_{t},x_{i,t}\right)+ \overline{\delta}_{i,t}\left(\overline{d}_{i,t}(b), h_{t}, x_{i,t}\right),
    \end{aligned}
\end{equation}
where $\mathtt{Mg}_{i,t}(\overline{d}_{i,t}(b), h_{t}, x_{i,t})$ is the maximum carrier at $\overline{d}_{i,t}(b)$ and $\overline{\delta}_{i,t}(\overline{d}_{i,t}(b), h_{t},x_{i,t})$ is given by (\ref{eq:uppted_delta_function}).

\begin{theorem}\label{thm:SPIR_conditions}
    %
    %
    Suppose that the following holds.
    \begin{itemize}
    \item[(i)] A boundary profile $c$ partitions each $S_{i,t}$ into $\Pi_{i,t}(c_{i,t})=\{\vec{\Gamma}_{i,t}(c_{i,t}), \Psi_{i,t}(c_{i,t})\}$, where the collection $\vec{\Gamma}_{i,t}(c_{i,t})$ constitutes $S^{\mathtt{off}}_{i,t}=\cup_{b\in[B]}\vec{\Gamma}^{b}_{i,t}(c_{i,t})$ and $\Psi_{i,t}(c_{i,t})$ is the collection of ONR${}^{d}$s.
    
    %
    \item[(ii)] Given the boundary profile $c$, each cutoff-switch $\phi_{i,t}(\cdot;b|c_{i,t})$ is given by (\ref{eq:cutoff_switch_horizontal}) for all $b\in[B]$.
        \item[(iii)] There exist carrier functions such that a mechanism $\Theta=\left<\sigma,\rho,\phi\right>$ satisfies \ref{itm:C1}-\ref{itm:C3}.
    \end{itemize}
    Then, the mechanism is IR-DOIC, if and only if, each $S^{\mathtt{off}}_{i,t}$ is an \textup{indifference and essential region} of $S_{i,t}$ with $\{\overline{d}_{i,t}(b)\}_{b\in[B]}$ as the \textup{essential points}.
    In addition, $\phi_{i,t}(b, h_{t}, x_{i,t}|c_{i,t}) = \phi_{i,t}(b',h_{t}, x_{i,t}|c_{i,t})$ for all $b,b'\in[B]$, $i\in\mathcal{N}$, $t\in\mathbb{T}$. 
\end{theorem}

Theorem \ref{thm:SPIR_conditions} provides sufficient conditions for the IR-DOIC.
When each cutoff-switch function is expressed in closed form by (\ref{eq:cutoff_switch_horizontal}), Theorem \ref{thm:SPIR_conditions} demonstrates necessary and sufficient conditions under which a delegation mechanism satisfying \ref{itm:C1}-\ref{itm:C3} concurrently satisfies OAIC (and hence IR-DOIC).
When agent $i$ is obedient in taking $a_{i,t}=\sigma_{i,t}(s_{i,t}, h_{t})$, \ref{itm:C1} necessitates that the marginal carrier $\zeta_{i,t}$\textemdash as explicitly defined by (\ref{eq:def_marginal_carrier})\textemdash congruently aligns with each agent $i$'s expected single-step utility.
Notwithstanding, the marginal carrier $\zeta_{i,t}$ remains decoupled from the coupling functions, and the congruence is meticulously established via the payoff-flow conservation condition imposed on the carrier functions and the task policies within the framework of the base game $\mathcal{G}$.
When the cutoff-switch function satisfies (\ref{eq:cutoff_switch_horizontal}), the sufficiency of IR-DOIC requires that each $S^{\mathtt{off}}_{i,t}$ is an indifference region, consequently leading to an on-rent of zero for agents with true states within $S^{\mathtt{off}}$ who are obedient in taking the regular action.
Furthermore, Theorem \ref{thm:SPIR_conditions} validates that under the conditions of payoff-flow conservation outlined in Theorem \ref{thm:payoff_flow_conservation_sufficient}, the cutoff switch $\phi_{i,t}(h_{t};b|c_{i,t})$\textemdash that adheres to (\ref{eq:cutoff_switch_horizontal})\textemdash for each sub-OFR${}^{d}$ is the same for all $b\in[B]$ given any $h_{t}\in H_{t}$.

By dropping the index $b$, we construct the cutoff switch function as follows: for all  $i\in\mathcal{N}$, $t\in\mathbb{T}$, $h_{t}\in H_{t}$,
\begin{equation}\label{eq:off_switch_monotone}
    \phi_{i,t}(h_{i,t},x_{i,t}|c_{i,t}) = \mathtt{Mg}_{i,t}(\underline{s}_{i,t},h_{t},x_{i,t}) +\overline{\delta}_{i,t}(\underline{s}_{i,t}, h_{t}, x_{i,t}), \textup{ where } c_{i,t}=\left\{\underline{s}_{i,t}\right\},
\end{equation}
where $\underline{s}_{i,t}$ is the smallest state in each $S_{i,t}$.
Hence, the corresponding OFR${}^{d}$ contains one state; i.e., $S^{\mathtt{off}}_{i,t}=\{\underline{s}_{i,t}\}$, for all $i\in\mathcal{N}$, $t\in\mathbb{T}$.
The following corollary follows Corollary \ref{lemma:monotone_dominated_region} and Theorem \ref{thm:SPIR_conditions}.
\begin{corollary}\label{corollary:monotone_SPIR}
Suppose that there exists a collection of carrier functions such that the mechanism $\Theta=\left<\sigma,\rho,\phi\right>$ satisfies \ref{itm:C1}-\ref{itm:C3}.
Then, in the \textup{monotone environment (Definition \ref{def:monotone_environment})}, the mechanism $\Theta>$ with each $\phi_{i,t}$ given by (\ref{eq:off_switch_monotone}) is OAIC.
\end{corollary}

\proof{Proof.}
The proof of Corollary \ref{corollary:monotone_SPIR} is straightforward.
Specifically, Corollary \ref{lemma:monotone_dominated_region} implies that the region $S^{\mathtt{off}}_{i,t}=\vec{\Gamma}_{i,t}(\{\underline{s}_{i,t}\})$ is an single-state essential region of $S_{i,t}$ where $\underline{s}_{i,t}$ is the corresponding essential point.
Thus, Theorem \ref{thm:SPIR_conditions} implies that the RAIC mechanism with each $\phi_{i,t}$ given by (\ref{eq:off_switch_monotone}) is OAIC.
%
%
\endproof

The following proposition leverages the conditions of Theorems \ref{thm:payoff_flow_conservation_sufficient} and \ref{thm:SPIR_conditions}.

\begin{proposition}\label{thm:doic_irod_sufficient_condition}
    Fix a base game $\mathcal{G}$, task policy profile $\sigma$, and a collection of carrier functions $g$.
    Suppose that each cutoff-switch function $\phi_{i,t}$ is given by (\ref{eq:cutoff_switch_horizontal}).
    Then, for any $\rho\in \textup{\ref{eq:feasible_rho_C1_original}}$, the mechanism $\Theta=\left<\sigma, \rho, \phi\right>$ is IR-DOIC if and only if
\begin{equation}\label{eq:sufficient_if_rho_c1}
        \begin{aligned}
            \mathtt{Mg}_{i,t}\left(\overline{D}_{i,t}\circ s_{i,t}, h_{t},x_{i,t}\right) + \overline{\delta}_{i,t}\left(\overline{D}_{i,t}\circ s_{i,t}, h_{t},x_{i,t}\right)\geq \mathtt{Mg}_{i,t}\left(a'_{i,t} , s_{i,t}, h_{t},x_{i,t}\right) + \overline{\delta}_{i,t}\left(a'_{i,t}, s_{i,t}, h_{t},x_{i,t}\right),
        \end{aligned}
    \end{equation}
    for all $i\in\mathcal{N}$, $t\in\mathbb{T}$, $s_{i,t}\in S_{i,t}$, $h_{t}\in H_{t}$, $a'_{i,t}\in A_{i,t}[\sigma]$.
\end{proposition}

%
Given a collection of carrier functions, Proposition \ref{thm:doic_irod_sufficient_condition} establishes sufficient conditions for IR-DOIC if $\rho$ and $\phi$ are constructed according to \ref{eq:feasible_rho_C1_original} and (\ref{eq:cutoff_switch_horizontal}), respectively.
The condition (\ref{eq:sufficient_if_rho_c1}) imposes constraints on the task policy profile $\sigma$ (explicitly and implicitly through the carrier functions).
The study of the closed-form formulation of each carrier function is deferred to Section \ref{sec:characterization_DOIC}.

\section{Switchability}\label{sec:switchability_D-SOD}

In this section, we depart from IR-DOIC and move to the general DOIC when agents are obedient in the sense of (\ref{eq:def_DOIC}), in which the principal desires $\mathbb{P}(\mathtt{OM}_{i,t}=1)>0$ for some $i\in\mathcal{N}$, $t\in\mathbb{T}$.
To distinguish from individual rationality, we use the notion of \textit{switchability} or \textit{$S^{\mathtt{off}}$-switchability} to capture the mechanism's capability of incentivizing the agents to take obedient OM actions from $\{0,1\}$, leading to the actual OFR that coincides with the principal-desired OFR OFR${}^{d}$; i.e., 
\[
\mathtt{S}_{i,t}\left(h_{t}, x_{i,t}\right) = S^{\mathtt{off}}_{i,t}, \forall i\in \mathcal{N}, t\in\mathbb{T}, h_{t}\in H_{t},
\]
where the OM strategy profile $\tau$ corresponding to $x=(x_{i,t})$ satisfies (\ref{eq:def_switchability_tau}) given $S^{\mathtt{off}}$.
%
We refer to the corresponding DOIC as $S^{\mathtt{off}}$-DOIC.

Throughout the remainder of this article, we adopt the default configuration where a boundary profile $c$ partitions each $S_{i,t}$ into $B'$ intervals, and the principal's OFR${}^d$ has $B$ sub-OFR${}^d$s and the rest $B'-B$ constitute the ONR${}^{d}$.
Recall that $S^{\mathtt{off}}=\bigcup_{b\in[B]}\vec{\Gamma}^{b}_{i,t}(c_{i,t})$, where $\vec{\Gamma}^{b}_{i,t}(c_{i,t})$ is the $b$-th sub-OFR${}^{b}$.
Let $\vec{\Gamma}_{i,t}(c_{i,t})=\{\vec{\Gamma}^{b}_{i,t}(c_{i,t})\}_{b\in[B]}$, and let $\Psi_{i,t}(c_{i,t})=\{\Psi^{e}_{i,t}(c_{i,t})\}_{e\in[B'-B]}$ denote the collection of sub-ONR${}^{d}$, where each $\Psi^{e}_{i,t}(c_{i,t})$ is the $e$-th sub-ONR${}^{d}$.
Hence, each boundary profile $c_{i,t}$ leads to a partition $\{\vec{\Gamma}_{i,t}(c_{i,t}), \Psi_{i,t}(c_{i,t})\}$ of each $S_{i,t}$.
In addition, we assume that each agent $i$ forms conjecture $x_{i,t}$ by considering all others are obedient.

\subsection{Horizontal Cutoff and Knowledgeable Principal}

In an imperfect information environment, the principal faces various challenges when designing off-switch functions for agents. There are two main constraints. First, the principal cannot rely on the stats that are privately realized at the beginning of every period for each agent. 
This limits the principal's ability to tailor off-switch functions to determine off-switch values according to the specific states.
Second, the ordering of OM and regular actions within each period prevents the principal from learning the true states of agents by observing their regular actions.
The OM actions occur before the agents' regular actions, meaning that the principal cannot observe agent behavior to deduce their states before deciding the off-switch value. 
This further complicates the principal's task, as she must make decisions about the off-switch value on limited information.
Consequently, the principal must design off-switch functions that work effectively across a range of possible states, without knowing the exact realized states.

To cope with the unobservability of the agents' private states, we consider two special settings: \textit{horizontal cutoff-switch} and \textit{knowledgeable principal}.

\begin{definition}[Horizontal Cutoff-Switch]\label{def:horizontal_cutoff_switch}
Suppose that a boundary profile $c$ leads to a partition $\{\vec{\Gamma}_{i,t}(c_{i,t}), \Psi_{i,t}(c_{i,t})\}$ of each $S_{i,t}$ for all $i\in\mathcal{N}$, $t\in\mathbb{T}$.
A cutoff-switch function $\phi_{i,t}(\cdot|c_{i,t})$ is a \textit{horizontal cutoff-switch} if for all $b, b'\in[B]$, $w,w'\in\{\ell,r\}$, and $\hat{s}_{i,t}\in \bigcup_{e\in[B'-B]}\Psi^{e}_{i,t}(c_{i,t})$,
\begin{itemize}
    \item[(i)]  $\zeta_{i,t}\left(s^{w,b}_{i,t}, h_{t},x_{i,t}\right) \leq \zeta_{i,t}\left(\hat{s}_{i,t}, h_{t},x_{i,t}\right)$, and
    \item[(ii)] $\zeta_{i,t}\left(s^{w,b}_{i,t}, h_{t},x_{i,t}\right) =\zeta_{i,t}\left(\overline{d}_{i,t}(b), h_{t},x_{i,t}\right) = \zeta_{i,t}\left(s^{w',b'}_{i,t},h_{t},x_{i,t}\right) =\zeta_{i,t}\left(\overline{d}_{i,t}(b'), h_{t},x_{i,t}\right)$, for all $s^{w,b}_{i,t}, s^{w',b'}_{i,t}$ $\neq \underline{s}_{i,t}$ or $\overline{s}_{i,t}$.
\end{itemize}
\hfill $\triangle$
\end{definition}

\begin{figure}
    \centering
    \begin{subfigure}[b]{\linewidth}
    \centering
    \includegraphics[width=0.7\linewidth]{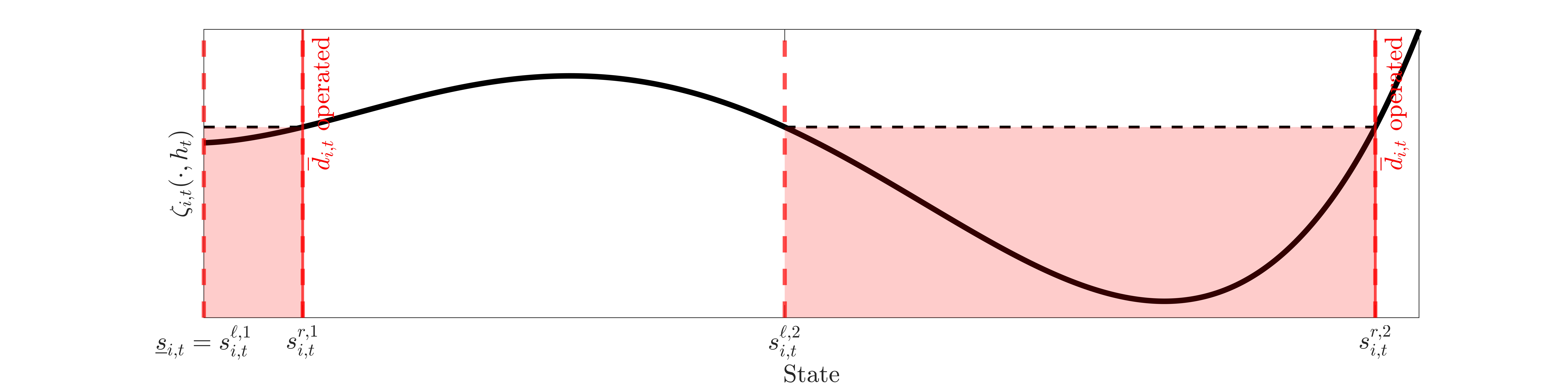}
    \caption{ \label{fig:horizontal_cut_eg} }
    \end{subfigure}
    \begin{subfigure}[b]{\linewidth}
    \centering
    \includegraphics[width=0.7\linewidth]{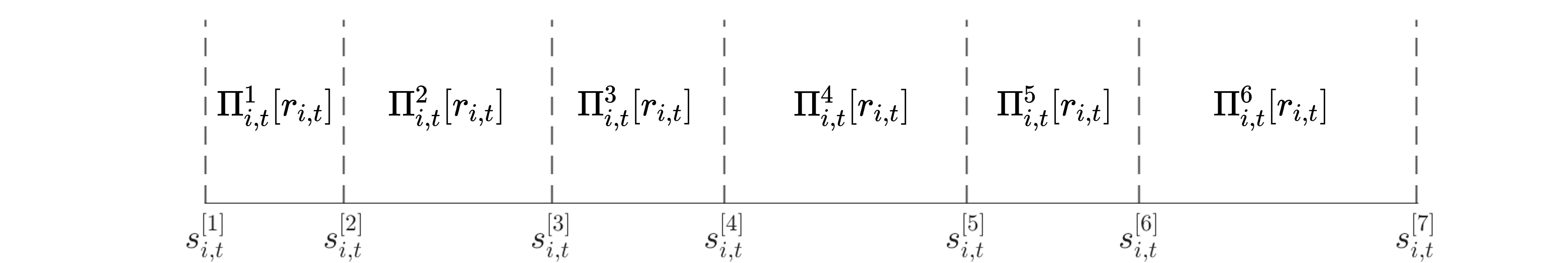}
    \caption{\label{fig:knowledge_state}}
    \end{subfigure}
    \caption{ (a): Example of the horizontal cut-off switch, in which $S^{\mathtt{off}}_{i,t}$ is composed of two sub-OFR${}^{d}$s, $[s^{\ell,1}_{i,t}, s^{r,1}_{i,t}]$ and $[s^{\ell,2}_{i,t}, s^{r,2}_{i,t}]$. (b): Example of a partition of $S_{i,t}$.
    }
\end{figure}

With the horizontal cutoff-switch (function) $\phi_{i,t}(\cdot|c_{i,t})$, each boundary $s^{w,b}_{i,t}$ for every $b\in[B]$ leads to the same value of marginal carrier and this value is the maximum one that can be achieved by all states in $S^{\mathtt{off}}_{i,t}$ (see, Fig. \ref{fig:horizontal_cut_eg}). 
In addition, this maximum marginal carrier value of the OFR${}^{d}$ is no greater than the minimum marginal carrier value that all states in the ONR${}^{d}$, $\bigcup_{e\in[B'-B]}\Psi^{e}_{i,t}(c_{i,t})$, can achieve.

With reference to Fig. \ref{fig:knowledge_state}, let $r=\{r_{i,t}\}$ be a collection of boundary profiles where each $r_{i,t}$ is given by (\ref{eq:ordered_sequence_boundary}) for $S'_{i,t}=S_{i,t}$.
Each $r_{i,t}$ leads to a partition $\Pi_{i,t}[r_{i,t}]$ of $S_{i,t}$ for all $i\in\mathcal{N}$, $t\in\mathbb{T}$.
Let $\Pi[r]=\{\Pi_{i,t}[r_{i,t}]\}$ denote the partition profile.
With abuse of notation, we denote $c_{i,t}\subseteq r_{i,t}$ if $\vec{\Gamma}_{i,t}(c_{i,t})$ is a collections of intervals chosen from $\Pi_{i,t}[r_{i,t}]=\{\Pi^{w}_{i,t}[r_{i,t}]\}_{w\in[B']}$.
That is, each $\vec{\Gamma}^{b}_{i,t}(c_{i,t})\in \Pi_{i,t}[r_{i,t}]$ for all $b\in [B]$, $i\in\mathcal{N}$, $t\in\mathbb{T}$.

\begin{definition}[Knowledgeable Principal]\label{def:knowledgeable_principal}
We say that the principal is \textit{knowledgeable} with $\Pi[r]$ if at the beginning of each period $t$ the principal knows which partitioned interval of $\Pi_{i,t}[r_{i,t}]$ where each agent $i$'s state is located for all $i\in\mathcal{N}$, $t\in\mathbb{T}$.
\hfill $\triangle$
\end{definition}

A knowledgeable principal still does not observe the realizations of the agents' true states.
However, the principal is \textit{more informed} about the agents' private states by knowing which partitioned intervals the agents' states are located at the beginning of each period.

\begin{remark}
    The knowledgeable principal can form beliefs (or knowledgeable beliefs) about the agents' contemporaneous states that contain more information about each agent's state realization than $F_{t}=\{F_{i,t}\}_{i\in\mathcal{N}}$ based on the public history.
    By leveraging the knowledgeable beliefs, our results for RAIC remain unchanged.
    Here, we omit the repetitive characterizations of RAIC when the principal is knowedgeable.
    \hfill $\triangle$
\end{remark}

\subsection{Jump Transformation}\label{sec:jump-transform}

In this section, we define the \textit{jump transform} operator by extending the up transform operator defined in (\ref{eq:def_up_transform}).
Let a boundary profile $r=\{r_{i,t}\}$ lead to a partition $\{\vec{\Gamma}_{i,t}(c_{i,t}), \Psi_{i,t}(c_{i,t})\}$ of $S_{i,t}$, where $c_{i,t}\subseteq r_{i,t}$ for all $i\in\mathcal{N}$, $t\in\mathbb{T}$.
We define an operator for all $e\in[B'-B]$,
\[
    \begin{aligned}
    &\underline{d}_{i,t}\left(e,h_{t}\middle|c_{i,t}\right) \equiv  \sup\left\{ s'_{i,t}\in\argmin\limits_{s''_{i,t} \in \Psi^{e}_{i,t}[c_{i,t}] } \zeta^{e}_{i,t}\left(s_{i,t}, h_{t}, x_{i,t}\right)  \right\}.
    \end{aligned}
\]

\begin{figure}
  \centering
    \includegraphics[width=0.7\linewidth]{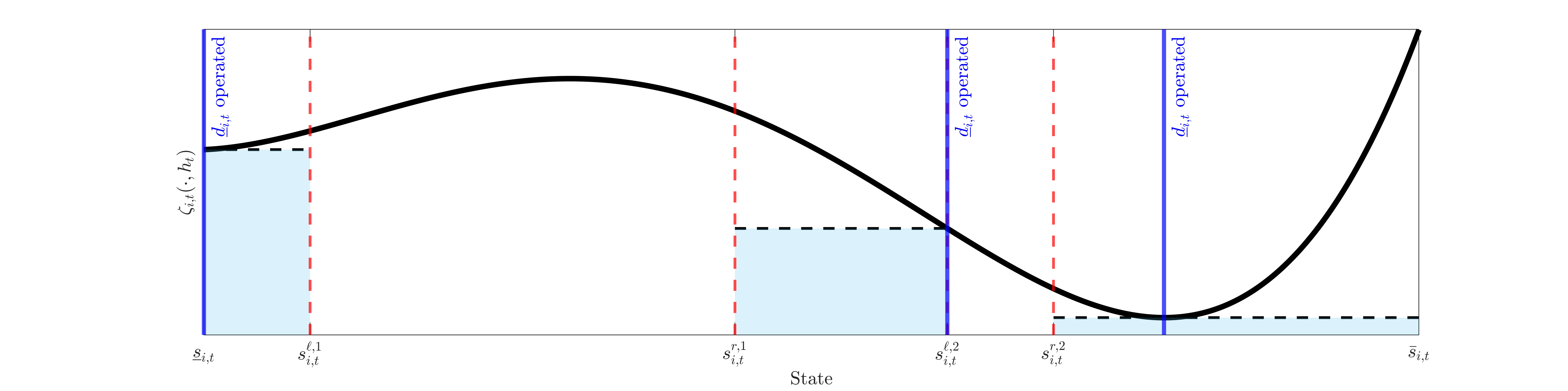}
    \caption{Example of $\underline{d}_{i,t}(e,h_{t}|c_{i,t})$, in which the state space $S_{i,t}$ is partitioned into $6$ regions. There are three sub-ONR${}^{d}$s $\Psi^{1}_{i,t}(c_{i,t})=[\underline{s}_{i,t}, s^{\ell, 1}_{i,t}]$, $\Psi^{2}_{i,t}(c_{i,t})=[s^{r, 1}_{i,t}, s^{\ell, 2}_{i,t}]$, and $\Psi^{2}_{i,t}(c_{i,t})=[s^{r, 1}_{i,t}, s^{\ell, 2}_{i,t}]$. } \label{fig:example_down_op}
\end{figure}

With reference to Fig. \ref{fig:example_down_op}, 
the operator $\underline{d}_{i,t}(e)=\underline{d}_{i,t}(e,h_{t},x_{i,t}|c_{i,t})$ identifies the state $s'_{i,t}$ in the $e$th sub-ONR${}^{d}$ $\Psi^{e}_{i,t}(c_{i,t})$ that minimizes the one-stage value of the marginal carrier $\zeta_{i,t}(s_{i,t}, h_{t}, x_{i,t})$.
When there are multiple maximizers, we choose the largest state.
Then, the the \textit{jump transform} (\texttt{jT}) operator for agent $i$ in period $t$ is defined by, for all $i\in\mathcal{N}$, $s_{i,t}\in S_{i,t}$, $b\in[B]$, $e\in[B'-B]$,
\begin{equation}\tag{$\mathtt{jT}$}\label{eq:region_jump_operator}
    \begin{aligned}
    &\overline{\underline{D}}_{i,t}\circ s_{i,t} \equiv 
    \begin{cases}
    \overline{d}_{i,t}(b),& \text{ if } s_{i,t}\in \vec{\Gamma}^{b}_{i,t}(c_{i,t}),\\
    \underline{d}_{i,t}(e),& \text{ if } s_{i,t}\in \Psi^{e}_{i,t}(c_{i,t}).
    \end{cases}
    \end{aligned}
\end{equation}

\subsection{$S^{\mathtt{off}}$-DOIC}

Recall that when agent $i$'s state is $s_{i,t}$, history is $h_{t}$, and his conjecture about others' OM actions is $x_{i,t}$, the continuing outcome perceived by agent $i$ is $\gamma^{\sigma}[s_{i,t}, h_{t}, x_{i,t}]$ (by assuming that agent $i$ is obedient in taking the regular actions).
With abuse of notation, for all the remaining agents participating in period $t-1$, let $\gamma^{\sigma}_{i,t}[s_{i,t}, h_{t}]\in \Delta\Big(\prod^{T}_{k=t} S_{k}\times A_{k} \Big)$ denote the continuing outcome perceived by agent $i$ when every agent $j\in \mathcal{N}_{t-1}$, including agent $i$ himself, takes $\tau_{j,t}(s_{j,t}, h_{t})=T+1$ (i.e., all agents in $\mathcal{N}_{t}$ participate and plan to participate from period $t+1$ onward).
For all $i\in\mathcal{N}_{t-1}$, let
\[
\gamma^{\sigma}_{i,t}[h_{t}]\equiv \mathbb{E}^{F_{i,t}}\Big[\gamma^{\sigma}_{i,t}[\tilde{s}_{i,t}, h_{t}]\Big|h_{t}\Big].
\]
Given any $h_{t}\in H_{t}$, $S^{\mathtt{off}}$ and $\gamma^{\sigma}_{i,t}[h_{t}]$ induce a probability distribution, $\chi_{i,t}(h_{t})=\big(\chi_{i,t}(k|h_{t}) \big)_{k\in\mathbb{T}_{t,T}}\in \Delta(\mathbb{T}_{t,T})$, where $\sum_{k=t}^{T} \chi_{i,t}(k|h_{t}) = 1$ and each $\chi_{i,t}(k|h_{t})\geq 0$ satisfies, for all $t'\in\mathbb{T}_{t,T}$,
\begin{equation}\label{eq:principal_desired_belief_sym}
    \begin{aligned}
        \begin{aligned}
        \chi_{i,t}\left(k\middle|h_{t}\right) =\mathbb{P}\left(s_{i,k}\in S^{\mathtt{off}}_{i,k} \middle|h_{t}\right),
\end{aligned}
    \end{aligned}
\end{equation}
where $\mathbb{P}(s_{i,k} \in S^{\mathtt{off}}_{i,k} |h_{t})$ is the probability of the event $\{s_{t'} \in S^{\mathtt{off}}_{t'}\}$ that is determined by the probability measure $\gamma^{\natural}_{i,t}[h_{t}]$ given the history $h_{t}$.
Given $S^{\mathtt{off}}$, $\chi_{i,t}(k|h_{t})$ is the (principal-desired) probability distribution of each agent $i$'s taking $\mathtt{OM}_{i,k}=1$ in period $k\in\mathbb{T}_{t,T}$ when the agent is obedient (in taking the OM and the regular actions).
Hence, the principal desires each agent $i$ to form a conjecture $x_{i,t}(\cdot)=\chi_{-i,t}(\cdot)$ at each period $t$ about others' obedient OM strategies.

Given any conjecture $x_{i,t}(h_{t})\in \Delta\big((\mathbb{T}_{t,T})^{|\mathcal{N}_{t}|-1}\big)$, we rewrite agent $i$'s period-$t$ expected payoff-to-go (\ref{eq:to_go_off_menu_switch}) as follows: for all $i\in\mathcal{N}$, $t\in\mathbb{T}$, $s_{i,t}\in S_{i,t}$, $L^{i}\in \mathbb{T}_{t,T+1}$,
\begin{equation}\label{eq:rewrite_payoff_togo}
    \begin{aligned}
        \Lambda_{i,t}\left(L^{i}, a'_{i,t}\middle|s_{i,t}, h_{t}, x_{i,t}\right)=\phi_{i,t}\left(h_{t},x_{i,t}\middle|c_{i,t}\right)\mathbf{1}_{\{L^{i}=t\}} +G_{i,t}\left(a_{i,t}\middle|s_{i,t}, h_{t}, L^{i}, x_{i,t}\right)\mathbf{1}_{\{L^{i}>t\}},
    \end{aligned}
\end{equation}
with $\Lambda_{i,t}(L^{i}|s_{i,t}, h_{i,t}, x_{i,t})=\Lambda_{i,t}(L^{i}, a_{i,t}|s_{i,t}, h_{i,t}, x_{i,t})$ when $a_{i,t}$ is obedient.

\subsubsection{Existence of Equilibrium OM Strategies}

Let $\epsilon_{k}=(\epsilon^{t}_{k}, \epsilon^{t+1}_{k},\dots,\epsilon^{T+1}_{k})\in \mathbb{R}^{|T-t+2|}$ be the vector with $\epsilon^{k}_{k}=1$ and $\epsilon^{\ell}_{k}=0$ for all $\ell\neq k \in \mathbb{T}_{t, T+1}$.
For any RAIC mechanism, the set of agent $i$'s best responses to the conjecture $x_{i,t}$ is defined by the correspondence, 
\begin{equation}\label{eq:correspondence_conj}
    \begin{aligned}
        C_{i,t}\left(s_{i,t}, h_{t}, x_{i,t}\right)\equiv\left\{\epsilon_{k}: \Lambda_{i,t}\left(k\middle|s_{i,t}, h_{t}, x_{i,t}\right)\geq \Lambda_{i,t}\left(k'\middle|s_{i,t}, h_{t}, x_{i,t}\right),\forall k'\in \mathbb{T}_{t,T+1}\right\},
    \end{aligned}
\end{equation}
for all $i\in\mathcal{N}$, $t\in\mathbb{T}$, $h_{t}\in H_{t}$.
Define the set
\begin{equation}\label{eq:set_correspondence_conj}
    \begin{aligned}
    &X_{i,t}\left(h_{t},x_{i,t}\right)\equiv\left\{\int_{s_{i,t}\in S_{i,t}} \varsigma_{i,t}\left(s_{i,t}, h_{t},x_{t}\right) d F_{i,t}\left(s_{i,t}|s_{i,t-1}, h_{t}\right):\varsigma^{\natural}_{t}\left(s_{t}, h_{t},x_{t}\right) \in C_{i,t}\left(s_{t}, h_{t}, x_{t}\right), \forall s_{i,t}\in S_{i,t}   \right\}.
\end{aligned}
\end{equation}
That is, $X_{i,t}(h_{t},x_{i,t})$ is the set of distributions of each agent $i$'s best-responding choices of $L^{i}$ to the conjecture $x_{i,t}$ given the history $h_{t}$.
Finally, define the set
\begin{equation}
    \begin{aligned}
        \mathtt{FE}_{t}\left(h_{t};\Theta\right)\equiv\Big\{\mu_{t}=\left(\mu_{i,t}\right): \mu_{i,t}\in X_{i,t}\left(h_{t},\mu_{-i,t}\right), \forall i\in\mathcal{N}\Big\},
    \end{aligned}
\end{equation}
where $\mu_{-i,t}$ is agent $i$'s conjecture about other agents' OM actions (i.e., $x_{i,t}=\mu_{-i,t}$) and $\mu_{i,t}$ is the conjecture of other agents about agent $i$'s OM actions.
If $\mathtt{FE}_{t}(h_{t};\Theta)\neq \emptyset$, then by the definition of $X_{i,t}$ in (\ref{eq:set_correspondence_conj}), there exists a measurable function profile $\tau_{t}=(\tau_{i,t})$ such that for every $\mu_{t}=(\mu_{i,t}, \mu_{-i,t})\in \mathtt{FE}_{t}(h_{t};\Theta)$, we have
\[
\begin{aligned}
    \Lambda_{i,t}\left( \tau_{i,t}(s_{i,t}, h_{t})\middle|s_{i,t}, h_{t}, \mu_{-i,t}\right)\geq \Lambda_{i,t}\left(k'\middle|s_{i,t}, h_{t}, \mu_{-i,t}\right),
\end{aligned}
\]
for all $i\in\mathcal{N}$, $t\in\mathbb{T}$, $s_{i,t}\in S_{i,t}$, $h_{t}\in H_{t}$.
The non-emptiness of $ \mathtt{FE}_{t}(h_{t};\Theta)$ is independent of the specific choice of $S^{\mathtt{off}}$ (given the cutoff-switch function profile) but imposes conditions for the base game $\mathcal{G}$ and the mechanism $\Theta$.

Let $\bm{X}_{t}(h_{t}, \mu_{t})\equiv\prod_{i\in \mathcal{N}} X_{i,t}(h_{t}, \mu_{-i,t})$, where $\mu_{t}=(\mu_{i,t})_{i\in \mathcal{N}}$ is the profile of agents' conjectures.
Hence, we say $\mu_{t}\in \bm{X}_{t}(h_{t}, \mu_{-i,t})$ if each $\mu_{i,t}\in X_{i,t}(h_{t}, \mu_{-i,t})$ for all $i\in\mathcal{N}$.

\begin{lemma}\label{lemma:existence_OM_strategies}
     For any RAIC mechanism $\Theta$, there exists $\mu^{*}_{t}=(\mu^{*}_{i,t}, \mu^{*}_{-i,t})\in \Delta\big(\mathbb{T}_{t,T+1} \big)^{n}$ such that $\mu^{*}_{i,t}\in \bm{X}_{t}(h_{t}, \mu^{*}_{-i,t})$ for all $i\in\mathcal{N}$; i.e., $\mathtt{FE}_{t}(h_{t};\Theta)\neq \emptyset$.
   %
\end{lemma}

Lemma \ref{lemma:existence_OM_strategies} implies that there always exists a profile $\tau_{t}$ such that $\mathtt{FE}_{t}(h_{t};\Theta)\neq \emptyset$.
Therefore, to ensure $S^{\mathtt{off}}$-DOIC, we seek the choices of $S^{\mathtt{off}}$ such that for all $t\in\mathbb{T}$, $h_{t}\in H_{t}$, 
\begin{equation}\tag{\texttt{FP}}\label{eq:fixed_point_general}
    \begin{aligned}
        \chi_{t}(\cdot|h_{t}) \in \bm{X}_{t}(h_{t}, \chi_{t}),
    \end{aligned}
\end{equation}
where $\chi_{t}=(\chi_{i,t}, \chi_{-i,t})$ with each $\chi_{i,t}$ satisfying (\ref{eq:principal_desired_belief_sym}) given $S^{\mathtt{off}}_{i,t}$, for all $i\in\mathcal{N}$, $t\in\mathbb{T}$.
It is straightforward to see that if the RAIC mechanism is $S^{\mathtt{off}}$-switchable, then (\ref{eq:fixed_point_general}) holds. That is,
\begin{equation}\label{eq:prob_equal_nece}
    \mathbb{P}\left(\tau_{i,t}(s_{i,t}, h_{t})=t\right) = \chi_{i,t}(t|h_{t}).
\end{equation}
However, the existence of (\ref{eq:fixed_point_general}) does not imply $S^{\mathtt{off}}$-switchability.
In particular, for a RAIC mechanism to be $S^{\mathtt{off}}$-switchable, we need to guarantee that there exists a profile $\tau_{t}=(\tau_{i,t}, \tau_{-i,t})$ such that \textit{(i)} each $\tau_{i,t}$ is a best response to $\chi_{-i,t}$, where each $\chi_{j,t}$ satisfies (\ref{eq:principal_desired_belief_sym}) given $S^{\mathtt{off}}_{j,t}$ for all $j\neq i$, and \textit{(ii)} for all $i\in\mathcal{N}$, $t\in\mathbb{T}$, $h_{t}\in H_{t}$,
\begin{equation}\label{eq:prob_equal_suffice}
    \mathbb{P}\Big(\tau_{i,t}(s_{i,t}, h_{t}) = t \textup{ and } s_{i,t}\in S^{\mathtt{off}}_{i,t} \big| h_{t} \Big) = \chi_{i,t}(t|h_{t}).
\end{equation}
%
%
Here, (\ref{eq:prob_equal_suffice}) implies (\ref{eq:fixed_point_general}), but not vice versa.
In particular, (\ref{eq:fixed_point_general}) only implies (\ref{eq:prob_equal_nece}) but does not necessarily align the probability of each agent $i$'s taking $\mathtt{OM}_{i,t}=1$ (according to $\tau_{i,t}$) with the probability of realizing a state $s_{i,t}\in S^{\mathtt{off}}_{i,t}$.

Construct the cutoff-switch function by,
for \textup{any} $b\in [B]$,
\begin{equation}\label{eq:cutoff_horizontal_general}
    \begin{aligned}
    \phi_{i,t}(b, h_{i,t}, \chi_{-i,t}|c_{i,t}) = \mathtt{Mg}_{i,t}( \overline{d}_{i,t}(b), h_{i,t}, \chi_{-i,t})+\overline{\delta}_{i,t}(\overline{d}_{i,t}(b), h_{i,t}, \chi_{-i,t}).
    \end{aligned}
\end{equation}
By Definition \ref{def:horizontal_cutoff_switch}, it is straightforward to see that the horizontal cutoff-switch $\phi_{i,t}(b,h_{i,t},\chi_{-i,t}|c_{i,t})$ given by (\ref{eq:cutoff_horizontal_general}) is independent of $b$. Hence, we omit $b$ and denote $\phi_{i,t}(h_{i,t},\chi_{-i,t}|c_{i,t})$.
Let $s_{i,t}[b]\in \vec{\Gamma}^{b}(c_{i,t})$ and $s_{i,t}[e]\in \Psi^{e}_{i,t}(c_{i,t})$, respectively, denote the typical states in the $b$th OFR${}^{d}$ and the $e$th ONR${}^{d}$.  
Construct the cutoff-switch function as follows: for $w\in \{b,e\}$, $b\in[B]$, $e\in[B'-B]$,
\begin{equation}\label{eq:cutoff_knowledgeable_general}
    \begin{aligned}
    \phi_{i,t}\left(w,h_{i,t},\chi_{-i,t}|c_{i,t}\right) = \mathtt{Mg}_{i,t}\left(\overline{\underline{D}}_{i,t}\circ s_{i,t}[w], h_{i,t}, \chi_{-i,t}\right)+ \overline{\delta}_{i,t}\left(\overline{\underline{D}}_{i,t}\circ s_{i,t}[w], h_{i,t},\chi_{-i,t}\right).
    \end{aligned}
\end{equation}

\begin{proposition}\label{prop:condition_conjecture_MAOS_gen}
    Suppose that $S^{\mathtt{off}}$ is desired by the principal, and let $\chi=(\chi_{i,t})$, where each $\chi_{i,t}\in \Delta\big(\mathbb{T}_{t,T+1} \big)$ satisfies (\ref{eq:principal_desired_belief_sym}) given $S^{\mathtt{off}}_{i,t}$ for all $i\in\mathcal{N}$, $t\in\mathbb{T}$.
    Suppose in addition that there exists a set of carrier functions such that the mechanism $\Theta=\left<\sigma, \rho, \phi\right>$ satisfies \ref{itm:C1}-\ref{itm:C3}, in which \textup{either of the following holds.} 
    \begin{itemize}
        \item[(i)] Each cutoff-switch $\phi_{i,t}$ is \textup{horizontal} and is given by (\ref{eq:cutoff_horizontal_general}).
        \item[(ii)] The principal is \textup{knowledgeable} and each cutoff-switch $\phi_{i,t}$ is given by (\ref{eq:cutoff_knowledgeable_general}).
    \end{itemize}
    Then, each $\chi_{i,t}$ satisfies (\ref{eq:prob_equal_suffice}) where each $\tau_{i,t}$ is a best response to $\chi_{-i,t}$ for all $i\in\mathcal{N}$, \textup{if and only if}, (\ref{eq:fixed_point_general}) holds for all $i\in\mathcal{N}$, $t\in\mathbb{T}$, $h_{t}\in H_{t}$. 
\end{proposition}

Proposition \ref{prop:condition_conjecture_MAOS_gen} establishes conditions under which (\ref{eq:fixed_point_general}) in an RAIC mechanism is equivalent to $S^{\mathtt{off}}$-switchability.

\subsubsection{Sufficient Conditions for $S^{\mathtt{off}}$-DOIC}

We boil down the conditions of state spaces when there exist certain essential regions given a base game model $\mathcal{G}$ and the task policy profile $\sigma$ by defining the following two sets:
\begin{equation}\label{set:horizontal_essential_regions}
    \mathtt{S}^{\mathtt{H}}[\sigma, \mathcal{G}]\equiv \left\{\begin{aligned}
    &S'=\prod\limits_{i,t}S'_{i,t}: S'_{i,t}=  \bigcup_{b\in [B]} \Pi^{b}_{i,t} \textup{ is an \textit{essential region} of } S_{i,t}, \textup{ where each } \Pi^{b}_{i,t}  \textup{ is the $b$-th} \\
    &\textup{ \textit{essential interval} of } S_{i,t} \textup{ with } \overline{d}_{t}(b) \textup{ as the \textit{essential point} }\forall i\in\mathcal{N}, t\in\mathbb{T}, b\in[B]
\end{aligned}\right\}.
\end{equation}

\begin{equation}\label{set:knowledgeable_essential_regions}
    \mathtt{S}^{\mathtt{K}}[\sigma, \mathcal{G}]\equiv \left\{\begin{aligned}
    &S'= \prod_{i,t}S'_{i,t}: S'_{i,t}=\left(\bigcup_{b\in[B]}\Pi^{b}_{i,t}\right)\bigcup\left(\bigcup_{e\in[B'-B]}\Pi^{e}_{i,t} \right), \textup{ where} \\
    &\textup{\textit{(i)} each }\Pi^{b}_{i,t}\backslash \{\overline{\underline{D}}\circ s_{i,t}[b]\} \textup{ is an \textit{essential region} of } \Pi^{b}_{i,t}, \textup{ and}\\
    & \textup{\textit{(ii)} each } \{\overline{\underline{D}}\circ s_{i,t}[e]\} \textup{ is an \textit{essential region} of } \Pi^{e}_{i,t} \forall i\in\mathcal{N}, t\in\mathbb{T}
\end{aligned}\right\}.
\end{equation}
The following theorem further characterizes the conditions of Proposition \ref{prop:condition_conjecture_MAOS_gen}.

\begin{theorem}\label{thm:MAO_switchability_conjecture_gen}
    Fix a base game model $\mathcal{G}$.
    Suppose that $S^{\mathtt{off}}$ is desired by the principal, and let $\chi=(\chi_{i,t})$, where each $\chi_{i,t}\in \Delta\big(\mathbb{T}_{t,T+1} \big)$ satisfies (\ref{eq:principal_desired_belief_sym}) given $S^{\mathtt{off}}_{i,t}$ for all $i\in\mathcal{N}$, $t\in\mathbb{T}$.
    Suppose in addition that there exists a collection of carrier functions such that the mechanism $\Theta=\left<\sigma, \rho, \phi\right>$ satisfies \ref{itm:C1}-\ref{itm:C3}.
    Then, the following holds.
    \begin{itemize}
        \item[(i)] Suppose that each cutoff-switch $\phi_{i,t}$ is \textup{horizontal} and is given by (\ref{eq:cutoff_horizontal_general}).
        Then, each $\chi_{t}=(\chi_{i,t})_{i\in\mathcal{N}}$ satisfies (\ref{eq:fixed_point_general}) \textup{if and only if} each $S^{\mathtt{off}}_{i,t}\in \mathtt{S}^{\mathtt{H}}[\sigma, \mathcal{G}]$ for all $i\in\mathcal{N}$, $t\in\mathbb{T}$.
        \item[(ii)] Suppose that the principal is \textup{knowledgeable} and each cutoff-switch $\phi_{i,t}$ is given by (\ref{eq:cutoff_knowledgeable_general}).
        Then, each $\chi_{t}=(\chi_{i,t})_{i\in\mathcal{N}}$ satisfies (\ref{eq:fixed_point_general}) \textup{if and only if} each $S^{\mathtt{off}}_{i,t}\in \mathtt{S}^{\mathtt{K}}[\sigma, \mathcal{G}]$ for all $i\in\mathcal{N}$, $t\in\mathbb{T}$.
    \end{itemize}
\end{theorem}

Theorem \ref{thm:MAO_switchability_conjecture_gen} establishes sufficient conditions for $S^{\mathtt{off}}$-DOIC that are placed on the task policy profile and the carrier functions, given the base game model.
In particular, under the assumption that the mechanism is RAIC by satisfying \ref{itm:C1}-\ref{itm:C3}, \textit{(i)} and \textit{(ii)} of Theorem \ref{thm:MAO_switchability_conjecture_gen} show necessary and sufficient conditions that incorporate essential regions for each $\chi_{t}$ (associated with $S^{\mathtt{off}}$) to satisfy (\ref{eq:fixed_point_general}).
Lemma \ref{lemma:existence_OM_strategies} implies that in a RAIC mechanism (by satisfying \ref{itm:C1}-\ref{itm:C3}), there always exists an equilibrium OM strategy profile.
Proposition \ref{prop:condition_conjecture_MAOS_gen} further shows that given certain formulations of the cutoff-switch functions (i.e., when it is horizontal or the principal is knowledgeable), the principal's desired $S^{\mathtt{off}}$ with the corresponding $\chi_{t}$ induces an equilibrium OM strategy profile $\tau=(\tau_{i,t}, \tau_{-i,t})$ such that (\ref{eq:fixed_point_general}) holds.
In other words, for any $S^{\mathtt{off}}$ with each $\chi_{t}$ satisfying (\ref{eq:fixed_point_general}), each agent $i$ forms conjecture $\chi_{-i,t}$ such that there exists a $\tau_{i,t}$ that is best-response to $\chi_{-i,t}$, and $\tau_{i,t}$ and $\chi_{i,t}$ satisfies (\ref{eq:prob_equal_suffice}).
Finally, Theorem \ref{thm:MAO_switchability_conjecture_gen} places conditions for the base game model $\mathcal{G}$, the principal's desired $S^{\mathtt{off}}$, the task policy profile, and the carrier functions, under which (\ref{eq:fixed_point_general}) guarantees \textit{(a)} the existence of an equilibrium OM strategy profile $\tau$ and \textit{(b)} $\tau$ is obedient.

\begin{proposition}\label{prop:switchability_sufficient_condition_new}
    Fix a base game $\mathcal{G}$, task policy profile $\sigma$, and a collection of carrier functions $g$.
    Suppose that $S^{\mathtt{off}}$ is desired by the principal, and let $\chi=(\chi_{i,t})$, where each $\chi_{i,t}\in \Delta\big(\mathbb{T}_{t,T+1} \big)$ satisfies (\ref{eq:principal_desired_belief_sym}) given $S^{\mathtt{off}}_{i,t}$ for all $i\in\mathcal{N}$, $t\in\mathbb{T}$.
    \begin{itemize}
        \item[(i)] Let each $\phi_{i,t}$ be constructed by (\ref{eq:cutoff_horizontal_general}).
        Suppose that each $S^{\mathtt{off}}_{i,t}\in \mathtt{S}^{\mathtt{H}}[\sigma, \mathcal{G}]$ for all $i\in\mathcal{N}$, $t\in\mathbb{T}$. 
        Then, for any $\rho\in \textup{\ref{eq:feasible_rho_C1_original}}$, the mechanism $\Theta=\left<\sigma, \rho, \phi\right>$ is $S^{\mathtt{off}}$-DOIC, if and only if, for all $i\in\mathcal{N}$, $t\in\mathbb{T}$, $s_{i,t}\in S_{i,t}$, $h_{t}\in H_{t}$, $a'_{i,t}\in A_{i,t}[\sigma]$,
        \begin{equation}
        \begin{aligned}
            \mathtt{Mg}_{i,t}(\overline{D}_{i,t}\circ s_{i,t}, h_{t},\chi_{-i,t}) + \overline{\delta}_{i,t}(\overline{D}_{i,t}\circ s_{i,t}, h_{t},\chi_{-i,t})\geq \mathtt{Mg}_{i,t}(a'_{i,t} , s_{i,t}, h_{t},\chi_{-i,t}) + \overline{\delta}_{i,t}(a'_{i,t}, s_{i,t}, h_{t},\chi_{-i,t}).
        \end{aligned}
    \end{equation}
    \item[(ii)] Let each $\phi_{i,t}$ be constructed by  (\ref{eq:cutoff_knowledgeable_general}).
    Suppose that each $S^{\mathtt{off}}_{i,t}\in \mathtt{S}^{\mathtt{K}}[\sigma, \mathcal{G}]$ for all $i\in\mathcal{N}$, $t\in\mathbb{T}$.
    Then, for any $\rho\in \textup{\ref{eq:feasible_rho_C1_original}}$, the mechanism $\Theta=\left<\sigma, \rho, \phi\right>$ is $S^{\mathtt{off}}$-DOIC, if and only if, for all $i\in\mathcal{N}$, $t\in\mathbb{T}$, $s_{i,t}\in S_{i,t}$, $h_{t}\in H_{t}$, $a'_{i,t}\in A_{i,t}[\sigma]$, $w\in \{b,e\}$, $b\in[B]$, $e\in[B'-B]$,
    \begin{equation}
        \begin{aligned}
            \mathtt{Mg}_{i,t}(\overline{D}_{i,t}\circ s_{i,t}[w], h_{i,t}, \chi_{-i,t})+ \overline{\delta}_{i,t}(\overline{D}_{i,t}\circ s_{i,t}[w], h_{i,t},\chi_{-i,t})\geq \mathtt{Mg}_{i,t}(a'_{i,t},s_{i,t}[w], h_{i,t}, \chi_{-i,t})+ \overline{\delta}_{i,t}(a'_{i,t}, s_{i,t}[w], h_{i,t},\chi_{-i,t}).
        \end{aligned}
    \end{equation}
    \end{itemize}
\end{proposition}

Proposition \ref{prop:switchability_sufficient_condition_new} extends Proposition \ref{thm:doic_irod_sufficient_condition} to $S^{\mathtt{off}}$-DOIC by leveraging the results in Theorem \ref{thm:MAO_switchability_conjecture_gen}.
In particular, the sufficient conditions in Proposition \ref{thm:doic_irod_sufficient_condition} only requires \ref{itm:C1} of Theorem \ref{thm:MAO_switchability_conjecture_gen}.

\section{Further Characterizations of \texorpdfstring{$S^{\mathtt{off}}$}{Lg}-DOIC}\label{sec:characterization_DOIC}

Sections \ref{sec:DOIC_D-IROD} and \ref{sec:switchability_D-SOD} characterize the DOIC when the mechanism is dynamically individually rational and $S^{\mathtt{off}}$-switchable, respectively, by establishing sufficient conditions.
These conditions rely on the existence of the carrier functions, without requiring explicit formulations to be provided. 
In this section, we further characterize the DOIC and provide explicit conditions placed on the task policy profile by obtaining a closed-form formulation for each carrier function.
We start by showing the following conditions used in our derivation.

\begin{condition}\label{cond:differentiable_reward}
    For all $i\in\mathcal{N}$, $t\in\mathbb{T}$, $a_{t}\in A_{t}[\sigma]$, the reward function $u_{i,t}(s_{i,t}, a_{t})$ is differentiable in $s_{i,t}$. In addition, each collection $\{u_{i,t}(\cdot, a_{t})\}_{a_{t}\in A_{t}[\sigma]}$ is equi-Lipschitz continuous on $S_{i,t}$, for all $i\in\mathcal{N}$, $t\in\mathbb{T}$; i.e., there exists a constant $c_{i,t}\in \mathbb{R}_{+}$ such that $|u_{i,t}(s_{i,t}, a_{t}) - u_{i,t}(s'_{i,t}, a_{t})|\leq c_{i,t}|s_{i,t}-s'_{i,t}|$, with $c_{i,t}\in\mathbb{R}_{+}$ and $\sum^{T}_{t=1}c_{i,t}<\infty$.
\end{condition}

\begin{condition}\label{cond:bounded_dynamic}
    For all $i\in\mathcal{N}$, $t\in\mathbb{T}$, $\omega_{i,t}\in \Omega$, $h_{t-1}\in H_{t-1}$, $\frac{\partial}{\partial v}\kappa_{i,t}(v, h_{t-1}, \omega_{t})\big|_{v=s_{i,t-1}}$ exists, and there exists $\bar{C}_{i,t}:\Omega\mapsto \mathbb{R}_{+}$ such that $\left|\frac{\partial}{\partial v}\kappa_{i,t}(v, h_{t-1}, \omega_{t})|_{v=s_{i,t-1}}\right|\leq \bar{C}_{i,t}(\omega_{i,t})$ with $\bar{C}_{i,t}(\omega_{i,t})$ be an integrable function and $\left|\bar{C}_{i,t}(\omega_{i,t})\right|<\hat{c}_{i,t}$, where $\hat{c}_{i,t}\in\mathbb{R}_{+}$ is a constant, and $\prod_{k=0}^{T}\big(\sum_{t=1}^{T} \hat{c}_{i,t}\big)<\infty$.
\end{condition}

Condition \ref{cond:differentiable_reward} assumes differentiability and equi-Lipschitz continuity of each agent's reward function at the true state in each period.
Condition \ref{cond:bounded_dynamic} places analogous assumptions on the persistence function to ensure that small changes in the current true state result in bounded impacts on future states.



\subsection{Necessary Conditions for \texorpdfstring{$S^{\mathtt{off}}$}{Lg}-DOIC}\label{sec:necessary_condition}

Given obedient conjecture $x_{i,t}=\chi_{-i,t}$, construct for all $i\in\mathcal{N}$, $s_{i,t},s'_{i,t}\in S_{i,t}$, $h_{t}\in H_{t}$, $L\in\mathbb{T}_{t,T}$,
\begin{align*}
    &MG_{i,t}\left(s'_{i,t}|s_{i,t},h_{t}, \chi_{-i,t}\right)\equiv \max\limits_{L\in\mathbb{T}_{t,T+1}} G_{i,t}\left(\sigma_{i,t}(s'_{i,t}, h_{t})\middle|s_{i,t},h_{t}, L,\chi_{-i,t}\right),\\
    &V_{i,t}\left(s_{i,t}, h_{t}, \chi_{-i,t}\right)\equiv \max\limits_{s'_{i,t}\in S_{i,t}} MG_{i,t}\left(s'_{i,t}\middle|s_{i,t},h_{t}, \chi_{-i,t}\right).
\end{align*}
Hence, if the mechanism is RAIC, then $V_{i,t}(s_{i,t}, h_{t}, \chi_{-i,t})=MG_{i,t}(s_{i,t}|s_{i,t},h_{t}, \chi_{-i,t})$.
For any $s_{i,t}, s'_{i,t}\in S_{i,t}$, $a'_{i,t}=\sigma_{i,t}(s'_{i,t}, h_{t})$, $h_{t}\in H_{t}$, $L\in \mathbb{T}_{t,T}$, let
\begin{equation}\label{eq:envelope_term}
\begin{aligned}
q_{i,t}(a'_{i,t}, s_{i,t}, h_{t}, L, \chi_{-i,t})\equiv \mathbb{E}^{\sigma}_{a'_{i,t}}\Big[\sum^{L}_{k=t}\frac{\partial u_{i,k}(v, \sigma_{i,k}(\tilde{s}_{i,k}, \tilde{h}_{i,k}) )}{\partial v} \big|_{ v = \tilde{s}_{i,k}}\mathtt{mp}_{i,t}(a_{i,t}, s_{i,t}, h_{t}, L, \chi_{-i,t})\Big| s_{i,t}, h_{t}, \chi_{-i,t} \Big],
\end{aligned}
\end{equation}
where $\mathtt{mp}_{i,t}(\cdot)$ given by
\begin{equation}\label{eq:partial_derivative_payoff}
    \begin{aligned}
    &\mathtt{mp}_{i,t}( a_{i,t}, s_{i,t}, h_{t}, L, \chi_{-i,t})\equiv\mathbb{E}^{\sigma}\Big[\prod_{k=t+1}^{L} \frac{\partial  \kappa_{i,k}(v, \sigma^{k-1}_{i}(\tilde{s}^{k-1}_{i}, \tilde{h}^{k-1}), \Tilde{a}^{k-1}_{-i},\tilde{\omega}_{i,k}) }{\partial  v }\big|_{ v = \tilde{s}_{i,k-1}}\Big|s_{i,t}, h_{t},\chi_{-i,t}\Big].
    \end{aligned} 
\end{equation}
Here, $q_{i,t}(\cdot)$ is independent of $\rho$ and $\phi$ and only depends on $\sigma$ and the base game model $\mathcal{G}$.
When agents are obedient in taking the regular actions, $q_{i,t}(s_{i,t}, h_{t}, L, \chi_{-i,t})=q_{i,t}(a_{i,t}, s_{i,t}, h_{t}, L, \chi_{-i,t})$ with obedient $a_{i,t}$  is the \textit{impulse response} (\cite{pavan2014dynamic}) that captures the marginal effects of current state $s_{i,t}$ on future (to a fixed horizon $L$).

\begin{proposition}\label{prop:necessary_DOIC_sigma_tau}
    Fix $\mathcal{G}$.
    Suppose that Conditions \ref{cond:differentiable_reward} and \ref{cond:bounded_dynamic} hold.
    Let $S^{\mathtt{off}}$ be desired by the principal, and let $\chi=(\chi_{i,t})$ satisfy (\ref{eq:principal_desired_belief_sym}) given $S^{\mathtt{off}}$.
    Suppose in addition that the game $\mathcal{G}^{\Theta}$ admits a PBE $\left<\tau, \pi\right>$.
    If the mechanism $\Theta$ is $S^{\mathtt{off}}$-DOIC, then \textit{(i)} $\pi=\sigma$ and each $\tau_{i,t}$ satisfies (\ref{eq:x_coincides_tau}) given $x_{i,t} = \chi_{-i,t}$, and \textit{(ii)} for all $i\in\mathcal{N}$, $t\in\mathbb{T}$, $L\in \mathbb{T}_{i,T}$, $s_{i,t}\in S_{i,t}$, $h_{t}\in H_{t}$,
\begin{equation}\label{eq:envelope_like_condition_sigma_tau}
        \begin{aligned}
            \frac{\partial}{\partial v} V_{i,t}(v, h_{t},\chi_{-i,t})\Big|_{v=s_{i,t}}= q_{i,t}\Big(\sigma_{i,t}(s_{i,t},h_{t}), s_{i,t}, h_{t}, \tau_{i,t}(s_{i,t},h_{t}), \chi_{-i,t}\Big).
        \end{aligned}
    \end{equation}
\end{proposition}

Proposition \ref{prop:necessary_DOIC_sigma_tau} shows an envelop-theorem-like necessary condition for $S^{\mathtt{off}}$-DOIC.
However, unlike classic mechanism design with quasilinear utility (see, e.g., \citet{myerson1981optimal,pavan2014dynamic}), where the envelope condition depends only on the allocation rule (i.e., $\sigma$ in our model), our condition (\ref{eq:envelope_like_condition_sigma_tau}) also depends on the optimality of the OM strategy profile $\tau$.


\subsection{Formulations of Carrier Functions}\label{sec:formulation_of_carrier_functions}

Let $\Theta_{i,t}=S_{i,t}$ for all $i\in\mathcal{N}$, $t\in\mathbb{T}$.
We construct each carrier function by
\begin{equation}\label{eq:construct_carrier_original}
\begin{aligned}
    g_{i,t}(a_{i,t}, s_{i,t}, h_{t}, L, \chi_{-i,t}|\theta_{i,t})=\begin{cases}
        \int\limits_{\theta_{i,t}}^{s_{i,t}} q_{i,t}\left(\sigma_{i,t}(v,h_{t}), v, h_{t},L, \chi_{-i,t}\right) dv, &\textup{ if } a_{i,t}=\sigma_{i,t}(s_{i,t}, h_{t})\\
       \int\limits_{\theta_{i,t}}^{s_{i,t}} q_{i,t}\left(a_{i,t}, v, h_{t},L, \chi_{-i,t}\right) dv, & \textup{ if } a_{i,t}\neq\sigma_{i,t}(s_{i,t}, h_{t})
    \end{cases}.
\end{aligned}
\end{equation}
%
%
%
Now, each carrier function, $g_{i,t}$, and the corresponding maximum carrier,  $\mathtt{Mg}_{i,t}$, are formulated in closed forms that depend only on the task policy profile $\sigma$ and are independent of $\rho$ and $\phi$.
%
%
%

Hence, by substituting each carrier function given by (\ref{eq:construct_carrier_original}), we obtain closed-form formulations of the cutoff-switch given by (\ref{eq:cutoff_horizontal_general}) and (\ref{eq:cutoff_knowledgeable_general}).
In addition, we can drop $g$ in \ref{eq:feasible_rho_C1_original} and denote it as $\mathcal{P}[\sigma,\mathcal{G}]$.
Given (\ref{eq:construct_carrier_original}), the maximum carrier $\mathtt{Mg}_{i,t}(s_{i,t}, h_{t},\chi_{-i,t}|\theta_{i,t})$ can be interpreted as the \textit{relative information rent} earned by agent $i$ with true state $s_{i,t}$ compared to if he takes action $a'_{i,t}=\sigma_{i,t}(s'_{i,t}, h_{t})$ (i.e., pretending to have $s'_{i,t}$ as the true state).
This relative rent is due to agent $i$'s possession of private information coupled with obedient strategic participation.

\begin{remark}\label{remark:closed_form_carrier}
The closed-form formulations of carrier functions given by (\ref{eq:construct_carrier_original}) are equivalent to the implicit carrier functions in maintaining all the previous results (which involve carrier functions) in Sections \ref{sec:DOIC_D-IROD} and \ref{sec:switchability_D-SOD} under Condition \ref{cond:differentiable_reward}.
\hfill $\triangle$
\end{remark}


\subsection{Sufficient Conditions for \texorpdfstring{$S^{\mathtt{off}}$}{Lg}-DOIC with Zero On-Rent}


Given the closed-form formulations of the carrier functions given by (\ref{eq:envelope_term}), we show sufficient conditions for $S^{\mathtt{\mathtt{off}}}$-DOIC in 
Sections \ref{sec:DOIC_D-IROD} and \ref{sec:switchability_D-SOD}, which are imposed on the task policy profile $\sigma$ and the base game $\mathcal{G}$ given the constructions of $\rho$ and $\phi$ in terms of $\sigma$.
In this section, we obtain another sufficient condition for $S^{\mathtt{\mathtt{off}}}$-DOIC when the mechanism induces zero on-rent for each agent $i$ with $s_{i,t}\in S^{\mathtt{off}}_{i,t}$ in every period $t$.

Given a base game $\mathcal{G}$, define the set for $\mathtt{X}\in\{\mathtt{H}, \mathtt{K}\}$,
\begin{equation}\label{set:indifference_region_X}
    \mathcal{S}^{\mathtt{IX}}[\mathcal{G}]\equiv \left\{\begin{aligned}
   &S^{\mathtt{off}}=(S^{\mathtt{off}}_{i,t})\in \mathtt{S}^{\mathtt{X}}[\sigma,\mathcal{G}]: \textup{ given $\mathcal{G}$, $\exists$  $\Theta=\left<\sigma,\rho,\phi\right>$ s.t. $\Theta$ is} \textup{ $S^{\mathtt{off}}$-DOIC,} \\
    &  \textup{and each } S^{\mathtt{off}}_{i,t} \textup{ is an \textit{indifference region.}}
\end{aligned}\right\},
\end{equation}
where $\mathtt{S}^{\mathtt{H}}[\sigma,\mathcal{G}]$ and $\mathtt{S}^{\mathtt{K}}[\sigma,\mathcal{G}]$ are defined by (\ref{set:horizontal_essential_regions}) and (\ref{set:knowledgeable_essential_regions}), respectively.
Then, we can define the set $\mathcal{S}^{\mathtt{IX}}[\sigma, \mathcal{G}]\subseteq \mathcal{S}^{\mathtt{IX}}[\mathcal{G}]$ by
\begin{equation}
    \mathcal{S}^{\mathtt{IX}}[\sigma, \mathcal{G}] \equiv \left\{\begin{aligned}
   &S^{\mathtt{off}}=(S^{\mathtt{off}}_{i,t})\in \mathtt{S}^{\mathtt{X}}[\sigma,\mathcal{G}]: \textup{ given $\mathcal{G}$, $\exists$ $\rho$ and $\phi$ s.t., $\Theta=\left<\sigma,\rho,\phi\right>$ is} \textup{ $S^{\mathtt{off}}$-DOIC,} \\
    &  \textup{and each } S^{\mathtt{off}}_{i,t} \textup{ is an \textit{indifference region.}}
\end{aligned} \right\}
\end{equation}
Hence, in every game $\mathcal{G}^{\Theta}$ with $S^{\mathtt{off}}\in \mathcal{S}^{\mathtt{IX}}[\sigma, \mathcal{G}]$, each agent $i$ with state $s_{i,t}$ earns a zero on-rent in every period $t$ if $s_{i,t}\in S^{\mathtt{off}}_{i,t}$.

\begin{theorem}\label{thm:sufficient_condition_without_MSO}
    Fix a base game $\mathcal{G}$. Suppose that Conditions \ref{cond:differentiable_reward} and \ref{cond:bounded_dynamic} hold.
    Suppose in addition that $\mathcal{S}^{\mathtt{IX}}[\sigma,\mathcal{G}]\neq \emptyset$, for either $\mathtt{X}\in\{\mathtt{H},\mathtt{K}\}$.
    Then, there exists $\rho$ and $\phi$ such that the mechanism is $S^{\mathtt{off}}$-DOIC for $S^{\mathtt{off}}\in \mathcal{S}^{\mathtt{IX}}[\sigma, \mathcal{G}]$ if $\sigma$ satisfies the following \textup{constrained monotone} condition, for all $i\in\mathcal{N}$, $t\in\mathbb{T}$, $s_{i,t}\in S_{i,t}$, $h_{t}\in H_{t}$,
    \begin{equation}\tag{\texttt{CM}}\label{eq:condition_RAIC_sigma_0}
        \begin{aligned}
            \max\limits_{L\in \mathbb{T}_{t,T}} \int^{s_{i,t}}_{s'_{i,t}}q_{i,t}\left(\sigma_{i,t}(v,h_{t}), v, h_{t}, L, \chi_{-i,t}\right)dv \geq \max\limits_{L\in \mathbb{T}_{t,T}}\int^{s_{i,t}}_{s'_{i,t}} q_{i,t}\left(\sigma_{i,t}(s'_{i,t},h_{t}), v, h_{t}, L,\chi_{-i,t}\right)dv,
        \end{aligned}
    \end{equation}
    where $\chi=(\chi_{i,t})$ satisfy (\ref{eq:principal_desired_belief_sym}) given $S^{\mathtt{off}}$.
    %
    %
    %
\end{theorem}

%
%
Theorem \ref{thm:sufficient_condition_without_MSO} shows a sufficient condition, known as the constrained monotone (\ref{eq:condition_RAIC_sigma_0}), when each $S^{\mathtt{off}}_{i,t}$ is an indifference region, under Conditions \ref{cond:differentiable_reward} and \ref{cond:bounded_dynamic} and the assumptions that $\mathcal{S}^{\mathtt{IX}}[\mathcal{G}]\neq \emptyset$, for either $\mathtt{X}\in\{\mathtt{H},\mathtt{K}\}$.
Here, $\mathcal{S}^{\mathtt{IX}}[\sigma,\mathcal{G}]\neq \emptyset$ ensures $\mathtt{S}^{\mathtt{X}}[\sigma,\mathcal{G}]\neq \emptyset$, which allows us to construct cutoff-switch functions such that the mechanism is $S^{\mathtt{off}}$-switchable.
Given $\mathcal{G}$, the condition (\ref{eq:condition_RAIC_sigma_0}) is placed only on the task policies without requiring the specific constructions of $\rho$ and $\phi$, and is independent of the agents' OM strategies.


\subsection{Application to Dynamic Canonical Mechanism}\label{sec:application_to_dynamic_classic}

In this section, we consider the dynamic extension of classic static mechanism design without off-switch functions, within the context of delegation.
We refer to such a model as a \textit{dynamic canonical mechanism} (DCM), denoted by $\Theta^{\mathtt{C}}=\left<\sigma, \rho\right>$, which includes a profile of task policies as well as a profile of coupling functions.

We start by demonstrating a form of uniqueness property associated with the cutoff-switch functions.
Given a base game $\mathcal{G}$, a task policy profile $\sigma$, and coupling policy profile $\rho$, define the set $\mathcal{S}[\sigma, \rho, \mathcal{G}]\subseteq \mathcal{S}[\sigma,  \mathcal{G}]$ by
\[
\mathcal{S}[\sigma, \rho, \mathcal{G}]\equiv \left\{ 
\begin{aligned}
    &S^{\mathtt{off}}=(S^{\mathtt{off}}_{i,t}): \textup{given $\mathcal{G}$, there is } \phi \textup{ such that } \Theta=\left<\sigma, \rho, \phi\right> \textup{ is } S^{\mathtt{off}}\textup{-DOIC}.
\end{aligned}
\right\}.
\]
%
%
Hence, for each $S^{\mathtt{off}}\in \mathcal{S}[\sigma, \rho,\mathcal{G}]$ (if non-empty), there always exists a profile $\phi$ (without referring to its specific formulations) such that the mechanism $\Theta=\left<\sigma, \rho, \phi\right>$ is $S^{\mathtt{off}}$-DOIC.
In addition, for any $S^{\mathtt{off}}\in \mathcal{S}[\sigma, \rho,\mathcal{G}]$ define the set of off-switch functions,
\[
\begin{aligned}
    \Phi[\sigma, \rho, \mathcal{G}; S^{\mathtt{off}}]\equiv \Big\{\phi: \textup{given } \mathcal{G}, \Theta=\left<\sigma, \rho, \phi\right> \textup{ is } S^{\mathtt{off}}\textup{-DOIC}. \Big\}.
\end{aligned}
\]
\begin{proposition}\label{prop:phi_uniqueness}
    Fix $\mathcal{G}$, $\sigma$, and $\rho$.
    Suppose that $\mathcal{S}[\sigma, \rho, \mathcal{G}]\neq \emptyset$.
    Suppose in addition that $\Phi[\sigma, \rho, \mathcal{G}; S^{\mathtt{off}}]\neq \emptyset$ for some $S^{\mathtt{off}}\in \mathcal{S}[\sigma, \rho, \mathcal{G}]$. 
    Then, for any two off-switch functions $\phi=(\phi_{i,t}), \phi'=(\phi'_{i,t})\in \Phi[\sigma, \rho, \mathcal{G}; S^{\mathtt{off}}]$, we have $\phi_{i,t}(h_{t}|c_{i,t})=\phi'_{i,t}(h_{t}|c_{i,t})$, for all $i\in\mathcal{N}$, $t\in\mathbb{T}$, given $h_{t}\in H_{t}$.
\end{proposition}

Proposition \ref{prop:phi_uniqueness} shows a uniqueness property of the off-switch functions (not necessarily cutoff-switch) in any $S^{\mathtt{off}}$-DOIC mechanism.
In particular, if any two off-switch functions associated with the same profiles of $\sigma$ and $\rho$ lead to $S^{\mathtt{off}}$-DOIC, then these two off-switch functions are the same.

\begin{corollary}
Suppose that $\rho\in\mathcal{P}[\sigma,\mathcal{G}]$ and the carrier functions are formulated by (\ref{eq:construct_carrier_original}).
Suppose in addition that $\mathcal{S}[\sigma, \rho, \mathcal{G}]\neq \emptyset$.
When the cutoff-switch is horizontal, let $\phi_{i,t}$ be constructed according to (\ref{eq:cutoff_horizontal_general}); when the principal is knowledgeable, let each $\phi_{i,t}$ be constructed according to (\ref{eq:cutoff_knowledgeable_general}).
Then, for all $S^{\mathtt{off}}\in \mathcal{S}[\sigma, \rho, \mathcal{G}]$, the mechanism $\Theta=\left<\sigma, \rho, \phi\right>$ is $S^{\mathtt{off}}$-DOIC.
%
\end{corollary}

By dropping the off-switch functions from the prospect function in (\ref{eq:offswitch_payoff_to_go}), each agent $i$'s expected payoff-to-go (\ref{eq:to_go_off_menu_switch}) in DCM becomes
\[
\begin{aligned}
    \Lambda^{\mathtt{C}}_{i,t}(\mathtt{OM}_{i,t}, a_{i,t}|s_{i,t}, h_{t}, x_{i,t}) &=\max\limits_{L\in\mathbb{T}} \mathbb{E}^{\sigma}_{\pi}\Big[\sum\limits_{k=t}^{L} z_{i,k}\big(\tilde{s}_{i,k}, \tilde{a}_{i,k}\big) \Big| s_{i,t}, h_{t}, x_{i,t} \Big] \mathbf{1}_{\{ \mathtt{OM}_{i,t} =0 \}},
\end{aligned}
\]
where $x_{i,t}$ is agent $i$'s conjecture of other agents' OM actions.
%
%
Hence, the best response of agent $i$ in period $t$ is given by the correspondence
\[
\begin{aligned}
    D^{\mathtt{C}}_{i,t}(s_{i,t},h_{t}, x_{i,t}|\mathtt{OM}_{-i},\pi_{-(i,t)}, \mathcal{G}^{\mathtt{C}})=\argmax_{\mathtt{OM}_{i,t}\in\{0,1\}, a_{i,t} \in A_{i,t}[\sigma] } \Lambda^{\mathtt{C}}_{i,t}(\mathtt{OM}_{i,t}, a_{i,t}|s_{i,t}, h_{t}, x_{i,t}).
\end{aligned}
\]
%

\begin{definition}
     A DCM $\Theta^{\mathtt{C}}$ is $S^{\mathtt{off}}$-DOIC if $\left<\tau^{d}_{i,t}, \sigma_{i,t}\right>\in$ $ D^{\mathtt{C}}_{i,t}(s_{i,t},h_{t}, \chi_{-i,t}|\mathtt{OM}_{-i},\pi_{-(i,t)}, \mathcal{G}^{\mathtt{C}}),$
     %
     where $\chi=(\chi_{i,t})$ and $\tau=(\tau_{i,t})$ satisfy (\ref{eq:x_coincides_tau}) and (\ref{eq:principal_desired_belief_sym}) given $S^{\mathtt{off}}$.
\end{definition}

Given a base game $\mathcal{G}$ and a task policy profile $\sigma$, define the following two sets, for $\mathtt{X}\in\{\mathtt{H}, \mathtt{K}\}$,
\[
\mathcal{S}^{\mathtt{CX}}[\sigma, \mathcal{G}]\equiv \left\{ 
\begin{aligned}
    &S^{\mathtt{off}}\in \mathtt{S}^{\mathtt{X}}: \textup{given $\mathcal{G}$, there exists $\rho$ such that the mechanism} \\
    & \Theta^{\mathtt{C}}=\left<\sigma, \rho\right> \textup{ is $S^{\mathtt{off}}$-DOIC, in which $\rho\in\mathcal{P}[\sigma,\mathcal{G}]$}. 
\end{aligned}
\right\},
\]
\[
\mathcal{S}^{\mathtt{OX}}[\sigma, \mathcal{G}]\equiv\left\{\begin{aligned}
    &S^{\mathtt{off}}\in \mathtt{S}^{\mathtt{X}}: \textup{given $\mathcal{G}$, there exists $\phi$ such that the mechanism} \\
    & \Theta=\left<\sigma, \rho,\phi\right> \textup{ is $S^{\mathtt{off}}$-DOIC, in which $\rho\in\mathcal{P}[\sigma,\mathcal{G}]$}. 
\end{aligned}  \right\}.
\]


\begin{proposition}\label{thm:sufficient_condition_positive_classic_mechanism}
    Fix a base game $\mathcal{G}$ and a task policy profile $\sigma$. 
    Suppose that $\mathcal{S}^{\mathtt{C}}[\sigma, \mathcal{G}]\neq\emptyset$.
    Then, the following holds.
    \begin{itemize}
        \item[(i)] Suppose that the cutoff-switch is \textup{horizontal} and $\mathcal{S}^{\mathtt{OX}}[\sigma, \mathcal{G}]\neq\emptyset$. Then, $\mathcal{S}^{\mathtt{CH}}[\sigma, \mathcal{G}] = \mathcal{S}^{\mathtt{OH}}[\sigma, \mathcal{G}]$ if and only if $\sigma$ satisfies, for all $i\in\mathcal{N}$, $t\in\mathbb{T}$, $h_{t}\in H_{t}$,
        \begin{equation}\label{eq:cutoff_condition_zero_H}
            \begin{aligned}
            \mathtt{Mg}_{i,t}(\overline{d}_{i,t}(b), h_{t}, \chi_{-i,t}) + \overline{\delta}_{i,t}(\overline{d}_{i,t}(b), h_{i,t}, \chi_{-i,t}) = 0.
            \end{aligned}
        \end{equation}
        \item[(ii)]  Suppose that the principal is \textup{knowledgeable} and $\mathcal{S}^{\mathtt{OX}}[\sigma, \mathcal{G}]\neq\emptyset$. Then, $\mathcal{S}^{\mathtt{CK}}[\sigma, \mathcal{G}] = \mathcal{S}^{\mathtt{OK}}[\sigma, \mathcal{G}]$ if and only if $\sigma$ satisfies, for all $i\in\mathcal{N}$, $t\in\mathbb{T}$, $h_{t}\in H_{t}$, $w\in\{b,e\}$, $b\in[B]$, $e\in[B'-B]$,
        \begin{equation}\label{eq:cutoff_condition_zero_K}
            \begin{aligned}
            \mathtt{Mg}_{i,t}(\overline{\underline{D}}_{i,t}\circ s_{i,t}[w], h_{t}, \chi_{-i,t}) + \overline{\delta}_{i,t}(\overline{\underline{D}}_{i,t}\circ s_{i,t}[w], h_{t}, \chi_{-i,t}) = 0.
            \end{aligned}
        \end{equation}
    \end{itemize}
\end{proposition}
Proposition \ref{thm:sufficient_condition_positive_classic_mechanism} obtains sufficient conditions for a DCM $\Theta^{\mathtt{C}}=\left<\sigma, \rho\right>$ to be $S^{\mathtt{off}}$-DOIC by imposing a set of additional constraints on $\sigma$ based on the $S^{\mathtt{off}}$-DOIC of our original mechanism $\Theta=\left<\sigma, \rho, \phi\right>$.
The conditions are sufficient because they require $\rho\in\mathcal{P}[\sigma,\mathcal{G}]$.
From Proposition \ref{prop:phi_uniqueness}, each cutoff-switch function is unique for every $S^{\mathtt{off}}\in \mathcal{S}[\sigma, \rho, \mathcal{G}]$ such that $\Phi[\sigma, \rho, \mathcal{G}; S^{\mathtt{off}}]\neq \emptyset$.
Suppose that a the mechanism $\Theta=\left<\sigma', \rho', \phi'\right>$ is $S^{\mathtt{off}}$-DOIC, in which $\rho\in\mathcal{P}[\sigma,\mathcal{G}]$.
Then, by directly imposing (\ref{eq:cutoff_condition_zero_H}) or (\ref{eq:cutoff_condition_zero_K}) (i.e., making each $\phi'_{i,t}(\cdot)=0$), we obtain an $S^{\mathtt{off}}$-DOIC DCM $\Theta^{\mathtt{C}}=\left<\sigma',\rho'\right>$. 
%

\section{Discussions}\label{sec:discussion_1st_order_approach}

In this section, we conclude this paper by discussing the application of the first-order approach and the mechanism design when the agents' voluntary participation decisions (i.e., the oM actions) are periodic ex-post.

\subsection{The First-Order Approach}
The first-order (FO) approach is one of the mainstream methods used in the literature on mechanism design with adverse selection (\citet{kapivcka2013efficient,pavan2014dynamic,pavan2017dynamic}).
In classic mechanism design problems, the FO approach encompasses three steps (\citet{pavan2017dynamic}).
\begin{itemize}
    \item[\textbf{Step 1.}] The first step is to establish a first-order necessary condition described by an envelope condition for incentive compatibility, which in turn allows one to construct the transfer rule as a function of the allocation rule.
    \item[\textbf{Step 2.}] The second step is then to formulate a relaxed mechanism design problem with only local incentive compatibility by using the first-order necessary condition.
    \item[\textbf{Step 3.}] The final step is to verify whether the allocation rule that solves the relaxed mechanism design problem with the transfer rule constructed according to Step 1 leads to a global incentive compatible and individually rational mechanism.  
\end{itemize}


Although the relaxed mechanism design problem in Step 3 is generally analytically intractable, especially in dynamic environments, the FO approach has proven effective in exploring the structural results of mechanism design with adverse selection.

The primary objective of this section is not to explicitly apply the FO approach in identifying the optimal solution for our dynamic delegation mechanism integrated with OM actions. 
In this section, our primary aim is to meticulously characterize the Step 1 for the mechanism $\Theta=\left<\sigma, \rho, \phi\right>$.
It is worth recalling that in Section \ref{sec:necessary_condition}, Proposition \ref{prop:necessary_DOIC_sigma_tau} delineates a necessary condition (\ref{eq:envelope_like_condition_sigma_tau}).
This condition is imposed on the task policy profile, which is analogous to the truthful reporting strategy that is prevalent in traditional mechanism design.
This necessary condition shares a close resemblance to the envelope condition. Nevertheless, it is important to note that the condition (\ref{eq:envelope_like_condition_sigma_tau}) is augmented by the obedient OM strategy profile $\tau$.

\subsubsection{Maximum-Sensitive Obedience}


Recall the set $\mathcal{S}^{\mathtt{IX}}[\mathcal{G}]$ defined by (\ref{set:indifference_region_X}) and the set $\mathcal{S}^{\mathtt{IX}}[\sigma, \mathcal{G}]\subseteq \mathcal{S}^{\mathtt{IX}}[\mathcal{G}]$, for $X\in\{\mathtt{H}, \mathtt{K}\}$.
The following corollary directly follows Proposition \ref{prop:necessary_DOIC_sigma_tau}.

%

\begin{corollary}\label{corollary:envelop_condition_indifference}
    Suppose that Conditions \ref{cond:differentiable_reward} and \ref{cond:bounded_dynamic} hold. 
    Suppose in addition that $\mathcal{S}^{\mathtt{IX}}[\sigma, \mathcal{G}]\neq \emptyset$.
    If the mechanism $\Theta$ is $S^{\mathtt{off}}$-DOIC with $S^{\mathtt{off}}\in \mathcal{S}^{\mathtt{IX}}[\sigma, \mathcal{G}]$, then for all $i\in\mathcal{N}$, $t\in\mathbb{T}$, $L\in \mathbb{T}_{i,T}$, $s_{i,t}\in S_{i,t}$, $h_{t}\in H_{t}$,
    \begin{equation}\label{eq:first_order_indifference}
    \begin{aligned}
            \frac{\partial}{\partial v} V_{i,t}(v, h_{t},\chi_{-i,t})\Big|_{v=s_{i,t}}=  q_{i,t}\big(\sigma_{i,t}(s_{i,t},h_{t}), s_{i,t}, h_{t}, T, \chi_{-i,t}\big).
        \end{aligned}
\end{equation}
\end{corollary}

Due to the zero on-rent (i.e., every $S^{\mathtt{off}}_{i,t}$ is an indifference region), Corollary \ref{corollary:envelop_condition_indifference} shows a first-order necessary condition for the task policy profile $\sigma$ that is independent of the agents' obedient OM strategy profile $\tau$.


In this section, we further analyze the first-order necessary condition by considering a fairly stringent condition known as \textit{maximum-sensitive obedience} to decouple the necessary condition from the agents' obedient OM strategy profile.

\begin{definition}[Maximum-Sensitive Obedience]\label{def:maximum_sensitive_obedience}
Fix a base game $\mathcal{G}$. Given any $s'=(s'_{i,t})$ with $s'_{i,t}\in S_{i,t}$, we say that the $S^{\mathtt{off}}$-DOIC mechanism $\Theta=\left<\sigma, \rho, \phi\right>$ induces \textit{maximum-sensitive obedience} (MSO) if the following holds at $s_{i,t}=s'_{i,t}$ for all $i\in\mathcal{N}$, $t\in\mathbb{T}$, $h_{t}\in H_{t}$,
    \begin{equation}\label{eq:def_MSO}
        \begin{aligned}
            \frac{\partial }{\partial v} \max\limits_{L\in\mathbb{T}_{t,T}} G_{i,t}\left(\sigma_{i,t}(s'_{i,t}, h_{t})\middle|v, h_{t}, L, \chi_{-i,t}\middle)\right|_{v=s_{i,t}}=\max\limits_{L\in\mathbb{T}_{t,T}}\frac{\partial }{\partial v} G_{i,t}\left(\sigma_{i,t}(s'_{i,t}, h_{t})\middle|v, h_{t}, L, \chi_{-i,t}\middle)\right|_{v=s_{i,t}}.
        \end{aligned}
    \end{equation}
    \hfill $\triangle$
\end{definition}

The MSO implies that the partial derivative of $MG_{i,t}(s'_{i,t}|s_{i,t}, h_{t}, \chi_{-i,t})$ with respect to $s_{i,t}$ when $s'_{i,t}=s_{i,t}$ (i.e., LHS of (\ref{eq:def_MSO})) can be computed by finding the maximum of the partial derivatives of $G_{i,t}(\sigma_{i,t}(s'_{i,t}, h_{t})|s_{i,t}, h_{t}, L)$ concerning $s_{i,t}$ at each possible choices of $\tau_{i,t}(s_{i,t}, h_{t})\in \mathbb{T}_{t,T+1}$ when $s'_{i,t}=s_{i,t}$ (i.e., RHS of (\ref{eq:def_MSO})).
If this property holds for the game $\mathcal{G}^{\Theta}$, then it enables us to focus on analyzing $G_{i,t}$ with the largest derivative (i.e., RHS of (\ref{eq:def_MSO})) to understand how the optimal $\tau_{i,t}(s_{i,t},h_{t})$ changes as the true state $s_{i,t}$ varies, when the agents are obedient in taking the regular actions.
Note that the LHS and the RHS of (\ref{eq:def_MSO}) in general have different $L's$ that reach the respective maxima.

\begin{corollary}\label{corollary:envelope_V_MSO}
    Suppose that Conditions \ref{cond:differentiable_reward} and \ref{cond:bounded_dynamic} hold.
    Let $S^{\mathtt{off}}$ be desired by the principal, and let $\chi=(\chi_{i,t})$ satisfy (\ref{eq:principal_desired_belief_sym}) given $S^{\mathtt{off}}$.
    If the mechanism $\Theta$ is $S^{\mathtt{off}}$-DOIC and induces MSO, then for all $i\in\mathcal{N}$, $t\in\mathbb{T}$, $L\in \mathbb{T}_{i,T}$, $s_{i,t}\in S_{i,t}$, $h_{t}\in H_{t}$,
\begin{equation}\label{eq:first_order_like_condition_envelope_MSO}
    \begin{aligned}
            \frac{\partial}{\partial v} V_{i,t}\left(v, h_{t},\chi_{-i,t}\right)\Big|_{v=s_{i,t}}= \max_{L\in\mathbb{T}_{t,T}} q_{i,t}\left(\sigma_{i,t}(s_{i,t},h_{t}), s_{i,t}, h_{t}, L, \chi_{-i,t}\right).
        \end{aligned}
\end{equation}
\end{corollary}

The necessary condition (\ref{eq:first_order_like_condition_envelope_MSO}) is now independent of $\rho$, $\phi$, and the agents' OM strategy profile $\tau$.
Let $\mathcal{S}^{\mathtt{M}}[\sigma, \mathcal{G}]\subseteq \mathcal{S}[\sigma,  \mathcal{G}]$ denote the set such that for every $S^{\mathtt{off}}\in \mathcal{S}^{\mathtt{M}}[\sigma, \mathcal{G}]$, there always exists an $S^{\mathtt{off}}$-DOIC mechanism $\Theta=\left<\sigma, \rho, \phi\right>$ that induces MSO.
It is worth noting that the RHS of (\ref{eq:first_order_like_condition_envelope_MSO}) is in general equivalent to the RHS of (\ref{eq:envelope_like_condition_sigma_tau}) without imposing MSO.

\begin{theorem}\label{thm:necessary_RAIC_condition}
    Fix a base game $\mathcal{G}$ and a task policy profile $\sigma$.
    Suppose that Conditions \ref{cond:differentiable_reward} and \ref{cond:bounded_dynamic} hold.
    Suppose in addition $\mathcal{S}^{\mathtt{M}}[\sigma,  \mathcal{G}]\neq \emptyset$.
    Then, $\sigma$ satisfies the \textup{constrained monotone condition (\ref{eq:condition_RAIC_sigma_0})} for any $S^{\mathtt{off}}\in \mathcal{S}^{\mathtt{M}}[\sigma,  \mathcal{G}]$.
\end{theorem}

%
Theorem \ref{thm:necessary_RAIC_condition} shows that the constrained monotone condition (\ref{eq:condition_RAIC_sigma_0}) in Theorem \ref{thm:sufficient_condition_without_MSO} is a necessary condition for $S^{\mathtt{off}}$-DOIC when the mechanism induces MSO.
In Theorem \ref{thm:sufficient_condition_without_MSO}, however, the sufficiency of (\ref{eq:condition_RAIC_sigma_0}) requires the mechanism to induce zero on-rent (i.e., $S^{\mathtt{off}}\in \mathcal{S}^{\mathtt{IX}}[\sigma, \mathcal{G}]$ for $\mathtt{X}\in\{\mathtt{H}, \mathtt{K}\}$).
Theorems \ref{thm:necessary_RAIC_condition} and \ref{thm:sufficient_condition_without_MSO} shed light on obtaining a necessary and sufficient condition for $S^{\mathtt{off}}$-DOIC placed solely on $\sigma$ given $\mathcal{G}$.

Denote $\mathcal{S}^{\mathtt{MIX}}[ \mathcal{G}]\subseteq\mathcal{S}^{\mathtt{IX}}[ \mathcal{G}]$ and $\mathcal{S}^{\mathtt{MIX}}[\sigma, \mathcal{G}]\subseteq\mathcal{S}^{\mathtt{IX}}[\sigma, \mathcal{G}]$, where $\mathcal{S}^{\mathtt{IX}}[\mathcal{G}]$ is given by (\ref{set:indifference_region_X}), for $\mathtt{X}\in\{\mathtt{H}, \mathtt{K}\}$, the sets in which the associated mechanisms induce MSO.
In addition, define the set of task profiles for all $\mathtt{X}\in\{\mathtt{H}, \mathtt{K}\}$,
\[
\begin{aligned}
    &\Sigma^{\mathtt{MX}}[\mathcal{G}]\equiv\Big\{\sigma: \textup{given $\mathcal{G}$, there exists } S^{\mathtt{off}}\in \mathcal{S}^{\mathtt{MIX}}[\sigma, \mathcal{G}] \textup{, s.t. } \sigma \textup{ satisfies (\ref{eq:condition_RAIC_sigma_0})}  \Big\}.
\end{aligned}
\]
The following proposition provides a condition placed on $\sigma$ that is both necessary and sufficient for $S^{\mathtt{off}}$-DOIC.

\begin{proposition}\label{prop:nece_suffic_indifference}
    Fix a base game $\mathcal{G}$.
Suppose that Conditions \ref{cond:differentiable_reward} and \ref{cond:bounded_dynamic} hold.
Suppose in addition that $\mathcal{S}^{\mathtt{MIX}}[\mathcal{G}]\neq \emptyset$, for either $\mathtt{X}\in\{\mathtt{H}, \mathtt{K}\}$.
Then, we have, for either $\mathtt{X}\in\{\mathtt{H}, \mathtt{K}\}$,
\begin{equation}\label{eq:prop_necessary_sufficient_indifference}
    \begin{aligned}
        \bigcup\limits_{\sigma\in \Sigma^{\mathtt{MX}}[\mathcal{G}]} \mathcal{S}^{\mathtt{MIX}}[\sigma,\mathcal{G}] = \mathcal{S}^{\mathtt{MIX}}[\mathcal{G}].
    \end{aligned}
\end{equation}
\end{proposition}

In Proposiiton \ref{prop:nece_suffic_indifference}, (\ref{eq:prop_necessary_sufficient_indifference}) implies that the set $\mathcal{S}^{\mathtt{MIX}}[\mathcal{G}]$ is equal to the union of the sets $\mathcal{S}^{\mathtt{MIX}}[\sigma,\mathcal{G}]$ over all task policy profiles in $\Sigma^{\mathtt{MX}}[\mathcal{G}]$.
In other words, the constrained monotone condition (\ref{eq:condition_RAIC_sigma_0}) is necessary and sufficient for $S^{\mathtt{off}}$-DOIC when the corresponding mechanism induces MSO and every $S^{\mathtt{off}}_{i,t}$ is an indifference region (i.e., leading to zero on-rent).
%
%
In Proposition \ref{prop:nece_suffic_indifference}, the necessity of (\ref{eq:condition_RAIC_sigma_0}) for $S^{\mathtt{off}}$-DOIC when each $S^{\mathtt{off}}_{i,t}$ is an indifference region of $S_{i,t}$ is the same as provided by Theorem \ref{thm:necessary_RAIC_condition}.
The sufficiency directly follows Theorem \ref{thm:sufficient_condition_without_MSO}.
%
%
Proposition \ref{prop:nece_suffic_indifference} shows that when the OFRs are indifference regions and the mechanism induces maximum-sensitive obedience, the constrained monotone (\ref{eq:condition_RAIC_sigma_0}) of the task policy profile $\sigma$ becomes necessary and sufficient for $S^{\mathtt{off}}$-DOIC.
Here, the necessity is independent of the specific formulations of $\rho$ and $\phi$. For sufficiency, we need to show that there exist $\rho$ and $\phi$ such that (\ref{eq:condition_RAIC_sigma_0}) leads to $S^{\mathtt{off}}$-DOIC, which requires the explicit formulations of $\rho$ and $\phi$.

Given $\mathcal{G}$ and $\sigma$, let $\mathcal{S}^{\mathtt{CM}}[\sigma, \mathcal{G}]\subseteq \mathcal{S}^{\mathtt{C}}[\sigma, \mathcal{G}]$ denote the set such that for each $S^{\mathtt{off}}\in \mathcal{S}^{\mathtt{CM}}[\sigma, \mathcal{G}]$ there exists an $S^{\mathtt{off}}$-DOIC DCM $\Theta^{\mathtt{C}}=\left<\sigma, \rho\right>$ that induces maximum-sensitive obedience.
\begin{corollary}\label{corollary:iff_MSO_DCM}
    Fix a base game $\mathcal{G}$. Suppose that $u_{i,t}(\cdot)\geq 0$ for all $i\in\mathcal{N}$, $t\in\mathbb{T}$. 
    Then, the following holds. 
    \begin{itemize}
        \item[(i)] Suppose that the cutoff-switch is \textup{horizontal}. Then, for any $S^{\mathtt{off}}\in \mathtt{S}^{\mathtt{H}}[\sigma, \mathcal{G}]$, $S^{\mathtt{off}}\in \mathcal{S}^{\mathtt{CM}}[\sigma, \mathcal{G}]$ if and only if $\sigma$ satisfies (\ref{eq:condition_RAIC_sigma_0}) and (\ref{eq:cutoff_condition_zero_H}).
        \item[(ii)] Suppose that the principal is \textup{knowledgeable}. Then, for a $S^{\mathtt{off}}\in \mathtt{S}^{\mathtt{K}}[\sigma, \mathcal{G}]$, $S^{\mathtt{off}}\in \mathcal{S}^{\mathtt{CM}}[\sigma, \mathcal{G}]$ if and only if $\sigma$ satisfies (\ref{eq:condition_RAIC_sigma_0}) and (\ref{eq:cutoff_condition_zero_K}).
    \end{itemize}
\end{corollary}

Corollary \ref{corollary:iff_MSO_DCM} extends the results of Proposition \ref{thm:sufficient_condition_positive_classic_mechanism} by imposing the requirement of maximum-sensitive obedience and establishes a set of necessary and sufficient conditions for $S^{\mathtt{off}}$-DOIC of the DCM $\Theta^{\mathtt{C}}$.

However,  the specific constructions of $\mathcal{G}$ and $\Theta$ that lead to MSO are difficult to universally characterize, as they depend on the intricate interplay between all components within various structures of the game $\mathcal{G}^{\Theta}$. A necessary condition for achieving MSO, as described in (\ref{eq:def_MSO}), is that the single-stage payoff function ($z_{i,t}(\cdot) = u_{i,t}(\cdot) + \rho_{i,t}(\cdot)$), the off-switch function, the state dynamics, and the task policies must interact in such a manner that the partial derivatives of each prospect function concerning the true state variable maintain the same sign for all possible $L \in \mathbb{T}_{t,T}$.
A potential sufficient condition for maximum-sensitive obedience pertains to the magnitudes of these partial derivatives. Specifically, in the game $\mathcal{G}^{\Theta}$, the magnitudes of the partial derivatives of each prospect function with respect to the true state needs to be simultaneously maximized by the obedient OM strategy $\tau_{i,t}$, over all possible $L\in \mathbb{T}_{t,T}$.
We provide two examples that exhibit MSO in the online supplementary document.

As in classic mechanism design settings, the FO approach finds implementable choice rules (e.g., the task policies in our model, and allocation rules in general resource allocation problems) only under fairly strong assumptions such as exogenous transitions, single-crossing utility functions, and positive serial correlation (\citet{pavan2014dynamic,pavan2017dynamic,zhang2012analysis}).
This imposes fundamental challenges to the validity of the ``relaxation" performed by the FO approach.
Nevertheless, the FO approach has been successfully applied in studying the structural results of the mechanism design (\citet{zhang2012analysis}).
Furthermore, it also facilitates ``reverse engineering" to identify primitive conditions under which the choice rules determined by Step 3 satisfy global incentive compatibility (\citet{pavan2017dynamic}).

\subsection{Periodic Ex-Post Voluntary Participation}

When each agent $i$'s voluntary participation is periodic ex-post, he first takes the regular action from the action menu $A_{i,t}[\sigma]$ and then takes the OM actions.
    Each agent $i$'s period-$t$ \textit{ex-post expected payoff-to-go} becomes
\begin{equation}\label{eq:payofftogo_expost}
    \begin{aligned}
        &\Lambda^{\texttt{ep}}_{i,t}\left(\mathtt{OM}_{i,t}, a_{i,t}|s_{i,t}, h_{t}, x_{i,t}\right)\\
        &\equiv \phi^{\natural}_{i,t}(a_{t}, h_{t})\mathbf{1}_{\{ \mathtt{OM}_{i,t} =1 \}}+ \max\limits_{L\in\mathbb{T}_{t+1,T}} \mathbb{E}^{\sigma}\left[\sum\limits_{k=t+1}^{L} z_{i,k}\big(\tilde{s}_{i,k}, \tilde{a}_{i,k}\big) +\phi^{\natural}_{i,L}(\tilde{a}_{L},\tilde{h}_{L}) \middle| s_{i,t}, h_{t}, x_{i,t} \right]\mathbf{1}_{\{\mathtt{OM}_{i,t}=0 \}}\\
        &=\phi^{\natural}_{i,t}(a_{t}, h_{t})\mathbf{1}_{\{ \mathtt{OM}_{i,t} =1 \}}+ \mathbb{E}^{\sigma}\left[ z_{i,t+1}\big(\tilde{s}_{i,t+1}, \tilde{a}_{t+1}\big) +\Lambda^{\texttt{ep}}_{i,t+1}\left(\widetilde{\mathtt{OM}}_{i,t+1}, \tilde{a}_{i,t+1}|\tilde{s}_{i,t+1}, \tilde{h}_{t+1}\right)\middle| s_{i,t}, h_{t+1}, x_{i,t} \right]\mathbf{1}_{\{\mathtt{OM}_{i,t}=0 \}},
    \end{aligned}
\end{equation}
where each $\phi^{\natural}_{i,t}: A_{i,t}[\sigma]\times H_{t}\mapsto \mathbb{R}$ is the off-switch function that determines an off-switch value when agent $i$ takes $\mathtt{OM}_{i,t}=1$.
When we construct $\rho\in \textup{\ref{eq:feasible_rho_C1_original}}$, the term $z_{i,t+1}$ on the RHS of (\ref{eq:payofftogo_expost}) coincides witht he marginal carrier $\zeta_{i,t+1}$ given by (\ref{eq:def_marginal_carrier}).
According to (\ref{eq:def_marginal_carrier}), it is straightforward to see that the period-$t$ $\Lambda^{\texttt{ep}}_{i,t}$ is agent $i$'s expectation of the period-$t+1$ interim payoff-to-go evaluated at the end of period $t$.
Hence, at the end of period $t$, each agent $i$'s decision of $\mathtt{OM}_{i,t}\in\{0,1\}$ can be interpreted as his \textit{expected OM action} for the next period.
Suppose that each marginal carrier $\zeta_{i,t+1}$ is a continuous and bounded function of the true state $s_{i,t+1}$.
In addition, each $S_{i,t+1}$ is compact.
Then, $\zeta_{i,t+1}$ attains its minimum and maximum over $S_{i,t+1}$ by the Extreme Value Theorem.
Moreover, by the Intermediate Value Theorem, $\zeta_{i,t+1}$ takes on every value between its minimum and maximum.
Thus, the expected value, $\mathtt{E}^{\sigma}\left[ \zeta_{i,t+1}(\tilde{s}_{i,t+1}, h_{t+1})\middle|s_{i,t}, h_{t+1}, x_{i,t}\right]$ corresponds to at least one $s^{\ddagger}_{i,t+1}\in S_{i,t+1}$;
i.e., $\zeta_{i,t+1}(s^{\ddagger}_{i,t+1},h_{t+1})=\mathbb{E}^{\sigma}\left[ \zeta_{i,t+1}(\tilde{s}_{i,t+1}, h_{t+1})\middle|s_{i,t}, h_{t+1}, x_{i,t}\right]$.
Define the following correspondence
\[
\begin{aligned}
    \mathcal{C}^{\ddagger}_{i,t}\left(s_{i,t}, h_{t}, x_{i,t}\right)\equiv  \max\left\{s'_{t}\in S_{i,t+1}: \zeta_{i,t+1}(s^{\ddagger}_{i,t+1},h_{t+1})=\mathbb{E}^{\sigma}\left[ \zeta_{i,t+1}(\tilde{s}_{i,t+1}, h_{t+1})\middle|s_{i,t}, h_{t+1}, x_{i,t}\right]\right\}.
\end{aligned}
\]
Then, we can rewrite $\Lambda^{\texttt{ep}}_{i,t}\left(\mathtt{OM}_{i,t}|s_{i,t}, h_{t}, x_{i,t}\right)$ (i.e., when $a_{i,t}$ is obedient) by
\begin{equation}\label{eq:payofftogo_expost_2}
    \begin{aligned}
        &\Lambda^{\texttt{ep}}_{i,t}\left(\mathtt{OM}_{i,t}|s_{i,t}, h_{t}, x_{i,t}\right)=\phi^{\natural}_{i,t}(\sigma_{i,t}(s_{i,t},h_{t}), a_{-i,t}, h_{t})\mathbf{1}_{\{ \mathtt{OM}_{i,t} =1 \}}\\
& + \left(\underbrace{\zeta_{i,t+1}\left(\mathcal{C}^{\ddagger}_{i,t}\left(s_{i,t}, h_{t}, x_{i,t}\right),h_{t+1}\right)+\mathbb{E}^{\sigma}\left[\Lambda^{\texttt{ep}}_{i,t+1}\left(\widetilde{\mathtt{OM}}_{i,t+1}, \tilde{a}_{i,t+1}|\tilde{s}_{i,t+1}, \tilde{h}_{t+1}\right)\middle| s_{i,t}, h_{t+1}, x_{i,t} \right]}_{\mathbf{E}_{i,t+1}\left( \mathcal{C}^{\ddagger}_{i,t}\left(s_{i,t}, h_{t}, x_{i,t}\right), h_{t+1}, x_{i,t} \right)}\right)\mathbf{1}_{\{\mathtt{OM}_{i,t}=0 \}}.
    \end{aligned}
\end{equation}
Since $h_{t+1}=\left(a_{t},  h_{t}\right)$, the reformulation of $\Lambda^{\mathtt{ep}}_{i,t}$ in (\ref{eq:payofftogo_expost_2}) becomes similar to $\Lambda_{i,t}$ given by (\ref{eq:to_go_off_menu_switch}) (or the recursive form given by (\ref{eq:app_thm_4_0})).
The $S^{\mathtt{off}}$-switchability requires, for all $i\in\mathcal{N}$, $t\in\mathbb{T}$,
\begin{equation}
    \Lambda^{\texttt{ep}}_{i,t}\left(\mathtt{OM}_{i,t}|s_{i,t}, h_{t}, x_{i,t}\right)=
\begin{cases}
   \phi^{\natural}_{i,t}(h_{t+1}) ,& \textup{ if } s_{i,t}\in S^{\mathtt{off}}_{i,t},\\
    \mathbf{E}_{i,t+1}\left( \mathcal{C}^{\ddagger}_{i,t}\left(s_{i,t}, h_{t}, x_{i,t}\right), h_{t+1}, x_{i,t} \right), & \textup{ if } s_{i,t}\in 
S_{i,t}\backslash S^{\mathtt{off}}_{i,t}.
\end{cases}
\end{equation}
For any $S^{\mathtt{off}}_{i,t}$, define the sets 
\[
\begin{aligned}
&S^{\ddagger}_{i,t+1}\left[S^{\mathtt{off}}_{i,t};h_{t}, x_{i,t}\right]\equiv\left\{s^{\ddagger}_{i,t+1} = \mathcal{C}^{\ddagger}_{i,t}\left(s_{i,t}, h_{t}, x_{i,t}\right): s_{i,t}\in S^{\mathtt{off}}_{i,t} \right\},\\
& S^{\sharp}_{i,t+1}\left[S^{\mathtt{off}}_{i,t};h_{t}, x_{i,t}\right]\equiv\left\{s^{\sharp}_{i,t+1} = \mathcal{C}^{\ddagger}_{i,t}\left(s_{i,t}, h_{t}, x_{i,t}\right): s_{i,t}\in S_{i,t}\backslash S^{\mathtt{off}}_{i,t} \right\}.
\end{aligned}
\]
Therefore, we have (with abuse of notation)
\[
\Lambda^{\texttt{ep}}_{i,t}\left(\mathtt{OM}_{i,t}, s^{\ddagger}_{i,t+1}|s_{i,t}, h_{t}, x_{i,t}\right)=
\begin{cases}
   \phi^{\natural}_{i,t}(h_{t+1}) ,& \textup{ if } s^{\ddagger}_{i,t+1}\in S^{\ddagger}_{i,t+1}[S^{\mathtt{off}}_{i,t};h_{t}, x_{i,t}],\\
    \mathbf{E}_{i,t+1}\left( s^{\ddagger}_{i,t+1}, h_{t+1}, x_{i,t} \right), & \textup{ if } s^{\ddagger}_{i,t+1}\in S^{\sharp}_{i,t+1}[S^{\mathtt{off}}_{i,t};h_{t}, x_{i,t}],
\end{cases}
\]
which exhibits a similar structure of our periodic interim $S^{\mathtt{off}}$-switchability.
However, when the mechanism is RAIC, the principal can infer each agent $i$'s true state $s_{i,t}$ from $a_{i,t}$.
Accordingly, with the established correspondence $\mathcal{C}^{\ddagger}{i,t}$, the principal obtains knowledge of $s^{\ddagger}{i,t+1}$ that corresponds to the actual state, $s_{i,t}$, of each agent $i$. This provides the principal with a significant informational edge when devising the mechanism when the agents' participation decisions are made on a periodic ex-post basis.
The added advantage would simplify the design principles introduced in this paper, especially with regard to the formulation of the off-switch functions.
The design methodology introduced in this paper can furnish a set of design principles for the mechanism design when the agents take the OM actions (i.e., the strategic participation decision) on a periodic ex-post basis.
Nonetheless, we omit the comprehensive characterizations of such mechanism design problems as they exceed the intended scope of this current study

\bibliographystyle{cas-model2-names} 
\bibliography{OM_Reference}


\appendix

\section{Example of Maximum-Sensitive Obedience}

In this section, we show an example of the base game model $\mathcal{G}$ and the mechanism $\Theta$ such that the maximum-sensitive obedience is achieved.


Let $M_{i,t}:S_{i,t} \mapsto \mathbb{R}$ be a linear function of the state and let $R_{i,t}: A_{t}\mapsto \mathbb{R}$ be a function of the joint action that is independent of the state.
We assume that the constant $\mathcal{M}_{i,t}=M'_{i,t}(s_{i,t})\geq0$ for all $i\in\mathcal{N}$, $t\in\mathbb{T}$, $s_{i,t}\in S_{i,t}$.
The agents' reward functions satisfy \textit{additive separation} if, for all $i\in\mathcal{N}$, $t\in\mathbb{T}$,
\begin{equation}\label{eq:add_separation}
    u_{i,t}(s_{i,t}, a_{t}) = M_{i,t}(s_{i,t}) + R_{i,t}(a_{t}).
\end{equation}
We assume that the choice of $\{M_{i,t}\}$ and $\{R_{i,t}\}$ satisfies Condition \ref{cond:differentiable_reward}.
The state-dynamic model satisfies \textit{additive separation} if, for all $i\in\mathcal{N}$, $t\in\mathbb{T}$,
\[
\kappa_{i,t+1}(s_{i,t}, h_{t+1}, \omega_{i,t+1}) = \widehat{M}_{i,t+1}(s_{i,t}) + \widehat{R}_{i,t+1}(h_{t+1}, \omega_{i,t+1}),
\]
where $\widehat{M}_{i,t+1}$ is a linear function of $s_{i,t}$ and $\widehat{R}_{i,t+1}$ is independent of $s_{i,t}$.
Suppose that the choice of $(\widehat{M}_{i,t+1})$ and $(\widehat{R}_{i,t+1})$ satisfies Condition \ref{cond:bounded_dynamic}.
The additive separation of the state-dynamic model implies
\[
    \begin{aligned}
    &\mathtt{mp}_{i,t}( a_{i,t}, s_{i,t}, h_{t}, L, \chi_{-i,t})=\prod_{k=t+1}^{L} \widehat{\mathcal{M}}_{i,k},
    \end{aligned} 
\]
where $\widehat{\mathcal{M}}_{i,k}=\widehat{\mathcal{M}}'_{i,k}(s_{i,k-1})$ is a constant for all $s_{i,k-1}\in S_{i,k-1}$, $i\in\mathcal{N}$, $k\in\mathbb{T}_{t+1,T}$.
We assume that the constant $\widehat{\mathcal{M}}_{i,k}\geq0$ for all $i\in\mathcal{N}$, $t\in\mathbb{T}$, $s_{i,t}\in S_{i,t}$.
Consider a $S^{\mathtt{off}}$-DOIC mechanism $\Theta$ with $S^{\mathtt{off}}\in \mathcal{S}^{\mathtt{IX}}[\mathcal{G}]$.
Let $\tau=(\tau_{i,t})$ be the OM strategy profile corresponding to $S^{\mathtt{off}}$, and let $\chi=(\chi_{i,t})$ satisfies (\ref{eq:x_coincides_tau}) given $\tau$. 
Then, given $L^{*}=\tau_{i,t}(s_{i,t}, h_{t})$,
\[
\begin{aligned}
    &\frac{\partial }{\partial v} \max_{L\in\mathbb{T}_{t,T}} G_{i,t}\left(\sigma_{i,t}(s'_{i,t}, h_{t})\middle|v, h_{t}, L, \chi_{-i,t}\right)\big|_{v=s_{i,t}}\\
    =&\frac{\partial }{\partial v}G_{i,t}\left(\sigma_{i,t}(s'_{i,t}, h_{t})\middle|v, h_{t}, L^{*}, \chi_{-i,t}\right)\big|_{v=s_{i,t}}= \frac{\partial }{\partial v}G_{i,t}\left(\sigma_{i,t}(s'_{i,t}, h_{t})\middle|v, h_{t}, T, \chi_{-i,t}\right)\big|_{v=s_{i,t}} \\
    =&\sum_{k'=t}^{T}\Big(\mathcal{M}_{i,k'}\prod_{k=t+1}^{k'} \widehat{\mathcal{M}}_{i,k}\Big)=\max\limits_{L\in\mathbb{T}_{t,T}}\sum_{k'=t}^{L}\Big(\mathcal{M}_{i,k'}\prod_{k=t+1}^{k'} \widehat{\mathcal{M}}_{i,k}\Big)
    =\max\limits_{L\in\mathbb{T}_{t,T}}\frac{\partial }{\partial v} G_{i,t}\left(\sigma_{i,t}(s'_{i,t}, h_{t})\middle|v, h_{t}, L, \chi_{-i,t}\right)\Big|_{v=s_{i,t}},
\end{aligned}
\]
where the second equality is due to the zero on-rent, and the fourth equality is due to $\mathcal{M}_{i,k}\geq 0$ and $\widehat{\mathcal{M}}_{i,k}\geq 0$, for all $k\in\mathbb{T}$.






\hfill $\square$

\section{Proof of Proposition \ref{prop:dynamic_revelation_principle}}

Since the mechanism $\Theta'$ is DD and induces a PBE $<\tau', \pi'>$, we have that for every $s_{i,t}\in S_{i,t}$ there exists $s'_{i,t}\in S_{i,t}$ such that $\pi_{i,t}(s_{i,t}, h_{t}) = \sigma_{i,t}(s'_{i,t}, h_{t})$ for any given $h_{t}\in H_{t}$.
Suppose that the mechanism $\Theta'$ leads to OFR $S^{\mathtt{off}}=\{S^{\mathtt{off}}_{i,t}\}$.

First, we construct a feasible mechanism $\Theta$. 
The task policy profile $\sigma=(\sigma_{i,t})$ satisfies $\sigma_{i,t}(s_{i,t}, h_{t}) = \pi'_{i,t}(s_{i,t}, h_{t})$, for all $i\in\mathcal{N}$, $t\in\mathbb{T}$, $s_{i,t}\in S_{i,t}$, $h_{t}\in H_{t}$.
For the coupling policies and the off-switch functions, we let $\rho_{i,t}(a_{t}, h_{t}) = \rho'_{i,t}(a_{t}, h_{t})$ and $\phi_{i,t}(h_{t}|c'_{i,t}) = \phi'_{i,t}(h_{t}|c'_{i,t})$, for all $i\in \mathcal{N}$, $t\in\mathbb{T}$, $s_{i,t}\in S_{i,t}$, $h_{t}\in H_{t}$.
Note that by letting $\rho_{i,t}(\cdot) = \rho'_{i,t}(\cdot)$ (resp. $\phi_{i,t}(\cdot|c'_{i,t}) = \phi'_{i,t}(\cdot|c'_{i,t})$), we do not assume that $\rho_{i,t}$ (resp. $\phi_{i,t}$) is exactly the same function of $\rho'_{i,t}$ (resp. $\phi'_{i,t}$).
%
%
%

Second, we show that $\Theta$ coupled with $<\tau, \sigma>$ leads to the same expected payoff for the principal as what $\Theta'$ with $<\tau', \sigma'>$ does.
To highlight the dependence on the mechanism, we add $\Theta$ in the notations of functions.
The period-$t$ interim expected payoff-to-go under $\Theta'$ is then given by
\[
    \begin{aligned}
    &\Lambda_{i,t}(\mathtt{OM}_{i,t}, a_{i,t}|s_{i,t}, h_{t}, x_{i,t}; \Theta')\equiv \phi'_{i,t}(h_{t}|c_{i,t}) \mathbf{1}_{\{ \mathtt{OM}_{i,t} =1 \}}+ \max\limits_{L\in\mathbb{T}_{t,T}} G_{i,t}(a_{i,t}|s_{i,t}, h_{t}, L, x_{i,t}; \Theta') \mathbf{1}_{\{ \mathtt{OM}_{i,t} =0 \}}.
    \end{aligned}
\]
By the construction of $\Theta$, the obedient policy $\pi_{i,t}(\cdot)=\sigma_{i,t}(\cdot)$ leads to the same action choice for every state as $\sigma'_{i,t}$ does, thereby inducing the same history of actions.
Given the same base game model $\mathcal{G}$, we have $\gamma^{\sigma'}_{\pi',\tau'} = \gamma^{\sigma}_{\sigma,\tau}$.
In addition, from the constructions of $\rho$ and $\phi$, respectively from $\rho'$ and $\phi'$, it holds that $Q(\gamma^{\sigma'}_{\pi',\tau'};\mathcal{G}^{\Theta'}) = Q(\gamma^{\sigma}_{\sigma,\tau};\mathcal{G}^{\Theta})$.

Third, we show that $<\tau, \sigma>$ constitutes a PBE of the game $\mathcal{G}^{\Theta}$.
Since $<\tau', \pi'>$ is a PBE of the game $\mathcal{G}^{\Theta'}$, 
\[
<\tau'_{i,t}(s_{i,t}, h_{t}),\pi'_{i,t}(s_{i,t}, h_{t})> \in D_{i,t}(s_{i,t},h_{t}|x'_{i,t}, \pi'_{-(i,t)}, \mathcal{G}^{\Theta'}),
\]
where each conjecture $x'_{i,t}$ satisfies (\ref{eq:x_coincides_tau}) given $\tau'$.
That is, for all $i\in\mathcal{N}$, $t\in\mathbb{T}$, $s_{i,t}\in S_{i,t}$, $\mathtt{OM}'_{i,t}\in\{0,1\}$, $a'_{i,t}\in A_{i,t}[\sigma']$,
\[
\Lambda_{i,t}(\mathtt{OM}_{i,t}, a_{i,t}|s_{i,t}, h_{t}, x_{i,t}; \Theta') \geq \Lambda_{i,t}(\mathtt{OM}'_{i,t}, a'_{i,t}|s_{i,t}, h_{t}, x'_{i,t}; \Theta'),
\]
where $\mathtt{OM}_{i,t}=\tau'_{i,t}(s_{i,t},h_{t})=\tau_{i,t}(s_{i,t}, h_{t})$ and $a_{i,t}=\pi_{i,t}(s_{i,t}, h_{t})=\sigma_{i,t}(s_{i,t},h_{t})$.
Due to $\tau_{i,t}(\cdot) = \tau_{i,t}(\cdot)$, the associated conjectures $x_{i,t}=x'_{i,t}$.
Since $\Lambda_{i,t}(\mathtt{OM}_{i,t}, a_{i,t}|s_{i,t}, h_{t}, x'_{i,t}; \Theta') =\Lambda_{i,t}(\mathtt{OM}_{i,t}, a_{i,t}|s_{i,t}, h_{t}, x_{i,t}; \Theta)$, it is straightforward to see that 
\[
<\tau_{i,t}(s_{i,t}, h_{t}),\sigma_{i,t}(s_{i,t}, h_{t})> \in D_{i,t}(s_{i,t},h_{t}|x_{i,t}, \sigma_{-(i,t)}, \mathcal{G}^{\Theta}).
\]
That is, $<\tau,\sigma>$ is a PBE of $\mathcal{G}^{\Theta}$.
%
\hfill $\square$

\section{Proof of Theorem \ref{thm:payoff_flow_conservation_sufficient}}\label{app:thm:payoff_flow_conservation_sufficient}

\textbf{Notations. }
We use $(a_{i,t}, s_{i,t})$ and $(\hat{a}_{i,t}, \hat{s}_{i,t})$ denote two pairs of obedient action and the corresponding state and use $(\hat{a}_{i,t}, s_{i,t})$ or $(a_{i,t}, \hat{s}_{i,t})$ to denote any pair of disobedient action and the corresponding state.
For ease of notation, we drop the conjectures $x=(x_{i,t})$ on the notations.
In addition, we make the following simplifications: $g_{i,t}(a'_{i,t}, L) = g_{i,t}(a'_{i,t}, s_{i,t}, h_{t}, L|\theta_{i,t})$; $\phi_{i,t}(h_{t}) = \phi_{i,t}(h_{t}|c_{i,t})$; $\sigma_{i,t}(s_{i,t}) = \sigma_{i,t}(s_{i,t}, h_{t})$ and $\rho_{i,t}(a_{t}) = \rho_{i,t}(a_{t}, h_{t})$; $\mathtt{Mg}_{i,t}(a_{i,t}) = \mathtt{Mg}_{i,t}(s_{i,t}, h_{t}|\theta_{i,t})$ for $a_{i,t}=\sigma_{i,t}(s_{i,t}, h_{t})$; $\mathbb{E}^{\sigma}[\cdot|a'_{i,t}, h_{t}] = \mathbb{E}^{\sigma}_{a'_{i,t}}[\cdot|s_{i,t}, h_{t}]$; unless otherwise stated.

For any $\hat{a}_{i,t}\in A_{i,t}[\sigma]$, $h_{t}\in H_{t}$, let
\[
\mathtt{E}\rho_{i,t}(\hat{a}_{i,t}, h_{t}) \equiv \mathbb{E}^{F_{-i,t}}\Big[\rho_{i,t}(\hat{a}_{i,t}, \sigma_{-i,t}(\tilde{s}_{-i,t}, h_{t}), h_{t})\Big|h_{t}\Big]. 
\]
That is, $\mathtt{E}\rho_{i,t}(a_{i,t}, h_{t})$ is agent $i$'s expected coupling value if he plays (obedient) $a_{i,t}$.
By rearranging (\ref{eq:feasible_rho_C1_original}), we obtain
\begin{equation}\label{eq:payment_marginal_relation_S_appendix_prf_thm1}
    \begin{aligned}
        &\mathtt{E}\rho_{i,t}(\hat{a}_{i,t}, h_{t})= \mathtt{Mg}_{i,t}(a_{i,t}) - \mathbb{E}^{\sigma}\Big[\mathtt{Mg}_{i,t+1}(\tilde{a}_{i,t+1})\Big|a_{i,t}, h_{t}\Big]- \mathbb{E}^{F_{-i,t}}\Big[ u_{i,t}(s_{i,t}, a_{i,t})\Big| h_{t} \Big].
    \end{aligned}
\end{equation}
We then construct the off-switch functions using (\ref{eq:payment_marginal_relation_S_appendix_prf_thm1}) by rearranging (\ref{eq:thm_payoff_conservation_phi}):
\begin{equation}\label{eq:thm_payoff_conservation_phi_S_appendix_prf_thm1}
    \begin{aligned}
    &\phi_{i,t}(h_{t}|c_{i,t})= \eta_{i,t}(h_{t}|c_{i,t}) + g_{i,t-1}(a_{i,t-1}, t-1)-\mathtt{Mg}_{i,t}(a_{i,t})+ \mathbb{E}^{\sigma }\Big[\mathtt{Mg}_{i,t+1}(\tilde{a}_{i,t+1}) \Big|a_{i,t}, h_{t}\Big].
    \end{aligned}
\end{equation}
From (\ref{eq:payment_marginal_relation_S_appendix_prf_thm1}),
\[
\begin{aligned}
    \mathtt{E}\rho_{i,t}(\hat{a}_{i,t}, h_{t}) - \mathtt{E}\rho_{i,t}(a_{i,t}, h_{t})&=\mathtt{Mg}_{i,t}(\hat{a}_{i,t})-\mathtt{Mg}_{i,t}(a_{i,t}) -u_{i,t}(s_{i,t}, \hat{a}_{i,t}) + u_{i,t}(s_{i,t}, a_{i,t})\\
    &- \mathbb{E}^{\sigma}\Big[\mathtt{Mg}_{i,t+1}(\tilde{a}_{i,t+1})\Big|\hat{a}_{i,t}, h_{t}\Big] + \mathbb{E}^{\sigma}\Big[\mathtt{Mg}_{i,t+1}(\tilde{a}_{i,t+1})\Big|a_{i,t}, h_{t}\Big].
\end{aligned}
\]
The condition \ref{itm:C3} implies
\begin{equation}\label{eq:app_thm_1_eq_1}
    \begin{aligned}
        &\mathtt{E}\rho_{i,t}(\hat{a}_{i,t}, h_{t}) - \mathtt{E}\rho_{i,t}(a_{i,t}, h_{t})\leq -u_{i,t}(s_{i,t}, \hat{a}_{i,t}) + u_{i,t}(s_{i,t}, a_{i,t})-\mathbb{E}^{\sigma}\Big[\mathtt{Mg}_{i,t+1}(\tilde{a}_{i,t+1})\Big|\hat{a}_{i,t}, h_{t}\Big]\\
        &+ \mathbb{E}^{\sigma}\Big[\mathtt{Mg}_{i,t+1}(\tilde{a}_{i,t+1})\Big|a_{i,t}, h_{t}\Big]+ \min\limits_{ L\in\mathbb{T}_{t,T}}\Big( \lambda_{i,t}(\hat{a}_{i,t}, \hat{a}_{i,t},L)-\lambda_{i,t}(a_{i,t}, \hat{a}_{i,t},L) -\eta_{i,L+1}(\tilde{h}_{L+1}|c_{i,L})\Big).
    \end{aligned}
\end{equation}
Due to (\ref{eq:feasible_rho_C1_original}), we have, for any $L\in \mathbb{T}_{t,T}$, $s_{i,t}\in S_{i,t}$, $a_{i,t},\hat{a}_{i,t}\in A_{i,t}[\sigma]$,
\[
\begin{aligned}
    &u_{i,t}\big(s_{i,t}, \hat{a}_{i,t}\big) + \mathbb{E}^{\sigma}\Big[\mathtt{Mg}_{i,t+1}(\tilde{a}_{i,t+1})\Big|\hat{a}_{i,t}, h_{t}\Big]= \mathbb{E}^{\sigma}\Big[ \sum\limits_{k=t}^{L} u_{i,s}(\tilde{a}_{i,k}, \sigma_{i,k}(\tilde{a}_{k})) +  \sum\limits_{k=t+1}^{L}\rho_{i,k}(\tilde{a}_{k} ) + \mathtt{Mg}_{i,L+1}(\tilde{a}_{i,L+1})  \Big|\hat{a}_{i,t}, h_{t}\Big].
\end{aligned}
\]
Then, the definition of $\lambda_{i,t}$ and the condition \ref{itm:C1} yields:
\begin{equation}\label{eq:app_thm_1_eq_2-1}
    \begin{aligned}
        &\min\limits_{ L\in \mathbb{T}_{t,T}}\Bigg( \lambda_{i,t}(\hat{a}_{i,t}, \hat{a}_{i,t},L|\hat{s}_{i,t}, h_{t}) - \mathbb{E}^{\sigma}\Big[\sum\limits_{k=t}^{L} u_{i,s}(\tilde{a}_{i,k}, \sigma_{i,k}(\tilde{a}_{k})) +  \sum\limits_{k=t+1}^{L}\rho_{i,k}(\tilde{a}_{k} )\\
        &+ \mathtt{Mg}_{i,L+1}(\Tilde{a}_{i,L+1})\Big|\hat{a}_{i,t}, h_{t}\Big] -\lambda_{i,t}(a_{i,t}, \hat{a}_{i,t},L|s_{i,t}, h_{t})-\eta_{i,L+1}(\tilde{h}_{L+1}|c_{i,L+1})\Bigg)\\
        &= \min\limits_{ L\in \mathbb{T}_{t,T}}\Bigg( -\lambda_{i,t}(a_{i,t}, \hat{a}_{i,t},L|s_{i,t}, h_{t}) -\eta_{i,L+1}(\tilde{h}_{L+1}|c_{i,L+1})\\
        &+\mathbb{E}^{\sigma}\Big[u_{i,L}\big( \tilde{s}_{i,L}, \tilde{a}_{L} \big)  +\rho_{i,L}( \tilde{a}_{L}) +\phi_{i,L+1}(\tilde{h}_{L+1})- \mathtt{Mg}_{i,L}(\tilde{a}_{i,L})\Big| \hat{a}_{i,t}, h_{t}\Big]  \Bigg).
    \end{aligned}
\end{equation}

Since $\rho_{i,L}$ and $\phi_{i,L+1}$ satisfy the conditions \ref{itm:C1} and \ref{itm:C2}
, respectively, we have, for any $L'\in \mathbb{T}_{L,T}$, 
\begin{equation}\label{eq:app_thm_1_eq_3}
    \begin{aligned}
        &\mathbb{E}^{\sigma}\Big[u_{i,L}\big( \tilde{s}_{i,L},  \tilde{a}_{L} \big) +\rho_{i,L}( \tilde{a}_{L}) +\phi_{i,L+1}\big(\tilde{h}_{L}\big|c_{L+1}\big)- \mathtt{Mg}_{i,L}(\tilde{a}_{i,L})  \Big| \hat{a}_{i,t}, h_{t} \Big]\\
        &=\mathbb{E}^{\sigma}\Big[ \eta_{i,L+1}(\Tilde{h}_{L+1}|c_{i,L+1}) + g_{i,L}(\Tilde{a}_{i,L}, L') - \mathtt{Mg}_{i,L}(\Tilde{a}_{i,L}) \Big| \hat{a}_{i,t}, h_{t} \Big].
    \end{aligned}
\end{equation}
Since $\mathbb{E}^{\sigma}[\cdot|a'_{i,t}, h_{t}] = \mathbb{E}^{\sigma}_{a'_{i,t}}[\cdot|s_{i,t}, h_{t}]$, (\ref{eq:app_thm_1_eq_3}) $\leq \mathbb{E}^{\sigma}\Big[ \eta_{i,L+1}(\Tilde{h}_{L+1}|c_{i,L+1})\Big| \hat{a}_{i,t}, h_{t} \Big]$.
Hence, the definition of $\lambda_{i,t}$ gives
\[
\begin{aligned}
    (\ref{eq:app_thm_1_eq_2-1})\leq& \min\limits_{ L\in \mathbb{T}_{t,T}} -\lambda_{i,t}(a_{i,t}, \hat{a}_{i,t},L|s_{i,t}, h_{t})
    =\min\limits_{ L\in \mathbb{T}_{t,T}}\Big(-J_{i,t}(\hat{a}_{i,t}, L|s_{i,t}, h_{t})-\mathbb{E}^{\sigma}\Big[ \phi_{i,L+1}\big(\tilde{h}_{L+1}\big)\Big|\hat{a}_{i,t}, h_{t} \Big]\Big) + \mathtt{E}\rho_{i,t}(\hat{a}_{i,t}, h_{t}),
\end{aligned}
\]
which implies
\[
\begin{aligned}
    (\ref{eq:app_thm_1_eq_1}) 
    &\leq \mathbb{E}^{\sigma}\Big[ \sum\limits_{k=t}^{T} u_{i,k}(\tilde{s}_{i,k}, \tilde{a}_{k}) +  \sum\limits_{k=t+1}^{T}\rho_{i,k}(\tilde{a}_{k})   \Big| a_{i,t}, h_{t}   \Big]\\
    &+ \min\limits_{ L\in \mathbb{T}_{t,T}}\Big(-J_{i,t}(\hat{a}_{i,t}, L|s_{i,t}, h_{t}) - \mathbb{E}^{\sigma}\Big[ \phi_{i,L+1}(\tilde{h}_{L+1}|c_{i,L+1})\Big|\hat{a}_{i,t}, h_{t} \Big] \Big)+ \mathtt{E}\rho_{i,t}(\hat{a}_{i,t}, h_{t})\\
    & \leq \max\limits_{L\in \mathbb{T}_{t,T}}\Big(  \sum\limits_{k=t}^{L} u_{i,k}(\tilde{s}_{i,k}, \tilde{a}_{k}) +  \sum\limits_{k=t+1}^{L}\rho_{i,k}(\tilde{a}_{k}) +   \phi_{i,L+1}(\tilde{h}_{L+1}|c_{i,L+1})\Big|a_{i,t}, h_{t} \Big]\Big)\\
    &+ \min\limits_{ L\in \mathbb{T}_{t,T}}\Big(-J_{i,t}(\hat{a}_{i,t}, L|s_{i,t}, h_{t}) - \mathbb{E}^{\sigma}\Big[ \phi_{i,L+1}(\tilde{h}_{L+1}|c_{i,L+1})\Big|\hat{a}_{i,t}, h_{t} \Big] \Big)+ \mathtt{E}\rho_{i,t}(\hat{a}_{i,t}, h_{t})\\
    &=\max\limits_{L\in \mathbb{T}_{t,T}}\Big(  \sum\limits_{k=t}^{L} u_{i,k}(\tilde{s}_{i,k}, \tilde{a}_{k}) +  \sum\limits_{k=t+1}^{L}\rho_{i,k}(\tilde{a}_{k}) +   \phi_{i,L+1}(\tilde{h}_{L+1}|c_{i,L+1})\Big|a_{i,t}, h_{t} \Big]\Big)\\
    &-\max\limits_{L\in \mathbb{T}_{t,T}}\Big(  \sum\limits_{k=t}^{L} u_{i,k}(\tilde{s}_{i,k}, \tilde{a}_{k}) +  \sum\limits_{k=t+1}^{L}\rho_{i,k}(\tilde{a}_{k}) +   \phi_{i,L+1}(\tilde{h}_{L+1}|c_{i,L+1})\Big|\hat{a}_{i,t}, h_{t} \Big]\Big).
\end{aligned}
\]
%
%
%
From the definition of $Z_{i,t}$ in (\ref{eq:function_z}), we have
\[
\begin{aligned}
    &Z_{i,t}(a_{i,t}|s_{i,t}, h_{t}; \sigma,\rho, \phi)=\mathtt{E}\rho_{i,t}(\hat{a}_{i,t}, h_{t}) +\max\limits_{L\in \mathbb{T}_{t,T}}\Big(  \sum\limits_{k=t}^{L} u_{i,k}(\tilde{s}_{i,k}, \tilde{a}_{k}) +  \sum\limits_{k=t+1}^{L}\rho_{i,k}(\tilde{a}_{k}) +   \phi_{i,L+1}(\tilde{h}_{L+1}|c_{i,L+1})\Big|a_{i,t}, h_{t} \Big]\Big).
\end{aligned}
\]
Therefore, we conclude that the conditions \ref{itm:C1}-\ref{itm:C3} implies the RAIC; i.e., for all $i\in\mathcal{N}$, $t\in\mathbb{T}$, $s_{i,t}\in S_{i,t}$, $h_{t}\in H_{t}$,
\[
Z_{i,t}(a_{i,t}|s_{i,t}, h_{t}; \sigma,\rho, \phi)\geq Z_{i,t}(\hat{a}_{i,t}|s_{i,t}, h_{t}; \sigma,\rho, \phi),
\]
where $a_{i,t}=\sigma_{i,t}(s_{i,t}, h_{t})$ is obedient.
%
%
\hfill $\square$

\section{ Proof of Proposition \ref{prop:essential_region_DR_SD} }\label{app:prop:essential_region_DR_SD}

Let $\hat{S}_{i,t} \equiv \overrightarrow{M}_{i,t}(h_{t}|\mathring{S}_{i,t}) = \mathring{S}_{i,t}$.
Let $s'_{i,t}\in \hat{S}_{i,t}$ and $s_{i,t}\in S_{i,t}\backslash \hat{S}_{i,t}$ denote two typical states in $\hat{S}_{i,t}$ and $\hat{S}_{i,t}$, respectively.
From the definition of $\vec{S}_{i,t}[y^{*}_{i,t}]$, we have $\zeta_{i,t}(s_{i,t},h_{t})\geq \zeta_{i,t}(s'_{i,t}, h_{t})$
%
%
which implies 
\begin{equation}\label{eq:app_prop_6_0}
    \Upsilon_{i,t}(s_{i,t},s'_{i,t} |h_{t})\geq  0.
\end{equation}
The (\ref{eq:dcr_conditions}) condition implies for all $i\in\mathcal{N}$, $t\in\mathbb{T}$, $\tau\in\mathbb{T}_{t,T}$,$h_{t}\in H_{t}$,
\[
\begin{aligned}
\int^{s_{i,t}}_{\theta_{i,t}}\mathbb{E}^{\sigma}\Big[ \beta^{\tau-t}_{i} \frac{ \partial  }{ \partial v } u_{i,\tau}(v, \tilde{\sigma}_{i,\tau}(\tilde{s}_{i,\tau}, \tilde{h}_{i,\tau} ) )\big|_{v=\tilde{s}_{i,\tau}} \Big|s_{i,t}, h_{t}\Big]\geq \int^{\hat{s}_{i,t}}_{\theta_{i,t}}\mathbb{E}^{\sigma}\Big[ \beta^{\tau-t}_{i}\frac{ \partial  }{ \partial v } u_{i,\tau}(v, \tilde{\sigma}_{i,\tau}(\tilde{s}_{i,\tau}, \tilde{h}_{i,\tau} ) )\big|_{v=\tilde{s}_{i,\tau}} \Big|s'_{i,t}, h_{t}\Big].
\end{aligned}
\]
Thus, 
\[
\begin{aligned}
    \mathbb{E}^{\sigma}\Big[\mathtt{Mg}_{i,t+1}(\tilde{s}_{i,t+1}, \tilde{h}_{t+1})\Big|s_{i,t}, h_{t}\Big] \geq  \mathbb{E}^{\sigma}\Big[\mathtt{Mg}_{i,t+1}(\tilde{s}_{i,t+1}, \tilde{h}_{t+1})\Big|s'_{i,t}, h_{t}\Big],
\end{aligned}
\]
which implies
\begin{equation}\label{eq:app_prop_6_1}
    \begin{aligned}
\overline{\Xi}_{i,t}(s_{i,t}, s'_{i,t}, h_{t})\geq 0.
\end{aligned}
\end{equation}
Hence, (\ref{eq:app_prop_6_0}) and (\ref{eq:app_prop_6_1}) yield 
\begin{equation}\label{eq:app_prop_6_3}
     \Upsilon_{i,t}(s_{i,t},s'_{i,t} |h_{t}) + \overline{\Xi}_{i,t}(s_{i,t}, s'_{i,t}, h_{t})\geq 0.
\end{equation}

Let $\hat{s}_{i,t}\in \overrightarrow{M}_{i,t}(h_{t}|\mathring{S}_{i,t})$ be a state such that $\overrightarrow{M}_{i,t}(h_{t}|\mathring{S}_{i,t})\backslash \{\vec{s}_{i,t}\}$ is a dominated region of $\overrightarrow{M}_{i,t}(h_{t}|\mathring{S}_{i,t})$.
Then, following (\ref{eq:app_prop_6_0}) and (\ref{eq:app_prop_6_1}), we have, for all $s''_{i,t}\in \overrightarrow{M}_{i,t}(h_{t}|\mathring{S}_{i,t})\backslash \{\vec{s}_{i,t}\}$
\begin{equation}\label{eq:app_prop_6_4}
    \begin{aligned}
    \Upsilon_{i,t}(\hat{s}_{i,t},s'_{i,t} |h_{t}) + \overline{\Xi}_{i,t}(\hat{s}_{i,t}, s''_{i,t}, h_{t})\geq 0.
    \end{aligned}
\end{equation}
Then, from Definition \ref{def:essential_region}, (\ref{eq:app_prop_6_3}) and (\ref{eq:app_prop_6_4}) show that such dominated region $\overrightarrow{M}_{i,t}(h_{t}|\mathring{S}_{i,t})$ is an essential region.
%
\hfill $\square$

\section{ Proof of Theorem \ref{thm:SPIR_conditions} }\label{app:thm:SPIR_conditions}

First, we show by the following lemma that when the cutoff-switch is given by (\ref{eq:cutoff_switch_horizontal}), each agent $i$'s period-$t$ expected payoff-to-go given by (\ref{eq:to_go_off_menu_switch}) can be represented in terms of the carrier function given by (\ref{eq:cutoff_switch_horizontal}) and $\overline{\delta}_{i,t}$ given by (\ref{eq:uppted_delta_function}).
For simplicity, we temporally drop the conjecture in the notations.

\begin{lemma}\label{lemma:app_proof_thm2}
Let $<\sigma, \rho, \phi>$ be a RAIC delegation mechanism with the horizontal cutoff-switch $\phi(\cdot|c)$ given by (\ref{eq:cutoff_switch_horizontal}).
Then, each expected payoff-to-go in (\ref{eq:to_go_off_menu_switch}) is equivalent to the following: for all $i\in \mathcal{N}$, $t\in\mathbb{T}$, $s_{i,t}\in S_{i,t}$, $h_{t}\in H_{t}$,
\begin{equation}\label{eq:app_prospect_represent}
    \begin{aligned}
        \Lambda_{i,t}(\mathtt{OM}_{i,t}, a_{i,t}|s_{i,t},  h_{t}) = \mathtt{Mg}_{i,t}(\overline{D}_{i,t}\circ s_{i,t}, h_{t})+ \overline{\delta}_{i,t}(\overline{D}_{i,t}\circ s_{i,t}, h_{t}).
    \end{aligned}
\end{equation}
\end{lemma}

\proof{Proof.}
    From (\ref{eq:offswitch_payoff_to_go}) and (\ref{eq:to_go_off_menu_switch}), we can represent the prospect function of agent $i$ in period $t$ in terms of the current marginal carrier function and the next-period expected payoff-to-go. That is,
\[
 \begin{aligned}
    &G_{i,t}(a_{i,t}|s_{i,t}, h_{t})=\zeta_{i,t}(s_{i,t}, h_{t}) + \mathbb{E}^{\sigma}\Big[  \Lambda_{i,t+1}(\widetilde{\mathtt{OM}}_{i,t+1}, \tilde{a}_{i,t+1}|\tilde{s}_{i,t+1}, \tilde{h}_{t+1})  \Big| s_{i,t}, h_{t}\Big],
    \end{aligned}
\]
where $a_{i,t} = \sigma_{i,t}(s_{i,t}, h_{t})$.
Hence, we can write $\Lambda_{i,t}(\mathtt{OM}_{i,t}, a_{i,t}|s_{i,t}, h_{t})$ in the form of dynamic programming:
\begin{equation}\label{eq:app_thm_4_0}
    \begin{aligned}
\Lambda_{i,t}(\mathtt{OM}_{i,t}, a_{i,t}|s_{i,t}, h_{t}) &= \phi_{i,t}(h_{t}|c_{i,t})\mathbf{1}_{\{\mathtt{OM}_{i,t}=1\}}\\
&+ \Big( \zeta_{i,t}(s_{i,t}, h_{t}) + \mathbb{E}^{\sigma}\Big[ \Lambda_{i,t+1}(\widetilde{\texttt{OM}}_{i,t+1}, \tilde{a}_{i,t+1}|\tilde{s}_{i,t+1}, \tilde{h}_{t+1})  \Big| s_{i,t}, h_{t}\Big] \Big)\mathbf{1}_{\{\mathtt{OM}_{i,t}=0 \}}.
\end{aligned}
\end{equation}

In period $T$, 
\begin{equation}\label{eq:app_thm_4_1}
    \begin{aligned}
&\Lambda_{i,T}(\mathtt{OM}_{i,T}, a_{i,T}|s_{i,T}, h_{T})=\phi_{i,T}(h_{T}|c_{i,T}) \mathbf{1}_{\{\mathtt{OM}_{i,T}=1\}} + \zeta_{i,T}(s_{i,T}, h_{T})\mathbf{1}_{\{\mathtt{OM}_{i,t}=0\}}.\\
\end{aligned}
\end{equation}
From (\ref{eq:cutoff_switch_horizontal}),
\[
\begin{aligned}
\phi_{i,T}(h_{T}|c_{i,T}) =& \mathtt{Mg}_{i,t}(\overline{d}_{i,T}(b), h_{T})
= g_{i,T}( \overline{d}_{i,T}(b), h_{T}, T|\theta_{i,T}),
\end{aligned}
\]
and from (\ref{eq:def_marginal_carrier}),
\[
\begin{aligned}
    \zeta_{i,T}(s_{i,T}, h_{T}) =& \mathtt{Mg}_{i,t}(s_{i,T}, h_{T})
    = g_{i,T}(s_{i,T}, h_{T}, T|\theta_{i,T}).
    \end{aligned}
\]
Then, the definition of up-persistence transformation in (\ref{eq:def_up_transform}) implies, for any $b\in[B]$,
\[
\begin{cases}
g_{i,T}( \overline{d}_{i,T}(b), h_{T}, T|\theta_{i,T}) \geq \zeta_{i,T}(s_{i,T}, h_{T}), & \text{ if } s_{i,T}\in \vec{\Gamma}^{b}_{i,T}(c_{i,T}),\\
g_{i,T}( \overline{d}_{i,T}(b), h_{T}, T|\theta_{i,T})< \zeta_{i,T}(s_{i,T}, h_{T}),& \text{ if } s_{i,T}\not\in \vec{\Gamma}^{b}_{i,T}(c_{i,T}).
\end{cases}
\]
Equivalently, we have
\[
\begin{cases}
\phi_{i,T}(h_{T}|c_{i,T}) \geq \zeta_{i,T}(s_{i,T}, h_{T}), & \text{ if } s_{i,T}\in \vec{\Gamma}^{b}_{i,T}(c_{i,T}),\\
\phi_{i,T}(h_{T}|c_{i,T}) < \zeta_{i,T}(s_{i,T}, h_{T}), & \text{ if } s_{i,T}\not\in \vec{\Gamma}^{b}_{i,T}(c_{i,T}).
\end{cases}
\]
Therefore, (\ref{eq:app_thm_4_1}) can be equivalently represented as follows:
\[
\begin{aligned}
&\Lambda_{i,T}(\mathtt{OM}_{i,T}, a_{i,T}|s_{i,T}, h_{T}) = \mathtt{Mg}_{i,T}(\overline{D}_{i,T}\circ s_{i,T}, h_{T}).
\end{aligned}
\]

Then, in period $T-1$, we have
\[
\begin{aligned}
&\Lambda_{i, T-1}(\mathtt{OM}_{i,T-1}, a_{i,T-1}|s_{i,T-1}, h_{T-1})=\phi_{i,T-1}(h_{T-1}|c_{i,T-1})\mathbf{1}_{\{\mathtt{OM}_{i,T-1}=1\}}\\
& +\Big( \zeta_{i,T-1}(s_{i,T-1}, h_{T-1})+ \mathbb{E}^{\sigma}\Big[ \texttt{Mg}_{i,T}(\overline{D}_{i,t}\circ \tilde{s}_{i,T}, h_{T})\Big| s_{i,T-1}, h_{T-1} \Big]  \Big)\mathbf{1}_{\{\mathtt{OM}_{i,T-1}=0\}}\\
&= \phi_{i,T-1}(h_{T-1}|c_{i,T-1})\mathbf{1}_{\{\mathtt{OM}_{i,T}=1\}} +\Big( \zeta_{i,T-1}(s_{i,T-1}, h_{T-1})+ \overline{\texttt{DE}}^{T-1,T}\Big[\texttt{Mg}_{i,T}(\widetilde{\texttt{us}}_{i,T}, \tilde{h}_{T})\Big|s_{i,T-1}, h_{T-1}  \Big]\Big)\mathbf{1}_{\{\mathtt{OM}_{i,T-1}=0\}}.
\end{aligned}
\]
Here, the cutoff-switch is
\[
\begin{aligned}
&\phi_{i,T-1}(h_{T-1}|c_{i,T-1}) = \mathtt{Mg}_{i,T-1}(\overline{d}_{i,T-1}(b), h_{T-1})+ \overline{\delta}_{i,T-1}(\overline{d}_{i,T-1}(b), h_{T-1})\\
&= \mathtt{Mg}_{i,T-1}(\overline{d}_{i,T-1}(b), h_{T-1})+ \overline{\mathtt{DE}}^{T-1,T}\Big[\mathtt{Mg}_{i,T-1}(\widetilde{\mathtt{us}}_{i,T}, \tilde{h}_{T})\Big|\overline{d}_{i,T-1}(b), h_{T-1} \Big] - \mathbb{E}^{\sigma}\Big[\mathtt{Mg}_{i,T-1}(\tilde{s}_{i,T}, \tilde{h}_{T})\Big| \overline{d}_{i,T-1}(b), h_{T-1} \Big].
\end{aligned}
\]
From the definition of \texttt{uT} in (\ref{eq:def_up_transform}), we have:
\begin{itemize}
    \item If $s_{i,T-1}\in \vec{\Gamma}^{b}_{i,T-1}(c_{i,T-1})$, then
\[
\begin{aligned}
&\mathtt{Mg}_{i,T-1}(\overline{d}_{i,T-1}(b), h_{T-1})-\mathbb{E}^{\sigma}\Big[\mathtt{Mg}_{i,T}(\tilde{s}_{i,T}, \tilde{h}_{T})\Big| \overline{d}_{i,T-1}(b), h_{T-1} \Big]\geq \zeta_{i,T-1}(s_{i,T-1}, h_{T-1}).
\end{aligned}
\]
    \item If $s_{i,T-1}\not\in \vec{\Gamma}^{b}_{i,T-1}(c_{i,T-1})$, then
\[
\begin{aligned}
&\mathtt{Mg}_{i,t}(\overline{d}_{i,T-1}(b), h_{T-1})-\mathbb{E}^{\sigma}\Big[\mathtt{Mg}_{i,t}(\tilde{s}_{i,T}, \tilde{h}_{T})\Big| \overline{d}_{i,T-1}(b), h_{T-1} \Big]< \zeta_{i,T-1}(s_{i,T-1}, h_{T-1}).
\end{aligned}
\]
\end{itemize}

In addition,
\[
\begin{aligned}
&\overline{\texttt{DE}}^{T-1,T}\Big[\texttt{Mg}_{i,T-1}(\widetilde{\texttt{us}}_{i,T}, \tilde{h}_{T})\Big|\overline{d}_{i,T-1}(b), h_{T-1} \Big]\mathbf{1}{\{\mathtt{OM}_{i,T-1}=1\}}+\overline{\texttt{DE}}^{T-1,T}\Big[\texttt{Mg}_{i,T}(\widetilde{\texttt{us}}_{i,T}, \tilde{h}_{T}) \Big|s_{i,T-1}, h_{T-1}  \Big]\mathbf{1}{\{\mathtt{OM}_{i,T-1}=0\}}\\
&= \overline{\mathtt{DE}}^{T-1,T}\Big[\mathtt{Mg}_{i,t}(\widetilde{\mathtt{us}}_{i,T}, \tilde{h}_{T})\Big|\overline{D}_{i,T-1}\circ s_{i,T-1}, h_{T-1} \Big].
\end{aligned}
\]

Hence, the expected payoff-to-go function in $T-1$ can be represented as:
\[
\begin{aligned}
    &\Lambda_{i, T-1}(\mathtt{OM}_{i,T-1}, a_{i,T-1}|s_{i,T-1}, h_{T-1})= \mathtt{Mg}_{i,t}(\overline{D}_{i,T-1}\circ s_{i,T-1}, h_{T-1})- \mathbb{E}^{\sigma}\Big[\mathtt{Mg}_{i,t}(\tilde{s}_{i,T}, \tilde{h}_{T}) \Big| \overline{D}_{i,T-1}\circ s_{i,T-1}, h_{T-1} \Big] \\
    &+ \overline{\mathtt{DE}}^{T-1,T}\Big[\mathtt{Mg}_{i,t}(\widetilde{\mathtt{us}}_{i,T}, \tilde{h}_{T})\Big|\overline{D}_{i,T-1}\circ s_{i,T-1}, h_{T-1} \Big]\\
    &= \mathtt{Mg}_{i,t}(\overline{D}_{i,T-1}\circ s_{i,T-1}, h_{T-1}) + \overline{\delta}_{i,T-1}(\overline{D}_{i,T-1}\circ s_{i,T-1}, h_{T-1}).
\end{aligned}
\]

Suppose, by hypothesis, the following holds, in period $t+1$,
\[
\begin{aligned}
\Lambda_{i,t+1}(\mathtt{OM}_{i,t+1}, a_{i,t+1}|s_{i,t+1}, h_{t+1}) = \mathtt{Mg}_{i,t}(\overline{D}_{i,t}\circ s_{i,t+1}, h_{t+1}) +  \overline{\delta}_{i,t+1}(\overline{D}_{i,t}\circ s_{i,t+1}, h_{t+1}).
\end{aligned}
\]
Then, in period $t$, we have
\[
\begin{aligned}
\Lambda_{i,t}(\mathtt{OM}_{i,t}, a_{i,t}|s_{i,t}, h_{t})
&= \phi_{i,t}(h_{t}|c_{i,t})\mathbf{1}{\{\mathtt{OM}_{i,t}=1 \}}+ \Big(\zeta_{i,t}(s_{i,t}, h_{t}) + \mathbb{E}^{\sigma}\Big[\mathtt{Mg}_{i,t+1}(\overline{D}_{i,t+1}\circ \tilde{s}_{i,t+1}, \tilde{h}_{t+1}) \\
& + \overline{\delta}_{i,t+1}(\overline{D}_{i,t+1}\circ \tilde{s}_{i,t+1}, \tilde{h}_{t+1}) \Big|s_{i,t}, h_{t}   \Big]    \Big)\mathbf{1}{\{\mathtt{OM}_{i,t}=0 \}}.
\end{aligned}
\]
From (\ref{eq:uppted_delta_function}), 
\[
\begin{aligned}
 &\mathbb{E}^{\sigma}\Big[\mathtt{Mg}_{i,t+1}(\overline{D}_{i,t+1}\circ \tilde{s}_{i,t+1}, \tilde{h}_{t+1}) + \overline{\delta}_{i,t+1}(\overline{D}_{i,t+1}\circ \tilde{s}_{i,t+1}, \tilde{h}_{t+1})\Big|s_{i,t}, h_{t}    \Big] \\
 &= \overline{\mathtt{DE}}^{t,T}\Big[g^{\sigma}_{i,t+1}( \widetilde{\mathtt{us}}_{i,t+1}, \tilde{h}_{t+1}, L)+\sum\limits_{k=t+2}^{T} Z^{\sigma}_{i,k}(\widetilde{\mathtt{us}}_{i,k}, \widetilde{\mathtt{us}}_{i,k-1}|\tilde{h}_{k})  \Big|s_{i,t}, h_{t}\Big],
\end{aligned}
\]
and the cutoff-switch in period $t$ is
\[
\begin{aligned}
&\phi_{i,t}(h_{t}|c_{i,t}) = \mathtt{Mg}_{i,t}(\overline{d}_{i,t}(b), h_{t})+ \overline{\mathtt{DE}}^{t,T}\Big[\sum\limits_{k=t+1}^{T} Z^{\sigma}_{i,k}(\widetilde{\mathtt{us}}_{i,k}, \widetilde{\mathtt{us}}_{i,k-1}|\tilde{h}_{k})\Big| \overline{d}_{i,t}(b), h_{t} \Big]\\
&= \mathtt{Mg}_{i,t}(\overline{d}_{i,t}(b), h_{t})- \mathbb{E}^{\sigma}\Big[ \mathtt{Mg}_{i,t+1}(\tilde{a}_{i,t+1}, \tilde{s}_{i,t+1}, \tilde{h}_{t+1})\Big| \overline{d}_{i,t}(b), h_{t}  \Big]\\
&+\overline{\mathtt{DE}}^{t,T}\Big[\mathtt{Mg}_{i,t+1}(\widetilde{\mathtt{us}}_{i,t+1}, \tilde{h}_{t+1}) + \sum\limits_{k=t+2}^{T} Z^{\sigma}_{i,k}(\widetilde{\mathtt{us}}_{i,k}, \widetilde{\mathtt{us}}_{i,k-1}|\tilde{h}_{k})\Big| \overline{d}_{i,t}(b), h_{t} \Big].
\end{aligned}
\]
The operator \texttt{uT} implies:
\begin{itemize}
    \item If $s_{i,t}\in \vec{\Gamma}^{b}_{i,t}(c_{i,t})$, then
\[
\begin{aligned}
&\mathtt{Mg}_{i,t}(\overline{d}_{i,t}(b), h_{t}) - \mathbb{E}^{\sigma}\Big[ \mathtt{Mg}_{i,t+1}(\tilde{a}_{i,t+1}, \tilde{s}_{i,t+1}, \tilde{h}_{t+1})\Big| \overline{d}_{i,t}(b), h_{t}  \Big]\geq \zeta_{i,t}(s_{i,t}, h_{t}).
\end{aligned}
\]
    \item If $s_{i,t}\not\in \vec{\Gamma}^{b}_{i,t}(c_{i,t})$, then
\[
\begin{aligned}
&\mathtt{Mg}_{i,t}(\overline{d}_{i,t}(b), h_{t})- \mathbb{E}^{\sigma}\Big[ \mathtt{Mg}_{i,t+1}(\tilde{s}_{i,t+1}, \tilde{h}_{t+1})\Big| \overline{d}_{i,t}(b), h_{t}  \Big]< \zeta_{i,t}(s_{i,t}, h_{t}).
\end{aligned}
\]
\end{itemize}
In addition, 
\[
\begin{aligned}
&\overline{\texttt{DE}}^{t,T}\Big[g^{\sigma}_{i,t+1}( \widetilde{\texttt{us}}_{i,t+1}, \tilde{h}_{t+1}, L)+\sum\limits_{k=t+2}^{T} Z^{\sigma}_{i,k}(\widetilde{\texttt{us}}_{i,k}, \widetilde{\texttt{us}}_{i,k-1}|\tilde{h}_{k})  \Big|\overline{d}_{i,t}(b) , h_{t}\Big]\mathbf{1}_{\{\mathtt{OM}_{i,t}=1\}}\\
&+\overline{\texttt{DE}}^{t,T}\Big[g^{\sigma}_{i,t+1}( \widetilde{\texttt{us}}_{i,t+1}, \tilde{h}_{t+1}, L)+\sum\limits_{k=t+2}^{T}Z^{\sigma}_{i,k}(\widetilde{\texttt{us}}_{i,k}, \widetilde{\texttt{us}}_{i,k-1}|\tilde{h}_{k})  \Big|s_{i,t}, h_{t}\Big]\mathbf{1}_{\{\mathtt{OM}_{i,t}=0\}}\\
&= \overline{\mathtt{DE}}^{t,T}\Big[g^{\sigma}_{i,t+1}( \widetilde{\mathtt{us}}_{i,t+1}, \tilde{h}_{t+1}, L)+\sum\limits_{k=t+2}^{T} Z^{\sigma}_{i,k}(\widetilde{\mathtt{us}}_{i,k}, \widetilde{\mathtt{us}}_{i,k-1}|\tilde{h}_{k})  \Big|\overline{D}_{i,t}\circ s_{i,t}, h_{t}\Big].
\end{aligned}
\]
Suppose that $\mathtt{OM}_{i,t} = 1$ for all $s_{i,t}\in \vec{\Gamma}_{i,t}(c_{i,t})$ and $\mathtt{OM}_{i,t} = 0$ for all $s_{i,t}\not\in \vec{\Gamma}_{i,t}(c_{i,t})$.
Then, from the definition of \texttt{uT} in (\ref{eq:def_up_transform}), the expected payoff-to-go function can be represented as (\ref{eq:app_prospect_represent}), for all $i\in\mathcal{N}$, $t\in\mathbb{T}$.
%
\endproof

\subsection{$(\Leftarrow)$}
%

To prove the \textit{if} part of Theorem \ref{thm:SPIR_conditions} we need to show that the essential region shown in Theorem \ref{thm:SPIR_conditions} is indeed the OFR.
From the definitions of the essential regions in Definition \ref{def:essential_region} and the indifference region in (\ref{eq:indifference_region_def}), we have, for all $i\in\mathcal{N}$, $t\in\mathbb{T}$, $s_{i,t}\in S_{i,t}$, $h_{t}\in H_{t}$,
\[
\begin{cases}
\textup{RHS of (\ref{eq:app_prospect_represent}) } =  \phi_{i,t}(h_{t}|c_{i,t}), & \textup{ if } s_{i,t}\in  \vec{\Gamma}_{i,t}(c_{i,t}),\\
\textup{RHS of (\ref{eq:app_prospect_represent}) } \geq  \phi_{i,t}(h_{t}|c_{i,t}), & \textup{ if } s_{i,t}\not\in  \vec{\Gamma}_{i,t}(c_{i,t}).
\end{cases}
\]
From Lemma \ref{lemma:app_proof_thm2} and the definition of $\Lambda^{\sigma, \rho}_{i,t}$ in (\ref{eq:to_go_off_menu_switch}), i.e.,
\[
 \begin{aligned}
    &\Lambda_{i,t}(\mathtt{OM}_{i,t}, a_{i,t}|s_{i,t}, h_{t})\equiv \phi_{i,t}(h_{t}|c_{i,t}) \mathbf{1}_{\{ \mathtt{OM}_{i,t} =1 \}}+ G_{i,t}(a_{i,t}|s_{i,t}, h_{t}) \mathbf{1}_{\{ \mathtt{OM}_{i,t} =0 \}},
    \end{aligned}
\]
it is clear to see that the value of $\Lambda^{\sigma,\rho}_{i,t}$ specified by (\ref{eq:app_prospect_represent}) coincides with the following choices of $\mathtt{OM}_{i,t}$:
\[
\begin{cases}
\mathtt{OM}_{i,t} = 1, & \textup{ if } s_{i,t}\in \vec{\Gamma}_{i,t}(c_{i,t}),\\
\mathtt{OM}_{i,t} = 0, & \textup{ if } s_{i,t}\not\in \vec{\Gamma}_{i,t}(c_{i,t}).
\end{cases}
\]
Thus, the essential region shown in Theorem \ref{thm:SPIR_conditions} is the OFR.
Recall that the principal decides $\mathtt{OM}_{i,t}$ for each agent $i$ when his state is in the indifference region.
Hence, the principal can incentive the agents to take $\mathtt{OM}_{i,t}=0$ for all $i\in \mathcal{N}$, $t\in\mathbb{T}$.
Therefore, the RAIC mechanism is OAIC.

\subsection{$(\Rightarrow)$}
%

Now, we prove the \textit{only if} part.
Suppose that the RAIC mechanism with the horizontal cutoff-switch $\phi(\cdot|c)$ given by (\ref{eq:cutoff_switch_horizontal}) is also OAIC.
Then, every OFR is a sub indifference region; i.e., $\vec{\Gamma}_{i,t}(c_{i,t})=\mathtt{Id}_{i,t}(c_{i,t})$, for all $i\in\mathcal{N}$, $t\in\mathbb{T}$.
Hence, we need to show that every such sub indifference region is an essential region in every period.
As in (\ref{eq:app_thm_4_0}), each $\Lambda^{\sigma,\rho}_{i,t}$ can be written as
\[
\begin{aligned}
&\Lambda_{i,t}(\mathtt{OM}_{i,t}, a_{i,t}|s_{i,t}, h_{t}) \\
&= \phi_{i,t}(h_{t}|c_{i,t})\mathbf{1}_{\{\mathtt{OM}_{i,t}=1\}}+ \Big( \zeta_{i,t}(s_{i,t}, h_{t}) + \mathbb{E}^{\sigma}\Big[ \Lambda_{i,t+1}(\widetilde{\texttt{OM}}_{i,t+1}, \tilde{a}_{i,t+1}|\tilde{s}_{i,t+1}, \tilde{h}_{t+1})  \Big| s_{i,t}, h_{t}\Big] \Big)\mathbf{1}_{\{\mathtt{OM}_{i,t}=0 \}}.
\end{aligned}
\]
Then, for any two states $\hat{s}_{i,t}, \hat{s}'_{i,t}\in \mathtt{Id}_{i,t}(c_{i,t})$, we have
\begin{equation}\label{eq:app_equal_phi}
        \begin{aligned}
\phi_{i,t}(h_{t}|c_{i,t}) 
&= \zeta_{i,t}(\hat{s}_{i,t}, h_{t}) + \mathbb{E}^{\sigma}\Big[ \Lambda_{i,t+1}(\widetilde{\mathtt{OM}}_{i,t+1}, \tilde{a}_{i,t+1}|\tilde{s}_{i,t+1}, \tilde{h}_{t+1})  \Big| \hat{s}_{i,t}, h_{t}\Big]\\
&= \zeta_{i,t}(\hat{s}'_{i,t}, h_{t}) + \mathbb{E}^{\sigma}\Big[ \Lambda_{i,t+1}(\widetilde{\mathtt{OM}}_{i,t+1}, \tilde{a}_{i,t+1}|\tilde{s}_{i,t+1}, \tilde{h}_{t+1})  \Big| \hat{s}'_{i,t}, h_{t}\Big].
\end{aligned}
\end{equation}
From Lemma \ref{lemma:app_proof_thm2}, we have
\[
\begin{aligned}
&\zeta_{i,t}(s_{i,t}, h_{t}) +\mathbb{E}^{\sigma}\Big[ \Lambda_{i,t+1}(\widetilde{\mathtt{OM}}_{i,t+1}, \tilde{a}_{i,t+1}|\tilde{s}_{i,t+1}, \tilde{h}_{t+1})  \Big|s_{i,t}, h_{t}\Big]=\mathtt{Mg}_{i,t}(s_{i,t}, h_{t})+ \overline{\delta}_{i,t}(s_{i,t}, h_{t}).
\end{aligned}
\]
Hence, (\ref{eq:app_equal_phi}) implies
\[
\begin{aligned}
\mathtt{Mg}_{i,t}(\hat{s}_{i,t}, h_{t}) - \mathtt{Mg}_{i,t}(\hat{s}'_{i,t}, h_{t})+ \Big(\overline{\delta}_{i,t}(\hat{s}_{i,t}, h_{t})-\overline{\delta}_{i,t}(\hat{s}'_{i,t}, h_{t})\Big)&= \Upsilon_{i,t}(\hat{s}_{i,t},\hat{s}'_{i,t} |h_{t})+\overline{\Xi}_{i,t}(\hat{s}_{i,t}, \hat{s}'_{i,t}, h_{t})=0,
\end{aligned}
\]
which implies part \textit{(i)} of Definition \ref{def:essential_region}.
For any $s_{i,t} \in S_{i,t}\backslash\mathtt{Id}_{i,t}(c_{i,t})$, we have
\[
\begin{aligned}
& \zeta_{i,t}(s_{i,t}, h_{t}) + \mathbb{E}^{\sigma}\Big[ \Lambda^{\sigma,\rho}_{i,t+1}(\widetilde{\mathtt{OM}}_{i,t+1}, \tilde{a}_{i,t+1}|\tilde{s}_{i,t+1}, \tilde{h}_{t+1})  \Big| s_{i,t}, h_{t}\Big]\geq \phi_{i,t}(h_{t}|c_{i,t}).
\end{aligned}
\]
Then, for any $\hat{s}_{i,t}\in \mathtt{Id}_{i,t}(c_{i,t})$,
\[
\begin{aligned}
&\Big(\mathtt{Mg}_{i,t}(\hat{s}_{i,t}, h_{t})  - \mathtt{Mg}_{i,t}(s_{i,t}, h_{t})\Big) + \Big(\overline{\delta}_{i,t}(\hat{s}_{i,t}, h_{t})-\overline{\delta}_{i,t}(s_{i,t}, h_{t})\Big)= \Upsilon_{i,t}(\hat{s}_{i,t},s_{i,t} |h_{t})+\overline{\Xi}_{i,t}(\hat{s}_{i,t}, s_{i,t}, h_{t})\leq 0.
\end{aligned}
\]
Therefore, we obtain the part \textit{(ii)} of Definition \ref{def:essential_region}.
Hence, we conclude that $\vec{\Gamma}_{i,t}(c_{i,t})=\mathtt{Id}_{i,t}(c_{i,t})$ is an essential region of the state space $S_{i,t}$, for all $i\in\mathcal{N}$, $t\in\mathbb{T}$.
%
%
\hfill $\square$

\section{ Proof of Proposition \ref{thm:doic_irod_sufficient_condition} }\label{app:thm:doic_irod_sufficient_condition}

Suppose that $\rho\in \textup{\ref{eq:feasible_rho_C1_original}}$ and each cutoff-switch satisfies (\ref{eq:cutoff_switch_horizontal}).
Suppose in addition that (\ref{eq:sufficient_if_rho_c1}) holds.
Following similar procedures to prove Lemma \ref{lemma:app_proof_thm2}, we have
\begin{equation}\label{eq:app_thm4_e0}
    \Lambda_{i,t}(\mathtt{OM}_{i,t}, a'_{i,t}|s_{i,t},  h_{t}) = \mathtt{Mg}_{i,t}(a'_{i,t} , s_{i,t}, h_{t},x_{i,t}) + \overline{\delta}_{i,t}(a'_{i,t}, s_{i,t}, h_{t},x_{i,t}),
\end{equation}
for any arbitrary $a'_{i,t}\in A_{i,t}[\sigma]$.
Then, (\ref{eq:sufficient_if_rho_c1}) implies 
\begin{equation}\label{eq:app_thm4_e1}
    \Lambda_{i,t}(\mathtt{OM}_{i,t}=0|s_{i,t},  h_{t}) \geq \Lambda_{i,t}(\mathtt{OM}_{i,t}=y, a'_{i,t}|s_{i,t},  h_{t}),
\end{equation}
for all $i\in\mathcal{N}$, $t\in\mathbb{T}$, $s_{i,t}\in S_{i,t}$, $h_{t}\in H_{t}$, $y\in[0,1]$, any $a'_{i,t}\in A_{i,t}[\sigma]$.
Hence, the mechanism is DOIC.

Conversely, suppose that the mechanism is DOIC. Then, (\ref{eq:app_thm4_e1}) holds. From Lemma \ref{lemma:app_proof_thm2} and (\ref{eq:app_thm4_e0}), we obtain (\ref{eq:sufficient_if_rho_c1}).
%
%
\hfill $\square$

\section{Proof of Lemma \ref{lemma:existence_OM_strategies}}\label{app:lemma:existence_OM_strategies}

In every period $t$ of $\mathcal{G}^{\Theta}$ with a RAIC $\Theta$, the agents' stage-game of choosing $\tau=(\tau_{i,t})$ is a static Bayesian game with continuous states $S_{t}$ and discrete actions $\mathbb{T}_{t,T+1}$.
Consider a period-$t$ OM-strategy profile $\tau_{t}=(\tau_{i,t})$.
Let $\tau^{\natural}_{i,t}(k) = \mathbb{E}^{\sigma}\left[\mathbf{1}_{\{\tau_{i,t}(\tilde{s}_{i,t}, h_{t}) = k\}}|h_{t}\right]$.
Hence, a profile $\tau^{\natural}_{t}=(\tau^{\natural}_{i,t})$ can be interpreted as a mixed strategy profile for a normal-form game with finite actions (i.e., $\mathbb{T}_{t,T+1}$).
Then, by Kakutani’s Fixed Point Theorem, there exists a Nash equilibrium profile (\citet{kakutani1941generalization,yuan2017fixed}). 
 \hfill $\square$

\section{Proof of Proposition \ref{prop:condition_conjecture_MAOS_gen} }\label{app:prop:condition_conjecture_MAOS_gen}

The \textit{only if} part is straightforward.
In particular, since each $\chi_{i,t}$ satisfies (\ref{eq:prob_equal_suffice}) where each $\tau_{i,t}$ is a best response to $\chi_{-i,t}$, we have
\[
\begin{cases}
    \tau_{i,t}(s_{i,t}, h_{t})=t &, \textup{ if } s_{i,t}\in S^{\mathtt{off}}_{i,t},\\
   \tau_{i,t}(s_{i,t}, h_{t})=k>t &, \textup{ if } s_{i,t}\in S_{i,t}\backslash S^{\mathtt{off}}_{i,t}.
\end{cases}
\]
Then, $\mathbb{P}\big(\tau_{i,t}(s_{i,t}, h_{t})=t|h_{t}\big) = \mathbb{P}\big(s_{i,t}\in S^{\mathtt{off}}_{i,t}|h_{t}\big)$.
Since $x_{i,t}$ satisfies (\ref{eq:principal_desired_belief_sym}), we obtain 
\[
 \mathbb{P}(\tau_{i,t}(s_{i,t}, h_{t})=t|h_{t})  = \chi_{i,t}(t|h_{t}).
\]
Hence, by definition of $\bm{X}_{t}$, (\ref{eq:fixed_point_general}) holds.

Now we prove the \textit{if}. 
Suppose that each $\chi_{t}$ satisfies (\ref{eq:fixed_point_general}).
Then, there exists a measurable function $\tau_{i,t}:S_{i,t}\times H_{t}\mapsto \mathbb{T}_{t,T}$ for each $i\in\mathcal{N}$ such that the principal's desired $\chi_{t}=(\chi_{i,t}, \chi_{-i,t})$ satisfies, for all $i\in\mathcal{N}$, $t\in\mathbb{T}$, $s_{i,t}\in S_{i,t}$, $h_{t}\in H_{t}$,
\begin{equation}\label{eq:app_proof_thm7_0}
    \begin{aligned}
    \Lambda_{i,t}(k|s_{i,t}, h_{i,t}, \chi_{-i,t})\geq \Lambda_{i,t}(k'|s_{i,t}, h_{i,t}, \chi_{-i,t}),
\end{aligned}
\end{equation}
and 
\begin{equation}\label{eq:app_proof_thm7_1}
    \mathbb{P}\big(\tau_{i,t}(s_{i,t}, h_{t})=t|h_{t}\big)  = \chi_{i,t}(t|h_{t}) = \mathbb{P}\big( s''_{i,t}\in S^{\mathtt{off}}_{i,t}\big| h_{t} \big).
\end{equation}

We proceed with the proof by establishing a contradiction.
Suppose that there exists a set $S'_{i,t}\subseteq S_{i,t}\backslash S^{\mathtt{off}}_{i,t}$ such that $\tau_{i,t}(s'_{i,t}, h_{t})=t$ for all $s'_{i,t}\in S'_{i,t}$.
Then, from (\ref{eq:app_proof_thm7_1}), there must exist $S''_{i,t}\subseteq S^{\mathtt{off}}_{i,t}$ such that $\tau_{i,t}(s'_{i,t}, h_{t})\neq t$ for all $s'_{i,t}\in S'_{i,t}$.
In addition, $\mathbb{P}(s_{i,t}\in S^{\mathtt{off}}_{i,t}|h_{t}) = \mathbb{P}(s_{i,t}\in S'_{i,t} \textup{ or } s_{i,t}\in S^{\mathtt{off}}_{i,t}\backslash S''_{i,t}|h_{t})$.
Then, from (\ref{eq:app_proof_thm7_0}), we have that there exists $\ell=\tau_{i,t}(s''_{i,t}, h_{t})> t$, such that
\begin{equation}\label{eq:app_proof_thm7_2}
    \Lambda_{i,t}(t |s''_{i,t}, h_{t}, \chi_{-i,t}(h_{t}))< \Lambda_{i,t}( \ell|s''_{i,t}, h_{t}, \chi_{-i,t}(h_{t})), 
\end{equation}
for all $s''_{i,t}\in S''_{i,t}$, $h_{t}\in H_{t}$.
By definition of $\Lambda_{i,t}$, the LHS of (\ref{eq:app_proof_thm7_2}) is $\phi_{i,t}(h_{t}, \chi_{-i,t}(h_{t})|c_{i,t})$.

Suppose that each cutoff-switch $\phi_{i,t}$ is horizontal and is given by (\ref{eq:cutoff_horizontal_general}).
Then, from Lemma \ref{lemma:app_proof_thm2}, the RHS of (\ref{eq:app_proof_thm7_2}) is equivalent to (by putting back the conjecture $\chi_{-i,t}$)
\[
\begin{aligned}
    \Lambda_{i,t}(\ell|s''_{i,t},  h_{t}, \chi_{-i,t}) = \mathtt{Mg}_{i,t}(\overline{D}_{i,t}\circ s''_{i,t}, h_{t}, \chi_{-i,t})+ \overline{\delta}_{i,t}(\overline{D}_{i,t}\circ s''_{i,t}, h_{t}, \chi_{-i,t}).
\end{aligned}
\]
Then, from the definition of $\overline{D}_{i,t}$, (\ref{eq:app_proof_thm7_2}) becomes $\phi_{i,t}(h_{t}|c_{i,t})< \mathtt{Mg}_{i,t}\big(\overline{d}_{i,t}(b), h_{t},\chi_{-i,t}\big)+ \overline{\delta}_{i,t}\big(\overline{d}_{i,t}(b), h_{t},\chi_{-i,t} \big)$, which contradicts the setting that each cutoff-switch $\phi^{\natural}_{t}$ is horizontal and is given by (\ref{eq:cutoff_horizontal_general}).
Thus, $S''_{i,t}=\emptyset$. Therefore, $\mathbb{P}\big(\tau_{i,t}(s_{i,t}, h_{t})=t \textup{ and } s_{i,t}\in S^{\mathtt{off}}_{i,t}|h_{t}\big) = \chi_{i,t}(t|h_{t})$.

To prove the proposition when the principal is knowledgeable, we first establish a lemma similar to Lemma \ref{lemma:app_proof_thm2}.

\begin{lemma}\label{lemma:app_proof_prop_single_agent_switchable}
Suppose that $S^{\mathtt{off}}$ is desired by the principal, and let $\chi=(\chi_{i,t})$, where each $\chi_{i,t}\in \Delta\big(\mathbb{T}_{t,T} \big)^{n-1}$ satisfies (\ref{eq:principal_desired_belief_sym}) given $S^{\mathtt{off}}_{i,t}$ for all $i\in\mathcal{N}$, $t\in\mathbb{T}$.
Suppose that in addition, the principal is knowledgeable.
Let $<\sigma, \rho, \phi>$ be a RAIC delegation mechanism with the cutoff-switch $\phi(\cdot|c)$ given by (\ref{eq:cutoff_knowledgeable_general} ).
Then, each expected payoff-to-go given by (\ref{eq:to_go_off_menu_switch}) is equivalent to the following: for all $i\in \mathcal{N}$, $t\in\mathbb{T}$, $s^{w}_{i,t}\in S_{i,t}$, $h_{t}\in H_{t}$, $w\in\{b,e\}$, $b\in[B]$, $e\in[B'-B]$,
\begin{equation}\label{eq:app_switchable_to_go}
    \begin{aligned}
        \Lambda_{i,t}(\mathtt{OM}_{i,t}, a_{i,t}|s_{i,t}[w],  h_{t},\chi_{-i,t}) = \mathtt{Mg}_{i,t}(\overline{D}_{i,t}\circ s_{i,t}[w], h_{t},\chi_{-i,t})+ \overline{\delta}_{i,t}(\overline{D}_{i,t}\circ s_{i,t}[w], h_{t},\chi_{-i,t}).
    \end{aligned}
\end{equation}
\end{lemma}

\proof{Proof}
To simplify notations, we omit the conjecture $\chi_{-i,t}$.
The proof of this lemma follows a similar procedure as in the proof of Lemma \ref{lemma:app_proof_thm2}..
Let $s_{i,t}[w]$ be a typical state in $\vec{\Gamma}^{w}_{i,t}(c_{i,t})$ or $\vec{\Psi}^{w}_{i,t}(c_{i,t})$, for $w\in\{b,e\}$.
To highlight the location of the state (i.e., $w$), we denote the expected payoff-to-go as (with a slight abuse of notation) $\Lambda_{i,t}(w,\mathtt{OM}_{i,t}, a_{i,t}|s_{i,t}[w], h_{t})$.
As in (\ref{eq:app_thm_4_0}), the expected payoff-to-go function can be represented in terms of dynamic programming:
\begin{equation}
    \begin{aligned}
&\Lambda_{i,t}(w,\mathtt{OM}_{i,t}, a_{i,t}|s_{i,t}[w], h_{t})= \phi_{i,t}(w, h_{t}|c_{i,t})\mathbf{1}{\{\mathtt{OM}_{i,t}=1\}}\\
&+ \Big( \zeta_{i,t}(s_{i,t}[w], h_{t}) + \mathbb{E}^{\sigma}\Big[ \Lambda_{i,t+1}(\tilde{w},\widetilde{\texttt{OM}}_{i,t+1}, \tilde{a}_{i,t+1}|\tilde{s}_{i,t+1}[\tilde{w}], \tilde{h}_{t+1})  \Big| s_{i,t}[w], h_{t}\Big] \Big)\mathbf{1}{\{\mathtt{OM}_{i,t}=0 \}}.
\end{aligned}
\end{equation}
In period $T$,
\[
\begin{aligned}
\Lambda_{i,T}\big(w, \texttt{OM}_{i,T}, a_{i,T}|s_{i,T}[w], h_{T}\big)&=\phi_{i,T}(w, h_{T}|c_{i,T}) \mathbf{1}{\{\mathtt{OM}_{i,T}=1\}}+ \zeta_{i,T}(s_{i,T}[w], h_{T})\mathbf{1}{\{\mathtt{OM}_{i,t}=0\}}.
\end{aligned}
\]
Since the cutoff-switch is given by (\ref{eq:cutoff_knowledgeable_general}), 
\[
\begin{aligned}
\phi_{i,T}(w,h_{T}|c_{i,T}) = g_{i,T}(\overline{\underline{D}}_{i,T}\circ s_{i,T}[w], h_{T}, L |\theta_{i,T}).
\end{aligned}
\]
From the definitions of jump transformation in (\ref{eq:region_jump_operator}) and the marginal carrier in (\ref{eq:def_marginal_carrier}), we have
\[
\begin{cases}
g_{i,T}(\overline{\underline{D}}_{i,T}\circ s_{i,T}[w], h_{T}, L |\theta_{i,T}) \geq \zeta_{i,T}(s_{i,T}[w],h_{T}), & \textup{ if } s_{i,T}[w]\in \vec{\Gamma}^{w}_{i,T}(c_{i,T}),\\
g_{i,T}(\overline{\underline{D}}_{i,T}\circ s_{i,T}[w], h_{T}, L |\theta_{i,T}) \leq \zeta_{i,T}(s_{i,T}[w],h_{T}), & \textup{ if } s_{i,T}[w]\in \vec{\Psi}^{w}_{i,T}(c_{i,T}).
\end{cases}
\]
Hence,
\[
\begin{cases}
\phi_{i,T}(w,h_{T}|c_{i,T})\geq \zeta_{i,T}(s_{i,T}[w],h_{T}), & \textup{ if } s_{i,T}[w]\in \vec{\Gamma}^{w}_{i,T}(c_{i,T}),\\
\phi_{i,T}(w,h_{T}|c_{i,T}) < \zeta_{i,T}(s_{i,T}[w],h_{T}), & \textup{ if } s_{i,T}[w]\in \vec{\Psi}^{w}_{i,T}(c_{i,T}).
\end{cases}
\]
Thus, the period-$T$ expected payoff-to-go can be represented as follows:
\[
\Lambda_{i,T}(w,\mathtt{OM}_{i,T}, a_{i,T}|s_{i,T}[w], h_{T}) = g_{i,T}(\overline{D}_{i,T}\circ s_{i,T}[w], h_{T}, L |\theta_{i,T}),
\]
where $\overline{D}_{i,T}$ is the up transform defined in (\ref{eq:def_up_transform}). Then, in period $T-1$, we have
\[
\begin{aligned}
&\Lambda_{i, T-1}(w,\mathtt{OM}_{i,T-1}, a_{i,T-1}|s_{i,T-1}[w], h_{T-1}) =\phi_{i,T-1}(w,h_{T-1}|c_{i,T-1})\mathbf{1}{\{\mathtt{OM}_{i,t}=1\}}\\
&+\Big( \zeta_{i,T-1}(s_{i,T-1}[w], h_{T-1}) + \mathbb{E}^{\sigma}\Big[ \texttt{Mg}_{i,t}(\overline{D}_{i,T}\circ \tilde{s}_{i,T}[\tilde{w}], h_{T})\Big| s_{i,T-1}[w], h_{T-1} \Big]\Big)\mathbf{1}{\{\mathtt{OM}_{i,T-1}=0\}}.
\end{aligned}
\]
Since the cutoff-switch is constructed by (\ref{eq:cutoff_knowledgeable_general}), we have
\[
\begin{aligned}
&\phi_{i,T-1}(w,h_{T-1}|c_{i,T-1})=\mathtt{Mg}_{i,T-1}(\overline{\underline{D}}_{i,T-1}\circ s_{i,T-1}[w], h_{T-1})+ \overline{\delta}_{i,T-1}(\overline{\underline{D}}_{i,T-1}\circ s_{i,T-1}[w], h_{T-1};\beta_{i})\\
&=\mathtt{Mg}_{i,T-1}(\overline{\underline{D}}_{i,T-1}\circ s_{i,T-1}[w], h_{T-1})+\overline{\mathtt{DE}}^{T-1,T}\Big[\mathtt{Mg}_{i,T}(\widetilde{\mathtt{us}}_{i,T}, \tilde{h}_{T})\Big|\overline{\underline{D}}_{i,t}\circ s_{i,t}[w], h_{T-1} \Big] \\
&- \mathbb{E}^{\sigma}\Big[\mathtt{Mg}_{i,T}(\tilde{s}_{i,T}[\tilde{w}], \tilde{h}_{T})\Big| \overline{\underline{D}}_{i,T-1}\circ s_{i,T-1}[w], h_{T-1} \Big].
\end{aligned}
\]
Due to (\ref{eq:region_jump_operator}) and (\ref{eq:def_marginal_carrier}), we obtain the following two cases:
\begin{itemize}
    \item If $s_{i,T-1}[w]\in \vec{\Gamma}^{w}_{i,T-1}(c_{i,T-1})$, then
    \[
\begin{aligned}
&\mathtt{Mg}_{i,T-1}(\overline{\underline{D}}_{i,T-1}\circ s_{i,T-1}[w], h_{T-1}) - \mathbb{E}^{\sigma}\Big[\mathtt{Mg}_{i,T}(\tilde{s}_{i,T}[\tilde{w}], \tilde{h}_{T})\Big| \overline{\underline{D}}_{i,T-1}\circ s_{i,T-1}[w], h_{T-1} \Big]\geq \zeta_{i,T-1}(s_{i,T-1}, h_{T-1}).
\end{aligned}
\]

    \item If $s_{i,T-1}[w]\in \vec{\Psi}^{w}_{i,T-1}(c_{i,T-1})$, then
    \[
\begin{aligned}
&\mathtt{Mg}_{i,T-1}(\overline{\underline{D}}_{i,T-1}\circ s_{i,T-1}[w], h_{T-1}) - \mathbb{E}^{\sigma}\Big[\mathtt{Mg}_{i,T}(\tilde{s}_{i,T}, \tilde{h}_{T})\Big| \overline{\underline{D}}_{i,T-1}\circ s_{i,T-1}[w], h_{T-1} \Big]\leq \zeta_{i,T-1}(s_{i,T-1}[w], h_{T-1}).
\end{aligned}
\]
\end{itemize}
In addition,

\begin{equation}\label{eq:app_jump_up_1}
    \begin{aligned}
        &\begin{aligned}
            \overline{\mathtt{DE}}^{T-1,T}\Big[ \mathtt{Mg}_{i,T}(\widetilde{\mathtt{us}}_{i,T}, \tilde{h}_{T})&\Big|\overline{\underline{D}}_{i,t}\circ s_{i,t}[w], h_{T-1} \Big]\mathbf{1}{\{\mathtt{OM}_{i,T-1}=1\}} \\
        &+ \mathbb{E}^{\sigma}\Big[ \texttt{Mg}_{i,T}(\overline{D}_{i,t}\circ \tilde{s}_{i,T}, h_{T})\Big| s_{i,T-1}[w], h_{T-1} \Big]  \mathbf{1}{\{\mathtt{OM}_{i,T-1}=0\}}
        \end{aligned}\\
        &\begin{aligned}
            =\overline{\texttt{DE}}^{T-1,T}\Big[ \texttt{Mg}_{i,T}(\widetilde{\texttt{us}}_{i,T}, \tilde{h}_{T})&\Big|\overline{\underline{D}}_{i,t}\circ s_{i,t}[w], h_{T-1} \Big]\mathbf{1}{\{\mathtt{OM}_{i,T-1}=1\}}\\
            &+ \mathbb{E}^{\sigma}\Big[ \texttt{Mg}_{i,T}(\widetilde{\texttt{us}}_{i,T}, \tilde{h}_{T})\Big| s_{i,T-1}[w], h_{T-1} \Big]  \mathbf{1}{\{\mathtt{OM}_{i,T-1}=0\}}.
        \end{aligned}
    \end{aligned}
\end{equation}
%
%
From the definition of $\overline{\underline{D}}_{i,t}$, we have
\[
\begin{aligned}
\overline{\underline{D}}_{i,t}\circ s_{i,t}[w]\mathbf{1}{\{\mathtt{OM}_{i,T-1}=1\}} = \overline{D}_{i,t}\circ s_{i,t}[w] \mathbf{1}{\{s_{i,t}[w]\in \vec{\Gamma}^{w}_{i,t}(c_{i,t})\}}.
\end{aligned}
\]
Thus, 
\[
\begin{aligned}
\textup{ (\ref{eq:app_jump_up_1}) } = \overline{\mathtt{DE}}^{T-1,T}\Big[\mathtt{Mg}_{i,T}(\widetilde{\mathtt{us}}_{i,T}, \tilde{h}_{T})\Big|\overline{D}_{i,T-1}\circ s_{i,T-1}[w], h_{T-1}  \Big].
\end{aligned}
\]
Hence, the period-$T-1$ expected payoff-to-go function can be represented as follows:
\[
\begin{aligned}
&\Lambda_{i, T-1}(\mathtt{OM}_{i,T-1}, a_{i,T-1}|s_{i,T-1}, h_{T-1})= \mathtt{Mg}_{i,T-1}(\overline{D}_{i,T-1}\circ s_{i,T-1}, h_{T-1})+ \overline{\delta}_{i,T-1}(\overline{D}_{i,T-1}\circ s_{i,T-1}, h_{T-1} ).
\end{aligned}
\]
Following the similar backward induction in Appendix (\ref{app:thm:SPIR_conditions}) obtains (\ref{eq:app_switchable_to_go}) which coincides with (\ref{eq:app_prospect_represent}).
%
%
\endproof

Now, suppose that the principal is knowledgeable and each cutoff-switch $\phi_{i,t}$ is given by (\ref{eq:cutoff_knowledgeable_general}).
Without loss of generality, let $S''_{i,t}\subseteq \vec{\Gamma}^{b}_{i,t}(c_{i,t})$, for some $b\in[B]$.
Then, from Lemma \ref{lemma:app_proof_prop_single_agent_switchable}, the RHS of (\ref{eq:app_proof_thm7_2}) is equivalent to
\[
\begin{aligned}
    \Lambda_{i,t}(\ell|s_{i,t}[b],  h_{t}, \chi_{-i,t}) = \mathtt{Mg}_{i,t}(\overline{D}_{i,t}\circ s_{i,t}[b], h_{t}, \chi_{-i,t}) + \overline{\delta}_{i,t}( \overline{D}_{i,t}\circ s_{i,t}[b], h_{t}).
\end{aligned}
\]
Similarly, the definition of $\overline{D}_{i,t}$ gives 
\[
\phi_{i,t}(h_{t}|c_{i,t})< \mathtt{Mg}_{i,t}\big(\overline{d}_{i,t}(b), h_{t},\chi_{-i,t}\big)+ \overline{\delta}_{i,t}\big(\overline{d}_{i,t}(b), h_{t}, \chi_{-i,t} \big),
\]
which contradicts the definition of $\phi_{i,t}(b,h_{t},x_{t}|c_{i,t})$ in (\ref{eq:cutoff_knowledgeable_general}); i.e.,
\[
\begin{aligned}
    \phi_{i,t}(b,h_{t},\chi_{-i,t}|c_{i,t}) =& \mathtt{Mg}_{i,t}(\overline{\underline{D}}_{i,t}\circ s_{t}[b], h_{t}, \chi_{-i,t})+ \overline{\delta}_{i,t}(\overline{\underline{D}}_{i,t}\circ s_{i,t}[b], h_{t},\chi_{-i,t})\\
    =& \mathtt{Mg}_{i,t}\big(\overline{d}_{i,t}(b), h_{t}, \chi_{-i,t}\big)+ \overline{\delta}_{i,t}\big(\overline{d}_{i,t}(b), h_{t}, \chi_{-i,t} \big).
\end{aligned}
\]
Therefore, $\mathbb{P}\big(\tau_{i,t}(s_{i,t}, h_{t})=t \textup{ and } s_{i,t}\in S^{\mathtt{off}}_{i,t}|h_{t}\big) = \chi_{i,t}(t|h_{t})$.
%
\hfill $\square$

\section{Proof of Theorem \ref{thm:MAO_switchability_conjecture_gen}} \label{app:thm:MAO_switchability_conjecture_gen}

\subsection{($\Rightarrow$)}

We start with proving the \textit{only if}.
Since $\chi_{t}$ satisfies (\ref{eq:fixed_point_general}), there exists a measurable function $\tau_{i,t}:S_{i,t}\times H^{\natural}_{t}\mapsto \mathbb{T}_{t,T}$ for each $i\in\mathcal{N}$ such that
\begin{equation}\label{eq:app_thm_MAOS_0}
    \begin{aligned}
\Lambda_{i,t}(\tau_{i,t}(s_{i,t}, h_{t}) |s_{i,t}, h_{t}, \chi_{-i,t})\geq \Lambda_{i,t}( \ell|s_{i,t}, h_{t}, \chi_{-i,t}), 
\end{aligned}
\end{equation}
for all $\ell\in \mathbb{T}_{t,T}$, $s_{i,t}\in S_{i,t}$, $h_{t}\in H_{t}$.
By Proposition \ref{prop:condition_conjecture_MAOS_gen}, (\ref{eq:app_thm_MAOS_0}) implies that given the conjecture $\chi_{-i,t}$ at least one of such OM strategies satisfies
\begin{equation}\label{eq:app_thm_MAOS_1_gen}
    \begin{cases}
    \tau_{i,t}(s_{i,t}, h_{t})=t &, \textup{ if } s_{i,t}\in S^{\mathtt{off}}_{i,t},\\
   \tau_{i,t}(s_{i,t}, h_{t})=k>t &, \textup{ if } s_{i,t}\in S_{i,t}\backslash S^{\mathtt{off}}_{i,t}.
\end{cases}
\end{equation}

Each agent $i$'s period-$t$ ex-interm expected payoff-to-go is given by
\[
\begin{aligned}
        \Lambda_{i,t}(L|s_{i,t}, h_{t}, \chi_{-i,t})= \phi_{i,t}(h_{t},\chi_{-i,t}|c_{i,t})\mathbf{1}_{\{ L = t \}} + G_{i,t}(L|s_{i,t}, h_{t},\chi_{-i,t}) \mathbf{1}_{\{ L > t \}}.
    \end{aligned}
\]
Then, (\ref{eq:app_thm_MAOS_1_gen}) implies
\[
\begin{aligned}
    \Lambda_{i,t}(L|s_{i,t}, h_{t}, \chi_{-i,t})=\begin{cases}
        \phi_{i,t}(h_{t},\chi_{-i,t}|c_{i,t})&, \textup{ if } s_{i,t}\in S^{\mathtt{off}}_{i,t},\\
       G_{i,t}(L|s_{i,t}, h_{t},\chi_{-i,t}) &, \textup{ if } s_{i,t}\not\in S^{\mathtt{off}}_{i,t}.
    \end{cases}
\end{aligned}
\]

\subsubsection{Part \textit{(i)}}

From Lemma \ref{lemma:app_proof_thm2}, the expected payoff-to-go function $\Lambda_{i,t}$ can be represented by
\[
\begin{aligned}
    \Lambda_{i,t}(L|s_{i,t},  h_{t}, \chi_{-i,t}) = \mathtt{Mg}_{i,t}(\overline{D}_{i,t}\circ s_{i,t}, h_{t},\chi_{-i,t})+ \overline{\delta}_{i,t}(\overline{D}_{i,t}\circ s_{i,t}, h_{t},\chi_{-i,t}).
\end{aligned}
\]
When $s_{i,t}\in S_{i,t}\backslash S^{\mathtt{off}}_{i,t}$, $\phi_{i,t}(h_{t},\chi_{-i,t}|c_{i,t})\leq \mathtt{Mg}_{i,t}(\overline{D}_{i,t}\circ s_{i,t}, h_{t},\chi_{-i,t})+ \overline{\delta}_{i,t}(\overline{D}_{i,t}\circ s_{i,t}, h_{t},\chi_{-i,t})$.
When $s_{i,t}\in S^{\mathtt{off}}_{i,t}$, $\phi_{i,t}(h_{t},\chi_{-i,t}|c_{i,t})\geq \mathtt{Mg}_{i,t}(\overline{D}_{i,t}\circ s_{i,t}, h_{t},\chi_{-i,t})+ \overline{\delta}_{i,t}(\overline{D}_{i,t}\circ s_{i,t}, h_{t},\chi_{-i,t})$.
Since the cutoff-switch function $\phi_{i,t}$ is horizontal and is given by (\ref{eq:cutoff_horizontal_general}), by the definition of $\overline{D}_{i,t}$, we have, for $s_{i,t}\in S_{i,t}\backslash S^{\mathtt{off}}_{i,t}$,
\begin{equation}\label{eq:app_thm_MAOS_onlyif_1_gen}
    \begin{aligned}
    \mathtt{Mg}_{i,t}(\overline{d}_{i,t}(b), h_{t}, \chi_{-i,t})+\overline{\delta}_{i,t}(\overline{d}_{i,t}(b), h_{i,t},\chi_{-i,t})\leq \mathtt{Mg}_{i,t}(s_{i,t}, h_{t}, \chi_{-i,t})+\overline{\delta}_{i,t}(s_{i,t}, h_{t},\chi_{-i,t}),
\end{aligned}
\end{equation}
\begin{equation}\label{eq:app_thm_MAOS_onlyif_2_gen}
    \begin{aligned}
    \mathtt{Mg}_{i,t}(\overline{d}_{i,t}(b), h_{t}, \chi_{-i,t})+\overline{\delta}_{i,t}(\overline{d}_{i,t}(b), h_{t},\chi_{-i,t})\geq \mathtt{Mg}_{i,t}(\overline{d}_{i,t}(b), h_{t}, \chi_{-i,t})+\overline{\delta}_{i,t}(\overline{d}_{i,t}(b), h_{t},\chi_{-i,t}).
\end{aligned}
\end{equation}
By definitions of $\Upsilon_{i,t}$ and $\overline{\Xi}_{i,t}$, (\ref{eq:app_thm_MAOS_onlyif_1_gen}) and (\ref{eq:app_thm_MAOS_onlyif_2_gen}) imply, respectively, for all $s'_{i,t}\in S^{\mathtt{off}}_{i,t}$ and $s_{i,t}\in S_{i,t}\backslash S^{\mathtt{off}}_{i,t}$,
\[
\Upsilon_{i,t}(\overline{d}_{i,t}(b), s_{i,t} |h_{t},\chi_{-i,t})+\overline{\Xi}_{i,t}(\overline{d}_{i,t}(b), s_{i,t}| h_{t},\chi_{-i,t})\leq0,
\]
\[
\begin{aligned}
    \Upsilon_{i,t}(\overline{d}_{i,t}(b),s'_{i,t} |h_{t},\chi_{-i,t})+\overline{\Xi}_{i,t}(\overline{d}_{i,t}(b), s'_{i,t}| h_{t},\chi_{-i,t})\geq0,
\end{aligned}
\]
where $\hat{s}_{i,t} = \overline{d}_{i,t}(b)$, for any $b\in[B]$.
Therefore, each $S^{\mathtt{off}}_{i,t}$ is an essential region (given the conjecture $\chi_{-i,t}$) of $S_{i,t}$ with $\{\overline{d}_{i,t}(b)\}_{b\in[B]}$ as essential points, for all $t\in\mathbb{T}$, $h_{t}\in H_{t}$.

\subsubsection{Part \textit{(ii)}}

From Lemma \ref{lemma:app_proof_prop_single_agent_switchable}, the expected payoff-to-go function $\Lambda^{\natural}_{t}$ when the principal is knowledgeable can be represented by for $w\in\{b,e\}$, $b\in[B]$, $e\in[B'-B]$,
\[
\Lambda_{i,t}(L|s_{i,t}[w],  h_{t},\chi_{-i,t}) = \mathtt{Mg}_{i,t}(\overline{D}_{i,t}\circ s_{i,t}[w], h_{t},\chi_{-i,t})+ \overline{\delta}_{i,t}(\overline{D}_{i,t}\circ s_{i,t}[w], h_{t},\chi_{-i,t}).
\]

When $s_{t}[e]\in \Psi^{e}_{i,t}(c_{i,t})$,
$\phi_{i,t}(e, h_{t},\chi_{-i,t}|c_{i,t})\leq$ $\mathtt{Mg}_{i,t}(\overline{D}_{i,t}\circ s_{i,t}[w], h_{t},\chi_{-i,t})+$ $\overline{\delta}_{i,t}(\overline{D}_{i,t}\circ s_{i,t}[w], h_{t},\chi_{-i,t})$; when $s_{i,t}[b]\in \vec{\Gamma}^{b}_{i,t}(c_{i,t})$,
$\phi_{i,t}(b,h_{t},\chi_{-i,t}|c_{i,t})\geq$ $\mathtt{Mg}_{i,t}(\overline{D}_{i,t}\circ s_{i,t}[b], h_{t},\chi_{-i,t})+$ $\overline{\delta}_{i,t}(\overline{D}_{i,t}\circ s_{i,t}[w], h_{t},\chi_{-i,t})$.
Since the cutoff-switch if given by (\ref{eq:cutoff_knowledgeable_general}) and the principal is knowledgeable,
\begin{align*}
    &\mathtt{Mg}_{i,t}(\underline{d}_{i,t}(e), h_{t}, \chi_{-i,t})+ \overline{\delta}_{i,t}(\underline{d}_{i,t}(e), h_{t},\chi_{-i,t})\leq \mathtt{Mg}_{i,t}(s_{i,t}[e], h_{t}, \chi_{-i,t})+ \overline{\delta}_{i,t}(s_{i,t}[e], h_{t},\chi_{-i,t}),\\
    &\mathtt{Mg}_{i,t}(\overline{d}_{i,t}(b), h_{t}, \chi_{-i,t})+ \overline{\delta}_{i,t}(\overline{d}_{i,t}(b), h_{t},\chi_{-i,t}) \geq \mathtt{Mg}_{i,t}(\overline{d}_{i,t}(b), h_{t}, \chi_{-i,t})+ \overline{\delta}_{i,t}(\overline{d}_{i,t}(b), h_{t},\chi_{-i,t}).
\end{align*}
Hence, we obtain that each $\{\overline{\underline{D}}_{i,t}\circ s_{i,t}[e]\}$ is an essential region of $\Psi^{e}_{i,t}(c_{i,t})$ for all $e\in[B'-B]$ and each $\vec{\Gamma}^{b}_{i,t}(c_{i,t})\backslash \{\overline{\underline{D}}_{i,t}\circ s_{i,t}[b]\}$ is an essential region of $\vec{\Gamma}^{b}_{i,t}(c_{i,t})$ for all $b\in[B]$.

\subsection{($\Leftarrow$)}

Now we proceed with the \textit{if}.

\subsubsection{Part \textit{(i)}}

Since each $S^{\mathtt{off}}_{i,t}$ is an essential region (given the conjecture $\chi_{-i,t}$) of $S_{i,t}$ with $\{\overline{d}_{i,t}(b)\}_{b\in[B]}$ as essential points, for all $t\in\mathbb{T}$, $h_{t}\in H_{t}$, we have for all $b\in[B]$, 
\[
\Upsilon_{i,t}(\overline{d}_{i,t}(b),s'_{i,t} |h_{t},\chi_{-i,t})+\overline{\Xi}_{i,t}(\overline{d}_{i,t}(b), s'_{i,t}, h_{t},\chi_{-i,t})\geq0, \forall s'_{i,t}\in S^{\mathtt{off}}_{i,t},
\]
\[
\Upsilon_{i,t}(\overline{d}_{i,t}(b),s_{i,t} |h_{t},\chi_{-i,t}) + \overline{\Xi}_{i,t}(\overline{d}_{i,t}(b), s_{i,t}, h_{t},\chi_{-i,t}) \leq 0, \forall s_{i,t}\in S_{i,t}\backslash S^{\mathtt{off}}_{i,t}.
\]
That is, for all $b\in[B]$,
\[
\begin{aligned}
    \mathtt{Mg}_{i,t}(\overline{d}_{i,t}(b), h_{t},\chi_{-i,t}) + \overline{\delta}_{i,t}( \overline{d}_{i,t}(b), h_{t},\chi_{-i,t}) \geq \mathtt{Mg}_{i,t}(s'_{i,t}, h_{t},\chi_{-i,t}) + \overline{\delta}_{i,t}( s'_{i,t}, h_{t},\chi_{-i,t}), \forall s'_{i,t}\in S^{\mathtt{off}}_{i,t},
\end{aligned}
\]
\[
\begin{aligned}
    \mathtt{Mg}_{i,t}(\overline{d}_{i,t}(b), h_{t},\chi_{-i,t}) + \overline{\delta}_{i,t}( \overline{d}_{i,t}(b), h_{t},\chi_{-i,t})\leq \mathtt{Mg}_{i,t}(s_{i,t}, h_{t},\chi_{-i,t}) + \overline{\delta}_{i,t}( s_{i,t}, h_{t},\chi_{-i,t}), \forall s_{i,t}\in S_{i,t}\backslash S^{\mathtt{off}}_{i,t}.
\end{aligned}
\]
Since each cutoff-switch $\phi_{i,t}(h_{t},\chi_{-i,t}|c_{i,t})$ is \textup{horizontal} and is given by (\ref{eq:cutoff_horizontal_general}), 
\[
\phi_{i,t}(h_{t}, \chi_{-i,t}|c_{i,t}) \geq \mathtt{Mg}_{i,t}(s'_{i,t}, h_{t},\chi_{-i,t}) + \overline{\delta}_{i,t}( s'_{i,t}, h_{t},\chi_{-i,t}), \forall s'_{i,t}\in S^{\mathtt{off}}_{i,t},
\]
\[
\phi_{i,t}(h_{t}, \chi_{-i,t}|c_{i,t})\leq \mathtt{Mg}_{i,t}(s_{i,t}, h_{t},\chi_{-i,t}) + \overline{\delta}_{i,t}( s_{i,t}, h_{t},\chi_{-i,t}), \forall s_{i,t}\in S_{i,t}\backslash S^{\mathtt{off}}_{i,t}.
\]

By Lemma \ref{lemma:app_proof_thm2}, the expected payoff-to-go function $\Lambda_{i,t}$ can be represented by
\[
\begin{aligned}
    \Lambda_{i,t}(L|s_{i,t},  h_{t}, \chi_{-i,t}) = \mathtt{Mg}_{i,t}(\overline{D}_{i,t}\circ s_{i,t}, h_{t},\chi_{-i,t})+ \overline{\delta}_{i,t}(\overline{D}_{i,t}\circ s_{i,t}, h_{t},\chi_{-i,t}).
\end{aligned}
\]
Then, we have
\[
\begin{aligned}
    \begin{cases}
       \Lambda_{i,t}(L|s_{i,t}, h_{t}, \chi_{-i,t})\leq \phi_{i,t}(h_{t},\chi_{-i,t}|c_{i,t})&, \textup{ if } s_{i,t}\in S^{\mathtt{off}}_{i,t},\\
      \Lambda_{i,t}(L|s_{i,t}, h_{t}, \chi_{-i,t}) \geq  \phi_{i,t}(h_{t},\chi_{-i,t}|c_{i,t})&, \textup{ if } s_{i,t}\not\in S^{\mathtt{off}}_{i,t}.
    \end{cases}
\end{aligned}
\]
However, the expected payoff-to-go function is defined by
\[
\begin{aligned}
        \Lambda_{i,t}(L|s_{i,t}, h_{t}, \chi_{-i,t})= \phi_{i,t}(h_{t},\chi_{-i,t}|c_{-i,t})\mathbf{1}_{\{ L = t \}} + G_{i,t}(L|s_{i,t}, h_{t},\chi_{-i,t}) \mathbf{1}_{\{ L > t \}}.
    \end{aligned}
\]
Thus, we obtain 
\[
\begin{aligned}
    \Lambda_{i,t}(L|s_{i,t}, h_{t}, \chi_{-i,t})=\begin{cases}
        \phi_{i,t}(h_{t},\chi_{-i,t}|c_{i,t})&, \textup{ if } s_{i,t}\in S^{\mathtt{off}}_{i,t},\\
       G_{i,t}(L|s_{i,t}, h_{t},\chi_{-i,t}) &, \textup{ if } s_{i,t}\not\in S^{\mathtt{off}}_{i,t},
    \end{cases}
\end{aligned}
\]
which implies the following best response to the conjecture $\chi_{-i,t}$ for every agent,
\[
    \begin{cases}
    \tau_{i,t}(s_{i,t}, h_{t})=t &, \textup{ if } s_{i,t}\in S^{\mathtt{off}}_{i,t},\\
   \tau_{i,t}(s_{i,t}, h_{t})=k>t &, \textup{ if } s_{i,t}\in S_{i,t}\backslash S^{\mathtt{off}}_{i,t}.
\end{cases}
\]
That is, for all $t\in\mathbb{T}$, $h_{t}\in H_{t}$, $\mathbb{P}(\tau_{i,t}(s_{i,t}, h_{t})=t |h_{t}) = \chi_{i,t}(t|h_{t})$; or equivalently, $\chi_{t}(\cdot|h_{t})\in \bm{X}_{t}(h_{t}, \chi_{t})$.

\subsubsection{Part \textit{(ii)}}

Since each $\vec{\Gamma}^{b}_{i,t}(c_{i,t})\backslash \{\overline{\underline{D}}_{i,t}\circ s_{i,t}[b]\}$ is an essential region (given the conjecture $\chi_{i,t}$) of $\vec{\Gamma}^{b}_{i,t}(c_{i,t})\subset S^{\mathtt{off}}_{i,t}$, for all $b\in[B]$, $t\in\mathbb{T}$, and each $\{\overline{\underline{D}}_{i,t}\circ s_{i,t}[e]\}$ is an essential region of $\Psi^{e}_{i,t}(c_{i,t})\subset S^{\mathtt{off}}_{i,t}$, for all $e\in[B'-B]$, $t\in\mathbb{T}$, we obtain for all $s_{i,t}[e]\in \Psi^{e}_{i,t}(c_{t})$, $s_{i,t}[b]\in\vec{\Gamma}^{b}_{i,t}(c_{i,t})$, $e\in[B'-B]$, $b\in[B]$,

\begin{align}
    &\begin{aligned}
        \mathtt{Mg}_{t}(\underline{d}_{i,t}(e), h_{t}, \chi_{-i,t})+ \overline{\delta}_{i,t}(\underline{d}_{i,t}(e), h_{t},\chi_{-i,t})\leq \mathtt{Mg}_{i,t}(s_{i,t}[e], h_{t}, \chi_{-i,t})+ \overline{\delta}_{i,t}(s_{i,t}[e], h_{t},\chi_{-i,t}),
    \end{aligned}\label{eq:app_thm_MAOS_onlyif_3_gen_1}\\
    &\begin{aligned}
         \mathtt{Mg}_{i,t}(\overline{d}_{i,t}(b), h_{t}, \chi_{-i,t})+ \overline{\delta}_{i,t}(\overline{d}_{i,t}(b), h_{t},\chi_{-i,t})\geq \mathtt{Mg}_{i,t}(s_{i,t}[b], h_{t}, \chi_{-i,t})+ \overline{\delta}_{i,t}(s_{i,t}[b], h_{t},\chi_{-i,t}).\label{eq:app_thm_MAOS_onlyif_3_gen_2}
    \end{aligned}
\end{align}

By Lemma\ref{lemma:app_proof_thm2}, for all $w\in\{b,e\}$, $b\in[B]$, $e\in[B'-B]$,
\[
\Lambda_{i,t}(L|s_{i,t}[w],  h_{t},\chi_{-i,t}) = \mathtt{Mg}_{i,t}(\overline{D}_{i,t}\circ s_{i,t}[w], h_{t},\chi_{-i,t})+ \overline{\delta}_{i,t}(\overline{D}_{i,t}\circ s_{i,t}[w], h_{t},\chi_{-i,t}).
\]
In addition, the cutoff-switch is given by (\ref{eq:cutoff_knowledgeable_general}).
Then, (\ref{eq:app_thm_MAOS_onlyif_3_gen_1}) and (\ref{eq:app_thm_MAOS_onlyif_3_gen_2}), respectively, become, for all $s_{i,t}[e]\in \Psi^{e}_{i,t}(c_{i,t})$, $s_{i,t}[b]\in\vec{\Gamma}^{b}_{i,t}(c_{i,t})$, $e\in[B'-B]$, $b\in[B]$,
\begin{align*}
    &\phi_{i,t}(e,h_{t},\chi_{-i,t}|c_{i,t})\leq \Lambda^{\natural}_{t}(L|s_{i,t}[e],  h_{t},\chi_{-i,t}) \textup{ and }\\
    &\phi_{i,t}(b,h_{t},\chi_{-i,t}|c_{i,t})\geq \Lambda_{i,t}(L|s_{i,t}[b],  h_{t},\chi_{-i,t}).
\end{align*}
Then, from the definition of $\Lambda_{i,t}$, we have
\[
\begin{aligned}
        \Lambda_{i,t}(L|s_{i,t}, h_{t}, \chi_{-i,t})=\begin{cases}
        \phi_{i,t}(b, h_{t},\chi_{-i,t}|c_{i,t})&, \textup{ if } s_{i,t}\in \vec{\Gamma}^{b}_{i,t}(c_{i,t}),\\
       G_{i,t}(L|s_{i,t}, h_{t},\chi_{-i,t}) &, \textup{ if } s_{i,t}\in \vec{\Psi}^{e}_{i,t}(c_{i,t}),
    \end{cases}
\end{aligned}
\]
which implies the following best response to the conjecture $\chi_{i,t}$ for every agent,
\[
    \begin{cases}
    \tau_{i,t}(s_{i,t}, h_{t})=t &, \textup{ if } s_{i,t}\in \vec{\Gamma}^{b}_{i,t}(c_{i,t}),\\
   \tau_{i,t}(s_{i,t}, h_{t})=k>t &, \textup{ if } s_{i,t}\in \vec{\Psi}^{e}_{i,t}(c_{i,t}).
\end{cases}
\]
That is, for all $t\in\mathbb{T}$, $h_{t}\in H_{t}$, $\mathbb{P}(\tau_{i,t}(s_{i,t}[w], h_{t})=t |h_{t}) = \chi_{i,t}(t|h_{t})$ for all $w\in\{e,b\}$, $b\in[B]$, $e\in[B'-B]$; or equivalently, $\chi_{t}(\cdot|h_{t})\in \bm{X}_{t}(h_{t}, \chi_{t})$.
%
\hfill $\square$

\section{Proof of Proposition \ref{prop:switchability_sufficient_condition_new}}\label{app:prop:switchability_sufficient_condition_new}

The proof of Proposition \ref{prop:switchability_sufficient_condition_new} follows the similar procedures as in the proof of \ref{thm:doic_irod_sufficient_condition} based on Lemma \ref{lemma:app_proof_prop_single_agent_switchable}.
Hence, we omit the detailed proofs.
\hfill $\square$

\section{Proof of Proposition \ref{prop:necessary_DOIC_sigma_tau}}\label{app:necessary_DOIC_sigma_tau}

For simplicity, we drop $\chi_{-i,t}$.
We start by showing that both $MG_{i,t}(s'_{i,t}|\cdot, h_{t})$ and $V_{i,t}(s_{i,t}, h_{t})$ are Lipschitz continuous and differentiable in the true state $s_{i,t}$.

Let $\omega^{L}_{i,t}=(\omega_{i,k})_{k=t}^{L}$ be a realized sequence of shocks from period $t$ to $L\in\mathbb{T}_{t,T}$.
Let $a^{L}_{-i,t}=((\hat{a}_{j,k})_{i\neq j\in\mathcal{N}_{k}})_{t'=t}^{L}$ and $\hat{a}^{L}_{i,t}=(\hat{a}_{i,t},(a_{i,k})_{k=t+1}^{L})$, respectively, be the sequences of other agents' obedient actions and agent $i$'s arbitrary actions from period $t$ to $L\in\mathbb{T}_{t,T}$.
In addition, let $\hat{s}^{L}_{i,t}=(\hat{s}_{i,k})_{k=t}^{L}$ denote the corresponding states such that each $\hat{a}_{i,k}=\sigma_{i,k}(\hat{s}_{i,k},h_{k})$ for all $k\in\mathbb{T}_{t, L}$, and let $s^{L}_{i,t}=(s_{i,k})_{k=t}^{L}$ be a realized sequence of true states from period $t$ to $L\in\mathbb{T}_{t,T}$.

By the definition of the persistence function, we can represent each $s_{i,k}=K_{i,t}(k, s_{i,t}, h_{t}, \hat{a}^{k}_{t+1}, \omega^{k}_{i,t+1})$, where $K_{i,t}(k, s_{i,t}, h_{t}, \omega^{k}_{i,t})$ is recursively determined by $\{\kappa_{i,t'}\}_{t'=t+1}^{k}$ given $s_{i,t}$, $h_{t}$, $\hat{h}^{k}_{t+1}$, $\omega^{k}_{i,t}$; e.g., when $k=t+2$, $K_{i,t}(t+2, s_{i,t}, h_{t}, \omega^{k}_{i,t})=\kappa_{i,t+2}\big(\kappa_{i,t+1}(s_{i,t}, \hat{h}_{t+1}, \omega_{i,t+1}), \hat{h}_{t+2}, \omega_{i,t+2}\big)$, where $\hat{h}_{t+1}=\{h_{t}, \hat{a}_{t}\}$ and $\hat{h}_{t+2} = \{h_{t}, \hat{a}^{t+2}_{t+1} \}$.
Since $\kappa_{i,k}(s_{i,k-1}, \hat{a}_{i,k-1}, a_{-i,k-1}, \omega_{i,k})$ differentiable in $s_{i,k-1}$, we have (by fixing $a^{L}_{-i,t}$ and $\hat{a}^{L}_{-i,t}$),
\[
\begin{aligned}
    &\frac{\partial}{\partial v} K_{i,t}(k, v, h_{t}, \hat{a}^{k}_{t+1} \omega^{k}_{i,t+1})\big|_{v=s_{i,t}} = \prod\limits_{t'=t+1}^{k} \frac{\partial}{\partial v'} \kappa_{i,t'}(v', \hat{a}_{i,t'-1}, a_{-i,t'-1}, \omega_{i,t'})\big|_{v'=s_{i,t'}}.
\end{aligned}
\]
When he takes $\mathtt{OM}_{i,t}=0$, agent $i$'s realized payoff-to-go up to period $L$, denoted by $RG_{i,t}$, given history $h_{t}$, $s^{L}_{i,t}$, $a^{L}_{-i,t}$, $\hat{a}^{L}_{-i,t}$ (and thus $\hat{s}^{L}_{i,t}$), and $\omega^{L}_{i,t}$ is given by
\[
\begin{aligned}
    RG_{i,t}(s^{L}_{i,t}, \hat{a}^{L}_{i,t}, a^{L}_{-i,t},h_{t}, \omega^{L}_{i,t}, L)&\equiv \sum\limits_{k=t}^{L}u_{i,k}\Big(\hat{a}_{i,k}, a_{-i,k}, K_{i,t}(k, s_{i,t}, h_{t}, \hat{h}^{k}_{t+1}, \omega^{k}_{i,t+1})\Big) \\
    &+ \sum\limits_{k=t}^{L}\rho_{i,k}(\hat{a}_{i,k}, a_{-i,k}, h_{k}) + \phi_{i,L+1}(h_{L+1}).
\end{aligned}
\]
Then, let 
\[
\begin{aligned}
    DG_{i,t}(\varsigma,s_{i,t}, L)\equiv \frac{1}{\varsigma}\Big(RG_{i,t}((s_{i,t}+\varsigma, s^{L[\varsigma]}_{i,t+1}), \hat{a}^{L}_{i,t},h_{t}, \omega^{L}_{i,t}, L- RG_{i,t}((s_{i,t}, s^{L}_{i,t+1}), \hat{a}^{L}_{i,t},h_{t}, \omega^{L}_{i,t}, L)  \Big).
\end{aligned}
\]
Given $\hat{a}^{L}_{i,t}$ and $a^{L}_{-i,t}$, it is straightforward to see that
\begin{equation}\label{eq:app_prove_envelop_1}
    \begin{aligned}
    DG_{i,t}(\varsigma,s_{i,t}, L)&=\frac{1}{\varsigma}\Big(\sum\limits_{k=t}^{L}u_{i,k}\Big(\hat{a}_{i,k}, a_{-i,k}, K_{i,t}(k, s_{i,t}+\varsigma, h_{t}, \hat{h}^{k}_{t+1}, \omega^{k}_{i,t+1})\Big)\\
    &- \sum\limits_{k=t}^{L}u_{i,k}\Big(\hat{a}_{i,k}, a_{-i,k}, K_{i,t}(k, s_{i,t}, h_{t}, \hat{h}^{k}_{t+1}, \omega^{k}_{i,t+1})\Big)\Big).
\end{aligned}
\end{equation}
Let us temporarily omit $h_{t}, \hat{h}^{k}_{t+1}, \omega^{k}_{i,t+1}$, and $a^{L}_{-i,t}$.
Then, we rewrite (\ref{eq:app_prove_envelop_1}) as
\[
\begin{aligned}
    DG_{i,t}(\varsigma,s_{i,t}, L)&=\frac{1}{\varsigma}\Big(\sum\limits_{k=t}^{L}u_{i,k}\Big(\hat{a}_{i,k}, K_{i,t}(k, s_{i,t}+\varsigma)\Big)- \sum\limits_{k=t}^{L}u_{i,k}\Big(\hat{a}_{i,k}, K_{i,t}(k, s_{i,t})\Big)\Big)\\
    &\times \frac{\big(K_{i,t}(k, s_{i,t}+\varsigma) - K_{i,t}(k, s_{i,t})\big)/\varsigma }{\big( K_{i,t}(k, s_{i,t}+\varsigma) - K_{i,t}(k, s_{i,t})\big)/\varsigma  }=\frac{K_{i,t}(k, s_{i,t}+\varsigma) - K_{i,t}(k, s_{i,t})}{\varsigma}\\
    &\times \frac{ \sum\limits_{k=t}^{L}u_{i,k}\Big(\hat{a}_{i,k}, K_{i,t}(k, s_{i,t}+\varsigma)\Big)- \sum\limits_{k=t}^{L}u_{i,k}\Big(\hat{a}_{i,k}, K_{i,t}(k, s_{i,t})\Big)/\varsigma    }{ \big( K_{i,t}(k, s_{i,t}+\varsigma) - K_{i,t}(k, s_{i,t})\big)/\varsigma  }.
\end{aligned}
\]
Due to Conditions \ref{cond:differentiable_reward} and \ref{cond:bounded_dynamic}, we have
\[
\begin{aligned}
   \frac{\partial RG_{i,t}(v,s^{L}_{i,t+1}, L)}{\partial v}\big|_{v=s_{i,t}}= &\lim_{\varsigma\rightarrow 0} DG_{i,t}(\varsigma,s_{i,t}, L)
   = \sum\limits_{k=t}^{L}\frac{\partial}{\partial v}u_{i,k}(v, a_{k})\frac{\partial}{\partial v'} K_{i,t}(k, v')\big|^{v'=s_{i,t}}_{v=s_{i,k}}.
\end{aligned}
\]
By Condition \ref{cond:bounded_dynamic}, 
\[
\sum\limits_{k=t}^{L}\frac{\partial}{\partial v}u_{i,k}(v, a_{k})\big|_{v=s_{i,k}}\frac{\partial}{\partial v} K_{i,t}(k, v)\big|_{v=s_{i,t}}\leq \sum\limits_{k=t}^{L}\frac{\partial}{\partial v}u_{i,k}(v, a_{k})\big|_{v=s_{i,k}}  \bar{C}_{i,k}(\tilde{\omega}_{i,k})
\]
Hence,
\[
\begin{aligned}
    &\left|RG_{i,t}(s_{i,t},s^{L}_{i,t+1}, L) -RG_{i,t}(s'_{i,t},s^{L}_{i,t+1}, L) \right| =\left|\int^{s_{i,t}}_{s'_{i,t}} \sum\limits_{k=t}^{L}\frac{\partial}{\partial v}u_{i,k}(v, a_{k})\big|_{v=s_{i,k}}\frac{\partial}{\partial v} K_{i,t}(k, v)\big|_{v=r} dr \right|\\
    &\leq\left|\int^{s_{i,t}}_{s'_{i,t}} \sum\limits_{k=t}^{L}\frac{\partial}{\partial v}u_{i,k}(v, a_{k})\big|_{v=s_{i,k}}\bar{C}_{i,k}(\tilde{\omega}_{i,k}) dr \right|\\
    &\leq \left|\int^{s_{i,t}}_{s'_{i,t}} \sum\limits_{k=t}^{L}\frac{\partial}{\partial v}u_{i,k}(v, a_{k})\big|_{v=s_{i,k}}\hat{c}_{i,k} dr \right| \leq \Big(\sum_{k=t}^{L} c_{i,k} \big(\prod^{k}_{t'=1}\hat{c}_{i,t'}\big)\Big)|s_{i,t}- s'_{i,t}|,
\end{aligned}
\]
where the second inequality is from Condition \ref{cond:bounded_dynamic}, and the third inequality is from Condition \ref{cond:differentiable_reward}.
Thus, $RG_{i,t}(s_{i,t},s^{L}_{i,t+1}, L)$ is Lipschitz continuous in $s_{i,t}$ with constant $\sum_{k=t}^{L} c_{i,k} \big(\prod^{k}_{t'=1}\hat{c}_{i,t'}\big)$.
By Conditions \ref{cond:differentiable_reward} and \ref{cond:bounded_dynamic}, Lebesgue dominated convergence theorem implies
\begin{equation}\label{eq:app_prove_envelop_2}
    \begin{aligned}
   \frac{\partial}{\partial v} G_{i,t}(a'_{i,t}|v, h_{t}, L) \Big|_{v=s_{i,t}} 
   &= \lim_{\varsigma \rightarrow 0}\mathbb{E}^{\sigma}_{a'_{i,t}}\Big[ \frac{1}{\varsigma}\Big( RG_{i,t}(s_{i,t}+\varsigma, \tilde{s}^{L}_{i,t+1}, L) - RG_{i,t}(s_{i,t}, \tilde{s}^{L}_{i,t+1}, L)
 \Big)\Big|s_{i,t},h_{t}\Big]\\
 &= \mathbb{E}^{\sigma}_{a'_{i,t}}\Big[\lim_{\varsigma \rightarrow 0} \frac{1}{\varsigma}\Big( RG_{i,t}(s_{i,t}+\varsigma, \tilde{s}^{L}_{i,t+1}, L) - RG_{i,t}(s_{i,t}, \tilde{s}^{L}_{i,t+1}, L)
 \Big)\Big|s_{i,t},h_{t}\Big]\\
 &=\mathbb{E}^{\sigma}_{a'_{i,t}}\Big[ \lim_{h\rightarrow 0} DG_{i,t}(\varsigma,s_{i,t}, L)\Big|s_{i,t},h_{t}\Big]\\
 &= \mathbb{E}^{\sigma}_{a'_{i,t}}\Big[ \sum\limits_{k=t}^{L}\frac{\partial}{\partial v}u_{i,k}(v, \tilde{a}_{k})\big|_{v=\tilde{s}_{i,k}}\frac{\partial}{\partial v} K_{i,t}(k, v)\big|_{v=\tilde{s}_{i,k}}\Big|s_{i,t},h_{t}\Big]\\
 &=q_{i,t}\big(\sigma_{i,t}(s_{i,t}, h_{t}), v, h_{t}, \tau_{i,t}(s_{i,t}, h_{t})\big).
\end{aligned}
\end{equation}
By $S^{\mathtt{off}}$-DOIC, we have 
\[
\begin{aligned}
   &V_{i,t}(s_{i,t}, h_{t}) = MG_{i,t}(s_{i,t}|s_{i,t},h_{t})= G_{i,t}(s_{i,t},h_{t}, \tau_{i,t}(s_{i,t},h_{t})),
\end{aligned}
\]
which implies (\ref{eq:envelope_like_condition_sigma_tau}).
%
\hfill $\square$

\section{Proof of Theorem \ref{thm:sufficient_condition_without_MSO}}\label{app:thm:sufficient_condition_without_MSO}

%
Here, we need to show that if the task policy satisfies (\ref{eq:condition_RAIC_sigma_0}), then there exist $\rho$ and $\phi$ such that the mechanism is $S^{\mathtt{off}}$-DOIC.
By definition, agent $i$'s period-$t$ expected payoff to go function is given by

\[
\begin{aligned}
    &MG_{i,t}(s'_{i,t}|s_{i,t},h_{t}, \chi_{-i,t})\equiv \max\limits_{L\in\mathbb{T}_{t,T+1}} G_{i,t}(\sigma_{i,t}(s'_{i,t}, h_{t})|s_{i,t},h_{t}, L,\chi_{-i,t}),\\
    &V_{i,t}(s_{i,t}, h_{t}, \chi_{i,t})\equiv \max\limits_{s'_{i,t}\in S_{i,t}} MG_{i,t}(s'_{i,t}|s_{i,t},h_{t}, \chi_{-i,t}).
\end{aligned}
\]
%
%
By constructing $\rho$ that satisfies \ref{itm:C1}, we have
\[
\begin{aligned}
    MG_{i,t}(s_{i,t}, h_{t})&= \max\limits_{L\in \mathbb{T}_{t,T}}g_{i,t}(s_{i,t}, h_{t}, L|\theta_{i,t}) \\
    &+ \max\limits_{L'\in \mathbb{T}_{t,T}}\mathbb{E}^{\sigma}\Big[ \phi_{i,L'+1}(\tilde{h}_{L'+1}) - \max\limits_{L''\in \mathbb{T}_{L'+1,T}}g_{i,L'+1}(\tilde{s}_{i,L'+1}, \tilde{h}_{L'+1}, L''|\theta_{i,L'+1}) \Big|s_{i,t}, h_{t}\Big].
\end{aligned}
\]
%

%

For simplicity, we temporally drop $\chi=(\chi_{i,t})$ in the notations of the related functions.
Based on \ref{itm:C1} and the setting of $g_{i,T+1}(\cdot)=0$ and $\phi_{i,T+1}(\cdot)=0$, it holds that by law of total expectation $\max\limits_{L\in \mathbb{T}_{k,T}}g_{i,k}(s_{i,k}, \tilde{h}_{L'+1}, L|\theta_{i,k})=G_{i,k}(s_{i,k}, h_{k}, T)$ for all $k\in\mathbb{T}$.
Then, we can represent $MG_{i,t}$ as
\begin{equation}\label{eq:app_O_1}
    \begin{aligned}
    &MG_{i,t}(s_{i,t}, h_{t})=\max\limits_{L\in \mathbb{T}_{t,T}}g_{i,t}( s_{i,t}, h_{t}, L|\theta_{i,t})+ \max\limits_{L'\in \mathbb{T}_{t,T}}\mathbb{E}^{\sigma}\Big[ \phi_{i,L'+1}(\tilde{h}_{L'+1})- G_{i,L'+1}(s_{i,L'+1}, h_{L'+1}, T)\Big|s_{i,t}, h_{t}\Big].
\end{aligned}
\end{equation}
%
%
%
%
Let us focus on horizontal cutoff-switch functions (the analyses follow similarly when the principal is knowledgeable).
Consider that the cutoff-switch function is formulated by (\ref{eq:cutoff_horizontal_general}).
By setting, $S^{\mathtt{X}}[\sigma,\mathcal{G}]$ is non-empty (due to $\mathcal{S}^{\mathtt{MIX}}[\mathcal{G}]\neq \emptyset$).
Then, when agent $i$ is obedient in taking regular actions, the OFRs induced by the cutoff-switch function formulated by (\ref{eq:cutoff_horizontal_general}) coincide with $S^{\mathtt{off}}\in S^{\mathtt{X}}[\sigma,\mathcal{G}]$.
Since agent $i$ makes a unilateral one-shot deviation from obedience and each $S^{\mathtt{off}}_{i,k}$ is an indifference region of $S^{\mathtt{off}}_{i,k}$ for all $k\in\mathbb{T}$, it holds that
\[
\max\limits_{L'\in \mathbb{T}_{t,T}}\mathbb{E}^{\sigma}_{a'_{i,t}}\Big[ \phi_{i,L'+1}(\tilde{h}_{L'+1}) - G_{i,L'+1}(s_{i,L'+1}, h_{L'+1}, T)\Big|s_{i,t}, h_{t}\Big]=0,
\]
for any $a'_{i,t}\in A_{i,t}[\sigma]$.
Thus, we have 
\[
\begin{aligned}
    MG_{i,t}(s_{i,t}, h_{t}) &= \max\limits_{L\in \mathbb{T}_{t,T}}g_{i,t}(s_{i,t}, h_{t}, L|\theta_{i,t})=G_{i,t}(s_{i,t}, h_{t}, T),
\end{aligned}
\]
in which the left-hand side is independent of $\rho$ while the right-hand side depends on $\rho$.
From (\ref{eq:app_prove_envelop_2}), 
\[
\begin{aligned}
    \frac{\partial}{\partial v} G_{i,t}(v, h_{t}, T) \Big|_{v=s_{i,t}} &=\mathbb{E}^{\sigma}_{a'_{i,t}}\Big[ \sum\limits_{k=t}^{T}\frac{\partial}{\partial v}u_{i,k}(v, a_{k})\big|_{v=\tilde{s}_{i,k}}\mathtt{mp}_{i,k}(a_{i,t}, s_{i,t}, h_{t}, k)\Big|s_{i,t},h_{t}\Big]\\
    &=q_{i,t}(\sigma_{i,t}(s_{i,t},h_{t}), s_{i,t}, h_{t}, T).
\end{aligned}
\]
Let 
\[
\begin{aligned}
    \mathtt{Dq}_{i,t}(s_{i,t}, s'_{i,t}, h_{t}, T)\equiv q_{i,t}(\sigma_{i,t}(s_{i,t},h_{t}), s_{i,t}, h_{t}, T) - q_{i,t}(\sigma_{i,t}(s'_{i,t},h_{t}), s_{i,t}, h_{t}, T).
\end{aligned}
\]
Then, it is straightforward to obtain
\[
\begin{aligned}
    G_{i,t}(s_{i,t}, h_{t}, T)-G_{i,t}(\sigma_{i,t}(s'_{i,t})|s_{i,t}, h_{t}, T)=\int_{s'_{i,t}}^{s_{i,t}}\mathtt{Dq}_{i,t}(v, s'_{i,t}, h_{t}, T)dv.
\end{aligned}
\]
That is,
\begin{equation}\label{eq:app_O_2}
    \begin{aligned}
    &G_{i,t}(s_{i,t}, h_{t}, T)-G_{i,t}(\sigma_{i,t}(s'_{i,t})|s_{i,t}, h_{t}, T)
    =\int^{s'_{i,t}}_{s_{i,t}}q_{i,t}(\sigma_{i,t}(v), v, h_{t}, T)dv-\int^{s'_{i,t}}_{s_{i,t}}q_{i,t}(\sigma_{i,t}(s_{i,t}), v, h_{t}, T)dv\\
&=\max\limits_{L\in\mathbb{T}_{t,T}}\int^{s'_{i,t}}_{s_{i,t}}q_{i,t}(\sigma_{i,t}(v), v, h_{t}, L)dv-\max\limits_{L\in\mathbb{T}_{t,T}}\int^{s'_{i,t}}_{s_{i,t}}q_{i,t}(\sigma_{i,t}(s_{i,t}), v, h_{t}, L)dv\geq 0,
\end{aligned}
\end{equation}
where the last inequality is due to the fact that (\ref{eq:condition_RAIC_sigma_0}) is satisfied, which implies that the mechanism is RAIC.
%
 \hfill $\square$

\section{Proof of Proposition \ref{prop:phi_uniqueness} 
 }\label{app:prop:phi_uniqueness}


We denote $G_{i,t}(\cdot; \Theta)$ and $V_{i,t}(\cdot; \Theta)$ to highlight their dependence on the mechanism $\Theta$.
First, we obtain the following lemma.

\begin{lemma}\label{lemma:equivalence_Lambda}
    Fix a $S^{\mathtt{off}}$ and $\sigma$.
    Let $<\rho^{1},\phi^{1}>$ and $<\rho^{2},\phi^{2}>$ be two different profiles, where $\phi^{1}$ and $\phi^{2}$ are cutoff-switch functions with the same boundary profile $c$ that lead to an essential partition of each $S^{\mathtt{off}}_{i,t}$ for all $i\in\mathcal{N}$, $t\in\mathbb{T}$.
    Suppose that $<\sigma,\rho^{1},\phi^{1}>$ and $<\sigma,\rho^{2},\phi^{2}>$ are both $S^{\mathtt{off}}$-DOIC.
    Let $\tau=\{\tau_{i,t}\}$ denote the OM strategy profile associated with $S^{\mathtt{off}}$.
    Then, for all $i\in\mathcal{N}$, $t\in\mathbb{T}$, obedient history $h_{t}\in H_{t}$, $s_{i,t}\in S_{i,t}$,
    \begin{equation}\label{eq:equivalence_Lambda_1}
        \begin{aligned}
            G_{i,t}(s_{i,t},h_{t}, L^{*},\chi_{-i,t}; \Theta^{1}) - G_{i,t}(s_{i,t},h_{t}, L^{*},\chi_{-i,t}; \Theta^{2})=C_{i,t}(h_{t}),
        \end{aligned}
    \end{equation}
    where $C_{i,t}(h_{t})\in\mathbb{R}$ is a constant depending only on $h_t$, $L^{*}=\tau^{d}_{i,t}(s_{i,t}, h_{t})$, $\tau^{d}=(\tau^{d}_{i,t})$ and $\chi=(\chi_{i,t})$ satisfy (\ref{eq:x_coincides_tau}) and (\ref{eq:principal_desired_belief_sym}) given $S^{\mathtt{off}}$.
\end{lemma}

\proof{Proof}

    Denote $\Theta^{1}=<\sigma,\rho^{1},\phi^{1}>$ and $\Theta^{2}=<\sigma,\rho^{2},\phi^{2}>$.
    Since the mechanisms are $S^{\mathtt{off}}$-switchable, they are RAIC.
    In addition, since $\Theta^{1}$ and $\Theta^{2}$ share the same task policy profile, $\Theta^{1}$ and $\Theta^{2}$ induces the same continuing outcome $\gamma^{\sigma}[s_{i,t}, h_{t}]$ for any $s_{i,t}$ and $h_{t}$.
    Since the mechanisms are RAIC, it follows straightforwardly from Proposition \ref{prop:necessary_DOIC_sigma_tau} that
    \[
    \begin{aligned}
            \frac{\partial}{\partial v} V_{i,t}(v, h_{t};\Theta^{1})\Big|_{v=s_{i,t}}=\frac{\partial}{\partial v} V_{i,t}(v, h_{t};\Theta^{2})\Big|_{v=s_{i,t}}= q_{i,t}(\sigma_{i,t}(s_{i,t},h_{t}), s_{i,t}, h_{t}, L^{*}).
    \end{aligned}
    \]
    where $q_{i,t}$ is given by (\ref{eq:envelope_term}).
    Since in $S^{\mathtt{off}}$-DOIC mechanism, $V_{i,t}(s_{i,t}, h_{t}; \Theta^{v})= G_{i,t}(s_{i,t}, h_{t}, L^{*}, \chi_{-i,t}; \Theta^{v})$ for all $v\in\{1,2\}$, there exists a constant $C_{i,t}(h_{t})$ for each obedient history $h_{t}$ such that (\ref{eq:equivalence_Lambda_1}) holds.
    \hfill $\square$
\endproof

By Lemma \ref{lemma:equivalence_Lambda}, we have
\[
\begin{aligned}
    \mathbb{E}^{\sigma}\Big[\sum\limits_{k=t}^{L^{*}-1}\rho^{1}_{i,k}(\Tilde{a}_{k}, \Tilde{h}_{t}) + \phi^{1}_{i,L^{*}}(\Tilde{h}_{t}|c_{i,L^{*}}) \Big|s_{i,t}, h_{t}\Big] = \mathbb{E}^{\sigma}\Big[\sum\limits_{k=t}^{L^{*}-1}\rho^{2}_{i,k}(\Tilde{a}_{k}, \Tilde{h}_{t}) + \phi^{2}_{i,L^{*}}(\Tilde{h}_{t}|c_{i,L^{*}}) \Big|s_{i,t}, h_{t}\Big] + C_{i,t}(h_{t}),
\end{aligned}
\]
where $L^{*}=\tau_{i,t}(s_{i,t}, h_{t})$
Suppose that $\Theta^{1}$ and $\Theta^{2}$ have the same coupling function profile (i.e., $\rho^{1}=\rho^{2}$).
Since the mechanism is $S^{\mathtt{off}}$-switchable, we have
\[
\begin{cases}
    \mathbb{E}^{\sigma}\Big[ \phi^{1}_{i,L}(\Tilde{h}_{t}|c_{i,L}) \Big|s_{i,t}, h_{t}\Big] - \mathbb{E}^{\sigma}\Big[ \phi^{2}_{i,L}(\Tilde{h}_{t}|c_{i,L}) \Big|s_{i,t}, h_{t}\Big] = C_{i,t}(h_{t}), &\textup{ if } s_{i,t}\not\in S^{\mathtt{off}}_{i,t}\\
    \phi^{1}_{i,t}(h_{t}|c_{i,t}) - \phi^{2}_{i,t}(h_{t}|c_{i,t}) = C_{i,t}(h_{t}),& \textup{ if } s_{i,t}\in S^{\mathtt{off}}_{i,t}
\end{cases},
\]
for all $s_{i,t}\in S_{i,t}$, $h_{t}\in H_{t}$, where $L=\tau^{d}_{i,t}(s_{i,t}, h_{t})\in \mathbb{T}\cup\{T+1\}$.
Hence,
\[
\begin{aligned}
    \phi^{1}_{i,t}(h_{t}|c_{i,t}) - \phi^{2}_{i,t}(h_{t}|c_{i,t}) =  \mathbb{E}^{\sigma}\Big[ \phi^{1}_{i,L}(\Tilde{h}_{t}|c_{i,L}) \Big|s_{i,t}, h_{t}\Big] - \mathbb{E}^{\sigma}\Big[ \phi^{2}_{i,L}(\Tilde{h}_{t}|c_{i,L}) \Big|s_{i,t}, h_{t}\Big],
\end{aligned}
\]
in which the left-hand side is independent of $s_{i,t}$.
Due to Assumption \ref{assp:full_support}, we have that $\mathbb{P}(\tau^{d}_{i,t}(s_{i,t}, h_{t})=T+1)>0$.
Then, it also holds 
\[
\begin{aligned}
    \phi^{1}_{i,t}(h_{t}|c_{i,t}) - \phi^{2}_{i,t}(h_{t}|c_{i,t}) =  \mathbb{E}^{\sigma}\Big[ \phi^{1}_{i,T+1}(\Tilde{h}_{t}|c_{i,T+1}) \Big|s_{i,t}, h_{t}\Big] - \mathbb{E}^{\sigma}\Big[ \phi^{2}_{i,T+1}(\Tilde{h}_{t}|c_{i,T+1}) \Big|s_{i,t}, h_{t}\Big]=0.
\end{aligned}
\]
Therefore, $\phi^{1}_{i,t}(h_{t}|c_{i,t}) = \phi^{2}_{i,t}(h_{t}|c_{i,t})$.
 \hfill $\square$

\section{Proof of Proposition \ref{thm:sufficient_condition_positive_classic_mechanism} }\label{app:thm:sufficient_condition_positive_classic_mechanism}

We provide proof for part \textit{(i)} only. The proof for part \textit{(ii)} can be easily obtained analogously.
Since $\mathcal{S}^{\mathtt{OX}}[\sigma, \mathcal{G}]\neq\emptyset$, for any $S^{\mathtt{off}}\in \mathcal{S}^{\mathtt{OH}}[\sigma, \mathcal{G}]$, we have a mechanism $\Theta=<\sigma, \rho, \phi>$ that is $S^{\mathtt{off}}$-DOIC, in which $\rho\in\mathcal{P}[\sigma,\mathcal{G}]$.
Due to Theorem \ref{thm:MAO_switchability_conjecture_gen} and the definition of $\mathcal{S}^{\mathtt{OH}}[\sigma,\mathcal{G}]$, if each cutoff-switch is given by (\ref{eq:cutoff_horizontal_general}), then the mechanism $\Theta=<\sigma, \rho, \phi>$ is $S^{\mathtt{off}}$-DOIC for any $S^{\mathtt{off}}\in \mathcal{S}^{\mathtt{OH}}[\sigma, \mathcal{G}]$.

First, we show the \textit{only if} ($\Rightarrow$).
From Proposition \ref{prop:phi_uniqueness}, we know that in any $S^{\mathtt{off}}$-DOIC mechanism, $\phi$ is unique. 
Thus, $\mathcal{S}^{\mathtt{CH}}[\sigma, \mathcal{G}] = \mathcal{S}^{\mathtt{OH}}[\sigma, \mathcal{G}]$ implies (\ref{eq:cutoff_condition_zero_H}).

Next, we show the \textit{if} ($\Leftarrow$).
Given $\sigma$ and $\mathcal{G}$, suppose that there exists $S^{\mathtt{off}}\in \mathcal{S}^{\mathtt{CH}}[\sigma, \mathcal{G}]$ such that 
\begin{equation*}
            \begin{aligned}
            \widehat{\phi}_{i,t}(h_{t};\sigma)=\mathtt{Mg}_{i,t}(\overline{d}_{i,t}(b), h_{t}, \chi_{-i,t}) + \overline{\delta}_{i,t}(\overline{d}_{i,t}(b), h_{i,t}, \chi_{-i,t}) \neq 0.
            \end{aligned}
        \end{equation*}
Note that both $\mathtt{Mg}_{i,t}$ and $\overline{\delta}_{i,t}$ depend on $\sigma$.
Then, for any $\rho\in\mathcal{P}[\sigma,\mathcal{G}]$, if we construct each $\hat{\phi}_{i,t}(h_{t}) =\widehat{\phi}_{i,t}(h_{t};\sigma)$, then Theorem \ref{thm:MAO_switchability_conjecture_gen} implies that the mechanism $\Theta=<\sigma, \rho, \hat{\phi}>$ is $S^{\mathtt{off}}$-DOIC, which contradicts to that $\sigma$ satisfies (\ref{eq:cutoff_condition_zero_H}).
\hfill $\square$

\section{Proof of Corollary \ref{corollary:envelop_condition_indifference}}\label{app:corollary:envelop_condition_indifference}

Let $\hat{\tau}=(\hat{\tau}_{i,t})$, in which each $\hat{\tau}_{i,t}(s_{i,t}, h_{t})=T+1$ for all $s_{i,t}\in S^{\mathtt{off}}_{i,t}$, $i\in\mathcal{N}$, $t\in\mathbb{T}$.
Since the mechanism $\Theta$ is $S^{\mathtt{off}}$-DOIC with $S^{\mathtt{off}}\in \mathcal{S}^{\mathtt{IX}}[\sigma,\mathcal{G}]$, each $S^{\mathtt{off}}_{i,t}$ induces an on-rent of zero for all $i\in\mathcal{N}$, $t\in\mathbb{T}$.
Then, $\hat{\tau}=(\hat{\tau}_{i,t})$ is one the OM strategy profile that can be incentivized by the mechanism $\Theta$. 
Following
(\ref{eq:envelope_like_condition_sigma_tau}) of Proposition \ref{prop:necessary_DOIC_sigma_tau}, we obtain (\ref{eq:first_order_indifference}).
\hfill $\square$

\section{ Proof of Corollary \ref{corollary:envelope_V_MSO}}\label{app:corollary:envelope_V_MSO}

From the proof of Proposition \ref{prop:necessary_DOIC_sigma_tau} in Appendix \ref{app:necessary_DOIC_sigma_tau}, we can easily have that the following difference quotient is bounded for all $\varsigma\neq 0$:
\[
\begin{aligned}
    \Big|\max\limits_{L\in\mathbb{T}_{t,T}}G_{i,t}(a'_{i,t}|s_{i,t}+\varsigma, h_{t}, L) - \max\limits_{L\in\mathbb{T}_{t,T}}G_{i,t}(a'_{i,t}|s_{i,t}, h_{t}, L) \Big|\leq \sum_{k=t}^{T} c_{i,k} \big(\prod^{k}_{t'=1}\hat{c}_{i,t'}\big)\big|\varsigma\big|,
\end{aligned}
\]
for all $s_{i,t},s'_{i,t}$ with $a'_{i,t}=\sigma_{i,t}(s'_{i,t}, h_{t})$.
which implies that $MG_{i,t}$ is Lipschitz continuous in $s_{i,t}$.
RAIC implies that $V_{i,t}(s_{i,t}, h_{t}, \chi_{i,t})=MG_{i,t}(s_{i,t}|s_{i,t},h_{t}, \chi_{-i,t})$.
Thus, $V_{i,t}(s_{i,t}, h_{t}, \chi_{-i,t})$ is Lipschitz continuous in $s_{i,t}$ with constant $\big(\prod^{k}_{t'=1}\hat{c}_{i,t'}\big)$.
Let $L(s_{i,t})\in \mathbb{T}_{t,T}$ satisfy $MG_{i,t}(s_{i,t}|s_{i,t},h_{t}, \chi_{-i,t})= G_{i,t}(s_{i,t},h_{t}, L(s_{i,t}),\chi_{-i,t})$.
Then, the following exists for every $s_{i,t}\in S_{i,t}$,

%
\[
\begin{aligned}
    &\lim_{\varsigma \rightarrow 0} \frac{MG_{i,t}(s_{i,t}|s_{i,t}+\varsigma, h_{t})- MG_{i,t}(s_{i,t}|s_{i,t}, h_{t}) }{\varsigma}= \lim_{\varsigma \rightarrow 0}\frac{\max_{L\in\mathbb{T}_{t,T}} G_{i,t}(s_{i,t}+\varsigma, h_{t}, L) - \max_{L\in\mathbb{T}_{t,T}}G_{i,t}(s_{i,t}, h_{t}, L) }{\varsigma}.
\end{aligned}
\]
Thus, we have for all $s_{i,t}\in S_{i,t}$, $\frac{\partial }{\partial v} MG_{i,t}(s_{i,t}|v, h_{t})\Big|_{v=s_{i,t}} = \frac{\partial }{\partial v} \max\limits_{L\in\mathbb{T}_{t,T}} G_{i,t}(s_{i,t}|v, h_{t}, L)\Big|_{v=s_{i,t}}$.
Since the mechanism induces maximum-sensitive obedience, it holds that
\[
\begin{aligned}
    \frac{\partial }{\partial v} MG_{i,t}(s_{i,t}|v, h_{t})\Big|_{v=s_{i,t}} &=\max\limits_{L\in\mathbb{T}_{t,T}}\frac{\partial }{\partial v} G_{i,t}(s'_{i,t}|v, h_{t}, L)\Big|_{v=s_{i,t}}
    \\
    &=\max\limits_{L\in\mathbb{T}_{t,T}}\mathbb{E}^{\sigma}\Big[ \sum\limits_{k=t}^{L}\frac{\partial}{\partial v}u_{i,k}(v, a_{k})\big|_{v=\tilde{s}_{i,k}}\mathtt{mp}_{i,k}(a_{i,t}, s_{i,t}, h_{t}, k)\Big|s_{i,t},h_{t}\Big].
\end{aligned}
\]
Then, we have
$\frac{\partial }{\partial v} MG_{i,t}(s_{i,t}|v, h_{t}, \chi_{-i,t})\Big|_{v=s_{i,t}}\leq \max\limits_{L\in\mathbb{T}_{t,T}}\sum_{k=t}^{L} c_{i,k} \big(\prod^{k}_{t'=1}\hat{c}_{i,t'}\big)$.
    Hence, $MG_{i,t}(s'_{i,t}|s_{i,t}, h_{t}, \chi_{-i,t})$ is Lipschitz continuous in $s_{i,t}$ with constant $\max\limits_{L\in\mathbb{T}_{t,T}}\sum_{k=t}^{L} c_{i,k} \big(\prod^{k}_{t'=1}\hat{c}_{i,t'}\big)$.
Since the mechanism is RAIC, $V_{i,t}(s_{i,t}, h_{t}, \chi_{-i,t})= MG_{i,t}(s_{i,t}|s_{i,t},h_{t}, \chi_{-i,t})$. 
That is, $V_{i,t}(\cdot)$ is Lipschitz continuous in $s_{i,t}$ with constant $\max\limits_{L\in\mathbb{T}_{t,T}}\sum_{k=t}^{L} c_{i,k} \big(\prod^{k}_{t'=1}\hat{c}_{i,t'}\big)$.
Then, By Theorem 1 of \citet{milgrom2002envelope}, $\frac{\partial }{\partial v} V_{i,t}(v, h_{t}, \chi_{i,t})\big|_{v=s_{i,t}} = $  $ \frac{\partial }{\partial v} MG_{i,t}(v, h_{t}, \chi_{-i,t})\Big|_{v=s_{i,t}}$.
Thus, we obtain (\ref{eq:first_order_like_condition_envelope_MSO}).
%
 \hfill $\square$

\section{Proof of Theorem \ref{thm:necessary_RAIC_condition} }\label{app:thm:necessary_RAIC_condition}

Since $\mathcal{S}^{\mathtt{M}}[\sigma,  \mathcal{G}]\neq \emptyset$, for any $S^{\mathtt{off}}\in \mathcal{S}^{\mathtt{M}}[\sigma,  \mathcal{G}]$, we have a mechanism $\Theta=<\sigma, \rho, \phi>$ such that the mechanism is $S^{\mathtt{off}}$-DOIC.
Then, \textit{(i)} each agent $i$ is obedient in taking the OM actions and has conjecture $\chi_{-i,t}$ about others' OM strategies, which satisfies (\ref{eq:principal_desired_belief_sym}) given $S^{\mathtt{off}}$, and \textit{(i)} the mechanism is RAIC.
We drop the obedient conjecture $\chi_{-i,t}$ for simplicity.

Due to RAIC, it holds that 
\[
\begin{aligned}
    V_{i,t}(s_{i,t}, h_{t})=MG_{i,t}(s_{i,t}|s_{i,t},h_{t})\geq MG_{i,t}(s'_{i,t}|s_{i,t},h_{t}),
\end{aligned}
\]
for all $i\in\mathcal{N}$, $t\in\mathbb{T}$, $s_{i,t}, s'_{i,t}\in S_{i,t}$, $h_{t}\in H_{t}$.
Corollary \ref{corollary:envelope_V_MSO} implies 
\[
\frac{\partial}{\partial v} V_{i,t}(v, h_{t})\Big|_{v=s_{i,t}}= \max_{L\in\mathbb{T}_{t,T}} q_{i,t}(\sigma_{i,t}(s_{i,t},h_{t}), s_{i,t}, h_{t}, L).
\]
Hence, we can construct $V_{i,t}$ by
\begin{equation}\label{eq:app_prove_monotone_sigma_1}
    \begin{aligned}
    V_{i,t}(s_{i,t}, h_{t})= V_{i,t}(\mathring{s}_{i,t}, h_{t}) + \int_{\mathring{s}_{i,t}}^{s_{i,t}}\max\limits_{L\in \mathbb{T}_{t,T}} q_{i,t}(\sigma_{i,t}(v), v, h_{t}, L)d v, 
\end{aligned}
\end{equation}
for all $\mathring{s}_{i,t}\in S_{i,t}$.

Next, we want to obtain results analogous to (\ref{eq:first_order_like_condition_envelope_MSO}) for agent $i$ with true state $s_{i,t}$ when he makes a unilateral one-shot deviation from obedient in taking a regular action in period $t$.
Suppose that agent $i$ takes $\mathtt{OM}_{i,t}=0$ and an action $a'_{i,t}\in A_{i,t}[\sigma]$ such that $a'_{i,t}=\sigma_{i,t}(s'_{i,t}, h_{t})$ for some $s'_{i,t}\in S_{i,t}$.
Suppose in addition that agent $i$ is obedient for all $t'\neq t$.
Then, define
\[
\begin{aligned}
    V^{\dagger}_{i,t}(a'_{i,t}|s_{i,t}, h_{t},\chi_{-i,t})\equiv \mathbb{E}^{\sigma}_{a'_{i,t}}\Bigg[z_{i,t}(a'_{i,t}, \tilde{a}_{-i,t}, s_{i,t}) +   V_{i,t+1}(\tilde{s}_{i,t+1}, \tilde{h}_{t+1})\Bigg|s_{i,t}, h_{t}, \chi_{-i,t}\Bigg],
\end{aligned}
\]
where $V_{i,t+1}$ satisfies (\ref{eq:app_prove_monotone_sigma_1}).
%

\begin{lemma}\label{lemma:first_order_like_condition_deviation}
    Suppose that Conditions \ref{cond:differentiable_reward} and \ref{cond:bounded_dynamic} hold.
    Let $S^{\mathtt{off}}$ is desired by the principal, and let $\chi=(\chi_{i,t})$ satisfy (\ref{eq:principal_desired_belief_sym}) given $S^{\mathtt{off}}$.
    If the mechanism $\Theta$ is $S^{\mathtt{off}}$-DOIC, then for all $i\in\mathcal{N}$, $t\in\mathbb{T}$, $s_{i,t}\in S_{i,t}$, $h_{t}\in H_{t}$,
\begin{equation}\label{eq:first_order_like_condition_deviation}
    \begin{aligned}
        &\frac{\partial}{\partial v} V^{\dagger}_{i,t}(a'_{i,t}|v, h_{t}, \chi_{-i,t})\big|_{v=s_{i,t}} = \max\limits_{L\in\mathbb{T}_{t,T}}q_{i,t}(a'_{i,t}, s_{i,t}, h_{t}, L, \chi_{-i,t}).
    \end{aligned}
\end{equation}
    %
\end{lemma}

\begin{proof}
    Again, we temporally drop $\chi_{-i,t}$ for simplicity.
    Let $h_{t+1}=(h_{t}, a'_{i,t}, a_{-i,t})$. It is clear that $h_{t+1}$ is independent of $s_{i,t}$ since $a'_{i,t}=\sigma_{i,t}(s'_{i,t}, h_{t})$.
    Then, taking derivative of $V_{i,t+1}(s_{i,t+1}, h_{t+1})$ with respect to $s_{i,t}$ gives
    \[
    \begin{aligned}
        &\frac{d}{dv} V_{i,t+1}(\kappa_{i,t+1}(v, h_{t+1}, \omega_{i,t+1}), h_{t+1})\Big|_{v=s_{i,t}}=\frac{\partial}{\partial v} V_{i,t+1}(v', h_{t+1})\frac{\partial}{\partial v} \kappa_{i,t+1}(v, h_{t+1}, \omega_{i,t+1})\Big|^{v'=s_{i,t+1}}_{v=s_{i,t}}\\
        &=\max\limits_{L\in\mathbb{T}_{t+1,T}} q_{i,t+1}(\sigma_{i,t+1}(s_{i,t+1},h_{t+1}), s_{i,t+1}, h_{t+1}, L)\frac{\partial}{\partial v} \kappa_{i,t+1}(v, h_{t+1}, \omega_{i,t+1})\Big|_{v=s_{i,t}}.
    \end{aligned}
    \]
Hence, by Lebesgue dominated convergence theorem,
\[
\begin{aligned}
   &\frac{\partial}{\partial v} V^{\dagger}_{i,t}(a'_{i,t}|v, h_{t}\big|_{v=s_{i,t}}\ = \mathbb{E}^{\sigma}_{a'_{i,t}}\Big[ \frac{\partial}{\partial v}u_{i,t}(v,a'_{i,t}, \tilde{a}_{-i,t})\Big|_{v=s_{i,t}}   + \frac{\partial}{\partial v} \kappa_{i,t+1}(v, h_{t+1}, \omega_{i,t+1})\Big|_{v=s_{i,t}}\\
   &\times\max_{L\in \mathbb{T}_{t+1,T}}\sum\limits_{k=t+1}^{L}\frac{\partial}{\partial v}u_{i,k}(v, \tilde{a}_{k})\big|_{v=\tilde{s}_{i,k}}\mathtt{mp}_{i,k}(\tilde{a}_{i,t+1}, \tilde{s}_{i,t+1}, \tilde{h}_{t+1}, k) \Big|s_{i,t}, h_{t}\Big]\\
   & =\max\limits_{L\in\mathbb{T}_{t,T}}\mathbb{E}^{\sigma}_{a'_{t}}\Big[ \sum\limits_{k=t}^{L}\frac{\partial}{\partial v}u_{i,k}(v, \tilde{a}_{k})\big|_{v=\tilde{s}_{i,k}}\mathtt{mp}_{i,t}(a_{i,t},s_{i,t}, h_{t}, k, \chi_{-i,t}) \Big|s_{i,t}, h_{t}, \chi_{-i,t}
 \Big],
\end{aligned}
\]
which obtains (\ref{eq:first_order_like_condition_deviation}).
\end{proof}

From Lemma \ref{lemma:first_order_like_condition_deviation}, we can construct $V^{\dagger}_{i,t}$ similar to (\ref{eq:app_prove_monotone_sigma_1}), for all $\mathring{s}_{i,t}\in S_{i,t}$,
\begin{equation}\label{eq:app_prove_monotone_sigma_2}
    \begin{aligned}
    V^{\dagger}_{i,t}(a'_{i,t}|s_{i,t}, h_{t}) = V^{\dagger}_{i,t}(a'_{i,t}|\mathring{s}_{i,t}, h_{t}) + \int^{s_{i,t}}_{\mathring{s}_{i,t}}\max\limits_{L\in \mathbb{T}_{t,T}} q_{i,t}(\sigma_{i,t}(s'_{i,t}, h_{t}), v, h_{t}, L)dv.
\end{aligned}
\end{equation}

Since the mechanism is RAIC, for all $s_{i,t},s'_{i,t}\in S_{i,t}$, it holds that
\begin{equation}\label{eq:app_prove_monotone_sigma_3}
    \begin{aligned}
    V^{\dagger}_{i,t}(\sigma_{i,t}(s'_{i,t},h_{t})|s'_{i,t}, h_{t})\geq V^{\dagger}_{i,t}(\sigma_{i,t}(s'_{i,t},h_{t})|s_{i,t}, h_{t}).
\end{aligned}
\end{equation}
\begin{equation}\label{eq:app_prove_monotone_sigma_4}
    \begin{aligned}
        V_{i,t}(s_{i,t}, h_{t}) \geq V^{\dagger}_{i,t}(\sigma_{i,t}(s_{i,t},h_{t})|s'_{i,t}, h_{t})
    \end{aligned}
\end{equation}
\begin{equation}\label{eq:app_prove_monotone_sigma_5}
    \begin{aligned}
        V_{i,t}(s'_{i,t}, h_{t})\geq V^{\dagger}_{i,t}(\sigma_{i,t}(s'_{i,t},h_{t})|s_{i,t}, h_{t})
    \end{aligned}
\end{equation}
Thus, (\ref{eq:app_prove_monotone_sigma_4})- (\ref{eq:app_prove_monotone_sigma_5}) yields
\begin{equation}\label{eq:app_prove_monotone_sigma_6}
    \begin{aligned}
     V_{i,t}(s_{i,t}, h_{t}) - V_{i,t}(s'_{i,t}, h_{t}) &\geq V^{\dagger}_{i,t}(\sigma_{i,t}(s_{i,t},h_{t})|s'_{i,t}, h_{t}) - V^{\dagger}_{i,t}(\sigma_{i,t}(s'_{i,t},h_{t})|s_{i,t}, h_{t})\\
     &\geq V^{\dagger}_{i,t}(\sigma_{i,t}(s_{i,t},h_{t})|s'_{i,t}, h_{t}) - V^{\dagger}_{i,t}(\sigma_{i,t}(s'_{i,t},h_{t})|s'_{i,t}, h_{t}),
\end{aligned}
\end{equation}
where the second inequality is due to (\ref{eq:app_prove_monotone_sigma_4}).
Thus, by choosing the same $\mathring{s}_{i,t}\in S_{i,t}$ for (\ref{eq:app_prove_monotone_sigma_1}) and (\ref{eq:app_prove_monotone_sigma_2}), we obtain
\[
\begin{aligned}
    V_{i,t}(s_{i,t}, h_{t}) - V_{i,t}(s'_{i,t}, h_{t}) =  \int^{s_{i,t}}_{s'_{i,t}}\max\limits_{L\in \mathbb{T}_{t,T}} q_{i,t}(\sigma_{i,t}(v,h_{t}), v, h_{t}, L)dv,
\end{aligned}
\]
\[
\begin{aligned}
    V^{\dagger}_{i,t}(\sigma_{i,t}(s_{i,t},h_{t})|s'_{i,t}, h_{t})& - V^{\dagger}_{i,t}(\sigma_{i,t}(s'_{i,t},h_{t})|s'_{i,t}, h_{t})= \int^{s_{i,t}}_{s'_{i,t}} \max\limits_{L\in \mathbb{T}_{t,T}}q_{i,t}(\sigma_{i,t}(s'_{i,t},h_{t}), v,h_{t}, L)dv,
\end{aligned}
\]
where we cancel the constant terms $V_{i,t}(\mathring{s}_{i,t}, h_{t})$ and $V^{\dagger}_{i,t}(a'_{i,t}|\mathring{s}_{i,t}, h_{t})$ that depends on $\rho$ and $\phi$.
Finally, from (\ref{eq:app_prove_monotone_sigma_6}), we obtain a necessary condition which is independent of $\rho$ and $\phi$,
\begin{equation}\label{eq:app_int_max_q}
\int^{s_{i,t}}_{s'_{i,t}}\max\limits_{L\in \mathbb{T}_{t,T}} q_{i,t}(\sigma_{i,t}(v,h_{t}), v, h_{t}, L)dv \geq \int^{s_{i,t}}_{s'_{i,t}} \max\limits_{L\in \mathbb{T}_{t,T}} q_{i,t}(\sigma_{i,t}(s'_{i,t},h_{t}), v, h_{t}, L)dv.
\end{equation}

The next lemma shows that $\max\limits_{L\in \mathbb{T}_{t,T}}$ and $\int^{s_{i,t}}_{s'_{i,t}}$ can be switched.

\begin{lemma}\label{lemma:app_switch_int_max}
    Fix a base game $\mathcal{G}$ and a task policy profile $\sigma$.
    Suppose that Conditions \ref{cond:differentiable_reward} and \ref{cond:bounded_dynamic} hold.
    Then, (\ref{eq:app_int_max_q}) is equivalent to
    \begin{equation}\label{eq:app_max_int_q}
    \max\limits_{L\in \mathbb{T}_{t,T}} \int^{s_{i,t}}_{s'_{i,t}}q_{i,t}(\sigma_{i,t}(v,h_{t}), v, h_{t}, L)dv \geq \max\limits_{L\in \mathbb{T}_{t,T}} \int^{s_{i,t}}_{s'_{i,t}} q_{i,t}(\sigma_{i,t}(s'_{i,t},h_{t}), v, h_{t}, L)dv.
    \end{equation}
\end{lemma}

\begin{proof}
    Since $\frac{\partial}{\partial v} \Big|G_{i,t}(a'_{i,t}|v, h_{t}, L) \big|_{v=s_{i,t}}\Big|\leq \sum_{k=t}^{L} c_{i,k} \big(\prod^{k}_{t'=1}\hat{c}_{i,t'}\big)\leq\sum_{k=t}^{T} c_{i,k} \big(\prod^{k}_{t'=1}\hat{c}_{i,t'}\big)$, for all $L\in\mathbb{T}_{t,T}$.
Then, $q_{i,t}$ is uniformly bounded; i.e., $|q_{i,t}(a'_{i,t},s_{i,t}, h_{t}, L)|\leq \sum_{k=t}^{T} c_{i,k} \big(\prod^{k}_{t'=1}\hat{c}_{i,t'}\big)$.
In addition, for all $\varsigma\neq 0$,
\[
\begin{aligned}
    \left|q_{i,t}(a'_{i,t},s_{i,t}+\varsigma, h_{t}, L)- q_{i,t}(a'_{i,t},s_{i,t}, h_{t}, L) \right|&= \left| \frac{\partial}{\partial v} G_{i,t}(a'_{i,t}|v, h_{t}, L) \big|_{v=s_{i,t}+\varsigma}- \frac{\partial}{\partial v} G_{i,t}(a'_{i,t}|v, h_{t}, L) \big|_{v=s_{i,t}}\right|\\&\leq 2\sum_{k=t}^{T} c_{i,k} \big(\prod^{k}_{t'=1}\hat{c}_{i,t'}\big)\leq 2\sum_{k=t}^{T} c_{i,k} \big(\prod^{k}_{t'=1}\hat{c}_{i,t'}\big)\big|\varsigma \big|.
\end{aligned}
\]
Hence, $q_{i,t}(a'_{i,t}, s_{i,t}, h_{t}, L)$ is Lipschitz continuous in $s_{i,t}$.
Similarly, we can show that $q_{i,t}(\sigma_{i,t}(s_{i,t},h_{t}), s_{i,t}, h_{t}, L)$ is also Lipschitz continuous in $s_{i,t}$.
 Then, from Conditions \ref{cond:differentiable_reward} and \ref{cond:bounded_dynamic}, we have 
    \[
    \begin{aligned}
        &q_{i,t}(a'_{i,t},s_{i,t}, h_{t}, L) = \mathbb{E}^{\sigma}_{a'_{i,t}}\Big[ \sum\limits_{k=t}^{L}\frac{\partial}{\partial v}u_{i,k}(v, a_{k})\big|_{v=\tilde{s}_{i,k}}\mathtt{mp}_{i,k}(a_{i,t}, s_{i,t}, h_{t}, k)\Big|s_{i,t},h_{t}\Big]\\
        &\leq \mathbb{E}^{\sigma}_{a'_{i,t}}\Big[ \sum\limits_{k=t}^{L}\frac{\partial}{\partial v}u_{i,k}(v, \tilde{a}_{k})\big|_{v=\tilde{s}_{i,k}}\Big]L\max\{\hat{c}_{i,k}\}_{k\in\mathbb{T}_{t,L}} \leq L^{2}\max\{\hat{c}_{i,k}\}_{k\in\mathbb{T}_{t,L}}\max\{c_{i,k}\}_{t,L} \mathbb{E}^{\sigma}_{a'_{i,t}}\Big[\frac{\partial}{\partial v}u_{i,t}(v, a'_{i,t}, \tilde{a}_{-i,t}) \Big],
    \end{aligned}
    \]
     where the first inequality is due to Condition \ref{cond:bounded_dynamic} and the second inequality is due to Condition \ref{cond:differentiable_reward}.
     Furthermore, Condition \ref{cond:differentiable_reward} implies that $\Big[\frac{\partial}{\partial v}u_{i,t}(v, a'_{i,t}, \tilde{a}_{-i,t}) \Big]$ is bounded and integrable. 
     Construct $\mathtt{Mq}_{i,t}(s_{i,t}) = \widehat{C}\Big|\frac{\partial}{\partial v}u_{i,t}(v, a'_{i,t}, \tilde{a}_{-i,t}) \Big|$, where $\widehat{C}=L^{2}\max\{\hat{c}_{i,k}\}_{k\in\mathbb{T}_{t,L}}\max\{c_{i,k}\}_{t,L}$.
     Hence, we have $q_{i,t}(a'_{i,t},s_{i,t}, h_{t}, L)\leq \mathtt{Mq}_{i,t}(s_{i,t})$, where $\mathtt{Mq}_{i,t}(s_{i,t})$ is an integrable function of $s_{i,t}$ with $\int_{s\in S_{i,t}}|\mathtt{Mq}_{i,t}(v)|dv<\infty$.
     Then, from Theorem 2.1 of \citet{strugarek2006interchange}, we have
     \[
     \max\limits_{L\in \mathbb{T}_{t,T}} \int^{s_{i,t}}_{s'_{i,t}} q_{i,t}(\sigma_{i,t}(s'_{i,t},h_{t}), v, h_{t}, L)dv= \int^{s_{i,t}}_{s'_{i,t}} \max\limits_{L\in \mathbb{T}_{t,T}} q_{i,t}(\sigma_{i,t}(s'_{i,t},h_{t}), v, h_{t}, L)dv.
     \]
     We can also show the interchange of $\max\limits_{L\in \mathbb{T}_{t,T}}$ and $\int^{s_{i,t}}_{s'_{i,t}}$ on the left-hand side of (\ref{eq:app_max_int_q}).
     Hence, we can conclude that (\ref{eq:app_int_max_q}) is equivalent to (\ref{eq:app_max_int_q}).
\end{proof}

Lemma \ref{lemma:app_switch_int_max} concludes the proof.
%
 \hfill $\square$

\section{Proof of Proposition \ref{prop:nece_suffic_indifference}}\label{app:prop:nece_suffic_indifference}

For any $\sigma$, if $\mathcal{S}^{\mathtt{MIX}}[\sigma,\mathcal{G}]\neq \emptyset$, then $\mathcal{S}^{\mathtt{MIX}}[\sigma,\mathcal{G}]\subseteq \mathcal{S}^{\mathtt{MIX}}[\mathcal{G}]$.
We proceed with the proof by establishing a contradiction.
Suppose that there is a $\hat{\sigma}\not\in \Sigma^{\mathtt{MX}}[\mathcal{G}]$ such that there exist $\hat{\rho}$ and $\hat{\phi}$ such that the mechanism $\hat{\Theta}=<\hat{\sigma}, \hat{\rho}, \hat{\phi}>$ is $\hat{S}^{\mathtt{off}}$-DOIC, for a $\hat{S}^{\mathtt{off}}\in \mathcal{S}^{\mathtt{MIX}}[\mathcal{G}]$.
Since the mechanism $\hat{\Theta}$ induces MSO, $\mathcal{S}^{\mathtt{M}}[\sigma, \mathcal{G}]\neq \emptyset$.
By Theorem \ref{thm:necessary_RAIC_condition}, $\hat{\sigma}$ satisfies the constrained monotone condition (\ref{eq:condition_RAIC_sigma_0}), which contradicts $\hat{\sigma}\not\in \Sigma^{\mathtt{MX}}[\mathcal{G}]$.
\hfill $\square$

\section{Proof of Corollary \ref{corollary:iff_MSO_DCM}}\label{app:corollary:iff_MSO_DCM}

We provide proof for part \textit{(i)} only. The proof for part \textit{(ii)} can be easily obtained analogously.

\subsection{($\Leftarrow$)}

Choose $\rho\in \mathcal{P}[\sigma,\mathcal{G}]$.
Then, following a similar way to obtaining (\ref{eq:app_O_1}) in Appendix \ref{app:thm:sufficient_condition_without_MSO}, we have
\[
\begin{aligned}
    &MG_{i,t}(s_{i,t}, h_{t})=\max\limits_{L\in \mathbb{T}_{t,T}}g_{i,t}( s_{i,t}, h_{t}, L|\theta_{i,t}) + \max\limits_{L'\in \mathbb{T}_{t,T}}\mathbb{E}^{\sigma}\Big[-G_{i,L'+1}(s_{i,L'+1}, h_{L'+1}, T)\Big|s_{i,t}, h_{t}\Big].
\end{aligned}
\]
Since $u_{i,t}(\cdot)\geq 0$ for all $i\in\mathcal{N}$, $t\in\mathbb{T}$, 
\[
\begin{aligned}
    \max\limits_{L'\in \mathbb{T}_{t,T}}\mathbb{E}^{\sigma}\Big[-G_{i,L'+1}(s_{i,L'+1}, h_{L'+1}, T)\Big|s_{i,t}, h_{t}\Big]&=\mathbb{E}^{\sigma}\Big[-G_{i,T+1}(s_{i,T+1}, h_{T+1}, T)\Big|s_{i,t}, h_{t}\Big]=0.
\end{aligned}
\]
Thus, $MG_{i,t}(s_{i,t}, h_{t})=\max\limits_{L\in \mathbb{T}_{t,T}}g_{i,t}( s_{i,t}, h_{t}, L|\theta_{i,t})$. 
Due to \ref{itm:C1}, $\max\limits_{L\in \mathbb{T}_{t,T}}g_{i,t}( s_{i,t}, h_{t}, L|\theta_{i,t})=G_{i,t}(s_{i,t}, h_{t},T)$, in which the left-hand side is independent of $\rho$ while the right-hand side depends on $\rho$.

Hence, following the similar way to obtaining (\ref{eq:app_O_2}) in Appendix \ref{app:thm:sufficient_condition_without_MSO}, we can show that 
\[
\begin{aligned}
    G_{i,t}(s_{i,t}, h_{t}, T)-G_{i,t}(\sigma_{i,t}(s'_{i,t})|s_{i,t}, h_{t}, T)&=\int^{s'_{i,t}}_{s_{i,t}}q_{i,t}(\sigma_{i,t}(v), v, h_{t}, T)dv-\int^{s'_{i,t}}_{s_{i,t}}q_{i,t}(\sigma_{i,t}(s_{i,t}), v, h_{t}, T)dv\\
&=\max\limits_{L\in\mathbb{T}_{t,T}}\int^{s'_{i,t}}_{s_{i,t}}q_{i,t}(\sigma_{i,t}(v), v, h_{t}, L)dv-\max\limits_{L\in\mathbb{T}_{t,T}}\int^{s'_{i,t}}_{s_{i,t}}q_{i,t}(\sigma_{i,t}(s_{i,t}), v, h_{t}, L)dv\\
&\geq 0,
\end{aligned}
\]
where the last inequality is due to the fact that (\ref{eq:condition_RAIC_sigma_0}) is satisfied, which implies that the mechanism is RAIC.

\subsection{($\Rightarrow$)}

Since $S^{\mathtt{off}}\in \mathcal{S}^{\mathtt{CM}}[\sigma,\mathcal{G}]$, there is a dynamic canonical mechanism $\Theta=<\sigma, \rho>$ in which $\rho\in \mathcal{P}[\sigma,\mathcal{G}]$ such that $\Theta$ is $S^{\mathtt{off}}$-DOIC.
Since $S^{\mathtt{off}}\in \mathtt{S}^{H}[\sigma, \mathcal{G}]$, we can construct each cutoff-switch function according to (\ref{eq:cutoff_horizontal_general}) such that the mechanism $\widehat{\Theta}=<\sigma, \rho, \phi>$ is $S^{\mathtt{off}}$-DOIC.
Then, the $S^{\mathtt{off}}$-DOIC of $\Theta=<\sigma, \rho>$ implies (\ref{eq:cutoff_condition_zero_H}).
\hfill $\square$

\section{Table of Notations}

\begin{table}[htp]
	\centering
	\caption{Summary of Notations}
	\begin{tabularx}{\textwidth}{c>{\raggedright}X}
		\toprule  
		Symbol & Meaning\tabularnewline
		\toprule  
		$ \mathcal{E}=<\mathcal{N}, \mathbb{T}, S, A>$ & Event model = $<$set of agents, set of times, joint space of states, joint space of actions$>$.
		\tabularnewline
		\midrule 
		$\mathcal{P}=<\{\kappa_{i,t},F_{i,t}\}_{i\in \mathcal{N}, t\in \mathbb{T}}, \Omega, W>$ & State-dynamic model = $<$ \{state dynamic function, transition probability function\}, space of shocks, probability distribution of shocks $>$.
  \tabularnewline \midrule
      $\mathcal{G}=<\mathcal{E}, \mathcal{P}, \mathcal{Z}>$    & Base game model.
  \tabularnewline \midrule
		$\mathtt{OM}_{i,t}\in\{0,1\}$ & Off-menu (OM) action: $\mathtt{OM}_{i,t}=0$ continue; $\mathtt{OM}_{i,t}=1$ quit.
		\tabularnewline
		\midrule 
	    $\tau_{i,t}(s_{i,t}, h_{t})$, $\pi_{i,t}(s_{i,t}, h_{t})$ & Strategies for OM and regular actions.
     \tabularnewline \midrule
        $\Theta=<\sigma, \rho, \phi>$, $\mathcal{G}^{\Theta}$  & Mechanism $=<$ task policy profile, coupling policy profile, off-switch function profile$>$, a new game induced by $\Theta$.
		 \tabularnewline \midrule
		 $A_{i,t}[\sigma,h_{t}]=A_{i,t}[\sigma]\subseteq A_{i,t}$ & Action menu specified by $\sigma_{i,t}$
   \tabularnewline \midrule
		 $x_{i,t}(\cdot|h_{t})\in \Delta\Big(\big(\mathbb{T}_{t,T}\big)^{|\mathcal{N}_{t}|-1}\Big)$ & Agent $i$'s conjecture about other agents' OM decision, where $\mathcal{N}_{t}\subseteq \mathcal{N}$ is the set of agents who have not quitted at the beginning of $t$.
   \tabularnewline \midrule
   $\mathbb{E}^{\sigma}_{\pi}[\cdot]$ & Ex-ante expectation operator.
   \tabularnewline \midrule
		 $\mathbb{E}^{\sigma}_{\pi}[\cdot|s_{i,t}, h_{t},x_{i,t}]$ & period-$t$ interim expectation operator perceived by agent $i$ given $s_{i,t}, h_{t},x_{i,t}$.
   \tabularnewline \midrule
   $G_{i,t}(a_{i,t}|s_{i,t}, h_{t}, L, x_{i,t})$  & Prospect function
   \tabularnewline \midrule
    $\Lambda_{i,t}(\mathtt{OM}_{i,t}, a_{i,t}|s_{i,t}, h_{t}, x_{i,t})$  & Agent $i$'s period-$t$ interim expected payoff-to-go.
    \tabularnewline \midrule
   $g_{i,t}$ & Carrier function. 
   \tabularnewline
		\midrule
		$\mathtt{Mg}_{i,t}$ & Maximum carrier function.
   \tabularnewline
		\midrule 
		$\Gamma_{i,t}(h_{t}, x_{i,t}; \phi_{i})$, $\vec{\Gamma}^{b}(c_{i,t})$ &  Off-region (OFR) induced by $\mathcal{G}^{\Theta}$, principal-desired $b$-th sub-OFR (OFR${}^{d}$). 
   \tabularnewline
		\midrule     $S^{\mathtt{off}}_{i,t}=\cup_{b\in[B]}\vec{\Gamma}^{b}(c_{i,t)}$,  $\vec{\Gamma}(c_{i,t})=\{\vec{\Gamma}^{b}(c_{i,t})\}_{b\in[B]}$  & OFR${}^{d}$, collections of sub-OFR${}^{d}$.
		 \tabularnewline
		\bottomrule  
  \hline
	\end{tabularx}
 \label{table:notations}
\end{table}

Table \ref{table:notations} lists the main notations.

\end{document}